\pdfoutput=1
\documentclass[11pt,a4paper]{article}
\usepackage{t1enc}
\usepackage[latin1]{inputenc}
\usepackage[english]{babel}
\normalfont
\usepackage{amsmath}
\usepackage{amssymb}
\usepackage{amsthm}
\usepackage{amsfonts}

\usepackage{mathabx}
\usepackage{bbm}
\usepackage{bm}
\usepackage{mathrsfs}
\usepackage{pifont}
\usepackage{hyperref}
\usepackage{pgf}
\usepackage{graphicx}
\usepackage[inline]{enumitem}
\usepackage{tikz-cd}
\setlist{nosep, itemsep=.1cm, topsep=.1cm}
\usepackage{stackengine}
\usepackage{relsize}

\newcommand{\be}[0]{\begin{equation}}
\newcommand{\ee}[0]{\end{equation}}

\newcommand*{\textcal}[1]{%
  \textit{\large \fontfamily{pzc}\selectfont#1}%
}

\numberwithin{equation}{section}

\theoremstyle{plain}
\newtheorem{theorem}{Theorem}[section]
\newtheorem{lemma}[theorem]{Lemma}
\newtheorem{proposition}[theorem]{Proposition}
\newtheorem{corollary}[theorem]{Corollary}

\newtheorem{assertion}[theorem]{Assertion}

\theoremstyle{definition}
\newtheorem{definition}[theorem]{Definition}

\newtheorem{remark}[theorem]{Remark}

\begin{document}
\sloppy

\vspace*{-1cm}
\thispagestyle{empty}
\vspace*{1.5cm}

\begin{center}
{\Large 
{\bf Topological field theories on manifolds with Wu structures}}
\vspace{2.0cm}

{\large Samuel Monnier}
\vspace*{0.5cm}

Section de Mathématiques, Université de Genève\\
2-4 rue du Lièvre, 1211 Genève 4, Switzerland\\
samuel.monnier@gmail.com

\vspace*{1cm}

{\bf Abstract}
\end{center}

We construct invertible field theories generalizing abelian prequantum spin Chern-Simons theory to manifolds of dimension $4\ell+3$ endowed with a Wu structure of degree $2\ell+2$. After analysing the anomalies of a certain discrete symmetry, we gauge it, producing topological field theories whose path integral reduces to a finite sum, akin to Dijkgraaf-Witten theories. We take a general point of view where the Chern-Simons gauge group and its couplings are encoded in a local system of integral lattices. 

The Lagrangian of these theories has to be interpreted as a class in a generalized cohomology theory in order to obtain a gauge invariant action. We develop a computationally friendly cochain model for this generalized cohomology and use it in a detailed study of the properties of the Wu Chern-Simons action. In the three-dimensional spin case, the latter provides a definition of the "fermionic correction" introduced recently in the literature on fermionic symmetry protected topological phases.

In order to construct the state space of the gauged theories, we develop an analogue of geometric quantization for finite abelian groups endowed with a skew-symmetric pairing.

The physical motivation for this work comes from the fact that in the $\ell =  1$ case, the gauged 7-dimensional topological field theories constructed here are essentially the anomaly field theories of the 6-dimensional conformal field theories with (2,0) supersymmetry, as will be discussed elsewhere.

\newpage

{\small \linespread{.6}
\tableofcontents
}

\section{Introduction and summary}

Gauge invariance requires the coupling of three-dimensional $U(1)$ Chern-Simons theories, traditionally called the level, to be an integer in suitable units. On spin manifolds, thanks to the evenness of the cup product pairing in dimension 4, one can define Chern-Simons theories with half-integer level, the so-called spin Chern-Simons theories. More generally, an abelian (classical) Chern-Simons theory, with its gauge group and couplings, can be encoded in an even lattice, while the same data for an abelian spin Chern-Simons theory is encoded in an arbitrary integral lattice \cite{Belov:2005ze}. The abelian gauge group is the torus $(\Lambda \otimes_\mathbb{Z} \mathbb{R}) / \Lambda$, and the couplings are determined by the lattice pairing.

Abelian Chern-Simons theories have higher dimensional analogues involving higher degree abelian gauge fields. The quadratic theories involving degree $2\ell+1$ abelian gauge fields live in $4\ell+3$ dimensions. In Section \ref{SecPreq}, we construct at the classical/prequantum level the higher dimensional analogues of spin Chern-Simons theories. The spin structure gets generalized to a degree $2\ell+2$ Wu structure, which can be understood as follows. A spin structure can be pictured as a trivialization of the degree 2 Wu class $\nu_2 := w_2$, which coincides with the second Stiefel-Whitney class. A degree $2\ell+2$ Wu structure is a trivialization of the degree $2\ell+2$ Wu class $\nu_{2\ell+2}$, a certain polynomial in the Stiefel-Whitney classes. Like in three dimensions, the gauge group and couplings of the theory can be specified by an arbitrary integral lattice $\Lambda$. In fact, we generalize the discussion to theories associated to a local system $\bar{\Lambda}$ of lattices over spacetime. We will come back later to the motivation for this generalization.

The rest of the paper is devoted to field theories obtained from the prequantum theories by the gauging of the action of the discrete group $\Gamma := \Lambda^\ast/\Lambda$. These theories have finite dimensional state spaces and their path integrals reduce to finite sums, as happens for Dijkgraaf-Witten theories \cite{Dijkgraaf:1989pz, Freed:1991bn, Monnier:2015qxa}. They are however clearly distinct from the latter. For instance the dimensions of the state spaces of the $\Gamma$-valued higher Dijkgraaf-Witten theory and of the theories constructed here on a closed $4\ell+2$-dimensional manifold are respectively $|H^{2\ell+1}(M, \Gamma)|$ and $|H^{2\ell+1}(M, \Gamma)|^{1/2}$. The full quantization of Wu Chern-Simons theories, which would require the gauging of the full group $(\Lambda \otimes_\mathbb{Z} \mathbb{R}) / \Lambda$, is not treated in the present paper. Quantum spin Chern-Simons theories in three dimensions have been studied in \cite{Belov:2005ze} for abelian gauge groups and in \cite{2006math......5239J} for non-abelian gauge groups. 

We construct the Euclidean field theories described above in the Atiyah-Segal framework, by associating a Hilbert space to each closed $4\ell+2$-dimensional manifolds (endowed with suitable structure) and a vector to each $4\ell+3$-dimensional manifold with boundary. We prove that these associations define field theory functors in Theorems \ref{ThPreqFTFunc} and \ref{ThMainTheorem}.

\paragraph{Relation to 6d SCFTs} The main motivation for the present work is the study of anomalies in six-dimensional (2,0) superconformal field theories (6d SCFTs). A $d$-dimensional anomalous field theory $\mathcal{F}$ can be pictured elegantly as a field theory taking value in a $d+1$-dimensional field theory $\mathcal{A}$, the anomaly field theory, which contains all the information about the anomalies of $\mathcal{F}$  \cite{Freed:2014iua, Monnierd}. The results of \cite{Monnier:2014txa} suggest that the anomaly field theories of the 6d SCFTs are essentially the 7-dimensional discretely gauged Wu Chern-Simons theories constructed in the present paper. This claim will be developed in a future paper. 

The existence of these gauged theories and their link to the 6d SCFTs was previously suggested in Section 4 of \cite{Witten:1998wy}.

This relation also provides the motivation for considering Chern-Simons theories associated to local systems of lattices. Indeed, the lattice $\Lambda$ is essentially the charge lattice of the 6d SCFT, and situations where this charge lattice is a non-trivial local system over spacetime are relevant, leading for instance to gauge theories with non simply laced gauge groups upon reduction to four or five dimensions \cite{Tachikawa:2011ch}. 

\paragraph{Action from generalized cohomology} Maybe one of the most interesting aspect of the theories constructed in the present paper is that their Lagrangians naively fail to be gauge invariant. More precisely, the Lagrangian is a top $\mathbb{R}/\mathbb{Z}$-valued cocycle on spacetime. One may construct a naive action $S^{\rm naive}$ by integrating this cocycle over the spacetime manifold. However, under a (large) gauge transformation of the Chern-Simons gauge field, $\exp 2\pi i S^{\rm naive}$ may change by a sign. 

The solution to this problem is the following. The Lagrangian determines a degree $4\ell+3$ class in a generalized cohomology theory, known as E-theory  \cite{2005math......4524J, Freed:2006mx}. We construct a "fundamental E-homology class" associated to a $4\ell+3$-dimensional Wu manifold, and show that the action obtained by pairing the Lagrangian with this class is gauge invariant. We also develop in Appendix \ref{SecEth} a non-abelian cochain model for E-theory that allows for concrete computations. We use it to study the action in detail in Sections \ref{SecAct} and \ref{SecWilHooftOp}. 

Chern-Simons actions are often constructed as the integral of a top differential form on a manifold $W$ bounded by the spacetime $M$. This approach is fine in three dimensions, where such bounded manifolds always exists, but is deficient in higher dimensions, where some spacetimes are not boundaries. We show in Section \ref{SecActBound} that in cases where the spacetime is a boundary, our action can be reexpressed as the integral of a top differential form on $W$. Another formalism expressing the spin Chern-Simons action without using a bounding four-manifold involves eta invariants \cite{2005math......4524J}, but does not generalize straightforwardly above three dimensions.

The same type of actions appeared in the literature on fermionic symmetry protected topological orders \cite{2012arXiv1201.2648G, 2015arXiv150505856G}. The lack of gauge invariance of the naive action was noted, and solved by invoking a "fermionic correction" to the action, which cannot be expressed as the integral of a top (ordinary) cohomology class. Our framework in terms of generalized cohomology can be seen as providing a precise definition for this fermionic correction. The relevance of generalized cohomology to this problem was previously suggested in \cite{2014arXiv1406.7278F}.

\paragraph{The action as a quadratic refinement} Another interesting characteristic of the action on a closed $4\ell+3$-dimensional manifold $M$ is that when restricted to flat gauge fields, it coincides with the pullback of a quadratic refinement of the linking pairing on the torsion cohomology group $H^{2\ell+2}_{\rm tors}(M; \bar{\Lambda})$ with value in the local system $\bar{\Lambda}$. This quadratic refinement generalizes classical constructions for three-dimensional spin manifolds \cite{Taylor1984259} and on $4\ell+3$-dimensional boundaries \cite{Brumfiel1973, Miller_Lee_1987}. Modulo a technical condition that in our context coincides with the absence of discrete gauge anomalies, quadratic refinements have associated Arf invariants in $\mathbb{Z}_8$. These Arf invariants are topological invariants of $4\ell+3$-dimensional manifolds endowed with degree $2\ell+2$-Wu structures, generalizing the Rokhlin invariant of spin 3-manifolds. The Arf invariant of the action should appear as a phase in the partition function of the quantum Wu Chern-Simons theory.
 
\paragraph{Geometric quantization of finite groups} The construction of the state space of the gauged theory on a $4\ell+2$-dimensional manifold $M$ involves a discrete analogue of geometric quantization. The problem is the following. Suppose one is given a lattice $\Theta$ endowed with a non-degenerate alternating pairing. Given a finite abelian group $\Xi$ endowed with a non-degenerate symmetric pairing, we can construct the group $\mathsf{G} := \Theta \otimes_\mathbb{Z} \Xi$, which carries an induced non-degenerate alternating pairing. One can see $\mathsf{G}$ as a discrete analogue  of a symplectic vector space and one may be interested in quantizing this structure. The rules of canonical quantization apply straightforwardly. The commutative algebra of functions on $\mathsf{G}$ gets quantized to the Heisenberg central extension ${\rm Heis}_{\mathsf{G}}$ of $\mathsf{G}$ determined by the alternating pairing. We can take the Hilbert space to be the unique (up to isomorphism) irreducible representation of ${\rm Heis}_{\mathsf{G}}$, of dimension $|\mathsf{G}|^{1/2}$. As in ordinary geometric quantization, most of the subtleties appear when one asks for a canonical construction of the representation of ${\rm Heis}_{\mathsf{G}}$. 

In our context, $\mathsf{G} = H^{2\ell+1}_{\rm free}(M;\mathbb{Z}) \otimes_{\mathbb{Z}} \Gamma_0$, where $\Gamma_0 := \Lambda^\ast_0/\Lambda_0$ and $\Lambda_0$ is the sublattice left invariant by the monodromy representation of the local system $\bar{\Lambda}$.
The irreducible representation of ${\rm Heis}_{\mathsf{G}}$, or more precisely a direct sum of such representations over certain torsion fluxes, is the value of the field theory functor on a $4\ell+2$-dimensional manifold $M$. We therefore do need a canonical construction for the functor to be well-defined.

It turns out that the prequantum theory of Section \ref{SecPreq} and the Wilson operators of Section \ref{SecWilHooftOp} provide a copy of the left and right regular representations of ${\rm Heis}_{\mathsf{G}}$ on a vector space $V$ of dimension $|\mathsf{G}|$. To construct the Hilbert space, we pick a Lagrangian (i.e. maximal isotropic) subgroup $\mathsf{L}$ of $\mathsf{G}$, a direct analogue of a choice of polarization in geometric quantization. We lift $\mathsf{L}$ to a commutative subgroup of $\tilde{\mathsf{L}}$ of ${\rm Heis}_{\mathsf{G}}$ and restrict to vectors in $V$ invariant under the right action of $\tilde{\mathsf{L}}$. The left action of ${\rm Heis}_{\mathsf{G}}$ on this invariant subspace provides a realization of the irreducible module. 

To ensure that the representation constructed in this way is independent of the choice of lifted Lagrangians, we need to construct a canonical intertwiner for each pair of lifted Lagrangian, with the property that they close under composition. The details of this construction depends on the arithmetic properties of $|\mathsf{G}|$ and are spelled out in Section \ref{SecStateSpace}. In certain cases, extra structures restricting the set of allowed lifted Lagrangians are required for this construction to be possible, signalling a Hamiltonian anomaly of the theory.

This problem was studied in the case of symplectic vector spaces over finite fields by Gurevich and Hadani \cite{Gurevich2009, 2008arXiv0808.1664G}.

\paragraph{Quadratic refinements} Amusingly, quadratic refinements (see the definition in Appendix \ref{SecPairFinAbGroups}) appear in four distinct contexts in the present work. We already mentioned that the action is a quadratic refinement of the linking pairing of the $4\ell+3$-dimensional spacetime. 

A quadratic refinement of the $\mathbb{Z}_2$ valued pairing on $H^{2\ell+1}_{\rm free}(M;\bar{\Lambda})/2H^{2\ell+1}_{\rm free}(M;\bar{\Lambda})$ appears in the construction of the Wilson operators in Section \ref{SecWilHooftOp}. It Arf invariant, valued in $\{0,1/2\}$, is a cobordism invariant of $4\ell+2$-dimensional manifolds endowed with a degree $2\ell+2$ Wu structure, generalizing Atiyah's Arf invariant of spin structures \cite{Atiyah1971}.

In Section \ref{SecStateSpace}, we use quadratic refinements to characterize the lifts of Lagrangian subgroups of $\mathsf{G}$ to commutative subgroups of ${\rm Heis}_{\mathsf{G}}$. In the same section, certain quadratic refinements are part of the extra structure needed on $4\ell+2$-dimensional manifolds to resolve the Hamiltonian anomaly of the gauged Wu Chern-Simons theories.

\paragraph{Quantum Wu Chern-Simons theories} The present work lays the foundations for the study of quantum Wu Chern-Simons theories, which were classified in the three dimensional case in \cite{Belov:2005ze}. One of the main results of \cite{Belov:2005ze} (see also \cite{Witten:2003ya}) was the quantum equivalence of certain distinct classical spin Chern-Simons theories. It would be very interesting so see how these equivalences generalize to higher dimensions. However we do not touch this problem in the present paper.

We can nevertheless note the following difference with the analysis of \cite{Belov:2005ze}. There, essentially due to the fact that any 3-manifold bounds, no spin structure was required to define the Chern-Simons theories associated to even lattices. As we explain in Section \ref{SecEvLat}, this is not necessarily true in higher dimensions, and only theories associated to lattices with $2\mathbb{Z}$-valued pairings are guaranteed to be well-defined without a Wu structure. Manifolds carrying a local system of even lattices have a preferred Wu structure, but the prequantum theory may be anomalous with this Wu structure, what would prevent quantization.

\paragraph{Organization of the paper} In Section \ref{SecBackStruct}, we describe the structures required on spacetime for the definition of the field theories of interest here. In Section \ref{SecFields}, we describe the classical Chern-Simons gauge fields, as well as the discrete gauge fields appearing in the gauged theories. We model gauge fields using differential cocycles with value in local systems of lattices. We also describe the measure of the discrete path integral over the discrete gauge fields, which is analogous to the measure entering (higher) abelian Dijkgraaf-Witten theory \cite{Dijkgraaf:1989pz, Monnier:2015qxa}.

In Section \ref{SecAct}, we define the Lagrangian of the theory, and remark that the corresponding naive action fails to be gauge invariant. We show how this problem is solved by seeing the Lagrangian as defining a class in a generalized cohomology. We show that on boundaries, the action admits a formula in terms of differential forms on the bounded manifold. We study the dependence of the action on the Wu structure and shows that it is the pullback of a quadratic refinement of the linking pairing on the degree $2\ell+2$ torsion cohomology of spacetime. We also study an anomaly with respect to the discrete symmetry that will later be gauged.

In Section \ref{SecPreq}, we construct the prequantum theory. After proving a general theorem characterizing field theory functors, we construct the prequantum theory using techniques used previously in the context of Dijkgraaf-Witten theory \cite{Freed:1991bn}, suitably generalized for Lagrangians valued in a generalized cohomology theory. On a $4\ell+2$-dimensional manifold, we show that the state space is trivial unless a condition on the torsion background flux is satisfied. This condition generalizes similar conditions guaranteeing the absence of a global gauge anomaly in the self-dual field theory \cite{Witten:1999vg, Monniera}. We prove that the prequantum theory defines a field theory functor in the sense of Atiyah-Segal in Theorem \ref{ThPreqFTFunc}.

From Section \ref{SecPartFunc} on, we turn to the construction of the gauged Wu Chern-Simons theories. We start by defining the partition function, and describe a gauge anomaly that forces it to vanish on certain manifolds.

In Section \ref{SecWilHooftOp}, we study Wilson operators. We construct them as non-flat field configurations on cylinders, which allows us to derive their properties from the action principle of the theory. On a $4\ell+2$-dimensional manifold, the group $\mathsf{G}_M := H^{2\ell+1}_{\rm free}(M; \bar{\Lambda}^\ast)/H^{2\ell+1}_{\rm free}(M; \bar{\Lambda})$ is endowed with a perfect alternating pairing. We show that the Wilson operators form a copy of the associated Heisenberg group.

In Section \ref{SecStateSpace}, we turn to the construction of the state space of the gauged theory. Informally, the state space is a certain direct sum over torsion background fluxes of irreducible modules for the Heisenberg group above. Its construction is however rather subtle and involves a discrete analogue of geometric quantization, as already discussed above. Like the partition function, the state space exhibits an anomaly, and in general some extra structure is required on $4\ell+2$-dimensional manifolds in order to construct it.

In Section \ref{SecVectManBound}, we explain how the path integral over the discrete gauge fields on a $4\ell+3$-dimensional manifold with boundary produces naturally a vector in the state space associated to the boundary. We encounter here an analogue of the gauge anomaly that was affecting the partition function.

In Section \ref{SecGluing}, we prove that the homomorphisms associated to manifolds with boundaries glue consistently. This involves in particular a non-trivial relation between the partition function anomaly and the Hamiltonian anomaly of the state space.

In Section \ref{SecHamAn}, we investigate in more detail the Hamiltonian anomaly affecting the state space. We explain informally how the state space can be constructed without invoking extra structures on $4\ell+2$-dimensional manifolds. The price to pay is that the state space is no longer a Hilbert space, but rather a "Hilbert space up to the tensor product with a Hermitian line", more precisely an object in a category linearly equivalent to the category of Hilbert space, but not canonically so.

In Section \ref{SecEx}, we point out some special cases in which our construction simplify. We show that the anomalies are absent either if the lattice $\Lambda$ has a $2\mathbb{Z}$-valued pairing, or if $\Gamma_0$ has odd order. In particular, the anomaly is always absent in $U(1)$ Wu Chern-Simons theories, which correspond to the lattices $\Lambda = \sqrt{2k}\mathbb{Z}$, for $k$ the half-integral level. For $k$ an integer, corresponding to ordinary $U(1)$ Chern-Simons theories, there is a always a unique Wu structure. This is consistent with the fact that these theories do not require extra structures on spacetime. 

Several appendices introduce concepts used in the main text and prove technical results. Appendix \ref{SecPairFinAbGroups} reviews pairings on abelian groups, quadratic refinements and their classification. Appendix \ref{SecPerPair} reviews the various cup product pairings used in the main text, as well as the linking pairing. We show that a $4\ell+3$-dimensional manifold determines a Lagrangian subgroup of the cohomology of its boundary for various cohomology groups. 

Appendix \ref{SecWuStruct} contains the definition of Wu classes and Wu structure for cohomology with value in a local system of lattices. We also review the higher cup products of Steenrod. In Appendix \ref{SecEth}, we define E-theory, the generalized cohomology relevant for the construction of the action. We construct a non-abelian cochain model for E-theory and construct a fundamental E-theory class that allows us to integrate top E-cohomology classes on Wu manifolds. In Appendix \ref{ApBordCatWMan}, we define the bordism categories that are the domains of the field theory functors constructed here. In Appendix \ref{AppHeisGr2Gr} and \ref{App2Groups}, we review some basic facts about Heisenberg groups and 2-groups.

\section{Background structures}

\label{SecBackStruct}

Let $\Lambda$ be an $n$-dimensional integral Euclidean lattice and let $\Lambda^\ast$ be the dual lattice. The $\mathbb{Z}$-valued pairing on $\Lambda$ can be extended to a $\mathbb{Q}$-valued pairing on $\Lambda^\ast$. $\Gamma = \Lambda^\ast/\Lambda$ is a finite abelian group, endowed with a pairing $(\bullet, \bullet)_\Gamma: \Gamma \times \Gamma \rightarrow \mathbb{Q}/\mathbb{Z}$ induced from the pairing on $\Lambda^\ast$ modulo $1$. $(\bullet, \bullet)_\Gamma$ is perfect, which means that it induces an isomorphism between $\Gamma$ and its group of characters, see Appendix \ref{SecPairFinAbGroups}. Write $\Gamma^{(2)}_{\rm deg} := \Lambda/2\Lambda$ endowed with the $\frac{1}{2}\mathbb{Z}/\mathbb{Z}$-valued pairing induced by half the pairing on $\Lambda$ modulo 1. Let $\Gamma^{(2)}$ be the quotient of $\Gamma^{(2)}_{\rm deg}$ by the radical of the pairing. $\Gamma^{(2)}$ is endowed with a perfect pairing as well.

The theories we will construct are $4\ell+3$-dimensional field theories defined on smooth oriented compact manifolds, possibly with boundary, endowed with the following data:
\begin{enumerate}
\item An orthogonal local system $\bar{\Lambda}$ of lattices with fibers isometric to $\Lambda$. "Orthogonal" means that its structure group is $O(\Lambda)$, the finite group of automorphisms of the lattice preserving its pairing (see Appendix \ref{SecPerPair}). 
\item A degree $2\ell+2$ Wu structure relative to $\bar{\Gamma}^{(2)}$, as defined in Appendix \ref{SecWuStruct}.
\end{enumerate}
Because of an anomaly, a third piece of data is actually required on $4\ell+2$-dimensional manifolds if a certain subgroup $\Gamma_0$ of $\Gamma$ has even order, as will be discussed Section \ref{SecExStruct}. Moreover, some constraints have to be imposed on the Wu structures of $4\ell+2$-dimensional manifolds, as explained in Sections \ref{SecActTori} and \ref{SecCobInv}.

In the following, we will call the data above on a manifold $M$ a \emph{$\mathfrak{W}$-structure}, including the orientation, the smooth structure and the compactness assumption. Such manifolds will be referred to as $\mathfrak{W}$-manifolds and unless otherwise noted, all the manifolds will be assumed to be $\mathfrak{W}$-manifolds. $\bar{\Lambda}$ induces local systems of groups $\bar{\Gamma}$ and $\bar{\Gamma}^{(2)}$ with respective fibers $\Gamma$ and $\Gamma^{(2)}$. A $\mathfrak{W}$-manifold also carries a flat vector bundle $\mathscr{V} := \bar{\Lambda} \otimes_\mathbb{Z} \mathbb{R}^n$.

We now collect a few remarks about Wu structures, most of them derived from the discussion in Appendix \ref{SecWuStruct}.
\begin{itemize} 
\item The Wu structures on $M$ are classified by homotopy classes of maps from $M$ into a classifying space $BSO[\Gamma^{(2)}, 2\ell+2]$ defined in Appendix \ref{SecWuStruct}, subject to a constraint.
\item As explained in Appendix \ref{SecWuStruct}, we endow the classifying space of Wu structures with a pair of cochains $(\hat{\eta}, \hat{\nu})$ valued in $\bar{\Gamma}^{(2)}$ and satisfying $d\hat{\eta} = \hat{\nu}$. These cochains can be pulled back to any $\mathfrak{W}$-manifold $M$ through its classifying map. (We use the same notation for the pulled-back cochains.) $\hat{\nu}$ is a cocycle representing the (trivial) Wu class on $M$, and the cochain $\hat{\eta}$ encodes the choice of Wu structure.
\item Any smooth compact oriented manifold of dimension $4\ell+3$ or less can be given a degree $2\ell+2$ Wu structure relative to any local system $\bar{\Lambda}$.
\item The set of degree $2\ell+2$ Wu structures on a closed manifold $M$ is a torsor for $H^{2\ell+1}(M; \bar{\Gamma}^{(2)})$. Under a change of Wu structure associated to a class in $H^{2\ell+1}(M; \bar{\Gamma}^{(2)})$, the cochain $\hat{\eta}$ changes by a cocycle representing this cohomology class.
\item A direct corollary of the above is that if the pairing on $\Lambda$ is valued in $2\mathbb{Z}$, $\Gamma^{(2)}$ is the trivial group and there is a unique Wu structure. In this case, the latter plays no role and we may as well drop it from the definition of $\mathfrak{W}$-manifolds.
\item If the pairing on $\Lambda$ is even, then $\hat{\nu} = 0$, so $\hat{\eta}$ is closed. There is a preferred Wu structure corresponding to $\hat{\eta}$ representing the trivial class in $H^{2\ell+1}(M; \bar{\Gamma}^{(2)})$. In other words, the torsor of Wu structures is canonically trivialized.
\item If $\bar{\Lambda}$ is the trivial local system with fiber $\mathbb{Z}$ with the generator having norm $1$, a Wu structure in the sense above coincides with a Wu structures in the sense of \cite{Miller_Lee_1987, Brumfiel1973, Sati:2011rw}. If moreover $\ell = 0$, then a Wu structure is simply a spin structure. 
\end{itemize}

A \emph{morphism} $\phi$ from a $\mathfrak{W}$-manifold $M$ to a $\mathfrak{W}$-manifold $N$ is an orientation-preserving diffeomorphism from $M$ to $N$ preserving the $\mathfrak{W}$-structures. By this, we mean first that the local systems $\bar{\Lambda}_M$ and $\bar{\Lambda}_N$ over $M$ and $N$ are related by $\bar{\Lambda}_M = \phi^\ast(\bar{\Lambda}_N)$. Second, the maps $\omega_M$ and $\omega_N$ classifying the Wu structures on $M$ and $N$ are such that $\omega_N \circ \phi$ is homotopic to $\omega_M$. A pair $(M,N)$ of $\mathfrak{W}$-manifolds is a $\mathfrak{W}$-manifold $M$ with a submanifold $N$ endowed with a $\mathfrak{W}$-structure such that the inclusion $N \hookrightarrow M$ is a morphism in the sense above.

\section{Fields}

\label{SecFields}

\subsection{Cohomology groups associated to a local system of lattices}

\label{SecCohomLat}

Consider the setup of Section \ref{SecBackStruct}. We have the following square of short exact sequences of abelian groups
\be
\label{EqSqExSeqLatGrps}
\begin{tikzcd}
\Lambda \arrow[r] \arrow[d] & \Lambda^\ast \arrow[r] \arrow[d] & \Gamma \arrow[d] \\
\Lambda \arrow[r] \arrow[d] & \mathbb{R}^n \arrow[r] \arrow[d] & \mathbb{R}^n/\Lambda \arrow[d] \\
1 \arrow[r] & \mathbb{R}^n/\Lambda^\ast \arrow[r] & \mathbb{R}^n/\Lambda^\ast
\end{tikzcd}
\ee
We have a corresponding square of exact sequences where the groups $\Lambda$ and $\Lambda^\ast$ are replaced by the local systems $\bar{\Lambda}$ and $\bar{\Lambda}^\ast$ and $\mathbb{R}^n$ is replaced with the vector bundle $\mathscr{V}$. The latter square induces a corresponding grid of long exact sequences for the relative cohomology of a pair $N \subset M$ of $\mathfrak{W}$-manifolds (for space reasons, we suppress the labels $(M,N)$ in the cohomology groups):
\be
\label{DiagCohomGroups}
\begin{tikzcd}
& ... \arrow[d] & ... \arrow[d] & ... \arrow[d] & ... \arrow[d] & \\
... \arrow[r] 
& H^p(\bar{\Lambda}) \arrow[r, "\iota_\Lambda"] \arrow[d, "{\rm id}"] 
& H^p(\bar{\Lambda}^\ast) \arrow[r, "\rho_{\Lambda^\ast}"] \arrow[d, "\iota_{\Lambda^\ast}"] \arrow[rd, dotted, "\iota"]
& H^p(\bar{\Gamma}) \arrow[d, "\iota_{\Gamma}"] \arrow[r, "\beta_\Gamma"] 
& H^{p+1}(\bar{\Lambda}) \arrow[r] \arrow[d] 
& ... \\
... \arrow[r]
& H^p(\bar{\Lambda}) \arrow[r, "\iota_{\Lambda^\ast} \circ \iota_\Lambda"] \arrow[d]
& H^p(\mathscr{V}) \arrow[r] \arrow[d] 
& H^p(\mathscr{V}/\Lambda) \arrow[r, "\beta"] \arrow[d] 
& H^{p+1}(\bar{\Lambda}) \arrow[r] \arrow[d] 
& ...\\
... \arrow[r] 
& 1 \arrow[r] \arrow[d]
& H^p(\mathscr{V}/\Lambda^\ast) \arrow[r] \arrow[d]
& H^p(\mathscr{V}/\Lambda^\ast) \arrow[r] \arrow[d]
& 1 \arrow[r] \arrow[d]
& ... \\
& ... & ... & ... & ... &
\end{tikzcd}
\ee
We have named some of the homomorphisms for future reference.

We define the torsion subgroup $H^p_{\rm tors}(M,N; \bar{\Lambda}^\ast) = {\rm ker}(\iota_{\Lambda^\ast}) \subset H^p(M,N; \bar{\Lambda}^\ast)$, and let $H^p_{\rm free}(M,N; \bar{\Lambda}^\ast) := H^p(M,N; \bar{\Lambda}^\ast)/H^p_{\rm tors}(M,N; \bar{\Lambda}^\ast)$ be the free quotient. We use similar notations for the cohomology groups valued in $\bar{\Lambda}$.

\subsection{Differential cochains valued in local systems}

We need a slight generalization of the concept of differential cocycle introduced in \cite{springerlink:10.1007/BFb0075216} (see also \cite{hopkins-2005-70}), adapted to our setup involving local systems of lattices. Let $N \subset M$ be a pair of $\mathfrak{W}$-manifolds, with $M$ possibly having a boundary and $N$ disjoint from $\partial M$. If $N = \emptyset$, we follow the usual convention and omit the corresponding label. Let $C^p(M,N; \bar{\Lambda})$ and $C^p(M,N; \mathscr{V})$ be the groups of degree $p$ smooth singular cochains on $M$ valued in $\bar{\Lambda}$ and $\mathscr{V}$, respectively, and vanishing on $N$. Let $\Omega^p(M,N; \mathscr{V})$ be the group of degree $p$ smooth differential forms on $M$ valued in $\mathscr{V}$ and vanishing on $N$. Consider the abelian group
\be
\check{C}^p(M,N; \bar{\Lambda}) := C^p(M,N; \bar{\Lambda}) \times C^{p-1}(M,N; \mathscr{V}) \times \Omega^p(M,N; \mathscr{V}) \;,
\ee
the group of \emph{differential cochains} on $M$ relative to $N$. We write differential cochains with a caron: $\check{c} \in \check{C}^p(M,N; \bar{\Lambda})$, while ordinary cochains and differential forms are written with a hat: $\check{c} = (\hat{a}, \hat{h}, \hat{\omega})$. We call $\hat{a} =: [\check{c}]_a$ the \emph{characteristic}, $\hat{h} =: [\check{c}]_h$ the \emph{connection} and $\hat{\omega} =: [\check{c}]_\omega$ the \emph{curvature} of $\check{c}$. A differential cochain is \emph{flat} if its curvature vanishes. It is \emph{topologically trivial} if its characteristic vanishes.

We define a differential on $\check{C}^p(M,N; \bar{\Lambda})$ by
\be
\label{EqDefDiffCompDiffCoch}
d(\hat{a},\hat{h},\hat{\omega}) = (d\hat{a}, \hat{\omega} - d\hat{h} - \hat{a}, d\hat{\omega}) \;.
\ee
In the connection component of the right hand side, we used implicitly the homomorphisms $C^p(M,N; \bar{\Lambda}) \rightarrow C^p(M,N; \mathscr{V})$ and $\Omega^p(M,N) \rightarrow C^p(M,N; \mathscr{V})$, where the latter map is given by integration over the simplices. We define the group $\check{Z}^p(M,N; \bar{\Lambda}) := {\rm ker}(d)$ of degree $p$ differential cocycles as the group of closed differential cochains. 

We define a groupoid structure on $\check{Z}^p(M,N; \bar{\Lambda})$, by having an arrow from $\check{c}_1$ to $\check{c}_2$ if and only if $\check{c}_2 - \check{c}_1 = d\check{e}$, for $\check{e} = (\hat{b}, \hat{j}, 0)$ a flat differential cochain vanishing on $N \sqcup \partial M$. In components:
\be
\label{EqArrowGroupoidDiffCoc}
\hat{a}_2 = \hat{a}_1 + d\hat{b} \;, \quad \hat{h}_2 = \hat{h}_1 - d\hat{j} - \hat{b} \;, \quad \omega_2 = \omega_1 \;.
\ee
Equivalently, the groupoid $\check{Z}^p(M; \bar{\Lambda})$ is the action groupoid with respect to the action \eqref{EqArrowGroupoidDiffCoc}.

The group of components of the groupoid $\check{Z}^p(M,N; \bar{\Lambda})$ is the group of degree $p$ \emph{differential cohomology classes} $\check{H}^p(M,N; \bar{\Lambda})$ on $M$ relative to $N$. We write differential cohomology classes, and cohomology classes in general, as lower cases without accents. If $\hat{a}$ or $\check{c}$ are cocycles, $a$ or $c$ are the corresponding cohomology classes. 

A flat differential cocycle $\check{c} = (\hat{a}, \hat{h}, 0) \in \check{Z}^p(M,N; \bar{\Lambda})$ satisfies $\hat{a} = -d\hat{h}$, which means that the image of $a$ in $H^p(M,N; \mathscr{V}) = H^p(M,N; \bar{\Lambda}) \otimes_\mathbb{Z} \mathbb{R}$ vanishes. Therefore $a$ is a torsion class in $H^p(M,N; \bar{\Lambda})$. 

We can define a cup product on the groups of differential cochains as follows \cite{springerlink:10.1007/BFb0075216, hopkins-2005-70}:
\begin{align}
\label{EqDefCupProdDiffCohom}
\cup: \: & \check{C}^p(M,N; \bar{\Lambda}) \times \check{C}^q(M,N; \bar{\Lambda}) \rightarrow \check{C}^{p+q}(M,N; \mathbb{Z}) \\
& \check{c}_1 \cup \check{c}_2 = (\hat{a}_1 \cup \hat{a}_2, (-1)^{{\rm deg} a_1} \hat{a}_1 \cup \hat{h}_2 + \hat{h}_1 \cup \hat{\omega}_2 + H^\wedge_\cup(\hat{\omega}_1, \hat{\omega}_2) , \hat{\omega}_1 \wedge \hat{\omega}_2) \notag
\end{align}
$\check{C}^{p+q}(M,N; \mathbb{Z})$ is the group of differential cochains valued in the trivial $\mathbb{Z}$-local system. The cup products on the right-hand side are the ones associated to the pairings on $\Lambda$ and $\mathbb{R}^n$, see Appendix \ref{SecPerPair}. $\wedge$ is the wedge product of forms. The wedge product is homotopically equivalent to the cup product, and $H^\wedge_\cup$ is any choice of equivalence, i.e. a homomorphism from $\Omega^\bullet(M,N; \mathscr{V}) \times \Omega^\bullet(M,N; \mathscr{V})$ to $C^\bullet(M,N; \mathscr{V})$ satisfying
\be
dH^\wedge_\cup(\hat{\omega}_1, \hat{\omega}_2) + H^\wedge_\cup(d\hat{\omega}_1, \hat{\omega}_2) + (-1)^{{\rm deg} \omega_1} H^\wedge_\cup(\hat{\omega}_1, d\hat{\omega}_2) = \hat{\omega}_1 \wedge \hat{\omega}_2 - \hat{\omega}_1 \cup \hat{\omega}_2 \;.
\ee 
In the first term on the right-hand side, we first perform the wedge product and see the resulting form as a cochain. In the second one, we first see $\hat{\omega}_1$ and $\hat{\omega}_2$ as cochains and take their cup product. The cup product \eqref{EqDefCupProdDiffCohom} passes to a well-defined cup product in differential cohomology.

\subsection{Background fields}

\label{SecFieldGroup}

The domains of the field theory functors to be constructed will be $\mathfrak{W}$-manifolds endowed with a differential cohomology class, a "background gauge field" in physical parlance. We introduce here the relevant notation.

We define $\mathsf{Z}_{M,N} := \check{Z}^{2\ell+2}(M,N;\bar{\Lambda})$ and $\mathsf{Y}_{M,N} := \check{H}^{2\ell+2}(M,N,\bar{\Lambda})$. We write $\mathsf{Z}_{M,N,\check{w}} \subset \mathsf{Z}_{M,N}$ for the subgroupoid consisting of differential cocycles restricting to $\check{w} \in \mathsf{Z}_{\partial M}$ on $\partial M$. Similarly, we define the subgroup $\mathsf{Y}_{M,N,\check{w}} \subset \mathsf{Y}_{M,N}$ consisting of differential cohomology classes restricting to $\check{w}$ on $\partial M$. (Recall that the arrows in the groupoid of differential cocycles are all relative to $\partial M$.) $\mathsf{Z}_{M,N,\check{w}}$ and $\mathsf{Y}_{M,N,\check{w}}$ are torsors for $\mathsf{Z}_{M,N,\check{0}}$ and $\mathsf{Y}_{M,N,\check{0}}$, where $\check{0}$ denotes the vanishing differential cocycle on $\partial M$. Writing $P := N \sqcup \partial M$, we have $\mathsf{Z}_{M,N,\check{0}} = \mathsf{Z}_{M,P}$ and $\mathsf{Y}_{M,N,\check{0}} = \mathsf{Y}_{M,P}$. There are short exact sequences \cite{springerlink:10.1007/BFb0075216,hopkins-2005-70}:
\be
\label{EqShExSeqChar}
\begin{tikzcd}
0 \arrow[r] & \Omega^{2\ell+1}(M,P;\mathscr{V})/\Omega_{\Lambda}^{2\ell+1}(M,P;\mathscr{V}) \arrow[r] & \mathsf{Y}_{M,P} \arrow[r,"{[\bullet]_a}"] & H^{2\ell+2}(M,P;\bar{\Lambda}) \arrow[r] & 0 
\end{tikzcd}
\ee
\be
\label{EqShExSeqFieldStr}
\begin{tikzcd}
0 \arrow[r] & H^{2\ell+1}(M,P;\mathscr{V}/\bar{\Lambda}) \arrow[r] & \mathsf{Y}_{M,P} \arrow[r,"{[\bullet]_\omega}"] & \Omega_{\Lambda}^{2\ell+2}(M,P;\mathscr{V}) \arrow[r] & 0 
\end{tikzcd}
\ee
where $\Omega_{\Lambda}^{2\ell+2}(M,P;\mathscr{V})$ denotes the subgroup of differential forms valued in $\mathscr{V}$ whose periods over $\bar{\Lambda}^\ast$-valued cycles are integral.

We also define $\mathsf{Z}_{M,N}^{\rm flat} \subset \mathsf{Z}_{M,N}$, the subgroupoid of flat differential cocycles. $\mathsf{Y}_{M,N}^{\rm flat} \subset \mathsf{Y}_{M,N}$ is the subgroup of flat differential cohomology classes. $\mathsf{Z}_{M,N,\check{w}}^{\rm flat} := \mathsf{Z}_{M,N}^{\rm flat} \cap \mathsf{Z}_{M,N,\check{w}}$, $\mathsf{Y}_{M,N,\check{w}}^{\rm flat} := \mathsf{Y}_{M,N}^{\rm flat} \cap \mathsf{Y}_{M,N,\check{w}}$, for $\check{w} \in \mathsf{Z}_{\partial M}^{\rm flat}$. The short exact sequence \eqref{EqShExSeqFieldStr} shows that
\be
\mathsf{Y}^{\rm flat}_{M,P} \simeq H^{2\ell+1}(M,P;\mathscr{V}/\bar{\Lambda}) \;.
\ee
If $x \in \mathsf{Y}^{\rm flat}_{M,P}$, with a representative cocycle $\check{x} = (\hat{a}, \hat{h}, 0)$, then $\beta(x) = -a$, see \eqref{DiagCohomGroups}. The image of $\beta$ is $H^{2\ell+2}_{\rm tors}(M,P;\bar{\Lambda})$.

A $\mathfrak{W}$-manifold $M$ endowed with a flat differential cocycle in $\mathsf{Z}^{\rm flat}_M$ will be called a $(\mathfrak{W},\mathsf{Z}^{\rm flat})$-manifold for brevity.

\subsection{Discrete gauge fields}

\label{SecDiscrGF}

The field theories to be constructed involve a quotient operation with respect to the action of discrete groups on $\mathsf{Y}^{\rm flat}_{M,N}$. In physical parlance, they are obtained through the gauging of a discrete symmetry. We describe here these discrete actions.

Let $\mathsf{F}_{M,N}$ be the subgroupoid of $\mathsf{Z}^{\rm flat}_{M,N}$ defined as follows. Its objects are flat differential cocycles $(\hat{a}, \hat{h}, 0) \in \mathsf{Z}^{\rm flat}_{M,N}$ satisfying the following discreteness constraint on their connection:
\be
\label{EqConstrFields}
\hat{h} \in C^{2\ell+1}(M,N; \bar{\Lambda}^\ast) \subset C^{2\ell+1}(M,N; \mathscr{V}) \;.
\ee
Its arrows are the flat differential cochains $(\hat{b}, \hat{j}, 0)$ of degree $2\ell + 1$ vanishing on $\partial M$, satisfying the analogous constraint
\be
\label{EqConstrFields2}
\hat{j} \in C^{2\ell}(M,N; \bar{\Lambda}^\ast) \;.
\ee
We write $\mathsf{E}_{M,N}$ for the group of components of $\mathsf{F}_{M,N}$.

As above, if $\check{w} \in \mathsf{F}_{\partial M}$, we will write $\mathsf{F}_{M,N,\check{w}}$ for the subgroupoid of $\mathsf{F}_{M,N}$ consisting of differential cocycles restricting to $\check{w}$ on $\partial M$. We write $\mathsf{E}_{M,N,\check{w}}$ for the corresponding group of components. $\mathsf{E}_{M,N,\check{w}}$ is a torsor for the finite abelian group $\mathsf{E}_{M,N,\hat{0}}$. $\mathsf{E}_{M,N} = \bigcup_{\check{w} \in \mathsf{F}_{\partial M}} \mathsf{E}_{M,N,\check{w}}$ is an abelian group and is finite if and only if $M$ is closed. (We assume that an empty boundary carries a unique vanishing cocycle.)

In the rest of this section, we study the structure of $\mathsf{E}_{M,N,\check{0}} \simeq \mathsf{E}_{M,P}$, $P := N \sqcup \partial M$. This yields the structure of $\mathsf{E}_{M,N,\check{w}}$ and $\mathsf{E}_{M,N}$ by the remarks above. 
\begin{proposition}
\label{PropCharEM}
We have an isomorphism
\be
\label{EqEInTermsOfGammaCohom}
\mathsf{E}_{M,P} \simeq H^{2\ell+1}(M,P; \bar{\Gamma})
\ee
\end{proposition}
\begin{proof}
Suppose $x \in \mathsf{E}_{M,\partial M}$ and pick a cocycle representative $\check{x} = (\hat{a}, \hat{h}, 0) \in \mathsf{F}_M$. Then $\hat{y} := \hat{h}$ mod $\Lambda$ is a $\bar{\Gamma}$-valued cocycle, associated to a class $y \in H^{2\ell+1}(M,\partial M; \bar{\Gamma})$.

Conversely, let $y \in H^{2\ell+1}(M,\partial M; \bar{\Gamma})$. Pick a cocycle representative $\hat{y}$ in $Z^{2\ell+1}(M; \bar{\Gamma})$ vanishing on $\partial M$. Lift it to a cochain $\hat{h} \in Z^{2\ell+1}(M; \bar{\Lambda}^\ast)$ vanishing on $\partial M$. Then $\hat{a}:= -d\hat{h} \in Z^{2\ell+1}(M; \bar{\Lambda})$ and vanishes on $\partial M$. Hence $\check{x} := (\hat{a}, \hat{h}, 0) \in \mathsf{F}_{M,\hat{0}}$ and its equivalence class is independent of the choices of lifts. 

The two maps above are inverse of each other, which proves the proposition.
\end{proof}

Let us write $\mathsf{K}_{M,P}$ for the image of $\beta_\Gamma$ in $H^{2\ell+2}(M, P; \bar{\Lambda})$ and $\mathsf{G}_{M,P}$ for the image of $\iota$, see \eqref{DiagCohomGroups}. Let us define in addition $\mathsf{C}_{M,P} := H^{2\ell+1}(M,P; \bar{\Lambda}^\ast)/\iota_\Lambda(H^{2\ell+1}(M,P; \bar{\Lambda}))$ and $\mathsf{T}_{M,P} := H^{2\ell+1}_{\rm tors}(M,P; \bar{\Lambda}^\ast)/\iota_\Lambda(H^{2\ell+1}(M,P; \bar{\Lambda}))$.
\begin{corollary}
\label{CorStructEM}
\begin{enumerate}
\item The Bockstein map $\beta_\Gamma$ vanishes on $\rho_{\Lambda^\ast}(\mathsf{C}_{M,P})$ and $\beta$ vanishes on $\mathsf{G}_{M,P}$.
\item $\mathsf{G}_{M,P}$ is isomorphic to $H^{2\ell+1}_{\rm free}(M,P; \bar{\Lambda}^\ast)/\iota_\Lambda(H^{2\ell+1}_{\rm free}(M,P; \bar{\Lambda}))$, so there is a short exact sequence
\be
\label{EqShExSeqCM}
\begin{tikzcd}
0  \arrow[r] & \mathsf{T}_{M,P} \arrow[hookrightarrow,r] & \mathsf{C}_{M,P} \arrow[r, "\iota"] & \mathsf{G}_{M,P}  \arrow[r] & 0 
\end{tikzcd}
\ee
\item $\mathsf{E}_{M,P}$ fits in the short exact sequence
\be
\label{EqShExSeqEM}
\begin{tikzcd}
0  \arrow[r] & \mathsf{C}_{M,P} \arrow[r, "\rho_{\Lambda^\ast}"] & \mathsf{E}_{M,P} \arrow[r, "\beta_\Gamma"] & \mathsf{K}_{M,P}  \arrow[r] & 0 
\end{tikzcd}
\ee
\item If $x \in \mathsf{E}_{M,P}$ has a representative $\check{x} = (\hat{a}, \hat{h},0) \in \mathsf{F}_{M,\hat{0}}$, then $\beta_\Gamma(x) = -a$. If moreover, $\hat{a} = 0$, $x$ is the image of $h \in H^{2\ell+1}(M,P; \bar{\Lambda}^\ast)$ in $\mathsf{C}_{M,P}$.
\end{enumerate}
\end{corollary}
\begin{proof}
The first two statements follow easily by diagram chasing on \eqref{DiagCohomGroups} and the definitions of the torsion and free part of the cohomology. Using in addition Proposition \ref{PropCharEM} yields the third statement. The fact that $\beta_\Gamma(x) = -a$ is a direct consequence of the definition of the Bockstein map $\beta_\Gamma$ and the fact that $\hat{a} = -d\hat{h}$.
\end{proof}

\begin{proposition}
\label{PropActFinGaugFields}
There is an action of $\mathsf{E}_{M,N}$ on $\mathsf{Y}^{\rm flat}_{M,N}$, whose stabilizer is $\mathsf{T}_{M,N}$. 
\end{proposition}
\begin{proof}
There is an obvious action of $\mathsf{F}_{M,N}$ on $\mathsf{Z}^{\rm flat}_{M,N}$ given by the addition of differential cocycles. This action passes down to an action of $\mathsf{E}_{M,N}$ on $\mathsf{Y}^{\rm flat}_{M,N}$. Moreover, as a differential cocycle in $\mathsf{F}_{M,N}$ is exact in $\mathsf{Z}^{\rm flat}_{M,N}$ if and only if it projects to $\mathsf{T}_{M,N} \subset \mathsf{E}_{M,N}$, $\mathsf{T}_{M,N}$ is the stabilizer.
\end{proof}
The following proposition describes in more detail the structure of $\mathsf{G}_{M,P}$. Let $\Lambda_0$ be the sublattice of $\Lambda$ left invariant by the monodromy representation $\rho: \pi_1(M) \rightarrow O(\Lambda)$ defining the local system $\bar{\Lambda}$. Let $\Gamma_0 := \Lambda_0^\ast/\Lambda_0$. 
\begin{proposition}
\label{PropStructGM}
We have:
\be
\label{EqStructGM}
\mathsf{G}_{M,P} \simeq H^{2\ell+1}_{\rm free}(M,P;\mathbb{Z}) \otimes \Gamma_0 \;.
\ee
\end{proposition}
\begin{proof}
Recall that $O(\Lambda)$ is a finite group. This means that there is a finite covering $\tilde{M} \stackrel{\pi}{\rightarrow} M$ such that the pullback $\pi^\ast \mathscr{V}$ is trivial. Cohomology classes in $H^{2\ell+1}(M,P; \mathscr{V})$ can be represented by closed equivariant differential forms on $\tilde{M}$ valued in $\pi^\ast \mathscr{V}$ modulo exact equivariant differential forms. The equivariant structure is the one defined by $\rho$. 

Write $\mathscr{V}_0$ for the (trivial) subbundle of $\mathscr{V}$ left invariant by $\rho$. Clearly, $H^{2\ell+1}(M,P; \mathscr{V}_0) \subset H^{2\ell+1}(M,P; \mathscr{V})$. This inclusion would be strict if there were some equivariant forms on $\tilde{M}$ valued in $\pi^\ast \mathscr{V}$ that would be exact but not equivariantly exact. But as the covering $\tilde{M}$ is finite, any non-equivariant trivializing form can be made equivariant by averaging over the monodromy action. We have therefore an isomorphism $H^{2\ell+1}(M,P; \mathscr{V}_0) \simeq H^{2\ell+1}(M,P; \mathscr{V})$.

The above implies immediately that $H^{2\ell+1}_{\rm free}(M,P;\bar{\Lambda}) \simeq H^{2\ell+1}_{\rm free}(M,P;\Lambda_0) \simeq H^{2\ell+1}_{\rm free}(M,P;\mathbb{Z}) \otimes \Lambda_0$, and similarly $H^{2\ell+1}_{\rm free}(M,P;\bar{\Lambda}^\ast) \simeq  H^{2\ell+1}_{\rm free}(M,P;\mathbb{Z}) \otimes \Lambda_0^\ast$. This yields the isomorphism \eqref{EqStructGM}. 
\end{proof}

\subsection{Measure}

\label{SecMeas}

The gauging procedure amounts to summing the prequantum field theory functor over the action of Proposition \ref{PropActFinGaugFields}, with a normalization factor that we describe here. Alternatively, this procedure can be interpreted as a finite path integral over the discrete gauge fields of Section \ref{SecDiscrGF}. The normalization factor derived here can then be understood as a measure factor for the path integral \cite{Freed:1991bn, Monnier:2015qxa}.

We consider again a pair $N \subset M$ of $\mathfrak{W}$-manifolds with $N$ disjoint from $\partial M$. We write $P := N \sqcup \partial M$ and define $\mathsf{E}^p_{M,N}$ to be the group of equivalence classes of the groupoid $\check{Z}^{2\ell+2+p}(M,N;\bar{\Lambda})$, so $\mathsf{E}_{M,N} = \mathsf{E}^0_{M,N}$. The arguments of Proposition \ref{PropCharEM} show that $\mathsf{E}^p_{M,P} \simeq H^{2\ell+2+p}(M,P;\bar{\Gamma})$. The long exact sequence in $\bar{\Gamma}$-valued cohomology associated to the pair $(M,N)$ immediately yields:
\begin{lemma}
There is a long exact sequence:
\be
\begin{tikzcd}
\label{EqSuiteExLongRelHom}
... \arrow[r] 
& \mathsf{E}^{-1}_{M,N \sqcup \partial M} \arrow[r] 
& \mathsf{E}^{-1}_{M,\partial M} \arrow[r,"\rho_N"] 
& \mathsf{E}^{-1}_{N} \arrow[r, "\delta"] 
& \mathsf{E}_{M,N \sqcup \partial M} \arrow[r,"\iota_{N,M}"] 
& \mathsf{E}_{M,\partial M} \arrow[r] & ...
\end{tikzcd}
\ee
where $\rho_N$ is the restriction to $N$, $\iota_{N,M}$ is induced by the inclusion $\check{Z}^\bullet(M,N;\bar{\Lambda}) \subset \check{Z}^\bullet(M;\bar{\Lambda})$ and $\delta$ is the connection homomorphism.
\end{lemma}

We define the following measure factor:
\be
\label{EqDefMeasFact}
\mu_{M,N} = \prod_{p = 1}^{2\ell+1} |\mathsf{E}^{-p}_{M, N \sqcup \partial M}|^{(-1)^{p}} \; \in \mathbb{Q} \;,
\ee
where $|G|$ denotes the order of a finite group $G$. As usual, we write simply $\mu_M$ for $\mu_{M,\emptyset}$. We now prove a proposition that will be instrumental in the proof of the gluing formula for the gauged theories in Section \ref{SecGluing}.
\begin{proposition}
\label{PropRelMeasures}
Using the notation of \eqref{EqSuiteExLongRelHom}, we have
\be
\label{EqOrderKerIotaNM}
|{\rm ker}(\iota_{N,M})| = \frac{\mu_{M}}{ \mu_N \: \mu_{M,N}} \;.
\ee
\end{proposition}
\begin{proof}
The long exact sequence \eqref{EqSuiteExLongRelHom} implies that
\begin{align}
|{\rm ker}(\iota_{N,M})| & \: = |{\rm im}(\delta)| = \frac{|\mathsf{E}^{-1}_N|}{|{\rm ker}(\delta)|} \notag \\
& = \frac{|\mathsf{E}^{-1}_N| |{\rm ker}(\rho_N)|}{|\mathsf{E}^{-1}_{M,\partial M}|} \\
& = ... \notag
\end{align}
Continuing in this way until we reach the left end of the sequence at $p = 2\ell+1$ yields \eqref{EqOrderKerIotaNM}.
\end{proof}

\begin{remark}
The reader may wonder why we defined discrete gauge fields on $M$ to be valued in $\mathsf{E}_{M,\partial M}$, with their action admitting a non-trivial stabilizer $\mathsf{T}_{M,\partial M}$. It may look like we might as well have took them to be valued in $\mathsf{E}_{M,\partial M}/\mathsf{T}_{M,\partial M}$ with a free action. The reason for this definition is the need for the long exact sequence \eqref{EqSuiteExLongRelHom}, leading to the formula \eqref{EqOrderKerIotaNM}. There is no analogue of \eqref{EqOrderKerIotaNM} if the action of the discrete gauge fields is taken to be free.
\end{remark}

\section{Action}

\label{SecAct}

In this section we define the action of the theories to be constructed, which is a function on the group $\mathsf{Z}_M$ of differential cocycles of degree $2\ell+2$ on spacetime.

\subsection{Lagrangian and E-theory}

Let $M$ be an $\mathfrak{W}$-manifold of dimension $4\ell+3$, possibly with boundary. Recall that it comes with a $\bar{\Gamma}^{(2)}$-valued cochain $\hat{\eta}$ trivializing the Wu cocycle. Let $\hat{\eta}_{\Lambda}$ be any lift of this cochain to $C^{2\ell+1}(M; \bar{\Lambda})$. Define $\check{\nu} = (d\hat{\eta}_\Lambda, -\hat{\eta}_\Lambda, 0) \in \mathsf{Z}_M$. Let $\check{x} \in \mathsf{Z}_M$ be a differential cocycle and write $\check{x} = (\hat{a}, \hat{h}, \hat{\omega})$.

Consider the differential cocycle $\check{x} \cup (\check{x} + \check{\nu})$ and let $\hat{l}(\check{x})$ be the $\mathbb{R}/\mathbb{Z}$-valued degree $4\ell+3$ cochain given by half its connection modulo 1:
\begin{align}
\label{EqDefLag}
\hat{l}(\check{x}) \: & := \frac{1}{2} [\check{x} \cup (\check{x} + \check{\nu})]_h \\
& = \frac{1}{2} \hat{a} \cup (\hat{h} - \hat{\eta}_{\Lambda}) + \frac{1}{2} \hat{h} \cup \hat{\omega} + \frac{1}{2}H^\cup_\wedge(\hat{\omega}, \hat{\omega}) \quad {\rm mod} \; 1 \notag \;,
\end{align}
where on the second line, we made explicit the cup product of differential cochains. In the following, all the equalities involving $\hat{l}$ are to be understood modulo 1 and we will drop this indication from the notation. $\hat{l}(\check{x})$ is the \emph{Lagrangian} of the field theories to be constructed, in the sense that the action will be given by a suitable integral of $\hat{l}(\check{x})$ over spacetime.

$\hat{l}$ is independent modulo 1 of the choice of lift $\hat{\eta}_\Lambda$, as any other choice of lift would differ by twice a cocycle valued in $\bar{\Lambda}$. We now investigate the dependence of $\hat{l}(\check{x})$ on the differential cocycle $\check{x}$. From \eqref{EqDefLag}, we have
\begin{align}
\label{EqLagSum}
\hat{l}(\check{x}_1 + \check{x}_2) = & \: \hat{l}(\check{x}_1) + \hat{l}(\check{x}_2) \notag \\ 
&  + \frac{1}{2}( \hat{a}_1 \cup \hat{h}_2 + \hat{a}_2 \cup \hat{h}_1 + \hat{h}_1 \cup \hat{\omega}_2 + \hat{h}_2 \cup \hat{\omega}_1 \\
&
+ H^\cup_\wedge(\hat{\omega}_1, \hat{\omega}_2) + H^\cup_\wedge(\hat{\omega}_2, \hat{\omega}_1) ) \notag \;.
\end{align}
Suppose now that $\check{w} = (\hat{b}, \hat{k}, 0)$ is a flat differential cochain of degree $2\ell+1$ vanishing on $\partial M$. Then 
\begin{align}
\label{EqFailGaugInvAct}
\hat{l}(\check{x} + d\check{w}) - \hat{l}(\check{x})  \: & = \Delta_1 \hat{l} + \Delta_2 \hat{l} \;, \\
\Delta_1 \hat{l}  \: &  = \frac{1}{2} \left( \hat{a} \cup (-\hat{b} - d\hat{k}) + d\hat{b} \cup \hat{h} - (-\hat{b}-d\hat{k}) \cup \hat{\omega} \right)\notag \\
\label{EqFailGaugInvAct1}
 \: & = \frac{1}{2} \left(\hat{a} \cup \hat{b} + \hat{b} \cup \hat{a} \right) + {\rm exact} \;, \\
\Delta_2 \hat{l} \: & =  \frac{1}{2} \left( d\hat{b} \cup (-\hat{b} - d\hat{k} - \hat{\eta}_{\Lambda}) \right) \notag \\
\label{EqFailGaugInvAct2}
& = \frac{1}{2} \left( \hat{b} \cup (d\hat{b} + \hat{\nu}_{\Lambda}) \right) + {\rm exact} \;,
\end{align}
where we split the variation into two terms for notational convenience and wrote $d\hat{\eta}_\Lambda = \hat{\nu}_\Lambda$. "Exact" means exact terms vanishing on $\partial M$. We also used the fact that $\hat{l}$ is defined modulo 1 and that the first terms in \eqref{EqFailGaugInvAct1} and \eqref{EqFailGaugInvAct2} are valued in $\frac{1}{2}\mathbb{Z}$ to eliminate some irrelevant signs. The variation vanishes on $\partial M$.
\begin{remark} 
\label{RemLagNotGaugeInv}
\eqref{EqFailGaugInvAct1} and \eqref{EqFailGaugInvAct2} show that $\hat{l}$ does not vary by an exact term when $\check{x}$ is varied within its equivalence class. This means that the naive action
\be
\label{EqNaiveAction}
S^{\rm naive}(\check{x}) = \langle \hat{l}(\check{x}), [M,\partial M] \rangle \quad {\rm mod} \; 1 \;,
\ee 
defined as the pairing of $\hat{l}$ with the fundamental class of $[M,\partial M]$ and seen as a function on $\mathsf{Z}_M$, does \emph{not} factor through a function on the group of components $\mathsf{Y}_M$. Its value can change by a half-integer when $\check{x}$ is changed to an equivalent differential cocycle. In physical terms, the action \eqref{EqNaiveAction} fails to be gauge invariant, and the failure of gauge invariance is a sign in the exponentiated action. 
\end{remark}

The key to solve the problem mentioned in Remark \ref{RemLagNotGaugeInv} is to see $\hat{l}$ as defining a class in a generalized cohomology theory. In Appendix \ref{SecEth}, we describe a family of generalized cohomology theories called E-theories. Certain members have been described previously in \cite{2005math......4524J, Freed:2006mx}. We show in the appendix that the E-cohomology associated to the (parametrized) spectrum $E_\nu[\Gamma^{(2)}, 2\ell+1]$ admits the following cochain model on a $\mathfrak{W}$-manifold $M$. The space of E-cochains is $\bar{C}^p_\nu(M) := C^p(M;\mathbb{R}/\mathbb{Z}) \times C^{p-2\ell-1}(M;\Gamma^{(2)})$, endowed with the non-commutative group law $\boxplus$ defined in \eqref{EqGrpLawCocModETh}. We write E-cochains with a bar: $\bar{x} \in \bar{C}^p_\nu(M)$. One can define a differential $d$ \eqref{EqDefTwDiffEThCoch} squaring to zero. Elements of the kernel and image of $d$ will be called respectively E-cocycles and exact E-cocycle. We show in Propositions \ref{PropCochModETh} and \ref{PropCochModTwETh} that the group of E-cocycles quotiented by the normal subgroup of exact E-cocycles coincides with the E-cohomology group. On a manifold with boundary, the action of exact E-cocycles is restricted to those vanishing on $\partial M$.

Recall that $\Gamma^{(2)}$ is the quotient of $\Lambda/2\Lambda$ by the radical of the induced $\frac{1}{2}\mathbb{Z}/\mathbb{Z}$-valued pairing. $\Gamma^{(2)}$ can alternatively be described as the quotient of $\Lambda$ by the subgroup $\Lambda_{\rm even}$ consisting of all the lattice vectors having even scalar product with every vector in $\Lambda$. It carries a non-degenerate $\frac{1}{2}\mathbb{Z}/\mathbb{Z}$-valued pairing. Let $\hat{a}_2$ denotes the $\Gamma^{(2)}$ valued cocycle obtained from $\hat{a}$ using the projection from $\Lambda$ to $\Gamma^{(2)}$. We have
\begin{proposition} 
\label{PropGaugInvLag}
$\bar{l}(\check{x}) := (\hat{l}(\check{x}), \hat{a}_2)$ defines an $E_\nu[\Gamma^{(2)}, 2\ell+1]$-cohomology class of degree $4\ell+3$. This cohomology class depends only on the differential cohomology class of $\check{x}$, i.e. it is gauge invariant.
\end{proposition}
\begin{proof}
We compute using \eqref{EqDefTwDiffEThCoch}:
\begin{align}
d\bar{l}(\check{x}) \: & = \left(d\hat{l}(\check{x}) + (\hat{a})_2 \cup ((\hat{a})_2 + (\hat{\nu})_2), 0 \right) \notag \\
& = \left(\frac{1}{2}d[\check{x} \cup (\check{x} + \check{\nu})]_h + \frac{1}{2} [\check{x} \cup (\check{x} + \check{\nu})]_a, 0 \right) \\
& = \left(\frac{1}{2}[\check{x} \cup (\check{x} + \check{\nu})]_\omega, 0 \right) = 0 \notag
\end{align}
where the last equality holds for degree reasons. Therefore $\bar{l}$ is an E-cocycle and it defines an E-cohomology class in $E_\nu[\Gamma^{(2)}, 2\ell+1]^{4\ell+3}(M)$. The variation of $\hat{l}(\check{x})$ under changes of the cocycle representative $\check{x}$ of $x$ is described in \eqref{EqFailGaugInvAct}-\eqref{EqFailGaugInvAct2}. We have:
\be
\hat{l}(\check{x} + d\check{w}) = \hat{l}(\check{x}) + \hat{b}_2 \cup (d\hat{b}_2 + \hat{\nu}) + \hat{a}_2 \cup \hat{b}_2 + \hat{b}_2 \cup \hat{a}_2 + d\hat{c} 
\ee 
for some cochain $\hat{c}$. Using the relation \eqref{EqRelHighCupProd} relating the higher cup products, we can rewrite the last three terms as $\hat{a}_2 \cup_1 d\hat{b}_2 + d\hat{c}'$ for another cochain $\hat{c}'$. The expressions \eqref{EqDefTwDiffEThCoch} for the differential and \eqref{EqGrpLawCocModETh} for the sum of E-cocycles then allow us to deduce that
\be
\label{EqVarLagEx}
\bar{l}(\check{x}+ d\check{w}) = \bar{l}(\check{x}) \boxplus d \bar{c} \;, 
\ee
with $\bar{c} = (\hat{c}', \hat{b}_2)$.
\end{proof}

\subsection{Action}

The interest of seeing the Lagrangian as an E-theory class is the following. Given a $\mathfrak{W}$-manifold $M$ of dimension $m < 4\ell+4$, we construct a distinguished $\mathbb{R}/\mathbb{Z}$-valued character $[M,\partial M]_E$ on the generalized cohomology group $E_\nu[\Gamma^{(2)}, 2\ell+1]^m(M,\partial M)$, which should be seen as a substitute for the fundamental homology class in the context of E-theory, see Proposition \ref{PropConstrFundHomEClass}. This means that any relative degree $m$ $E_\nu[\Gamma^{(2)}, 2\ell+1]$-cohomology class on $M$ can be "integrated" over $M$ by pairing it with $[M,\partial M]_E$. Specializing to $m = 4\ell+3$, we define the action by 
\be
\label{EqDefActEth}
S(M,\check{x}) = \langle \bar{l}(\check{x}), [M,\partial M]_E \rangle \in \mathbb{R}/\mathbb{Z} \;,
\ee
where $\check{x}$ is a differential cocycle vanishing on $\partial M$. This is a pairing in cohomology, so $S(M,\check{x})$ depends only on the E-cohomology class of $\bar{l}$. As the latter depends only on the differential cohomology class of $\check{x}$ by Proposition \ref{PropGaugInvLag}, we have:
\begin{proposition}
The function $S(M, \bullet)$ on $\mathsf{Z}_{M,\partial M}$ factors through a function on $\mathsf{Y}_{M,\partial M}$, i.e. it is gauge invariant.
\end{proposition}
\noindent If $\check{x}$ does not vanish on $\partial M$, we can only pair $\bar{l}(\check{x})$ with a chain representative of $[M,\partial M]_E$ (see Appendix \ref{SecEth}). The value of the action (as a complex number) then depends both on the boundary value of $\check{x}$ and on the chain representative of $[M,\partial M]_E$. We will discuss this dependence in more detail in Section \ref{SecSubPreqTh}.

We define the \emph{exponentiated action} by
\be
\label{EqExpAct}
S^e(M,\check{x}) = \exp 2 \pi i S(M,\check{x}) \in U(1) \;.
\ee
We will not distinguish in notation between $S(M, \bullet)$, $S^e(M,\bullet)$ and the functions they factor through on equivalence classes of cocycles.

\begin{remark}
The type of action considered here, in the case $\ell = 0$ and for a trivial local system, has been recently studied in the condensed matter literature, in the context of fermionic symmetry protected topological order \cite{2012arXiv1201.2648G, 2015arXiv150505856G}. The point of view in these papers is to consider the naive action of Remark \ref{RemLagNotGaugeInv} and to invoke the existence of a "fermionic correction", an extra term in the action that restore gauge invariance. By definition, this term cannot be expressed by the integral of an ordinary cohomology class over the manifold. The formalism developed above using E-theory can be seen as providing a concrete definition of this fermionic correction. The relevance of generalized cohomology in this context was already noted in \cite{2014arXiv1406.7278F}.
\end{remark}

\subsection{Dependence on the Wu structure}

We show in Appendix \ref{SecWuStruct} that the set of Wu structures on $M$ is a torsor for $H^{2\ell+1}(M,\partial M;\bar{\Gamma}^{(2)})$. 
\begin{proposition}[Change of Wu structure]
\label{PropChangWu}
Suppose the $\mathfrak{W}$-manifold $M'$ can be obtained from the $\mathfrak{W}$-manifold $M$ by a change of Wu structure associated to $\delta \in H^{2\ell+1}(M,\partial M;\bar{\Gamma}^{(2)})$. Then 
\be
\label{EqChActChWuStruct}
S(M',\check{x}) = S(M,\check{x}) - \langle (a)_2 \cup \delta, [M,\partial M] \rangle \;.
\ee
where the second term involves the cup product pairing with coefficients in $\bar{\Gamma}^{(2)}$.
\end{proposition}
\begin{proof}
Let $\hat{\delta}$ be any lift of $\delta$ to a $\bar{\Lambda}$-valued cochain. Suppose we made a choice choice of lift $\hat{\eta}_{\Lambda}$ in the Lagrangian \eqref{EqDefLag} on $M$. Then we may choose the lift $\hat{\eta}_{\Lambda}' := \hat{\eta}_\Lambda + \hat{\delta}$ to compute the Lagrangian on $M'$. The new Lagrangian is 
\be
\hat{l}'(\check{x}) = \hat{l}(\check{x}) - \frac{1}{2}\hat{a} \cup \hat{\delta} \;.
\ee
The corresponding E-cochain $\bar{l}'(\check{x}) := (\hat{l}'(\check{x}), \hat{a}_2)$ can be written
\be
\bar{l}'(\check{x}) = \bar{l}(\check{x}) \boxplus \left (-\frac{1}{2}\hat{a} \cup \hat{\delta}, 0 \right) \;.
\ee
Remembering that the pairing with the E-homology class coincides with the pairing with the ordinary homology classes on E-chains of the type $(\hat{s},0)$, we obtain  
\be
S(M',\check{x}) = S(M,\check{x}) - \frac{1}{2} \langle \hat{a} \cup \hat{\delta}, [M,\partial M] \rangle \quad {\rm mod} \; 1 \;.
\ee
But the second term modulo 1 coincides with $\langle (a)_2 \cup \delta, [M,\partial M] \rangle$.
\end{proof}

\subsection{Action on boundaries} 

\label{SecActBound}

We now show that on spacetimes bounded by a $4\ell+4$-manifold $W$, the action $S$ can be expressed in terms of differential forms on $W$.

Assume that $M$ is a closed $\mathfrak{W}$-manifold of dimension $4\ell+3$ that is the boundary of a $4\ell+4$-dimensional manifold $W$, and that the local system of lattices on $M$, $\bar{\Lambda}_M$, is the restriction of a local system of lattices $\bar{\Lambda}_W$ on $W$.
The Wu class $\nu$ is in general non-zero on $W$ so $W$ does not admit a Wu structure. We pick a lift $\hat{\nu}_\Lambda$ of the Wu cocycle $\hat{\nu}$ as a $\bar{\Lambda}$-valued cocycle on $W$ and assume that $d\hat{\eta}_\Lambda = \nu_\Lambda|_M$. (Note that such a lift always exists, because $0 = t \cup t = t \cup \nu$ whenever $t \in H^{2\ell+2}(W,M;\bar{\Gamma}^{(2)})$ is the reduction of an element of $H^{2\ell+2}_{\rm tors}(W,M;\bar{\Lambda})$, so $\nu$ can always be expressed as the reduction of an element of $H^{2\ell+2}(W,M;\bar{\Lambda})$.) We extend $\hat{\eta}_\Lambda$ to a cochain $\hat{\mu}$ on $W$. $\hat{\mu}$ and $\hat{\nu}_\Lambda$ can be chosen so that $\hat{\lambda} := \hat{\nu}_\Lambda  - d\hat{\mu}$ is a smooth cocycle. We therefore have a differential cocycle $\check{\nu} := (\hat{\nu}_\Lambda, \hat{\mu}, \hat{\lambda})$. Remark that $\hat{\lambda}$ vanishes on $M$, so defines a relative class $\lambda \in H^{2\ell+2}(W,M; \mathscr{V})$. Assume also that $\check{x}$ extends to a differential cocycle $\check{z} = (\hat{a}_W, \hat{h}_W, \hat{\omega}_W)$ on $W$.
\begin{proposition} 
\label{PropActFromBoundedMan}
The action can be computed on $W$ as follows:
\be
\label{EqACtFromBoundedMan}
S(M, \check{x}) = \frac{1}{2} \left \langle \hat{\omega}_W \wedge (\hat{\omega}_W + \hat{\lambda}), [W,M] \right \rangle \quad {\rm mod} \; 1 \;.
\ee 
When $\check{x} \in \mathsf{Y}^{\rm flat}_M$, $\hat{\omega}_W$ defines a relative class $\omega_W \in H^{2\ell+2}(W,M; \mathscr{V})$ and we have the cohomological expression  
\be
\label{EqActBManField}
S(M, \check{x}) = \frac{1}{2} \langle \omega_W \wedge (\omega_W + \lambda), [W,M] \rangle \quad {\rm mod} \; 1 \;.
\ee
\end{proposition}
\noindent
Note that to keep the notations uniform we use a cochain notation for the integration map:
\be
\frac{1}{2} \left \langle \hat{\omega}_W \wedge (\hat{\omega}_W + \hat{\lambda}), [W,M] \right \rangle := \frac{1}{2} \int_W  \hat{\omega}_W \wedge (\hat{\omega}_W + \hat{\lambda})
\ee
This particular degree $4\ell+4$ form has appeared before in the context of the self-dual field theory \cite{Witten:1996hc} (see also \cite{Brumfiel1973}).
\begin{proof}
Let $\bar{l}(\check{z}) = (\hat{l}(\check{z}), (\hat{a}_W)_2)$, with $\hat{l}(\check{z})$ given by \eqref{EqDefLag}. By the definition of E-chains and E-cochains in Appendix \ref{SecEth}, we have
\be
S(M, \check{x}) = \langle \bar{l}(\check{x}), [M]_E \rangle = \langle d\bar{l}(\check{z}), [W,M]_E \rangle \;.
\ee
From \eqref{EqDefTwDiffEThCoch}, we have 
\be
d\bar{l}(\check{z}) = \left(d\hat{l}(\check{z}) + (\hat{a}_W)_2 \cup ((\hat{a}_W)_2 + (\hat{\nu}_\Lambda)_2) ,0 \right) \;,
\ee
because $(\hat{a}_W)_2$ is closed. As shown in Appendix \ref{SecEth}, E-cochains of the form $\bar{s} = (\hat{s}, 0)$ satisfy $\langle \bar{s}, [W,M]_E \rangle = \langle \hat{s}, [W,M] \rangle$. We have therefore
\begin{align}
S(M, \check{x}) \: & = \left \langle d\hat{l}(\check{z}) + \frac{1}{2} \hat{a}_W \cup (\hat{a}_W + \hat{\nu}_\Lambda), [W,M] \right \rangle \\
& = \frac{1}{2} \left \langle d[\check{z} \cup (\check{z} + \check{\nu})]_h + [\check{z} \cup (\check{z} + \check{\nu})]_a, [W,M] \right \rangle \\
& = \frac{1}{2} \left \langle [\check{z} \cup (\check{z} + \check{\nu})]_\omega, [W,M] \right \rangle \\
& = \frac{1}{2} \left \langle \hat{\omega}_W \wedge (\hat{\omega}_W + \hat{\lambda}), [W,M] \right \rangle
\end{align}
modulo 1. The second part of the proposition is immediate.
\end{proof}

\subsection{Action as a quadratic refinement}

\label{SecActQR}

In this section, we restricts ourselves to flat differential cocycles: $\check{x}_1, \check{x}_2 \in \mathsf{Z}^{\rm flat}_{M,\partial M}$.
\begin{proposition}
\label{PropActQuadrRef}
$S(M, \bullet)$ is a quadratic refinement of the pairing
\be
\label{EqDefLinkPairFlat}
L(\check{x}_1, \check{x}_2) := \langle \hat{a}_1 \cup \hat{h}_2, [M,\partial M] \rangle \quad  {\rm mod} \; 1
\ee
on $\mathsf{Z}^{\rm flat}_{M,\partial M}$.
\end{proposition}
\begin{proof}
We have from \eqref{EqLagSum}
\begin{align}
\bar{l}(\check{x}_1 + \check{x}_2) \: & = (\hat{l}(\check{x}_1 + \check{x}_2), (\hat{a}_1)_2 + (\hat{a}_2)_2) \\
& = \left(\hat{l}(\check{x}_1) + \hat{l}(\check{x}_2) + \frac{1}{2} \hat{a}_1 \cup \hat{h}_2 + \frac{1}{2} \hat{a}_2 \cup \hat{h}_1, (\hat{a}_1)_2 + (\hat{a}_2)_2 \right) \;.
\end{align}
On the other hand, we deduce from \eqref{EqGrpLawCocModETh} that
\be
\bar{l}(\check{x}_2) \boxplus \bar{l}(\check{x}_2) = \left(\hat{l}(\check{x}_1) + \hat{l}(\check{x}_2) + \frac{1}{2}\hat{a}_1 \cup_1 \hat{a}_2 ,(\hat{a}_1)_2 + (\hat{a}_2)_2 \right) \;,
\ee
where $\cup_1$ is the higher cup product described in Appendix \ref{SecWuStruct}. Using $d\hat{h}_i = -\hat{a}_i$ and \eqref{EqRelHighCupProd}, we have
\be
\hat{a}_1 \cup_1 \hat{a}_2 = -\hat{h}_1 \cup \hat{a}_2 + \hat{a}_2 \cup \hat{h}_1 + d(\hat{h}_1 \cup_1 \hat{a}_2)\;.
\ee
This means that 
\be
\bar{l}(\check{x}_1 + \check{x}_2) \boxminus \bar{l}(\check{x}_2) \boxminus \bar{l}(\check{x}_2) = \left(\frac{1}{2} \hat{a}_1 \cup \hat{h}_2 + \frac{1}{2}\hat{h}_1 \cup \hat{a}_2 , 0 \right) \boxplus \; d\left(\frac{1}{2}\hat{h}_1 \cup_1 \hat{a}_2,0 \right)
\ee
and
\be
\label{EqActAsAQuadrRef}
S(M; \check{x}_1 + \check{x}_2) - S(M; \check{x}_1) - S(M; \check{x}_2) = \langle \hat{a}_1 \cup \hat{h}_2, [M,\partial M] \rangle \;,
\ee
using again the fact that the pairing of E-cochains of the form $(\hat{s},0)$ with the fundamental E-homology class coincides with the pairing of $\hat{s}$ with the fundamental homology class.
\end{proof}

Recall Corollary \ref{CorStructEM} describing the structure of $\mathsf{Y}^{\rm flat}_{M,\partial M}$: given $\hat{x} \in \mathsf{Z}^{\rm flat}_{M,\partial M}$, the Bockstein map $\beta$ defines an element $\beta(x) \in H^{2\ell+2}_{\rm tors}(M,\partial M; \bar{\Lambda})$. Let $\tilde{L}$ be the linking pairing on $H^{2\ell+2}_{\rm tors}(M,\partial M; \bar{\Lambda})$, defined in Appendix \ref{SecPerPair}. We have
\begin{proposition}
$L$ passes to a well-defined pairing on $\mathsf{Y}^{\rm flat}_{M,\partial M}$. Moreover, it is the pullback of $\tilde{L}$ through $\beta$.
\end{proposition}
\begin{proof}
$-\hat{h}_2$ trivializes the torsion cocycle $\check{a}_2$. If $a_2$ has order $k$, $\hat{y}_2 := -k\hat{h}_2$ is a cocycle trivializing $k\check{a}_2$. We have
\be
L(\check{x}_1, \check{x}_2) = -\frac{1}{k} \langle \hat{a}_1 \cup \hat{y}_2, [M,\partial M] \rangle \quad {\rm mod} \; 1 \;,
\ee
which is the linking pairing between the torsion classes $a_1$ and $a_2$, see \eqref{EqDefLinkPair}. This shows that $L$ is a pull-back from $H^{2\ell+2}_{\rm tors}(M,\partial M; \bar{\Lambda})$ and in particular that it factors through a pairing on $\mathsf{Y}^{\rm flat}_{M,\partial M}$.
\end{proof}
\begin{proposition}
\label{PropActQRLP}
On $\mathsf{Y}^{\rm flat}_{M,\partial M}$, $S(M, \bullet)$ is the pull-back of a quadratic refinement $q_M$ of $\tilde{L}$.
\end{proposition}
\begin{proof}
Let $\check{x}_1, \check{x}_2 \in \mathsf{Y}^{\rm flat}_{M,\partial M}$ be such that $\beta(x_1) = \beta(x_2)$. We can choose the differential cocycles $\check{x}_1$ and $\check{x}_2$ such that $\hat{a}_1 = \hat{a}_2$. Then $\check{x}_2 = \check{x}_1 + (0, \hat{c},0)$, where $\hat{c} \in Z^{2\ell+1}(M,\partial M;\mathscr{V})$ and $\hat{l}(\check{x}_2) = \hat{l}(\check{x}_1) + \frac{1}{2} \hat{a} \cup \hat{c}$. Therefore
\be
\bar{l}(\check{x}_2) = \bar{l}(\check{x}_1) \boxplus d \left(-\frac{1}{2} \hat{h} \cup \hat{c}, 0 \right) \;,
\ee
and $S(M,\check{x}_2) = S(M,\check{x}_1)$.
\end{proof}

It follows from \eqref{EqDefLinkPairFlat} that $\mathsf{E}_{M,\partial M} \subset \mathsf{Y}^{\rm flat}_{M,\partial M}$ is the radical of $L$. Similarly, $\mathsf{K}_{M,\partial M}$ is the radical of $\tilde{L}$. Basic facts about quadratic refinements reviewed in Appendix \ref{SecPairFinAbGroups} imply that $q_M$ is a $\mathbb{Z}_2$-valued character of $\mathsf{K}_{M,\partial M}$, and similarly that $S(M,\bullet)$ is a $\mathbb{Z}_2$-valued character of $\mathsf{E}_{M,\partial M}$. These quadratic refinements are \emph{tame} if the corresponding characters are trivial. The Gauss sum constructed from $q_M$ is non-vanishing if and only if $q_M$ is tame. In this case, the associated Arf invariant is a topological invariant of $\mathfrak{W}$-manifolds. 

Remark that reduction modulo $2\Lambda$ followed with the quotient by the radical provides a homomorphism from $\mathsf{K}_{M,\partial M} \subset H^{2\ell+2}(M,\partial M; \bar{\Lambda})$ into $H^{2\ell+2}(M,\partial M; \bar{\Gamma}^{(2)})$. The $\mathbb{Z}_2$-character on $\mathsf{K}_{M,\partial M}$ determined by the action can be seen as an element $\delta_\mathsf{K}$ of $H^{2\ell+1}(M;\bar{\Gamma}^{(2)})/\mathsf{W}_M$, where $\mathsf{W}_M$ is the orthogonal of the image of $\mathsf{K}_{M,\partial M}$ in $H^{2\ell+2}(M;\bar{\Gamma}^{(2)})$. If $\delta_\mathsf{K}|_{\partial M} = 0$, then we pick a lift $\delta$ of $\delta_\mathsf{K}$ in $H^{2\ell+1}(M,\partial M;\bar{\Gamma}^{(2)})$ and use it to shift the Wu structure of $M$. By Proposition \ref{PropChangWu}, the resulting Wu structure is such that $q_M$ is tame. We therefore have
\begin{proposition}
\label{PropExWuStructqMTame}
If $\delta_\mathsf{K}|_{\partial M} = 0$, there is a $\mathsf{W}_M$-torsor of Wu structures for which $q_M$ is tame.
\end{proposition}

We call \emph{admissible} the Wu structures on $M$ for which $q_M$ is tame. $\delta_\mathsf{K}|_{\partial M} = 0$ if and only if $q_M$ vanishes on elements of $\mathsf{K}_{M,\partial M}$ of the form $a = a' \cup b$. Here $a' \in H^{2\ell+1}_{\rm tors}(\partial M; \bar{\Lambda})$ and $b \in H^1(M,\partial M; \mathbb{Z})$ is obtained as the differential of a 0-cochain taking value 1 on $\partial M$ and vanishing outside a tubular neighborhood of $\partial M$. As $b$ has support in a tubular neighborhood of $\partial M$, this situation can be more conveniently studied on a torus $\partial M \times S^1$, which we will do in Section \ref{SecActTori}.

\paragraph{Relation to existing literature} Let us briefly compare the quadratic refinement $q_M$ to existing constructions in the literature. As far as we know, the existing constructions involve only the case where $\bar{\Lambda}$ is the trivial $\mathbb{Z}$ local system endowed with the unimodular pairing. In this case, $\mathsf{K}_M = 0$ and $q_M$ is automatically tame.

Brumfiel and Morgan define a quadratic refinement of the linking pairing in \cite{Brumfiel1973} in the case where $M$ is the boundary of a $4\ell+4$ manifold $W$. (See also Section 5 of \cite{Morgan1974}.) The data required is a lift of the Wu class of $W$ as a relative cohomology class. In our case, such a lift is provided by the universal choice of trivialization of the Wu class on the classifying space of Wu structures. Proposition \ref{PropActFromBoundedMan} shows that the quadratic refinement defined by their formula on $W$ coincides with our intrinsic definition.

An dual construction by Taylor on homology in the 3-dimensional case appears in \cite{Taylor1984259}, using spin structures. Unlike the approach of Brumfiel and Morgan, it does not rely on manifolds being boundaries.

In principle, $q_M$ should coincide with the quadratic refinement constructed by Hopkins and Singer (Proposition 5.66 of \cite{hopkins-2005-70}). Proposition \ref{PropActFromBoundedMan} shows that they coincide on boundaries, but in general we do not know how to make contact with their formalism.

Topological field theories involving quadratic functions on finite groups were constructed by Deloup in \cite{DELOUP2003199}, but their relation to the present work is unclear to us.

\subsection{Action on tori}

\label{SecActTori}

Let $N$ be a closed $4\ell+2$-dimensional $\mathfrak{W}$-manifold, set $M = N \times S^1$, with the local system pulled back from $N$. If $\ell > 0$, $S^1$ has a unique Wu structure, and we take the Wu structure on $M$ to be  the one induced from $N$. If $\ell = 0$, the Wu structures are spin structures. We take the spin structure on $M$ induced by the spin structure on $N$ and the non-bounding spin structure on $S^1$. Let $\check{x} = (\hat{a}, \hat{h},0) \in \mathsf{Z}^{\rm flat}_M$ be a cocycle representative of $x \in \mathsf{Y}^{\rm flat}_M$, so that $\hat{a} = -d\hat{h}$ and $a$ is a torsion class. Recall that $\beta(x) = a$. We use the fact that the action is a pull back with respect to $\beta$ and freely write $S(M,a)$ for $S(M,x)$. We have a decomposition
\be
\label{EqDecompCohomCyl}
H^{2\ell+2}_{\rm tors}(M; \bar{\Lambda}) = H_{\rm tors}^{2\ell+2}(N; \bar{\Lambda}) \oplus H_{\rm tors}^{2\ell+1}(N; \bar{\Lambda}) \cup g_{S^1} \;,
\ee
where $g_{S^1}$ is the generator of $H^1(S^1;\mathbb{Z})$. Write $\pi$ for the projection on the first summand. The decomposition \eqref{EqDecompCohomCyl} induces a decomposition $\mathsf{K}_M = \mathsf{K}_N \oplus \mathsf{K}'$ with $\mathsf{K}' \subset H_{\rm tors}^{2\ell+1}(N; \bar{\Lambda}) \cup g_{S^1}$.
\begin{proposition}
\label{PropActVanishPullBackTorus}
$S(M,a) = 0$ for any $a$ in the first summand of \eqref{EqDecompCohomCyl}.
\end{proposition}
\begin{proof}
Remark that if $x$ lies in the first summand of \eqref{EqDecompCohomCyl}, we can extend it to $W := N \times D^2$, where $D^2$ is a disk. We can therefore use  \eqref{EqActBManField} to compute the action. But the integrand is pulled back from $N$ and vanishes for degree reasons.
\end{proof}
Reduction modulo 2 provides a homomorphism $\mathsf{K}' \rightarrow H^{2\ell+1}(N;\bar{\Gamma}^{(2)}) \cup g_{S^1} \simeq H^{2\ell+1}(N;\bar{\Gamma}^{(2)})$. Let $\mathsf{W}_N$ be the orthogonal of the image of this inclusion with respect to the non-degenerate cup product pairing on $H^{2\ell+1}(N;\bar{\Gamma}^{(2)})$.
\begin{proposition}
\label{PropDefAdmWuStruct}
There is a Wu structure on $N$, determined up to an element of $\mathsf{W}_N$, such that $S(M,a) = 0$ for all $a \in \mathsf{K}_M$
\end{proposition}
\begin{proof}
As $S$ is linear on $\mathsf{K}_M$ and vanishes on $\mathsf{K}_N$ by Proposition \ref{PropActVanishPullBackTorus}, it is sufficient to show that $S$ vanishes on $\mathsf{K}'$ for some appropriate choice of Wu structure on $N$.

Let $\chi$ be the $\mathbb{Z}_2$-valued character of $\mathsf{K}'$ defined by the action. $\chi$ factors to a character of the image of $\mathsf{K}'$ in $H^{2\ell+1}(N;\bar{\Gamma}^{(2)})$. This ensures that we can identify $\chi$ with a coset of $\mathsf{W}_N$ in $H^{2\ell+1}(N;\bar{\Gamma}^{(2)})$. After shifting the Wu class of $N$ by any element of this coset, the action vanishes on $\mathsf{K}'$, hence on $\mathsf{K_M}$. 
\end{proof}
We call the Wu structures characterized by Proposition \ref{PropDefAdmWuStruct} \emph{admissible}. From now on, we always assume that all $4\ell+2$-dimensional manifolds are endowed with admissible Wu structures. 
\begin{proposition}
\label{PropPrelPreqThHermLine}
Let us fix $\pi(a)$ and consider $S$ as a function of $j = a - \pi(a) \in H_{\rm tors}^{2\ell+1}(N; \bar{\Lambda}) \cup g_{S^1}$. Then there is a unique coset $a_N + \mathsf{K}_N \subset H_{\rm tors}^{2\ell+2}(N; \bar{\Lambda})$ such that $S(M, a_N + j) = 0$ for all $j$.
\end{proposition}
\begin{proof}
The summands in \eqref{EqDecompCohomCyl} are isotropic. The linking pairing $\tilde{L}$ induces a perfect pairing $L'$ on $H^{2\ell+2}_{\rm tors}(M; \bar{\Lambda})/\mathsf{K}_{M}$, with respect to which $H_{\rm tors}^{2\ell+2}(N; \bar{\Lambda})/\mathsf{K}_N$ and $(H_{\rm tors}^{2\ell+1}(N; \bar{\Lambda}) \cup g_{S^1})/\mathsf{K}'$ are Lagrangian. $L'$ identifies them as dual to each other.

$S$ is linear when restricted $H_{\rm tors}^{2\ell+1}(N; \bar{\Lambda}) \cup g_{S^1}$ because the latter subgroup is isotropic. Moreover, we know that $S$ vanishes on $\mathsf{K}_M$. This implies that the linear functional defined by $S$ on $H_{\rm tors}^{2\ell+1}(N; \bar{\Lambda}) \cup g_{S^1}$ is canonically associated to a coset $-a_N + \mathsf{K}_N \in H_{\rm tors}^{2\ell+2}(N; \bar{\Lambda})/\mathsf{K}_N$ under $L'$. Given $f \in H^{2\ell + 1}(N;\bar{\Lambda})$, we see that
\be
S(M,a_N + f \cup g_{S^1}) = S(M,a_N) + S(M,f \cup g_{S^1}) + \tilde{L}(a_N, f \cup g_{S^1}) = S(M,a_N) \;,
\ee
using the defining relation for the quadratic refinement. But Proposition \ref{PropActVanishPullBackTorus} shows that $S(M,a_N) = 0$.
\end{proof}

\section{Prequantum theory} 

\label{SecPreq}

In this section, we construct a field theory functor $\mathcal{S}^e$ out of the exponentiated action function described in Section \ref{SecAct}. This field theory should be thought of as the prequantum Wu Chern-Simons field theory associated to the action \eqref{EqDefActEth}.

\subsection{Field theories on $\mathfrak{W}$-manifolds}

The discretely gauged Wu Chern-Simons field theories of interest in this paper will be constructed in the Atiyah-Segal formalism, as functors 
\be
\label{EqFuncTopFieldTh}
\mathcal{R}: \mathcal{B}_{\mathfrak{W}, {\rm flat}}^{4\ell+3,1} \rightarrow \mathcal{H} \mbox{ or } \mathcal{B}_{\mathfrak{W}, {\rm flat}, \epsilon}^{4\ell+3,1} \rightarrow \mathcal{H} \;,
\ee
depending on whether the order of $\Gamma_0$ is odd or even. The category $\mathcal{B}_{\mathfrak{W}, {\rm flat}}^{4\ell+3,1}$, defined in Appendix \ref{ApBordCatWMan}, is the bordism category of $4\ell+2$-dimensional $(\mathfrak{W}, \mathsf{Z}^{\rm flat})$-manifolds, i.e. $\mathfrak{W}$-manifolds $M$ endowed with an element of $\mathsf{Z}^{\rm flat}_M$. The category $\mathcal{B}_{\mathfrak{W}, {\rm flat}, \epsilon}^{4\ell+3,1}$ is a closely related category in which the $4\ell+2$-dimensional manifolds are endowed with an extra structure $\epsilon$ to be described in Section \ref{SecExStruct}. The target category $\mathcal{H}$ is the category of finite dimensional Hilbert spaces. $\mathcal{B}_{\mathfrak{W}, {\rm flat}}^{4\ell+3,1}$, $\mathcal{B}_{\mathfrak{W}, {\rm flat}, \epsilon}^{4\ell+3,1}$ and $\mathcal{H}$ carry symmetric monoidal structures, given by the disjoint union of manifolds and bordisms and by the tensor product, respectively. The functor $\mathcal{R}$ is required to be monoidal. In addition, the three categories carry dagger structures given by orientation reversal and complex conjugation, respectively (see Appendix \ref{ApBordCatWMan} for a comment on this point). $\mathcal{R}$ is required to preserve the dagger structures as well. We will refer to functors preserving both monoidal and dagger structures as \emph{field theory functors}.

The intermediate step toward the construction of $\mathcal{R}$, taken in the present section, is the construction of an invertible field theory functor
\be
\mathcal{S}^e: \mathcal{B}_{\mathfrak{W},{\rm flat}}^{4\ell+3,1} \rightarrow \mathcal{H} 
\ee
out of the exponentiated action $S^e$. We will call the associated field theory the \emph{prequantum theory} associated to the action $S^e$. This terminology comes from geometric quantization, the state space of the theory being the fiber of the prequantum line bundle \cite{MR1183739}

In order to check that $\mathcal{S}^e$ and $\mathcal{R}$ are field theory functors, we will use the following proposition which offers a characterization slightly different from the one given directly by the axioms of a functor. It is well-known, but we have not been able to find an explicit proof in the literature, so we provide it here. The following notation will be helpful. For all pairs $((M,\check{x}),(N,\check{x}|_N))$ of $(\mathfrak{W},\mathsf{Z}^{\rm flat})$-manifold, where $M$ has dimension $4\ell+3$, $N$ has codimension 1 and is disjoint from $\partial M$, let $M_N$ be the $(\mathfrak{W},\mathsf{Z}^{\rm flat})$-manifold obtained by cutting $M$ along $N$, with boundary $\partial M_N = \partial M \sqcup N \sqcup -N$. The flat differential cocycle on $M_N$ is $\pi^\ast(\check{x})$, where $\pi : M_N \rightarrow M$ is the gluing map, identifying the boundary components $N$ and $-N$ of $M_N$ (see Figure \ref{Fig-Prop51-1}). 
\begin{figure}[t]
  \begin{center}
    \includegraphics[width=.9\textwidth]{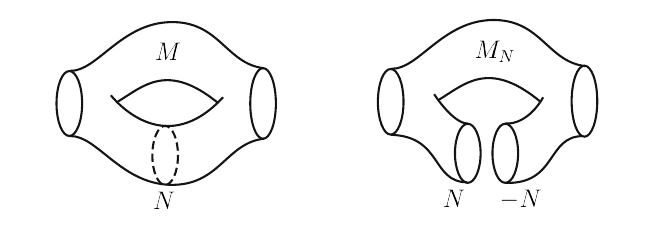}
    \vspace{-1cm}
    \caption{The manifold $M$ and the manifold $M_N$ obtained by cutting $M$ along $N$.}
    \label{Fig-Prop51-1}
  \end{center}
\vspace{-.5cm}
\end{figure}
To simplify the notation, we suppress the flat differential cocycles in the following. It is understood that $\partial M$ is endowed with the restriction of the flat differential cocycle on $M$.
\begin{proposition}
\label{PropAltCondCheckFunc}
Let $\mathcal{F}$ be a function assigning a Hilbert space $\mathcal{F}(N)$ to each closed $4\ell+2$-dimensional $(\mathfrak{W},\mathsf{Z}^{\rm flat})$-manifold $N$ and a vector $\mathcal{F}(M)$ in $\mathcal{F}(\partial M)$ to each $4\ell+3$-dimensional $(\mathfrak{W},\mathsf{Z}^{\rm flat})$-manifold $M$. Then $\mathcal{F}$ is a field theory functor from $\mathcal{B}_{\mathfrak{W},{\rm flat}}^{4\ell+3,1}$ to $\mathcal{H}$ if and only if the following conditions holds.
\begin{enumerate}
\item $\mathcal{F}$ is compatible with the monoidal structures, i.e.
\be
\mathcal{F}(N_1 \sqcup N_2) = \mathcal{F}(N_1) \otimes \mathcal{F}(N_2) \;,
\ee
\be
\label{EqCompMonStructDimD}
\mathcal{F}(M_1 \sqcup M_2) = \mathcal{F}(M_1) \otimes \mathcal{F}(M_2) \in \mathcal{F}(\partial M_1 \sqcup \partial M_2) \;.
\ee
\item $\mathcal{F}$ is compatible with the dagger structures, i.e.
\be
\mathcal{F}(-N) = \overline{\mathcal{F}(N)} \;, 
\ee
\be
\label{EqCompDagStructDimD}
\mathcal{F}(-M) = \overline{\mathcal{F}(M)} \;, 
\ee
where the bar denotes the complex conjugation of vectors and vector spaces.
\item (Gluing condition) For all pairs $(M,N)$ as above, 
\be
\label{EqConstrGluing}
{\rm Tr}_{\mathcal{F}(N)} \mathcal{F}(M_N) = \mathcal{F}(M) \;.
\ee
${\rm Tr}_{\mathcal{F}(N)}$ is here the canonical pairing between the tensor factors $\mathcal{F}(N)$ and $\mathcal{F}(-N) = \overline{\mathcal{F}(N)} \simeq \mathcal{F}(N)^\ast$ of $\mathcal{F}(\partial M_N) \simeq \mathcal{F}(\partial M) \otimes \mathcal{F}(N) \otimes \mathcal{F}(-N)$. 
\end{enumerate}
\end{proposition}
\noindent The conditions given in Proposition \ref{PropAltCondCheckFunc} are natural from the point of view of Euclidean quantum field theory, where the distinction between the incoming components $\partial^-M$ and outgoing components $\partial^+ M$ of a bordism $M$ is arbitrary. The proposition holds for any bordism category, with identical proof up to notational details. In particular, the proposition also holds for the category $\mathcal{B}_{\mathfrak{W},{\rm flat},\epsilon}^{4\ell+3,1}$
\begin{proof}
Suppose that $\mathcal{F}$ satisfies the conditions above. Let $(M, \theta^-, \theta^+)$ be a bordism from $N^-$ to $N^+$. (See Appendix \ref{ApBordCatWMan} for the notation.) We see that
\begin{align}
\label{EqDerFuncBoundToHomFunc}
\mathcal{F}(\partial M) \: & = \mathcal{F}(\partial^- M) \otimes \mathcal{F}(\partial^+ M) \simeq \mathcal{F}(-N^-) \otimes \mathcal{F}(N^+) \\
& \simeq \mathcal{F}(N^-)^\ast \otimes \mathcal{F}(N^+) = {\rm Hom}(\mathcal{F}(N^-), \mathcal{F}(N^+)) \;, \notag
\end{align}
so $\mathcal{F}$ assigns a homomorphism from $\mathcal{F}(N^-)$ to $\mathcal{F}(N^+)$ to the bordism $M$, as a functor should. 

Assume now that we have a bordism $(M_1, \theta^-_1, \theta^+_1)$ from $N^-$ to $N$ and a bordism $(M_2, \theta^-_2, \theta^+_2)$ from $N$ to $N^+$. By setting $M_N = M_1 \sqcup M_2$ and $M = M_1 \sqcup_{(\theta^-_2)^{-1} \circ \theta^+_1} M_2$ in \eqref{EqConstrGluing} (see Appendix \ref{ApBordCatWMan}), we see that $\mathcal{F}$ preserves the composition in $\mathcal{B}_{\mathfrak{W},{\rm flat}}^{4\ell+3,1}$ and $\mathcal{H}$, hence is a functor. The first two conditions ensure that it is a field theory functor. This proves one direction of the proposition.

Suppose now that $\mathcal{F}$ is a field theory functor. The functor axioms imply that on a bordism $M$ with $\partial^- M = \emptyset$, $\partial^+ M = \partial M$, $\mathcal{F}(M) \in {\rm Hom}(\mathbb{C}, \mathcal{F}(\partial M)) \simeq \mathcal{F}(\partial M)$. The fact that $\mathcal{F}$ is a field theory functor ensures that the first two conditions are satisfied. 

To check the third one, pick another codimension 1 submanifold $N' \subset M$ disjoint from $N$ and $\partial M$, and such that after cutting $M$ along $N \sqcup N'$, pairs of points on each side of the cut belong to disconnected components, say $M_1$ and $M_2$. In particular, we have a decomposition $\partial M = N_1 \sqcup N_2$, $N_1 \subset M_1$, $N_2 \subset M_2$ (see Figure \ref{Fig-Prop51-2}).
\begin{figure}[t]
  \begin{center}
    \includegraphics[width=.9\textwidth]{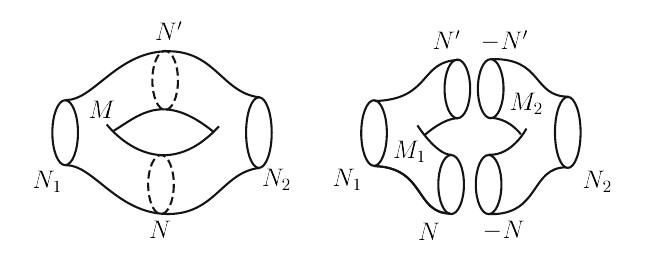}
    \vspace{-1cm}
    \caption{The manifold $M$ cut into two disjoint components.}
    \label{Fig-Prop51-2}
  \end{center}
\vspace{-.5cm}
\end{figure}
We define bordism structures on $M_1$ and $M_2$ by setting $\partial^- M_1 = N_1$, $\partial^+ M_1 = \partial^- M_2 = N \sqcup N'$, $\partial^+ M_2 = N_2$ and a bordism structure on $M$ by setting $\partial^- M = N_1$, $\partial^+ M = N_2$. (The choices of isomorphisms $\theta^\pm_{1,2}$ are irrelevant as long as they are made consistently with the gluing.) Then the bordism $M$ can be obtained by the composition of $M_1$ and $M_2$, from which we deduce that
\be
\label{EqFM}
\mathcal{F}(M) = \mathcal{F}(M_2) \circ_{N \sqcup N'} \mathcal{F}(M_1) = {\rm Tr}_{\mathcal{F}_{N \sqcup N'}} (\mathcal{F}(M_2) \otimes \mathcal{F}(M_1)) \;.
\ee
By making the alternative choice $\partial^- M_1 = N_1 \sqcup N$, $\partial^+ M_1 = \partial^- M_2 = N'$, $\partial^+ M_2 = N_2 \sqcup N$, and setting a bordism structure on $M_N$ with $\partial^- M_N = N_1 \sqcup N$, $\partial^+ M_N = N_2 \sqcup N$, we deduce that 
\be
\label{EqFMN}
\mathcal{F}(M_N) = \mathcal{F}(M_2) \circ_{N'} \mathcal{F}(M_1) = {\rm Tr}_{\mathcal{F}_{N'}} (\mathcal{F}(M_2) \otimes \mathcal{F}(M_1))\;.
\ee
As $\mathcal{F}(N \sqcup N') = \mathcal{F}(N) \otimes \mathcal{F}(N')$, 
\be
\label{EqProdTrace}
{\rm Tr}_{\mathcal{F}_{N \sqcup N'}} = {\rm Tr}_{\mathcal{F}_{N}}{\rm Tr}_{\mathcal{F}_{N'}}
\ee
Combining \eqref{EqFM}, \eqref{EqFMN} and \eqref{EqProdTrace} yields the third condition of the proposition, completing the proof.
\end{proof}

\subsection{Prequantum theory}

\label{SecSubPreqTh}

\paragraph{Closed $4\ell+3$-dimensional manifold} Let $(M,\check{x})$ be a closed $(\mathfrak{W},\mathsf{Z}^{\rm flat})$-manifold of dimension $4\ell+3$.
Such a manifold can be seen as a bordism from the empty object of $\mathcal{B}_{\mathfrak{W},{\rm flat}}^{4\ell+3,1}$ to itself. It is not difficult to see that the compatibility of the field theory functor with the monoidal structure requires it to assign $\mathbb{C}$ to the empty object. The value of the field theory functor on such bordisms is therefore a complex number, the \emph{partition function}. The partition function of the prequantum theory is simply the exponentiated action:
\be
\mathcal{S}^e(M,\check{x}) := S^e(M,\check{x}) \;.
\ee
As the exponentiated action depends only on the differential cohomology class $x$ of $\check{x}$, the same is true for the partition function $\mathcal{S}^e(M,\check{x})$ of the prequantum theory. We will therefore sometimes freely write $\mathcal{S}^e(M,x)$.

\paragraph{Closed $4\ell+2$-dimensional} Let $(M,\check{x})$ be a closed $(\mathfrak{W},\mathsf{Z}^{\rm flat})$-manifold of dimension $4\ell+2$. We construct a Hermitian line $\mathcal{S}^e(M, x)$ out of this data as follows. As the notation suggests, $\mathcal{S}^e(M, x)$ depends only on the cohomology class $x$ of $\check{x}$.

Consider the category $\mathcal{C}_M$ whose objects are pairs $(\bar{m}, \check{x}')$, where $\bar{m}$ is an E-chain representing the fundamental E-homology class of $M$ and $\check{x}' \in \mathsf{Z}^{\rm flat}_M$ is a differential cocycle representative of $x$. E-chains are defined simply as homomorphisms from the E-cochain group into $\mathbb{R}/\mathbb{Z}$, see Appendix \ref{SecEth}. The morphisms in $\mathcal{C}_M$ from $(\bar{m}_0, \check{x}_0)$ to $(\bar{m}_1, \check{x}_1)$ are cylinders $M \times I$ endowed with a pair $(\bar{n}, \check{y})$, where $\bar{n}$ is an E-chain representing the fundamental E-homology class of $M \times I$ and $\check{y} \in \mathsf{Z}^{\rm flat}(M \times I)$. $(\bar{n}, \check{y})$ is required to satisfy $\partial \bar{n} = \bar{m}_1 - \bar{m}_0$, $\check{y}|_{M \times \{0\}} = \check{x}_0$ and $\check{y}|_{M \times \{1\}} = \check{x}_1$. $\mathcal{C}_M$ is a transitive groupoid whose automorphisms are given by tori $M \times S^1$ endowed with the data above. 

We write $\bar{l}$ for the Lagrangian associated to a morphism $(M \times I, \check{y})$. $\bar{l}$ is an E-cocycle of degree $4\ell+3$. Consider the functor $\textcal{f}: \mathcal{C}_M \rightarrow \mathcal{H}$ assigning $\mathbb{C}$ to each object and $\exp 2\pi i \langle \bar{l}, \bar{n} \rangle$ to any morphism $\bar{n}$. 

A \emph{section} of $\textcal{f}$ is a function $v$ assigning to each object $(\bar{m}, \check{x})$ of $\mathcal{C}_M$ a vector $v(\bar{m}, \check{x})$ in $\textcal{f}\:(\bar{m}, \check{x})$. A section $v$ is \emph{invariant} if $v \circ (\bar{n}, \check{y}) = \textcal{f}\:(\bar{n}, \check{y}) \circ v$ for any morphism $(\bar{n},\check{y})$ in $\mathcal{C}_M$. We define $\mathcal{S}^e(M,x)$ to be the Hilbert space of invariant sections of $\textcal{f}$. 

Recall that by considering the partition function on tori of the form $M \times I$ in Section \ref{SecActTori}, we defined a coset $a_M + \mathsf{K}_M \subset H^{2\ell+2}_{\rm tors}(M;\bar{\Lambda})$ on $M$.
\begin{proposition}
\label{PropDimHilbSp}
The Hilbert space $\mathcal{S}^e(M,x)$ is a Hermitian line if $a \in a_M + \mathsf{K}_M$, otherwise it is the zero Hilbert space.
\end{proposition}
\begin{proof}
As the groupoid $\mathcal{C}_M$ is transitive and $\textcal{f}$ sends each object to a 1-dimensional Hilbert space, the space of invariant section is either 1-dimensional or 0-dimensional. It is 1-dimensional if and only if all the automorphisms of $\mathcal{C}_M$ are sent to the identity by $\textcal{f}\:$. Proposition \ref{PropPrelPreqThHermLine} shows that this is the case if and only if $a \in a_M + \mathsf{K}_M$.
\end{proof}
The vanishing of the Hilbert space when $a \notin a_M + \mathsf{K}_M$ can be understood as a global gauge anomaly. In case $\Lambda = \Lambda^\ast = \mathbb{Z}$ with trivial local systems, it boils down to the one described in \cite{Witten:1999vg, Monniera} in the context of the self-dual field theory.

\paragraph{$4\ell+3$-dimensional manifolds with boundary} Let $(M,\check{w})$ be a $(\mathfrak{W},\mathsf{Z}^{\rm flat})$-manifold of dimension $4\ell + 3$ with boundary. Let $w$ be the differential cohomology class of $\check{w}$.

\begin{proposition}
\label{PropPQFTdDimWB}
$(M,w)$ determines an element $\mathcal{S}^e(M,w) \in \mathcal{S}^e(\partial M, w|_{\partial M})$.
\end{proposition}
\begin{proof}
Pick an E-chain representative $\bar{p}$ of the fundamental class $[M,\partial M]_E$. Write $\bar{m} = \partial \bar{p}$. Pick as well any cocycle representative $\check{w}' \in \mathsf{Z}^{\rm flat}_M$ of $w$. As the exponentiated pairing $\exp 2\pi i \langle \bar{l}, \bar{p} \rangle$ is invariant on closed manifolds under changes of $\bar{p}$ and $\check{w}'$, by excision it is also invariant on $M$ as long as $\bar{m}$ and $\check{w}|_{\partial M}$ are kept fixed. However, as the Lagrangian $\bar{l}$ is not a relative E-cocycle in general, $\exp 2\pi i \langle \bar{l}, \bar{p} \rangle$ fails to be invariant for generic variations of $\bar{p}$ and $\check{w}'$.

Suppose we change the representative of the fundamental class from $\bar{p}_0$ to $\bar{p}_1 = \bar{p}_0 + \bar{n}$, with $\partial \bar{n} = \bar{m}_1 - \bar{m}_0$. Suppose also that we also change the cocycle representative from $\check{w}_0$ to $\check{w}_1$, with $\check{w}_i|_{\partial M} = \check{x}_i$, $i = 0,1$. Up to a change of $(\bar{p}_1, \check{w}_1)$ away from $\partial M$ leaving the exponentiated pairing invariant, we can implement these changes by gluing along $\partial M$ the cylinder $(\partial M  \times I, \bar{n}, \check{y})$ corresponding to a morphism in the category $\mathcal{C}_{\partial M}$  from $(\bar{m}_0, \check{x}_0)$ to $(\bar{m}_1, \check{x}_1)$. Under this operation, the exponentiated pairing gets multiplied by $\exp 2\pi i \langle \bar{l}(\check{y}), \bar{n} \rangle$.  This shows that the exponentiated pairing defines an invariant section $\mathcal{S}^e(M, w)$ in $\mathcal{S}^e(\partial M, w|_{\partial M})$.
\end{proof}

Note that the proof does not use the fact that the differential cocycle $\check{w}$ is flat. We have therefore the following Corollary that will be useful in Section \ref{SecWilHooftOp}.
\begin{corollary}
\label{CorHomFromManWBoundWNFDiffCoc}
Let $M$ be a $\mathfrak{W}$-manifold of dimension $4\ell+3$ with boundary endowed with a not necessarily flat differential cohomology class $w$, whose restriction to $\partial W$ is flat. Then the pair $(M,w)$ determines an element $\mathcal{S}^e(M,w) \in \mathcal{S}^e(\partial M, w|_{\partial M})$.
\end{corollary}

\paragraph{Compatibility with the monoidal structures} Suppose $(M,\check{x})$ is a $(\mathfrak{W},\mathsf{Z}^{\rm flat})$-manifold of dimension $4\ell+3$, possibly with boundary, that decomposes into two disjoint components $(M,\check{x}) = (M_1,\check{x}_1) \sqcup (M_2,\check{x}_2)$, with $\check{x}_i$, $i = 1,2$ supported on $M_i$. Let $\bar{n}_1$ and $\bar{n}_2$ be E-chains representing $[M_1,\partial M_1]_E$ and $[M_2, \partial M_2]_E$ and extend them by zero to $M$. Proposition \ref{PropFundHomEClassDisUn} ensures that $[M,\partial M]_E = [M_1,\partial M_1]_E + [M_2, \partial M_2]_E$, so $\bar{n} := \bar{n}_1 + \bar{n}_2$ is an E-chain representative of $[M,\partial M]_E$. As $\bar{l}(\check{x}) = \bar{l}(\check{x}_1) + \bar{l}(\check{x}_2)$ and $\langle \bar{l}(\check{x}_1), \bar{n}_2 \rangle = \langle \bar{l}(\check{x}_2), \bar{n}_1 \rangle = 0$, we have
\be
\langle \bar{l}(\check{x}), \bar{n} \rangle = \langle \bar{l}(\check{x}_1), \bar{n}_1 \rangle + \langle \bar{l}(\check{x}_2), \bar{n}_2 \rangle \;.
\ee
The formula above together with the definitions of $\mathcal{S}^e$ on $4\ell+2$ and $4\ell+3$-dimensional $\mathfrak{W}$-manifolds imply immediately
\begin{proposition}
\label{PropCompPreqThMonoidStruct}
If $(N,\check{y})$ is a $4\ell+2$-dimensional $(\mathfrak{W},\mathsf{Z}^{\rm flat})$-manifold that decomposes into two disjoint components $(N,\check{y}) = (N_1,\check{y}_1) \sqcup (N_2,\check{y}_2)$, then
\be
\mathcal{S}^e(N,y) = \mathcal{S}^e(N_1,y_1) \otimes \mathcal{S}^e(N_2,y_2) \;.
\ee
Moreover, if $(M,\check{x})$ is as above, then
\be
\mathcal{S}^e(M,x) = \mathcal{S}^e(M_1,x_1) \otimes \mathcal{S}^e(M_2,x_2) \;.
\ee
\end{proposition}

\noindent Hence the functor $\mathcal{S}^e$ is compatible with the monoidal structures on $\mathcal{B}^{4\ell+3,1}_{\mathfrak{W},{\rm flat}}$ and $\mathcal{H}$.

\paragraph{Compatibility with the dagger structures} A bar over a Hilbert space will denote the complex conjugate Hilbert space, and a bar over a vector in a Hilbert space will denote the complex conjugate vector in the complex conjugate Hilbert space. (We recall that a bar is also used to denote E-chains and E-cochains, hopefully the context will allow the reader to distinguish these two meanings unambiguously.) Let $(M,\check{x})$ be a $(\mathfrak{W},\mathsf{Z}^{\rm flat})$-manifold of dimension $4\ell+3$, possibly with boundary. Let $\bar{n}$ be an E-chain representing $[M,\partial M]_E$. Proposition \ref{PropFundHomEClassOrFl} ensures that $- \bar{n}$ is an E-chain representing $[-M,\partial (-M)]_E$. 
\begin{proposition}
\label{PropCompPreqDagStruct}
If $(N,\check{y})$ is a $4\ell+2$-dimensional $(\mathfrak{W},\mathsf{Z}^{\rm flat})$-manifold, then
\be
\mathcal{S}^e(-N,y) = \overline{\mathcal{S}^e(N,y)} \;.
\ee
Moreover, if $(M,\check{x})$ is as above, then
\be
\mathcal{S}^e(-M,x) = \overline{\mathcal{S}^e(M,x)} \;.
\ee
\end{proposition}
\begin{proof}
The factors $\exp \langle \bar{l}(\check{x}), \bar{n} \rangle$ involved in the construction of $\mathcal{S}^e(N,y)$ and $\mathcal{S}^e(M,x)$ turn into their complex conjugate under an orientation flip.
\end{proof}

\paragraph{Gluing} We now show that $\mathcal{S}^e$ satisfies the 3rd condition of Proposition \ref{PropAltCondCheckFunc}. 
Let $(M,\check{x})$, $(N,\check{x}|_N)$ and $(M_N,\pi^\ast(\check{x}))$ be $(\mathfrak{W},\mathsf{Z}^{\rm flat})$-manifolds as in Proposition \ref{PropAltCondCheckFunc}. Recall that the gluing map $\pi: M_N \rightarrow M$ identifies the points of the components $N$ and $-N$ of $\partial M_N$ and is bijective away from $N \sqcup -N$. We write $\check{y} := \check{x}|_N$, $\check{z} := \pi^\ast(\check{x})$. 
\begin{proposition}
\label{PropGluPreq}
The 3rd condition of Proposition \ref{PropAltCondCheckFunc} holds:
\be
\label{Eq3rdCondSe}
{\rm Tr}_{\mathcal{S}^e(N,y)}(\mathcal{S}^e(M_N,z)) = \mathcal{S}^e(M,x) \;.
\ee
\end{proposition}
\begin{proof}
Let us choose E-chains $\bar{n}$ and $\bar{p} := \pi_\ast(\bar{n})$ representing $[M_N,\partial M_N]_E$ and $[M,\partial M]_E$. Together with the differential cocycles $\check{x}|_{\partial M}$ and $\check{z}|_{\partial M_N}$, their restriction to the boundary provide trivializations of $\mathcal{S}^e(\partial M, x|_{\partial M})$ and $\mathcal{S}^e(\partial M_N, z|_{\partial M_N})$. Moreover we have by definition
\be
\langle \bar{l}(\check{x}), \bar{p} \rangle = \langle \bar{l}(\check{z}), \bar{n} \rangle \;.
\ee
These two quantities are respectively $\mathcal{S}^e(M,x)$ and $\mathcal{S}^e(M_N,z)$ in the trivializations above. As the trace from $\mathbb{C} \simeq \mathbb{C} \otimes \mathbb{C} \stackrel{{\rm Tr}_{\mathbb{C}}}{\rightarrow} \mathbb{C}$ is the identity map, we obtain \eqref{Eq3rdCondSe}.
\end{proof}

\noindent The results of this section combined with Proposition \eqref{PropAltCondCheckFunc} yields
\begin{theorem}
\label{ThPreqFTFunc}
$\mathcal{S}^e: \mathcal{B}_{\mathfrak{W},{\rm flat}}^{4\ell+3,1} \rightarrow \mathcal{H}$ is a field theory functor.
\end{theorem}

\section{Partition function}

\label{SecPartFunc}

We now proceed to the construction of the field theory functor $\mathcal{R}$ out of the prequantum theory of Section \ref{SecPreq}. $\mathcal{R}$ has a $\mathbb{C}$-valued truncation to closed $4\ell + 3$-dimensional manifolds, the partition function, which we describe in the present section.

\subsection{Definition}

Given a closed $(\mathfrak{W},\mathsf{Z}^{\rm flat})$-manifold $(M,\check{x})$ of dimension $4\ell+3$, the partition function of the quantum theory $\mathcal{R}$ on $M$ is obtained by summing the exponentiated action over the coset of $x + \mathsf{E}_M$, with a normalization factor given by $\mu_M$:
\be
\label{EqPartFuncThR}
\mathcal{R}(M,x) = \mu_M \sum_{w \in x + \mathsf{E}_M} \mathcal{S}^e(M,w) \in \mathbb{C} \;.
\ee
The partition function depends only on the differential cohomology class $x$ of $\check{x}$ because $\mathcal{S}^e(M,x)$ has this property. Moreover, $\mathcal{R}(M,x)$ depends by construction only on the equivalence class of $x$ in $\mathsf{Y}^{\rm flat}_M/\mathsf{E}_M$.

Proposition \ref{PropActQRLP} states that the action $S^e$ is the pull-back through $\beta$ of a quadratic refinement $q_M$ of the linking pairing on $H_{\rm tors}^{2\ell+2}(M; \bar{\Lambda})$. We can therefore rewrite the partition function \eqref{EqPartFuncThR} as follows:
\be
\mathcal{R}(M,x) = \mu_M |\mathsf{C}_M| \sum_{y \in \beta(x) + \mathsf{K}_M} \exp 2\pi i \, q_M(y) \;.
\ee

\subsection{Partition function anomaly}

\label{SecPartFuncAnom}

Recall that $\mathsf{K}_M$ is the radical of the linking pairing on $H_{\rm tors}^{2\ell+2}(M; \bar{\Lambda})$, so $q_M|_{\mathsf{K}_M}$ is a character of $\mathsf{K}_M$. If this character is trivial,
\be
\mathcal{R}(M,x) = \mu_M |\mathsf{E}_M| \exp 2\pi i \, q_M(x)\;.
\ee
If $q_M|_{\mathsf{K}_M}$ is non-trivial, the partition function of the prequantum theory is not invariant under the action of the discrete symmetry $\mathsf{E}_M$. In other words, this symmetry is anomalous. As a result, the partition function of the gauged theory $\mathcal{R}$ vanishes identically in this case.

As $q_M$ is $\mathbb{Z}_2$-valued on $\mathsf{K}_M$, it can be non-trivial only if the order of $\mathsf{K}_M$ is even. This is possible only if $\Gamma := \Lambda^\ast/\Lambda$ has even order, see \eqref{DiagCohomGroups}. Therefore the theories with $|\Gamma|$ odd are free of this anomaly. Moreover, Proposition \ref{PropExWuStructqMTame} shows that there exist always admissible Wu structures on $M$, with which the theory is anomaly-free.

We will see that the partition function anomaly described here has a Hamiltonian counterpart on $4\ell+2$-dimensional $\mathfrak{W}$-manifolds.

\section{Wilson operators}

\label{SecWilHooftOp}

Given a closed $4\ell+2$-dimensional $(\mathfrak{W},\mathsf{Z}^{\rm flat})$-manifold $(M,\check{w})$, the prequantum theory yields a Hilbert space $\mathcal{S}^e(M,w)$. Proposition \ref{PropDimHilbSp} shows that $\mathcal{S}^e(M,w)$ is a Hermitian line if $\beta(w) \in a_M + \mathsf{K}_M$, otherwise it is the trivial Hilbert space. Let us consider the vector space
\be
\label{EqDefVecSpV}
V_{M,w} := \bigoplus_{x \in w + \mathsf{G}_M} \mathcal{S}^e(M,x) \;.
\ee
If $\beta(w) \notin a_M + \mathsf{K}_M$, ${\rm dim}(V_{M,v}) = 0$. Otherwise ${\rm dim}(V_{M,w}) = |\mathsf{G}_M|$. In the following, we will always assume that $\beta(w) \in a_M + \mathsf{K}_M$.

$\mathsf{G}_M$ is endowed with a alternating pairing (see Appendix \ref{SecPerPair}) and has therefore an associated Heisenberg extension $\mathsf{H}_M$ (see Appendix \ref{AppHeisGr2Gr}). The aim of the present section is to show that $(M,\hat{w})$ determines canonically a representation of $\mathsf{H}_M$ on $V_{M,w}$.

\subsection{Homomorphisms from cylinders}

In this section, we show that certain non-flat differential cocycles on cylinders define operators on $V_{M,w}$ associated to each element of $H^{2\ell+1}_{\rm free}(M; \bar{\Lambda}^\ast)$. These operators can be thought of as the Wilson operators for the background gauge field $\check{x}$.

Let $\check{x} = (\hat{a}, \hat{h}, 0) \in \beta^{-1}(a_M + \mathsf{K}_M)$. Let $\Omega_{\bar{\Lambda}^\ast}^{2\ell+1}(M; \mathscr{V})$ be the subgroup of closed differential forms whose pairing with classes in $H_{2\ell+1}(M; \bar{\Lambda})$ is integer-valued.
Let $\hat{j} \in \Omega_{\bar{\Lambda}^\ast}^{2\ell+1}(M; \mathscr{V})$ and write $\check{j}$ for the flat topologically trivial differential cocycle $(0, \hat{j},0)$. Let $I := [0,1]$ and define the path $p_{\check{x},\hat{j}}: I \rightarrow \check{Z}^{2\ell+2}(M;\bar{\Lambda})$, 
\be
\label{EqPathDiffCoc}
p_{\check{x},\hat{j}}(t) = \check{x} + t \check{j} = (\hat{a}, \hat{h} + t\hat{j},0) \;, \quad t \in I \;,
\ee 
interpolating linearly between the differential cocycles $\check{x}$ and $\check{x} + \check{j}$. Here $t$ is a coordinate on $I$ running from $0$ to $1$ and that has all its derivative vanishing at the endpoints. $p_{\check{x},\hat{j}}$ defines a non-flat differential cocycle $\check{y}_{\check{x},\hat{j}} \in \check{Z}^{2\ell+2}(M \times I)$ given by 
\be
\label{EqDefCocYOnCyl}
\check{y}_{\check{x},\hat{j}} = (\hat{a}, \hat{h} + \hat{j} \wedge \hat{t}, -\hat{j} \wedge d\hat{t}) \;.
\ee
For notational consistency, we write $\hat{t} = t$ for the degree 0 $\mathbb{R}$-valued differential form on $M \times I$ determined by the function $t$. We also write all the wedge products explicitly.

Let $j$ be the cohomology class of $\hat{j}$ in $H^{2\ell+1}_{\rm free}(M;\bar{\Lambda}^\ast)$. There is an action of $\Omega_{\bar{\Lambda}^\ast}^{2\ell+1}(M; \mathscr{V})$ on $\mathsf{Z}^{\rm flat}_M$ given by picturing the differential form $\hat{j}$ as a differential cocycle $\check{j}$ as above and using the addition of differential cocycles. The action passes to an action of $H^{2\ell+1}_{\rm free}(M;\bar{\Lambda}^\ast)$ on $\mathsf{Y}^{\rm flat}_M$, which we write additively: $x + j \in \mathsf{Y}^{\rm flat}_M$. By Corollary \ref{CorHomFromManWBoundWNFDiffCoc}, $\mathcal{S}^e(M \times I, \check{y}_{\check{x},\hat{j}})$ is an isomorphism from $\mathcal{S}^e(M, x)$ to $\mathcal{S}^e(M,x + j)$. If we change $\check{x}$ by the differential of a flat differential cochain, $\check{y}_{\check{x},\hat{j}}$ changes by the differential of a flat differential cochain, hence $\mathcal{S}^e(M \times I, \check{y}_{\check{x},\hat{j}})$ depends only on the differential cohomology class $x$.
\begin{proposition}
\label{PropDepQCohomClassj}
$\mathcal{S}^e(M \times I, \check{y}_{\check{x},\hat{j}})$ depends only on the cohomology class $j$.
\end{proposition}
\begin{proof}
Any differential form in the same class as $\hat{j}$ in $H^{2\ell+1}_{\rm free}(M;\bar{\Lambda}^\ast)$ can be written $\hat{j} + d\hat{k}$ for some $\hat{k} \in \Omega^{2\ell}(M, \mathscr{V})$. Then
\be
\check{y}_{\check{x},\hat{j}+d\hat{k}} = (\hat{a}, \hat{h} + (\hat{j} + d\hat{k}) \wedge \hat{t}, -(\hat{j}  + d\hat{k}) \wedge dt) \;.
\ee
Define
\be
\check{y}' = (\hat{a}, \hat{h} + \hat{j} \wedge \hat{t} - \hat{k} \wedge dt, -(\hat{j} + d\hat{k}) \wedge dt ) \;.
\ee
Then $\check{y}'$ and $\check{y}_{\check{x},\hat{j}+d\hat{k}}$ differ by an exact differential cocycle and Corollary \ref{CorHomFromManWBoundWNFDiffCoc} shows that $\mathcal{S}^e(M \times I,\check{y}_{\check{x},\hat{j}+d\hat{k}})$ and $\mathcal{S}^e(M \times I,\check{y}')$ are equal as isomorphisms from $\mathcal{S}^e(M, x)$ to $\mathcal{S}^e(M,x + j)$. We can compute the composition $\mathcal{S}^e(M \times I,\check{y}') \circ (\mathcal{S}^e(M \times I,\check{y}_{\check{x},\hat{j}}))^{-1}$ by gluing the corresponding cylinders into a torus $M \times S^1$, with an orientation flip on the second one, and evaluating the exponentiated action $S^e$ on the corresponding differential cocycle. Up to exact terms, the latter is
\be
\check{y}_{S^1} = (\hat{a}, \hat{h} + \hat{j} \wedge \hat{t}_{S^1} - \hat{k} \wedge \hat{g}^{(1)}, -d\hat{k} \wedge \hat{g}^{(1)} - \hat{j}\wedge (\hat{g}^{(1)} - \hat{g}^{(2)}))\;.
\ee
$\hat{t}_{S^1}$ is the smooth function obtained from the functions $\hat{t}$ on the two cylinders after gluing. $\hat{g}^{(1)}$ et $\hat{g}^{(2)}$ are the cocycle representatives of the generator of $H^1(S^1;\mathbb{Z})$ restricting to $d\hat{t}$ on each interval respectively. 

We can compute the Lagrangian
\begin{align}
\label{EqLagActCyl}
\bar{l}(\check{y}_{S^1}) = \: & \left(\hat{l}(\check{y}_{S^1}), (\hat{a})_2 \right) \;, \notag \\
\hat{l}(\check{y}_{S^1}) = \: & \frac{1}{2} \hat{a} \cup (\hat{h} + \hat{j} \wedge \hat{t}_{S^1} - \hat{k} \wedge \hat{g}^{(1)} - \hat{\eta}) \\ 
& + \frac{1}{2} (\hat{h} + \hat{j} \wedge \hat{t}_{S^1} - \hat{k} \wedge \hat{g}^{(1)}) \cup (-d\hat{k} \wedge \hat{g}^{(1)} - \hat{j}\wedge \hat{g}^{(1)} + \hat{j} \wedge \hat{g}^{(2)}) 
+ H^\wedge_\cup(\hat{\omega}, \hat{\omega})   \notag
\end{align}
where we wrote $\hat{\omega} = -d\hat{k} \wedge \hat{g}^{(1)} - \hat{j} \wedge (\hat{g}^{(1)} - \hat{g}^{(2)})$. As $\hat{a} = -d\hat{h}$, 
\be
\frac{1}{2} \hat{a} \cup (- \hat{k} \wedge \hat{g}^{(1)}) + \frac{1}{2} \hat{h} \cup (-d\hat{k} \wedge \hat{g}^{(1)}) = \frac{1}{2} d\left(\hat{h} \cup (\hat{k} \wedge \hat{g}^{(1)}) \right) \;,
\ee
\be
\frac{1}{2} \hat{a} \cup (\hat{j} \wedge \hat{t}_{S^1}) + \frac{1}{2} \hat{h} \cup (-\hat{j} \wedge (\hat{g}^{(1)} - \hat{g}^{(2)})) = -\frac{1}{2} d\left(\hat{h} \cup (\hat{j} \wedge \hat{t}_{S^1})) \right) \;,
\ee
We can always get rid of such exact terms by picking an equivalent E-cocycle, as this does not change the value of the action. The remaining terms on the second line of \eqref{EqLagActCyl} are
\begin{align}
\frac{1}{2} (\hat{j} \wedge \hat{t}_{S^1} - \hat{k} \wedge \hat{g}^{(1)}) \cup \hat{\omega} + H^\wedge_\cup(\hat{\omega}, \hat{\omega}) \: & = \frac{1}{2} (\hat{j} \wedge \hat{t}_{S^1} - \hat{k} \wedge \hat{g}^{(1)}) \wedge \hat{\omega} + {\rm exact} \notag \\
& = - \frac{1}{2} (\hat{j} \wedge \hat{t}_{S^1} \wedge d\hat{k} \wedge \hat{g}^{(1)}) + {\rm exact} \;,
\end{align}
where we discarded the vanishing terms involving $\hat{g}^{(1)} \wedge \hat{g}^{(1)}$, $\hat{g}^{(1)} \wedge \hat{g}^{(2)}$ or $\hat{j} \wedge \hat{j}$. The remaining term is exact as well because $d\hat{t}_{S^1} \wedge \hat{g}^{(1)} = 0$. The Lagrangian is equivalent to an E-cocycle of the form 
\be
\label{EqPullBackELag}
\left( \frac{1}{2} \hat{a} \cup (\hat{h} - \hat{\eta}_\Lambda),(\hat{a})_2 \right) \;,
\ee
which is pulled-back from $M$. Proposition \ref{PropActVanishPullBackTorus} shows that the action vanishes. Therefore $\mathcal{S}^e(M \times I,\check{y}_{\check{x},\hat{j}}) = \mathcal{S}^e(M \times I,\check{y}_{\check{x},\hat{j}+d\hat{k}})$ as isomorphisms from $\mathcal{S}^e(M, x)$ to $\mathcal{S}^e(M,x + j)$, proving the proposition.
\end{proof}
\begin{proposition}
\label{PropQIndepPar}
$\mathcal{S}^e(M \times I, \check{y}_{\check{x},\hat{j}})$ does not depend on the parametrization of the path \eqref{EqPathDiffCoc}.
\end{proposition}
\begin{proof}
Suppose we have two parametrization $t_1$ and $t_2$. We evaluate the exponentiated action on $M \times S^1$ endowed with the differential cocycle
\be
\check{y}_{S^1} = (\hat{a}, \hat{h} + \hat{j} \wedge (\hat{t}_2 - \hat{t}_1), - \hat{j} \wedge (d\hat{t}_2 - d \hat{t}_1)) \;.
\ee
Here we decompose $M \times S^1$ into $M \times I$ and $M \times -I$. $\hat{t}_1$ ($\hat{t}_2$) is the function on $M \times S^1$ that take value $t$ on the first (second) copy of $M \times I$ and vanishes on the complement. The Lagrangian reads
\be
\label{EqLagRepInv}
\hat{l}(\check{y}_{S^1}) = \frac{1}{2} \hat{a} \cup (\hat{h} + \hat{j} \wedge (\hat{t}_2 - \hat{t}_1) - \hat{\eta}_\Lambda) - \frac{1}{2} (\hat{h} + \hat{j} \wedge (\hat{t}_2 - \hat{t}_1)) \cup (\hat{j} \wedge (d\hat{t}_2 - d\hat{t}_1)) + H^\wedge_\cup(\hat{\omega},\hat{\omega}) \;,
\ee
with $\hat{\omega} := -\hat{j} \wedge (d\hat{t}_2 - d\hat{t}_1)$. The terms proportional to $\hat{a} \cup (\hat{j} \wedge (\hat{t}_2 - \hat{t}_1))$ and $\hat{h} \cup (\hat{j} \wedge (d\hat{t}_2 - d\hat{t}_1))$ combine into an exact term and can therefore be dropped.  The first term of \eqref{EqLagRepInv} becomes $\frac{1}{2} \hat{a} \cup (\hat{h} - \hat{\eta}_\Lambda)$ and the remaining terms combine up to an exact term into
\be
-\frac{1}{2} \hat{j} \wedge (\hat{t}_2 - \hat{t}_1) \wedge \hat{j} \wedge (d\hat{t}_2 - d\hat{t}_1) \;,
\ee
which vanishes because it involves $\hat{j} \wedge \hat{j}$. The Lagrangian is therefore equivalent to an E-cocycle of the form \eqref{EqPullBackELag}.  Proposition \ref{PropActVanishPullBackTorus} shows that  the action evaluated on the torus vanishes, hence that the exponentiated actions on the two cylinders yield the same homomorphism.
\end{proof}
Let us simplify the notation and define
\be
Q(x,j) := \mathcal{S}^e(M \times I,\check{y}_{\check{x},\hat{j}}) : \mathcal{S}^e(M, x) \rightarrow \mathcal{S}^e(M, x + j) \;.
\ee

We can repeat the construction above on a cylinder $M \times [0,\epsilon]$ for any $\epsilon \in \mathbb{R}_+$. The resulting operator is the same, as the theory is topological. We can therefore take the limit $\epsilon \rightarrow 0$ and the operators $Q(x,j)$ can be pictured as codimension 1 defects in the field theory. But we can do better. The proof of Proposition \ref{PropStructGM} shows that $H^{2\ell+1}_{\rm free}(M;\bar{\Lambda}^\ast) \simeq H^{2\ell+1}_{\rm free}(M;\mathbb{Z}) \otimes \Lambda^\ast_0$, where $\Lambda^\ast_0$ is the sublattice of $\Lambda^\ast$ left invariant by the monodromy representation defining the local system $\bar{\Lambda}$. Let us therefore consider classes of the form $j  = j' \otimes \lambda$, $j' \in H^{2\ell+1}_{\rm free}(M;\mathbb{Z})$, $\lambda \in \Lambda^\ast_0$. $j'$ is Poincaré dual to the fundamental homology class of a $2\ell+1$-dimensional submanifold $J \subset M$, so we can perform a similar limiting procedure and take the cocycle $\hat{j}$ to be supported on $J$. We can therefore picture $Q(x,j)$ as a $2\ell+1$-dimensional defect operator in the theory. Those are Wilson operators on $J$ with charge $\lambda$ for the degree $2\ell+1$ gauge field $\check{x}$.

\subsection{Refinement} 

\label{SecSkewSymRef}

We study here the groupoid generated by the operators $Q$ defined in the previous section. Define a $U(1)$-valued pairing $\psi$ on $H_{\rm free}^{2\ell+1}(M;\bar{\Lambda}^\ast)$ by
\be
\label{EqDefPsiPair}
\psi(j_1, j_2) := Q(x + j_1 + j_2, -j_1 - j_2) Q(x + j_1, j_2) Q(x, j_1) \;.
\ee
Remark that $\psi(j_1, j_2)$ is an automorphism of the Hermitian line $\mathcal{S}^e(M, x)$, hence can be canonically identified with an element of $U(1)$. 
\begin{proposition}
\label{PropPsiIndepx}
$\psi(j_1, j_2)$ is independent of $x$.
\end{proposition}
\begin{proof}
Consider the triangle
\be
\Delta := \{ (t_1,t_2) \in \mathbb{R}^2 | 0 \leq t_1 \leq 1, 0 \leq t_2 \leq t_1 \}
\ee
and endow $M \times \Delta$ with the differential cocycle 
\be
\check{x}_W = (\hat{a}_W, \hat{h}_W, \hat{\omega}_W) := \check{x} + (0, -\hat{j}_1 \wedge \hat{t}_1 - \hat{j}_2 \wedge \hat{t}_2, \hat{j}_1 \wedge d\hat{t}_1 + \hat{j}_2 \wedge d\hat{t}_2) \;.
\ee
We write $\hat{t}_i$ for the degree 0 $\mathbb{R}$-valued differential forms on $M \times \Delta$ given by the pullback of the coordinate functions $t_i$. We picked differential form representatives $\hat{j}_i$ of $j_i$. 

Remark that $\mathcal{S}^e(M \times \partial \Delta, \check{x}_W)$ computes the inverse of the right-hand side of \eqref{EqDefPsiPair}. We now use Proposition \ref{PropActFromBoundedMan} to get
\begin{align}
\label{EqExplFormPsi}
\psi(j_1, j_2) = & \: \exp \pi i \langle \hat{\omega}_W \wedge \hat{\omega}_W + \hat{\omega}_W \wedge \hat{\lambda}, [M \times \Delta, M \times \partial \Delta] \rangle \\
= & \: \exp \pi i \langle (\hat{j}_1 \wedge d\hat{t}_1 + \hat{j}_2 \wedge d\hat{t}_2) \wedge (\hat{j}_1 \wedge d\hat{t}_1 + \hat{j}_2 \wedge d\hat{t}_2), [M \times \Delta, M \times \partial \Delta] \rangle \notag \;,\\
= & \: \exp -\pi i \langle \hat{j}_1 \wedge \hat{j}_2 , [M] \rangle \notag \;.
\end{align}
which makes it obvious that $\psi(j_1, j_2)$ is independent of $x$.
\end{proof}
\begin{corollary}
\label{CorPropSkewRefPsi}
$\psi$ is bimultiplicative, i.e.
\be
\psi(j_1 + j_2, j_3) = \psi(j_1, j_3)\psi(j_2,j_3) \;,
\ee
\be
\psi(j_1, j_2 + j_3) = \psi(j_1, j_2)\psi(j_1,j_3) \;.
\ee
$\psi$ is skew-symmetric, i.e. 
\be
\psi(j_1, j_2) = (\psi(j_2, j_1))^{-1} \;.
\ee
It is also alternating, because it is defined on a free group: $\psi(j, j) = 1$. $\psi$ refines the bimultiplicative $U(1)$-valued skew-symmetric pairing on $H^{2\ell+1}_{\rm free}(M;\bar{\Lambda}^\ast)$:
\be
\label{EqRelPsiPairing}
\psi(j_1, j_2)(\psi(j_2, j_1))^{-1} = \exp -2\pi i \langle j_1 \wedge j_2, [M] \rangle \;.
\ee
\end{corollary}
The above definitions can be interpreted geometrically as follows. We can think of $Q$ as a connection on the line bundle $\mathcal{S}^e(M,\bullet)$ over the discrete set $x + \mathsf{G}_M$, in the sense that $Q$ defines a parallel transport between the fibers of $\mathcal{S}^e(M,\bullet)$. $\psi(j_1, j_2)$ describes the holonomies of this non-flat connection along elementary triangular paths, hence can bee seen as a discrete analogue of the curvature of the connection $Q$.

\subsection{A cobordism invariant}

\label{SecCobInv}

We now turn to the properties of the operators $Q(x,b)$ for $b \in H^{2\ell+1}_{\rm free}(M;\bar{\Lambda})$. In this case, $Q(x,b)$ maps the Hermitian line $\mathcal{S}^e(M,x)$ to itself, hence can be canonically identifed with a complex number.
\begin{proposition}
\label{PropQ2beq1}
$Q(x,\bullet)$, seen as a complex valued function on $H^{2\ell+1}_{\rm free}(M;\bar{\Lambda})$, passes to a function on $H^{2\ell+1}_{\rm free}(M;\bar{\Lambda})/2H^{2\ell+1}_{\rm free}(M;\bar{\Lambda})$.
\end{proposition}
\begin{proof}
Consider the following function on $I$:
\be
\theta(t) = \left \{ \begin{array}{ll} 0 & \mbox{ if } t < 1/2 \\
1 & \mbox{ if } t \geq 1/2 \;. \end{array} \right. \;, 
\ee
and let $\hat{\theta}$ be the corresponding $\mathbb{Z}$-valued 0-cochain. Let $b \in H^{2\ell+1}(M;\bar{\Lambda})$. The cocycle 
\be
\check{y}_{\check{x},2\hat{b}} = (\hat{a}, \hat{h} + 2\hat{b} \wedge \hat{t}, -2\hat{b} \wedge d\hat{t}) \;,
\ee
is equivalent to 
\be
\label{EqDefCocYThet}
\check{y} = (\hat{a} -  2 \hat{b} \cup d\hat{\theta}, \hat{h} + 2\hat{b} \wedge \hat{t}- 2\hat{b} \cup \hat{\theta}, -2\hat{b} \wedge d\hat{t}) \;,
\ee
where we picked a cocycle representative $\hat{b}$ for $b$. Proposition \ref{PropPQFTdDimWB} ensures that $Q(x,2b)$ can be computed by evaluating the exponentiated action on $M \times I$ endowed with the differential cocycle $\check{y}$. $\check{y}$ has the advantage that its restrictions to the boundaries of $M \times I$ coincide, so $Q(x,2b)$ can be computed by evaluating the exponentiated action on $M \times S^1$ endowed with $\check{y}$. The Lagrangian reads
\begin{align}
\label{EqDefLagTorus}
\hat{l}_{\check{x}, 2\hat{b}} := \: & \frac{1}{2} (\hat{a} - 2\hat{b} \cup d\hat{\theta}) \cup (\hat{h} + 2\hat{b} \wedge \hat{t}- 2\hat{b} \cup \hat{\theta} - \hat{\eta}_\Lambda) \\
& + \frac{1}{2} (\hat{h} + 2\hat{b} \wedge \hat{t} -  2\hat{b} \cup \hat{\theta})  \cup (-2\hat{b} \wedge d\hat{t}) + \frac{1}{2} H^\wedge_\cup(-2\hat{b} \wedge d\hat{t},-2\hat{b} \wedge d\hat{t}) \notag \;.
\end{align}
After simplifying using the Leibniz rule and dropping the exact terms, this expression reduces to
\begin{align}
\frac{1}{2} \hat{a} \cup (\hat{h} - \hat{\eta}_\Lambda) \: & - \hat{h} \cup (\hat{b} \cup d\hat{\theta}) 
- (\hat{b} \cup d\hat{\theta}) \cup \hat{h} + 2\hat{b} \cup d\hat{\theta} \cup \hat{b} \cup \hat{\theta}
+ (\hat{b} \cup d\hat{\theta}) \cup \hat{\eta}_\Lambda \\
& + 2(\hat{b} \wedge \hat{t}) \cup (\hat{b} \wedge d\hat{t}) + 2 H^\wedge_\cup(\hat{b} \wedge d\hat{t},\hat{b} \wedge d\hat{t}) \notag
\end{align}
The last four terms on the first line vanish modulo 1. The terms on the second line combine into $2\hat{b} \wedge \hat{t} \wedge \hat{b} \wedge d\hat{t}$, which vanishes by the skew-symmetry of the wedge product. The E-theory class determined by the Lagrangian has a representative of the form \eqref{EqPullBackELag}. Proposition \ref{PropActVanishPullBackTorus} ensures that the action vanishes. Therefore, $Q(x,2b) = 1$.
\end{proof}
\noindent \eqref{EqRelPsiPairing} shows that $\psi$ restricted to $H^{2\ell+1}_{\rm free}(M; \bar{\Lambda})/2H^{2\ell+1}_{\rm free}(M; \bar{\Lambda})$ is a symmetric pairing.
\begin{proposition}
\label{PropQQuadrRef}
$Q(x,\bullet)$ restricted to $H^{2\ell+1}_{\rm free}(M; \bar{\Lambda})/2H^{2\ell+1}_{\rm free}(M; \bar{\Lambda})$, is a quadratic refinement of $\psi$.
\end{proposition}
\begin{proof}
\eqref{EqDefPsiPair} shows that for fixed $x$, $Q(x,b) \in Hom(S^e(M,x), S^e(M,x)) \simeq \mathbb{C}$ is a refinement of $\psi$. We have $(Q(x,b))^{-2} = Q(x,2b)(Q(x,b))^{-2} = \psi(b,b)$, from which we deduce that $Q(x,\bullet)$ is quadratic, see Appendix \ref{SecPairFinAbGroups}. 
\end{proof} 
\noindent Note that as $\psi(b,b) = 1$, the above also shows
\begin{corollary}
$Q(x,b) = \pm 1$ for $b \in H^{2\ell+1}_{\rm free}(M; \bar{\Lambda})$.
\end{corollary}
\begin{proposition}
\label{PropQIndepJLambVal}
For $b \in H^{2\ell+1}_{\rm free}(M; \bar{\Lambda})$, $Q(x,b)$ depends only on the equivalence class of $x$ in $\mathsf{Y}^{\rm flat}_M/\mathsf{C}_M$.
\end{proposition}
\begin{proof}
\eqref{EqDefPsiPair} implies that for $j \in H^{2\ell+1}_{\rm free}(M; \bar{\Lambda}^\ast)$,
\begin{align}
\label{EqAlphIndepX2}
Q(x, j) Q(x,b) & \: = Q(x, j + b) \psi(b,j) \notag \\
& = \psi(b,j)/\psi(j,b) \: Q(x+j, b) Q(x,j) \\
&= Q(x+j, b) Q(x,j)\;, \notag
\end{align}
because $\psi(b,j)/\psi(j,b) = \exp 2\pi i \langle b \wedge j, [M] \rangle = 1$. $Q(x,b)$ and $Q(x+j, b)$ both act by multiplication by a complex number, so \eqref{EqAlphIndepX2} implies that they are equal. As the action of $H^{2\ell+1}_{\rm free}(M; \bar{\Lambda}^\ast)$ on $\mathsf{Y}^{\rm flat}_M$ has the same orbits as the action of $\mathsf{C}_M$, $Q(x,b)$ depends only on the equivalence class of $x$ in $\mathsf{Y}^{\rm flat}_M/\mathsf{C}_M$.
\end{proof}
\noindent At the expense of a less elegant proof, we can show
\begin{proposition}
\label{PropQIndepJLambVal2}
$Q(x,b)$ depends only on the equivalence class of $x$ in $\mathsf{Y}^{\rm flat}_M/\mathsf{E}_M$.
\end{proposition}
\begin{proof}
Suppose $w \in \mathsf{E}_M$, with a lifting cocycle $\check{w} = (\hat{c}, \hat{j}, 0) \in \mathsf{F}_M$, $d\hat{j} = -\hat{c}$ and $\hat{j}$ valued in $\bar{\Lambda}^\ast$. Using the arguments of the proof of Proposition \ref{PropQ2beq1} and the notation \eqref{EqDefLagTorus}, we can compute $Q(x + w, b)$ and $Q(x, b)$ by evaluating the actions associated respectively to the Lagrangians $\hat{l}_{\check{x} + \check{w}, \hat{b}}$ and $\hat{l}_{\check{x}, \hat{b}}$ on $M \times S^1$.

We have
\be
\hat{l}_{\check{x} + \check{w}, \hat{b}} = \hat{l}_{\check{x}, \hat{b}} + \Delta \hat{l} \;,
\ee
with (up to exact terms)
\begin{align}
\Delta \hat{l} \: & = \frac{1}{2} \hat{c} \cup (\hat{j} - \hat{\eta}_\Lambda) + \frac{1}{2} (\hat{a} - \hat{b} \cup d\hat{\theta}) \cup \hat{j} + \frac{1}{2} \hat{j} \cup (\hat{a} - \hat{b} \cup d\hat{\theta}) \\
& = \frac{1}{2} \hat{c} \cup (\hat{j} - \hat{\eta}_\Lambda) + \frac{1}{2} \hat{c} \cup_1 (\hat{a} - \hat{b} \cup d\hat{\theta}) \;.
\end{align}
Writing 
\be
\bar{l}_{\check{x}, \hat{b}} = \left( \hat{l}_{\check{x}, \hat{b}}, (\hat{a} - \hat{b} \cup d\hat{\theta})_2 \right) \;, \quad \bar{l}_{\check{x} + \check{w}, \hat{b}} = \left( \hat{l}_{\check{x} + \check{w}, \hat{b}}, (\hat{a} + \hat{c} - \hat{b} \wedge d\hat{\theta})_2 \right) \;, 
\ee
we have
\be
\bar{l}_{\check{x} + \check{w}, \hat{b}} = \bar{l}_{\check{x}, \hat{b}} \boxplus \left(\frac{1}{2} \hat{c} \cup (\hat{j} - \hat{\eta}_\Lambda),\hat{c}_2 \right) \;. 
\ee
The second term on the right hand side is pulled back from $M$, so Proposition \ref{PropActVanishPullBackTorus} ensures that the actions constructed from $\bar{l}_{\check{x} + \check{w}, \hat{b}}$ and $\bar{l}_{\check{x}, \hat{b}}$ are equal, proving the proposition.
\end{proof}
\noindent Recall that the image of $\beta(x)$ in $H^{2\ell+2}_{\rm tors}(M;\bar{\Lambda})/\mathsf{K}_M$ is fixed by the requirement that $V_{M,x}$ is a non-trivial vector space. With this requirement, $x$ should be chosen in a torsor for the group $H^{2\ell+1}(M;\mathscr{V})/H^{2\ell+1}(M;\bar{\Lambda}^\ast)$, which is connected. As $Q(x,b) \in \{\pm 1\}$ and is continuous in $x$, we obtain
\begin{proposition}
\label{PropQxbIndepX}
$Q(x,b)$ is independent of $x$ for $b \in H^{2\ell+1}_{\rm free}(M; \bar{\Lambda})$.
\end{proposition}

We write $\mathfrak{q}_M$ for the multiplicative quadratic refinement of $\psi$ defined by $Q$ on $H^{2\ell+1}_{\rm free}(M; \bar{\Lambda})/2H^{2\ell+1}_{\rm free}(M; \bar{\Lambda})$. Let us write $\mathsf{G}^{(2)}_M$ for the quotient of $H^{2\ell+1}_{\rm free}(M; \bar{\Lambda})/2H^{2\ell+1}_{\rm free}(M; \bar{\Lambda})$ by the radical $\mathsf{R}^{(2)}_M$ of the pairing $\psi$. 
\begin{proposition}
\label{PropqMCobTame}
$\mathfrak{q}_M$ is a tame quadratic refinement, i.e. it vanishes on $\mathsf{R}^{(2)}_M$.
\end{proposition}
\begin{proof}
If $b \in \mathsf{R}^{(2)}_M$, its lifts to $H^{2\ell+1}_{\rm free}(M; \bar{\Lambda})$ have even pairing with any element of $H^{2\ell+1}_{\rm free}(M; \bar{\Lambda})$. Replacing $2b$ by $b$ in the proof of Proposition \ref{PropQ2beq1}, we find that $Q(x,b) = 1$.
\end{proof}

Proposition \ref{PropqMCobTame} ensures $\mathfrak{q}_M$ has a well-defined Arf invariant $\mathfrak{a}_M \in \{0,1\}$. We will refer to $\mathfrak{a}_M \in \{0,1\}$ as the Arf invariant of the $4\ell+2$-dimensional $\mathfrak{W}$-manifold $M$. 
\begin{proposition}
\label{PropArfInvCobInv}
Assume there exists a $4\ell+3$-dimensional $\mathfrak{W}$-manifold $W$ such that $\partial W = M$ as $\mathfrak{W}$-manifolds. Then $\mathfrak{a}_M = 0$. In other words, $\mathfrak{a}_M$ is a cobordism invariant of $4\ell+2$-dimensional $\mathfrak{W}$-manifolds.
\end{proposition}
\begin{proof}
As $\mathfrak{q}_{M}$ is tame, it reduces to a quadratic refinement $\mathfrak{q}'_{M}$ on $\mathsf{G}^{(2)}_M$, on which $\psi$ is non-degenerate. A symmetric $\mathbb{Z}_2$-valued pairing is also skew-symmetric, so $\mathsf{G}^{(2)}_M$ admits Lagrangian subgroups. As explained in Appendix \ref{SecPairFinAbGroups}, such a quadratic refinement has Arf invariant zero if and only if there is a Lagrangian subgroup on which it vanishes identically. Proposition \ref{PropLMLagSubMod2} shows that the elements of $\mathsf{G}^{(2)}_M$ that extend to $W$ form a Lagrangian subgroup $\mathsf{L}^{(2)}_W$.

Let $b \in \mathsf{L}^{(2)}_W$. Following the argument in the proof of Proposition \ref{PropQ2beq1}, we can compute $\mathfrak{q}'_M(b) = Q(x,b)$ by evaluating the exponentiated action on $M \times S^1$ endowed with the differential cocycle
\be
\check{y} = (\hat{a} -  \hat{b} \cup d\hat{\theta}, \hat{h} + \hat{b} \wedge \hat{t} - \hat{b} \cup \hat{\theta}), -\hat{b} \wedge d\hat{t})
\ee
Writing $\hat{c}$ for the extension of $\hat{b}$ to $W$, $\check{y}$ extends to a differential cocycle on $W \times S^1$ with curvature given by $-\hat{c} \wedge d\hat{t}$. We can therefore evaluate the action using Proposition \ref{PropActFromBoundedMan}. As $W \times S^1$ is a direct product, its Wu class vanishes, so we can take $\lambda = 0$. Moreover $(\hat{c} \wedge d\hat{t}) \wedge (-\hat{c} \wedge d\hat{t}) = 0$, so the first term in \eqref{EqACtFromBoundedMan} vanishes as well. The exponentiated action is therefore equal to 1. This shows that $\mathfrak{q}'_M$ vanishes on $\mathsf{L}^{(2)}_W$ and that $M$ has Arf invariant $0$.
\end{proof}

Consider a 2-torus with a distinguished basis of its first homology. Pick the non-bounding spin structure along the generators. The resulting spin structure is actually independent of the choice of basis, and the corresponding Arf invariant is $1$. This example shows that the cobordism invariant defined above can be non-trivial. It also shows that the cobordism invariant $\mathfrak{a}_M$ is a generalization of the Arf invariant of spin structures \cite{Atiyah1971}.

\begin{proposition}
\label{PropChgWuStructCobInvVan}
Suppose there is a $4\ell+2$-dimensional $\mathfrak{W}$-manifold $M$ with an admissible Wu structure such that $\mathfrak{a}_M = 1$. One can always change the Wu structure to an admissible one such that $\mathfrak{a}_M = 0$. 
\end{proposition}
\begin{proof}
Pick a Lagrangian subgroup $\mathsf{L}$ of $\mathsf{G}^{(2)}_M$. $\mathfrak{q}'_M|_\mathsf{L}$ is linear, hence coincides with a character $\chi_\mathsf{L} = ( \bullet, x'_{\mathsf{L}})_{\mathsf{G}^{(2)}_M}$ for some $x'_{\mathsf{L}} \in \mathsf{G}^{(2)}_M$. $x'_{\mathsf{L}}$ can be lifted to an element of $H^{2\ell+1}(M;\bar{\Lambda})$ and then reduced mod $2\Lambda$ to an element $x_{\mathsf{L}}$ of $H^{2\ell+1}(M;\bar{\Gamma}^{(2)})$. We can use $x_{\mathsf{L}}$ to shift the Wu structure of $M$. 

The quadratic refinement associated to the shifted Wu structure is $\mathfrak{q}_M - \chi_\mathsf{L}$. As it vanishes on $\mathsf{L}$, it has Arf invariant zero.

It remains to check that the Wu structure obtained in this way is admissible. Recall that admissible Wu structures differ by elements of the group $\mathsf{W}_N$, defined above Proposition \ref{PropDefAdmWuStruct}. $\mathsf{W}_N$ contains in particular the elements of $H^{2\ell+1}(M;\bar{\Gamma}^{(2)})$ orthogonal to the image of the torsion subgroup $H^{2\ell+1}_{\rm tors}(M;\bar{\Lambda})$ into $H^{2\ell+1}(M;\bar{\Gamma}^{(2)})$. From the construction above, it is clear that $x_{\mathsf{L}}$ belongs to $\mathsf{W}_M$. As the original Wu structure is admissible by assumption, the shifted Wu structure is admissible as well.
\end{proof}
\noindent We will see in Section \ref{SecStateSpace} that the state space of the gauged theory can be constructed only on $4\ell+2$-dimensional manifolds with vanishing Arf invariant. By the proposition above, this is a constraint on the Wu structure, not on the underlying manifold.

\subsection{Push-down}

\label{SecPushDown}

In the discussion so far, we have defined operators $Q(x,j)$ on $V_{M,x}$, with $j \in H^{2\ell+1}_{\rm free}(M; \Lambda^\ast)$. Let $\mathsf{G}_{M,2\Lambda} := H^{2\ell+1}_{\rm free}(M; \Lambda^\ast)/2H^{2\ell+1}_{\rm free}(M; \Lambda)$. We have
\begin{lemma}
$Q(x,j) = Q(x,j+2b)$, for $b \in H^{2\ell+1}_{\rm free}(M; \Lambda)$, hence $Q(x,j)$ is well-defined for $j \in \mathsf{G}_{M,2\Lambda}$.
\end{lemma}
\begin{proof}
This follows immediately from the discussions in the previous section, for instance from the first line of \eqref{EqAlphIndepX2} after substituting $2b$ for $b$ there.
\end{proof}
\noindent The aim of the present section is to define analogous operators parametrized by $k \in \mathsf{G}_M := H^{2\ell+1}_{\rm free}(M; \Lambda^\ast)/H^{2\ell+1}_{\rm free}(M; \Lambda)$. 

Let $t \in \mathbb{C}$, $j \in \mathsf{G}_{M,2\Lambda}$  and consider the operators $tQ(x,j)$. For fixed $x$ and $j$ they form a copy of the trivial Hermitian line $\mathbb{C}$, which we write $\mathbb{C}_{(x,j)}$. Suppose that $j_2 = j_1 + b$, for $b \in H^{2\ell+1}_{\rm free}(M; \bar{\Lambda})/2H^{2\ell+1}_{\rm free}(M; \Lambda)$. $Q(x, j_1)$ and $Q(x, j_2)$ have the same domain and target, so $Q(x, j_2)(Q(x, j_1))^{-1} =: \alpha(x, j_1, b) \in \mathbb{C}$. This provides an identification of the Hermitian lines $\mathbb{C}_{(x,j_1)}$ and $\mathbb{C}_{(x,j_2)}$. Identifying all the trivial Hermitian lines $\mathbb{C}_{(x,j+b)}$, for $b \in H^{2\ell+1}_{\rm free}(M; \bar{\Lambda})/2H^{2\ell+1}_{\rm free}(M; \Lambda)$, we obtain a Hermitian line $L_{(x,k)}$, where $k \in \mathsf{G}_M$ is the reduction of $j$ modulo $\Lambda$. Remark that the lines $L_{(x,0)}$ are canonically trivialized, as there is a preferred isomorphism to $\mathbb{C}_{(x,0)}$.

Given $l \in L_{(x,k)}$, we obtain an operator
\be
P(x, k, l) := t Q(x,j) \;,
\ee
where $t$ is the element of $\mathbb{C}_{(x,j)}$ corresponding to $l$. One can check that the definition of $P$ is independent of the choice of lift $j$.

\begin{proposition}
\label{PropHeisGrpLineIndepx}
$L_{(x,k)}$ is independent of $x$.
\end{proposition}
\begin{proof}
Concretely, we need to check that $\alpha(x, j, b)$ is independent of $x$. To show that, remark that \eqref{EqDefPsiPair} implies that
\be
\label{EqAlphIndepX}
\alpha(x, j, b) = Q(x, j+b)(Q(x, j))^{-1} = Q(x+j, b) \psi^{-1}(j,b) \;.
\ee
Proposition \ref{PropQxbIndepX} ensures that $Q(x+j, b)$ is independent of $x$, hence the right-hand side of \eqref{EqAlphIndepX} is independent of $x$. 
\end{proof}

From now on we write $L_{k}$ for $L_{(x,k)}$. Remark that the Hermitian lines $L_{k}$ are trivialized by a choice of lift $j: \mathsf{G}_M \rightarrow \mathsf{G}_{M,2\Lambda}$: $L_{k} \simeq \mathbb{C}_{(x,j(k))}$. Such a lift is in general not a homomorphism and given $k_1, k_2 \in \mathsf{G}_M$, there is a unique $b(k_1, k_2) \in H^{2\ell+1}_{\rm free}(M; \bar{\Lambda})/2H^{2\ell+1}_{\rm free}(M; \Lambda)$ such that 
\be
\label{EqDefCocb}
j(k_1) + j(k_2) = j(k_1 + k_2) + b(k_1, k_2) \;.
\ee
Define
\be
\label{EqDefPhi}
\phi(k_1, k_2) = \big(Q(x, j(k_1 + k_2))\big)^{-1} Q(x + k_1, j(k_2)) Q(x, j(k_1)) \;.
\ee
Although we do not make this fact explicit in the notation, $\phi$ depends on the choice of trivialization $j$.
\begin{proposition}
$\phi$ is a 2-cocycle on $\mathsf{G}_M$ valued in $U(1)$, i.e.
\be
\label{EqPhiCocycle}
\phi(k_1, k_2) \phi(k_1 + k_2, k_3) = \phi(k_1, k_2 + k_3) \phi(k_2, k_3) \;.
\ee
Moreover $\phi$ is independent of $x$, it is skew-symmetric
\be
\phi(k_2, k_1)= \big(\phi(k_1,k_2)\big)^{-1}\;,
\ee
and 
\be
\label{EqPhiSkewRef}
\phi(k_1, k_2)\big(\phi(k_2,k_1)\big)^{-1} = \exp -2\pi i \langle k_1 \cup k_2, [M] \rangle \;.
\ee
\end{proposition}
\begin{proof}
Remark that
\be
\label{EqProofPhiCoc1}
\phi(k_1, k_2) = \psi(j(k_1), j(k_2)) \:  \frac{Q(x, j(k_1) + j(k_2))}{Q(x, j(k_1 + k_2))} \;,
\ee
\eqref{EqDefPsiPair} implies
\be
\label{EqProofPhiCoc2}
\frac{Q(x, j(k_1) + j(k_2))}{Q(x, j(k_1 + k_2))} = \psi(b(k_1, k_2), j(k_1) + j(k_2)) \: \left( Q(x, b(k_1, k_2)) \right)^{-1} \;,
\ee
which shows that the second factor in \eqref{EqProofPhiCoc1} is a sign. We deduce that $\phi$ is skew-symmetric because $\psi$ is. \eqref{EqPhiSkewRef} is a direct consequence of \eqref{EqProofPhiCoc1}. \eqref{EqProofPhiCoc2} also show that $\phi$ is independent of $x$.

To check the cocycle condition \eqref{EqPhiCocycle}, we need to simplify the notation a bit. We write $j_a := j(k_a)$, $b_{a|b} := b(k_a, k_b)$, $b_{ab|c} := b(k_a + k_b, k_c)$, $a,b,c \in 1,2,3$. We also write $Q^{-1}(b_{a|b}) := (Q(x,b(k_a, k_b))^{-1}$. We can now compute each of the factors in \eqref{EqPhiCocycle}, using \eqref{EqProofPhiCoc1}, \eqref{EqProofPhiCoc2} and the relation $j(k_a + k_b) = j_a + j_b + b_{a|b}$.
\begin{align}
\phi(k_1, k_2) = \, & \psi(j_1, j_2) \psi(b_{1|2}, j_1 + j_2) Q^{-1}(b_{1|2}) \;,\notag \\
\phi(k_1 + k_2, k_3)  = \, & \psi(j_1, j_3) \psi(j_2, j_3) \psi(b_{1|2}, j_3) \psi(b_{12|3}, j_1) \psi(b_{12|3}, j_2)\notag \\
& \psi(b_{12|3}, b_{1|2}) \psi(b_{12|3}, j_3) Q^{-1}(b_{12|3}) \;, \\
\phi(k_1, k_2 + k_3)  = \, & \psi(j_1, j_2) \psi(j_1, j_3) \psi(j_1, b_{2|3}) \psi(b_{1|23}, j_1) \psi(b_{1|23}, j_2) \notag \\
& \psi(b_{1|23}, j_3) \psi(b_{1|23}, b_{2|3}) Q^{-1}(b_{1|23}) \;, \notag \\
\phi(k_2, k_3) = \, & \psi(j_2, j_3) \psi(b_{2|3}, j_2 + j_3) Q^{-1}(b_{2|3}) \;,\notag
\end{align}
Plugging the expressions above into \eqref{EqPhiCocycle}, all the factors of the form $\psi(j_a, j_b)$ cancel. Using \eqref{EqDefPsiPair} and the fact that $\psi(j_1, b_{2|3}) = \pm 1$, hence $\psi(j_1, b_{2|3}) = \psi(b_{2|3}, j_1)$, we can gather factors and obtain
\be
\psi(b_{1|2} + b_{12|3}, j_1 + j_2 + j_3) Q^{-1}(b_{1|2} + b_{12|3}) \stackrel{?}{=} \psi(b_{1|23} + b_{2|3}, j_1 + j_2 + j_3) Q^{-1}(b_{1|23} + b_{2|3}) \;.
\ee
This equality is true because $b_{1|2} + b_{12|3} = b_{1|23} + b_{2|3}$, as is obvious from the definition of $b$ in \eqref{EqDefCocb}.
\end{proof}

\subsection{The Heisenberg group and its representations}

\label{SecHeisGroupAndRep}

We now describe the canonical representation of the Heisenberg group associated to $\mathsf{G}_M$ on $V_{M,w}$. In Appendix \ref{AppHeisGr2Gr}, we review some basic facts about Heisenberg groups. 

There is a bimultiplicative pairing on $\mathsf{G}_M$ given by
\be
\label{DefBimultPairG}
B(k_1, k_2) := \exp -2\pi i \langle \hat{k}_1 \wedge \hat{k}_2, [M] \rangle \;,
\ee
where $\hat{k}_1$ and $\hat{k}_2$ are differential form representatives of $k_1, k_2 \in \mathsf{G}_M$. $B$ is independent of the choice of representatives. We write $\mathsf{H}_M$ for the corresponding Heisenberg group. Given a lift $j$ as in Section \ref{SecPushDown}, the associated cocycle $\phi$ provides an explicit realization of $\mathsf{H}_M$ on the set of pairs $\left(k, t \right) \in \mathsf{G}_M \times U(1)$, with the multiplication law given by
\be
\label{EqExplRealHeisGrp}
\left(k_1, t_1\right) \cdot \left(k_2, t_2\right) = \left(k_1 + k_2, \phi(k_1, k_2)t_1t_2 \right) \;.
\ee
This realization depends on the lift $j$ through $\phi$.

Given $w \in \mathsf{Y}^{\rm flat}_M$, we can obtain a canonical realization of $\mathsf{H}_M$ as follows.  Let $k \in \mathsf{G}_M$ and $l \in L_{k}$. Let us define
\be
\label{EqDefRho}
\rho(k,l) = \bigoplus_{x \in w + \mathsf{G}_M} P(x,k,l): V_{M,w} \rightarrow V_{M,w} \;.
\ee
\be
\label{EqDefRhoT}
\tilde{\rho}(k,l) =  \bigoplus_{x \in w + \mathsf{G}_M} B(k,x - w) P(x,k,l): V_{M,w} \rightarrow V_{M,w} \;.
\ee
Remark that $x - w \in \mathsf{G}_M$ so $B(k,x - w)$ is well-defined.
\begin{proposition}
\label{PropVRegRepHeisGrp}
$V_{M,w}$ carries the regular representation of $\mathsf{H}_M$, with the right regular representation given by $\rho$ and the left regular representation given by $\tilde{\rho}$. 
\end{proposition}
\begin{proof}
In this proof, all the sums over $x$ are understood to run over $w + \mathsf{G}_M$. Pick a lift $j$ of $k$ as in Section \ref{SecPushDown} and let $\phi$ be the associated 2-cocycle. Then each $l \in L_{k}$ can be canonically identified with a complex number. In particular, there is an element $l(k) \in L_{k}$ associated to $1 \in \mathbb{C}_{(x,j(k))}$. Then $P(x,k,l) = l/l(k) Q(x,j(k))$, where $l/l(k) \in \mathbb{C}$.

We can now check that $\rho$ is a right action of $\mathsf{H}_M$:
\begin{align}
\rho(k_2, l_2) \circ \rho(k_1, l_1) \: & = t_1 t_2 \bigoplus_x Q(x + k_1, j(k_2)) Q(x,j(k_1)) \notag \\
& = t_1 t_2 \bigoplus_x \phi(k_1, k_2) Q(x, j(k_1 + k_2)) \\
& = \rho\big( (k_1, l_1) \cdot (k_2, l_2) \big) \notag\;,
\end{align}
where $t_i = l_i/l(k_i)$.
$\tilde{\rho}$ is a left action of $\mathsf{H}_M$:
\begin{align}
\tilde{\rho}(k_2, l_2) \circ \tilde{\rho}(k_1, l_1) \: &  = t_1 t_2 \bigoplus_x B(k_2,x - w + k_1) Q(x + k_1, j(k_2)) B(k_1, x - w) Q(x,j(k_1)) \notag \\
& = t_1 t_2 B(k_2, k_1) \bigoplus_x B(k_2, x - w)B(k_1, x - w)\phi(k_1, k_2) Q(x, j(k_1 + k_2)) \notag \\
& = t_1 t_2 \bigoplus_x B(k_1 + k_2, x - w)\phi(k_2, k_1) Q(x, j(k_1 + k_2)) \\
& = \tilde{\rho} \big( (k_2, l_2) \cdot (k_1, l_1) \big) \;. \notag
\end{align}
Moreover, the two actions commute:
\begin{align}
\tilde{\rho}(k_2, l_2) \circ \rho(k_1, l_1) & \: = t_1 t_2 \bigoplus_x B(k_2,x - w + k_1) Q(x + k_1, j(k_2)) Q(x,j(k_1)) \notag \\
& = t_1 t_2 \bigoplus_x B(k_2, x - w)\phi(k_2, k_1) Q(x, j(k_1 + k_2)) \\
& = t_1 t_2 \bigoplus_x B(k_2, x - w) Q(x + k_2, j(k_1)) Q(x,j(k_2)) \\
& = \rho(k_1, l_1) \circ \tilde{\rho}(k_2, l_2) \;.
\end{align}
The dimension of $V_{M,w}$ is equal to $|\mathsf{G}_M|$ by Proposition \ref{PropDimHilbSp}, so $V_{M,w}$ coincides indeed with the regular representation of $\mathsf{H}_M$.
\end{proof}
Note that the elements of norm 1 of a Hermitian line is a $U(1)$-torsor. Let us write $L_k^{(1)}$ for the $U(1)$-torsor determined by $L_k$. 
\begin{corollary}
$\bigcup_{k \in \mathsf{G}_M} L_{k}^{(1)}$ endowed with the product given by the composition of the morphisms $\rho(k,l)$ is a group isomorphic to $\mathsf{H}_M$.
\end{corollary}
This is the canonical realization of the Heisenberg group associated to the $(\mathfrak{W},\mathsf{Z}_{\rm flat})$-manifold $(M,\check{x})$.

\section{State space}

\label{SecStateSpace}

Let $(M,x)$ be a closed $4\ell+2$-dimensional $(\mathfrak{W},\mathsf{Z}^{\rm flat})$-manifold. Up to isomorphism, $\mathsf{H}_M$ has a unique irreducible module $\mathsf{I}_M$ (see Appendix \ref{AppHeisGr2Gr}).  Informally, the state space $\mathcal{R}(M)$ of the topological field theory $\mathcal{R}$ is the Hilbert space given by the direct sum of $|\mathsf{K}_M|$ copies of $\mathsf{I}_M$. However, characterizing $\mathcal{R}(M)$ up to isomorphisms is not sufficient, we need to construct explicitly a canonical Hilbert space from $M$. 

Recall that $\Lambda_0$ is the sublattice of $\Lambda$ left invariant by the monodromy representation defining $\bar{\Lambda}$, and that $\Gamma_0 := \Lambda^\ast_0/\Lambda_0$. As we will see, when $|\Gamma_0|$ is even, a canonical construction of the Hilbert space is impossible without picking an extra structure on $M$, because of the existence of a Hamiltonian anomaly. The latter is the counterpart on $4\ell+2$-dimensional manifolds of the partition function anomaly on $4\ell+3$-dimensional manifolds described in Section \ref{SecPartFuncAnom}, in a sense that will be made precise in Section \ref{SecProofGluinFormAnomCase}. 

In the present section, we will choose the relevant extra structure on $M$ and show how to construct a Hilbert space given this choice. In Section \ref{SecHamAn}, we will describe the state space in a canonical way, as an object in a category linearly equivalent to the category of Hilbert spaces, rather than as a Hilbert space.

\subsection{Strategy}

\label{SecStrat}

Recall Proposition \ref{PropStructGM}, stating that $\mathsf{G}_{M} \simeq H^{2\ell+1}_{\rm free}(M;\mathbb{Z}) \otimes \Gamma_0$. The skew-symmetric pairing on $\mathsf{G}_{M}$ is the tensor product of the skew-symmetric pairing on $H^{2\ell+1}_{\rm free}(M;\mathbb{Z})$ with the symmetric pairing on $\Gamma_0$. All these pairings are non-degenerate. Whenever we consider Lagrangian subgroups of $\mathsf{G}_{M}$, we always mean subgroups of the form $\mathsf{L} = \mathsf{L'} \otimes \Gamma_0$, where $\mathsf{L}'$ is a Lagrangian subgroup of $H^{2\ell+1}_{\rm free}(M;\mathbb{Z})$.

Explicit models for $\mathsf{I}_M$ can be obtained as follows. The Heisenberg group $\mathsf{H}_M$ admits finite commutative subgroups lifting the Lagrangian subgroups of $\mathsf{G}_M$. We will call such subgroups \emph{LL subgroups}. Given a LL subgroup $\tilde{\mathsf{L}} \subset \mathsf{H}_M$, the subspace of vectors in $V_{M,w}$ invariant under $\rho(\tilde{\mathsf{L}})$ form an irreducible $\mathsf{H}_M$-module $\mathsf{I}(\tilde{\mathsf{L}})$ under the action $\tilde{\rho}(\mathsf{H}_M)$.

There is no canonical way to pick a preferred LL subgroup $\tilde{\mathsf{L}}$. The plan is therefore to use the invariant section construction, which we already encountered in the construction of the prequantum theory, to construct a canonical irreducible module. When $|\Gamma_0|$ is even, the existence of invariant sections requires us to pick extra structures on $M$ restricting the set of Lagrangians considered.

This construction can be understood as a discrete analogue of geometric quantization \cite{MR1183739}. The LL subgroups play the role of the polarizations. We are constructing a flat parallel transport on a bundle of Hilbert spaces over the (discrete) space of polarizations. The space of flat sections forms the canonical Hilbert space that we aim to construct. We will also see an analogue of the metaplectic correction of geometric quantization.

The results in this section are inspired by the work of Gurevich and Hadani \cite{Gurevich2009, 2008arXiv0808.1664G} on the quantization of symplectic vector spaces over finite fields. Their results are extended in two directions. First, we need to generalize vector spaces over finite fields to free modules over finite rings of the form $\mathbb{Z}_r$, $r = p^m$ for $p$ a prime and $m$ a positive integer. This introduces some complications that can be dealt with using the results of Taylor in \cite{Taylor}. Second, in characteristic 2, we obtain a global trivialization of the bundle of Heisenberg modules (the "strong Von Neumann property" in the language of \cite{Gurevich2009, 2008arXiv0808.1664G}), at the expense of restricting ourselves to a subclass of LL subgroups. In the present context, this subclass is ultimately determined by the Wu structure of $M$, as well as the extra structure to be described.

\subsection{Extra structure for the construction of the Hilbert space}

\label{SecExStruct}

The extra structure required for the construction of the Hilbert space associated to $M$ consists of two pieces. 

First, we need a lift $j: \mathsf{G}_M \rightarrow \mathsf{G}_{M,2\Lambda} := H^{2\ell+1}_{\rm free}(M; \bar{\Lambda}^\ast)/2H^{2\ell+1}_{\rm free}(M; \bar{\Lambda})$. We consider as equivalent any two choices of lifts $j_1$ and $j_2$ such that $Q(x,j_1(k)) = Q(x,j_2(k))$ for all $k \in \mathsf{G}_M$.
In general, such lifts cannot be homomorphisms, but if the order of $\Gamma_0$ is odd, there is a canonical choice of lift, which is a homomorphism. To construct it, we compose the division by 2 in $\Gamma_0$, which is well-defined in odd order abelian groups, with the map $\Gamma_0 \rightarrow \Gamma_{0,2\Lambda}$ given by multiplication by two. In the odd case, we will always choose this canonical lift.

The second extra structure needed to construct the state space is non-trivial only if the order of $\mathsf{G}_M$, hence of $\Gamma_0$, is even. Using the classification theorem for finite abelian groups, we can write 
\be
\Gamma_0 = \bigoplus_r \Gamma_{0,r} \;, \quad \Gamma_{0,r} := \mathbb{Z}_r^{m_r} \;,
\ee
where $r$ runs over the power of primes and the $m_r$ are positive integers. We deduce from Proposition \ref{PropStructGM} that
\be
\label{EqDecompGM}
\mathsf{G}_M = \bigoplus_r \mathsf{G}_{M,r} \;, \quad \mathsf{G}_{M,r} := \mathbb{Z}_r^{2n_r} \;,
\ee
for some positive integers $n_r := m_r m$, where $m$ is half the number of generators in $H^{2\ell+1}_{\rm free}(M;\mathbb{Z})$. The pairing $B$ on $\mathsf{G}_M$ decomposes into pairings $B_r$ on $\mathsf{G}_{M,r}$. For $r$ even, $(B_r)^{r/2}$ is a $\{\pm 1\}$-valued bimultiplicative pairing on $\mathsf{G}_{M,r}/2\mathsf{G}_{M,r}$. The extra structure is a choice of decomposition \eqref{EqDecompGM}, together with a (multiplicative) quadratic refinement $q_r$ of $(B_r)^{r/2}$ of Arf invariant zero, for each even $r$ appearing in the decomposition.

We will write $\epsilon := (j, \{q_r\})$ for the choice of extra structure. A $4\ell+2$-dimensional $\mathfrak{W}$-manifold endowed with such an extra structure will be called a $(\mathfrak{W}, \epsilon)$-manifold. Remark that effectively, no extra structure is needed to construct the Hilbert space of the theory on $M$ when $\mathsf{G}_M$, hence $\Gamma_0$, has odd order. In particular, $\Gamma_0$ has odd order if $\Gamma$ has odd order, which is a property of the field theory, valid for any local system on any $4\ell+2$-dimensional manifold $M$.

\subsection{Admissible Lagrangian subgroups}

\label{SecAdmLagSubg}

$\psi$ is a perfect bimultiplicative pairing on $\mathsf{G}_{M,2\Lambda}$ and it is in fact a square root of the pairing $B$ on $\mathsf{G}_M$, as shown by \eqref{EqRelPsiPairing}. By the proof of Proposition \ref{PropStructGM}, we can also write $\mathsf{G}_{M,2\Lambda} \simeq H^{2\ell+1}_{\rm free}(M;\mathbb{Z}) \otimes \Gamma_{0,2\Lambda}$, where $\Gamma_{0,2\Lambda} := \Lambda^\ast_0/2\Lambda_0$. When speaking of Lagrangian subgroups of $\mathsf{G}_{M,2\Lambda}$, we will always understand subgroups of the form $\mathsf{L}' \otimes \Gamma_{0,2\Lambda}$, where $\mathsf{L}'$ is a Lagrangian subgroup of $H^{2\ell+1}_{\rm free}(M;\mathbb{Z})$. If $\mathsf{L}_{2\Lambda}$ is such a Lagrangian subgroup, then its projection to $\mathsf{G}_M$ is a Lagrangian subgroup $\mathsf{L}$. Suppose we picked the extra structure $\epsilon = (j, \{q_r\})$ on $M$. Then $\mathsf{L}$ decomposes into a direct sum of Lagrangian subgroups $\mathsf{L}_r \subset \mathsf{G}_{M,r}$.

\begin{definition}
The pair $(\mathsf{L}_{2\Lambda}, \mathsf{L})$ is called \emph{admissible} with respect to $\epsilon$ if the following conditions are satisfied.
\begin{enumerate}
\item $j(\mathsf{L}) \subset \mathsf{L}_{2\Lambda}$.
\item If $b \in \mathsf{L}_{2\Lambda} \cap H^{2\ell+1}_{\rm free}(M;\Lambda)/2H^{2\ell+1}_{\rm free}(M;\Lambda)$, then $\mathfrak{q}_M(b) = 1$, where $\mathfrak{q}_M$ is the multiplicative quadratic refinement of the cup product pairing mod 2 defined in Section \ref{SecCobInv}.
\item $q_r|_{\mathsf{L}_r/2\mathsf{G}_{M,r}} = 1$.
\end{enumerate}
\end{definition}
We will also call admissible those Lagrangians $\mathsf{L} \subset \mathsf{G}_M$ fitting into an admissible pair $(\mathsf{L}_{2\Lambda}, \mathsf{L})$. Note that given $\mathsf{L}$ we can always choose a lift $j$ such that the first admissibility condition is satisfied. Thanks to the Arf invariant zero constraint on $q_r$, there are always Lagrangian subgroups satisfying the third admissibility condition. The second admissibility condition can be satisfied if the Arf invariant $\mathfrak{a}_M$ vanishes. Proposition \ref{PropChgWuStructCobInvVan} then shows that given any Lagrangian subgroup $\mathsf{L}_{2\Lambda}$, one can always pick a Wu structure so that $\mathfrak{a}_M = 0$ and the second admissibility condition holds. Therefore the existence of admissible Lagrangian subgroups is a constraint on the Wu structure only.
The lift $j$ determines a 2-cocycle $\phi$ on $\mathsf{G}_M$, as explained in Section \ref{SecPushDown}.
\begin{proposition}
\label{PropLAdmPhiTriv}
If $\mathsf{L}$ is admissible, then $\phi|_\mathsf{L} = 1$.
\end{proposition}
\begin{proof}
Recall the expression of $\phi$ given by \eqref{EqProofPhiCoc1} and \eqref{EqProofPhiCoc2}, which amount to
\be
\label{EqRelPhiPsi}
\phi(k_1, k_2) = \psi(j(k_1), j(k_2)) \psi(b(k_1, k_2), j(k_1) + j(k_2)) \: \left( Q(x, b(k_1, k_2)) \right)^{-1} \;.
\ee
If $k_1, k_2 \in \mathsf{L}$, then the first two factors are equal to 1 by the isotropy of $\mathsf{L}_{2\Lambda}$. The last factor is equal to 1 thanks to the second admissibility condition.
\end{proof}

\subsection{LL subgroups and orientations}

\label{SecLLSubAndOr}

We obtain from $\phi$ a realization of $\mathsf{H}_M$ on $\mathsf{G}_M \times U(1)$ given by \eqref{EqExplRealHeisGrp}. A LL subgroup of $\mathsf{H}_M$ is then specified by a pair $(\mathsf{L}, q_\mathsf{L})$, where $\mathsf{L}$ is a Lagrangian subgroup of $\mathsf{G}_M$ and $q_\mathsf{L}$ is a quadratic refinement of $\phi|_\mathsf{L}$. More explicitly, the elements of $\mathsf{G}_M \times U(1)$ corresponding to $\tilde{\mathsf{L}}$ are pairs $(l, q_\mathsf{L}(l))$ with $l \in \mathsf{L}$. 

Proposition \ref{PropLAdmPhiTriv} shows that if $\mathsf{L}$ is admissible, there is a preferred associated LL subgroup, given by the trivial quadratic refinement, $q_\mathsf{L}(k) = 1$ for all $k \in \mathsf{L}$. We call this preferred LL subgroup \emph{admissible} and write it $\tilde{\mathsf{L}}$. 

Recall as well that $j$ determines a trivialization of the Hermitian lines $L_k$. We will find it convenient to write for $t \in \mathbb{C}$
\be
\rho(k,t) = \bigoplus_{x \in w + \mathsf{G}_M} t Q(x,j(k)) \;,
\ee
\be
\tilde{\rho}(k,t) =  \bigoplus_{x \in w + \mathsf{G}_M} t B(k,x - w) Q(x,j(k)) \;,
\ee
which are simply \eqref{EqDefRho} and \eqref{EqDefRhoT} expressed in the trivialization of $L_k$ determined by $j$. We will also write $\rho(k)$ and $\tilde{\rho}(k)$ for $\rho(k,1)$ and $\tilde{\rho}(k,1)$.

Let $r$ be a power of a prime $p$. $\mathbb{Z}_r$ is a ring, and let $\Xi$ be a free $\mathbb{Z}_r$ module, of the form $\mathbb{Z}_r^n$. An \emph{orientation} of $\Xi$ is a generator of $\bigwedge^{n} \Xi \simeq \mathbb{Z}_r$. For a generic finite abelian group $\Xi \simeq \bigoplus_{r} \mathbb{Z}_r^n$, an orientation is a collection of orientations for each of the summands. We will consider below orientations of Lagrangian subgroups of $\mathsf{G}_M$. By extension, an orientation of an LL subgroup is an orientation of the underlying Lagrangian subgroup.

\subsection{Bundle of Heisenberg modules and parallel transport}

\label{SecBunHeisParTrans}

Let $\mathcal{C}_{o,\epsilon}$ be the category defined as follows. Its objects are the oriented LL subgroups of $\mathsf{H}_M$ admissible with respect to $\epsilon$ and there is a single morphism between each pair of objects. Let $\mathcal{I}$ be the category of irreducible modules for $\mathsf{H}_M$, with morphisms given by intertwining isomorphisms. Our aim is to define a functor $\mathcal{F}_{o,\epsilon}: \mathcal{C}_{o,\epsilon} \rightarrow \mathcal{I}$. Given such a functor, we can take its space of invariant sections to get a canonical irreducible module for $\mathsf{H}_M$. ("Canonical" is understood here and in the following up to a choice of extra structure $\epsilon$.) We can picture $\mathcal{F}_{o,\epsilon}$ as a bundle of Heisenberg modules over $\mathcal{C}_{o,\epsilon}$ endowed with a holonomy-free parallel transport operation. The invariant sections are then the flat sections of $\mathcal{F}_{o,\epsilon}$. 

We define $\mathcal{F}_{o,\epsilon}$ on objects as follows. $\mathcal{F}_{o,\epsilon}$ associates the irreducible $\mathsf{H}_M$-module $\mathsf{I}(\tilde{\mathsf{L}})$ described in Section \ref{SecStrat} to the object $\tilde{\mathsf{L}}$ of $\mathcal{C}_{o,\epsilon}$. Note that this assignment is independent of the orientation.

We will now describe $\mathcal{F}_{o,\epsilon}$ on morphisms. Let $\tilde{\mathsf{L}}_1$ and $\tilde{\mathsf{L}}_2$ be two oriented admissible LL subgroups. Recall that $\mathsf{I}(\tilde{\mathsf{L}}_i)$ is composed of the vectors in $V_{M,w}$ invariant under the action of $\rho(\tilde{\mathsf{L}}_i)$. We can therefore define a homomorphism $F_{\tilde{\mathsf{L}}_2, \tilde{\mathsf{L}}_1}: \mathsf{I}(\tilde{\mathsf{L}}_1) \rightarrow \mathsf{I}(\tilde{\mathsf{L}}_2)$ by
\be
F_{\tilde{\mathsf{L}}_2, \tilde{\mathsf{L}}_1}(v) = \sum_{k \in \mathsf{L}_2} \rho(k)v \;, \quad v \in \mathsf{I}(\tilde{\mathsf{L}}_1) \;.
\ee
By irreducibility, these homomorphisms compose up to a complex factor: given admissible oriented LL subgroups $\tilde{\mathsf{L}}_1$, $\tilde{\mathsf{L}}_2$, $\tilde{\mathsf{L}}_3$, we have
\be
\label{EqDefCoca}
F_{\tilde{\mathsf{L}}_3, \tilde{\mathsf{L}}_2} \circ F_{\tilde{\mathsf{L}}_2, \tilde{\mathsf{L}}_1} = a(\tilde{\mathsf{L}}_3, \tilde{\mathsf{L}}_2, \tilde{\mathsf{L}}_1) F_{\tilde{\mathsf{L}}_3, \tilde{\mathsf{L}}_1} \;, \quad a(\tilde{\mathsf{L}}_3, \tilde{\mathsf{L}}_2, \tilde{\mathsf{L}}_1) \in \mathbb{C} \;.
\ee
\begin{remark}
In the context of Lagrangians of a real symplectic vector space, the quantity $a(\tilde{\mathsf{L}}_3, \tilde{\mathsf{L}}_2, \tilde{\mathsf{L}}_1)$ is essentially the \emph{Maslov index} of the triplet of Lagrangians $(\mathsf{L}_3, \mathsf{L}_2, \mathsf{L}_1)$, see for instance \cite{lion2013weil}. We will use this terminology in the following.
\end{remark}
In Section \ref{SecTrivMasInd}, Proposition \ref{PropFactorMaslovIndex}, we will prove
\begin{assertion}
\label{PropTrivialCompCoc}
There is a $\mathbb{C}^\ast$-valued function $b$ depending on ordered pairs of admissible oriented LL subgroups such that 
\be
\label{EqTrivialCompCoc}
a(\tilde{\mathsf{L}}_3, \tilde{\mathsf{L}}_2, \tilde{\mathsf{L}}_1) = b(\tilde{\mathsf{L}}_3, \tilde{\mathsf{L}}_2) b(\tilde{\mathsf{L}}_2, \tilde{\mathsf{L}}_1) (b(\tilde{\mathsf{L}}_3, \tilde{\mathsf{L}}_1))^{-1} \;.
\ee
\end{assertion}
We will refer to a function $b$ with the property \eqref{EqTrivialCompCoc} as a "trivialization" of the Maslov index. Given such a trivialization,
\be
\label{EqCompCompT}
T_{\tilde{\mathsf{L}}_2, \tilde{\mathsf{L}}_1} := (b(\tilde{\mathsf{L}}_2, \tilde{\mathsf{L}}_1))^{-1} F_{\tilde{\mathsf{L}}_2, \tilde{\mathsf{L}}_1}
\ee
satisfies 
\be
\label{EqCompatCompT}
T_{\tilde{\mathsf{L}}_3, \tilde{\mathsf{L}}_2} \circ T_{\tilde{\mathsf{L}}_2, \tilde{\mathsf{L}}_1} = T_{\tilde{\mathsf{L}}_3, \tilde{\mathsf{L}}_1} \;.
\ee
We can now define the functor $\mathcal{F}_{o,\epsilon}$ on morphisms:
\be
\mathcal{F}_{o,\epsilon}(\tilde{\mathsf{L}}_1 \rightarrow \tilde{\mathsf{L}}_2) := T_{\tilde{\mathsf{L}}_2, \tilde{\mathsf{L}}_1} \;.
\ee 
\eqref{EqCompatCompT} shows that $\mathcal{F}_{o,\epsilon}$ is compatible with the composition of morphisms in $\mathcal{C}_{o,\epsilon}$ and $\mathcal{I}$, hence is indeed a functor.

\subsection{Maslov index}

\label{SecMasInd}

The aim of the present section and of the next one is to prove Assertion \ref{PropTrivialCompCoc}. The first step, taken in the present section, is to study the Maslov index $a(\tilde{\mathsf{L}}_3, \tilde{\mathsf{L}}_2, \tilde{\mathsf{L}}_1)$, written $a_{321}$ in the following.

We start by introducing some notation and stating some basic facts about triplets of Lagrangian subgroups $(\mathsf{L}_1, \mathsf{L}_2, \mathsf{L}_3)$. Given such a triplet, define
\be
\mathsf{L}_{12} := \mathsf{L}_1 \cap \mathsf{L}_2 \;, \quad \mathsf{L}_{23} := \mathsf{L}_2 \cap \mathsf{L}_3  \;, \quad \mathsf{L}_{31} := \mathsf{L}_3 \cap \mathsf{L}_1 \;, 
\ee
and $\mathsf{L}_{123} = \mathsf{L}_1 \cap \mathsf{L}_2 \cap \mathsf{L}_3$. We write $\mathsf{S}_2$ for the subgroup of elements of $\mathsf{L}_2$ that can be written as sums of elements in $\mathsf{L}_1$ and in $\mathsf{L}_3$, $\mathsf{R}_2 := \mathsf{L}_2/\mathsf{S}_2$ and $\mathsf{U}_2 = \mathsf{S}_2/(\mathsf{L}_{12} \cup \mathsf{L}_{23})$. (By $\mathsf{L}_{12} \cup \mathsf{L}_{23}$, we mean the smallest subgroup of $\mathsf{L}_2$ containing $\mathsf{L}_{12}$ and $\mathsf{L}_{23}$.) By cyclic permutations of the indices, we define $\mathsf{S}_i$, $\mathsf{R}_i$ and $\mathsf{U}_i$ for $i = 1,3$. 

Consider a pair of group homomorphisms $p_{21}: \mathsf{S}_2 \rightarrow \mathsf{S}_3$ and $p_{23}: \mathsf{S}_2 \rightarrow \mathsf{S}_1$ with the property that
\be
\label{EqDefRelp}
k_2 = p_{21}(k_2) + p_{23}(k_2)
\ee
for each $k_2 \in \mathsf{S}_2$. Remark that if $k_2 \in \mathsf{L}_{12} \cap \mathsf{S}_2$, then the fact that $\mathsf{L}_1$ is a subgroup implies that $p_{23}(k_2) \in \mathsf{L}_{31} \cap \mathsf{S}_3$ and $p_{21}(k_2) \in (\mathsf{L}_{12} \cup \mathsf{L}_{31}) \cap \mathsf{S}_1$. Similarly, if $k_2 \in \mathsf{L}_{23} \cap \mathsf{S}_2$, then $p_{21}(k_2) \in \mathsf{L}_{31} \cap \mathsf{S}_1$ and $p_{23}(k_2) \in (\mathsf{L}_{23} \cup \mathsf{L}_{31}) \cap \mathsf{S}_3$. This implies that the homomorphisms $p_{21}$ and $p_{23}$ induce well-defined homomorphisms
\be
\tilde{p}_{21}: \mathsf{U}_2 \rightarrow \mathsf{U}_1 \;, \quad \tilde{p}_{23}: \mathsf{U}_2 \rightarrow \mathsf{U}_3 \;.
\ee
\begin{proposition}
$\tilde{p}_{21}$ and $\tilde{p}_{23}$ are isomorphisms. Moreover, they are independent of the choices of homomorphisms $p_{21}$ and $p_{23}$.
\end{proposition}
\begin{proof}
We prove this proposition for $\tilde{p}_{21}$, the proof for $\tilde{p}_{23}$ being obtained by a permutation of the indices.

$\tilde{p}_{21}$ is injective. Indeed, let us write $\tilde{k}_2$ for the equivalence class in $\mathsf{U}_2$ of $k_2 \in \mathsf{S}_2$. Assume that $\tilde{p}_{21}(\tilde{k}_2) = 0$ in $\mathsf{U}_1$, so $p_{21}(k_2) \in (\mathsf{L}_{12} \cup \mathsf{L}_{31}) \cap \mathsf{S}_1$. If $p_{21}(k_2) \in \mathsf{L}_{12} \cap \mathsf{S}_1$, then $p_{23}(k_2) \in \mathsf{L}_{23} \cap \mathsf{S}_3$ so $k_2 \in (\mathsf{L}_{12} \cup \mathsf{L}_{23}) \cap \mathsf{S}_2$. If $p_{21}(k_2) \in \mathsf{L}_{31} \cap \mathsf{S}_1$, then $k_2 \in \mathsf{L}_{23} \cap \mathsf{S}_2$. By the homomorphism property, we have $k_2 \in (\mathsf{L}_{12} \cup \mathsf{L}_{23}) \cap \mathsf{S}_2$ whenever $p_{21}(k_2) \in (\mathsf{L}_{12} \cup \mathsf{L}_{31}) \cap \mathsf{S}_1$, so $\tilde{k}_2 = 0$.

Repeating the construction with permuted indices, we obtain a sequence of injective homomorphisms $\mathsf{U}_2 \stackrel{\tilde{p}_{21}}{\rightarrow} \mathsf{U}_1 \rightarrow \mathsf{U}_3 \rightarrow \mathsf{U}_2$. This can occur only if $|\mathsf{U}_1| = |\mathsf{U}_2| = |\mathsf{U}_3|$ and $\tilde{p}_{21}$ is an isomorphism.

To see that $\tilde{p}_{21}$ is independent of the choice of $p_{21}$, remark that the constraint \eqref{EqDefRelp} requires any other choice $p'_{21}$ to satisfy $p'_{21}(k_2) - p_{21}(k_2) \in \mathsf{L}_{31} \cap \mathsf{S}_1$, hence the induced homomorphisms between $\mathsf{U}_2$ and $\mathsf{U}_1$ agree.
\end{proof}
\noindent Therefore $\mathsf{U}_1$, $\mathsf{U}_2$ and $\mathsf{U}_3$ are canonically isomorphic.
\begin{proposition} 
\label{PropPairOnUNonDeg}
The following pairing on $\mathsf{S}_2$ is symmetric and bimultiplicative:
\be
\label{EqDefB321}
B_{321}(\bullet, \bullet) := B(p_{21}(\bullet), p_{23}(\bullet)) = B(\bullet, p_{23}(\bullet)) = B(p_{21}(\bullet), \bullet) \;.
\ee
It passes to a well-defined non-degenerate pairing on $\mathsf{U}_2$.
\end{proposition}
\begin{proof}
Using the isotropy of $\mathsf{S}_i$, the defining relation \eqref{EqDefRelp} for $p_{21}$ and $p_{23}$ and the bimultiplicative property of $B$, one easily proves the chain of equalities in \eqref{EqDefB321}. The isotropy of $\mathsf{S}_i$ makes it also clear that $B_{321}(\bullet, \bullet)$ passes to a well-defined pairing on $\mathsf{U}_2$.  

$B_{321}$ is bimultiplicative, as $B$ is bimultiplicative and $p_{21}$ and $p_{23}$ are group homomorphisms. We have for $k,k' \in \mathsf{U}_2$
\be
1 = B(p_{21}(k)+p_{23}(k), p_{21}(k') + p_{23}(k')) = (B(p_{21}(k'),p_{23}(k)))^{-1} B(p_{21}(k), p_{23}(k'))) \;,
\ee
so $B_{321}$ is symmetric. Finally, suppose that there is $u  \in \mathsf{U}_2$ such that $B_{321}(u,u') = 1$ for all $u' \in \mathsf{U}_2$. Let $k$ be a lift of $u$ to $\mathsf{S}_2$. Then $p_{23}(k)$ is an element of $\mathsf{L}_3$ in $\mathsf{L}_2^\perp = \mathsf{L}_2$. Hence $p_{23}(k) \in \mathsf{L}_{23}$, $\tilde{p}_{23}(u) = 0$ and $u = 0$, so $B_{321}$ is non-degenerate on $\mathsf{U}_2$.
\end{proof}
\noindent Depending on the context, we will consider $B_{321}$ as a non-degenerate pairing on $\mathsf{U}_2$ or as a possibly degenerate pairing on $\mathsf{S}_2$. The main result of this section is the following.
\begin{proposition}
\label{PropaAsGaussSum}
$a_{321}$ can be expressed as a Gauss sum,
\be
\label{EqaAsGaussSum}
a_{321} = {\rm Gauss}(q_{321}) \;,
\ee
where $q_{321}$ is the quadratic function on $\mathsf{S}_2$ defined by
\be
\label{EqDefq321}
q_{321}(k) :=  \phi(k,p_{23}(k)) \;,
\ee
for $k \in \mathsf{S}_2$. 
\end{proposition}
\begin{proof}
Consider 
\be
\label{EqVectIL1}
v = \sum_{k \in \mathsf{L}_1} \rho(k) v_0 \in \mathsf{I}(\tilde{\mathsf{L}}_1)\;,
\ee
where $v_0$ is a vector generating $\mathcal{S}^e(M, w) \subset V_{M,w}$. We also write $\delta_0$ for the linear functional on $V_{M,w}$ equal to $1$ on $v_0$ and to zero on $S^e(M,x)$, $x \neq w$. Let us compute
\be
\label{EqCompRHSProjForA}
\delta_0 \left(  F_{\tilde{\mathsf{L}}_3, \tilde{\mathsf{L}}_1} v \right) = \delta_0 \left( \sum_{k' \in \mathsf{L}_3} \sum_{k \in \mathsf{L}_1} \phi(k',k) \rho(k'+k) v_0 \right) = |\mathsf{L}_{13}| \;.
\ee
Indeed, $k+k' = 0$ only if $k,k' \in \mathsf{L}_{13}$. Combining \eqref{EqDefCoca} and \eqref{EqCompRHSProjForA}, we have
\begin{align}
\label{EqCompA1}
|\mathsf{L}_{13}| a_{321}  \: & = \delta_0 \left( F_{\tilde{\mathsf{L}}_3, \tilde{\mathsf{L}}_2} \circ F_{\tilde{\mathsf{L}}_2, \tilde{\mathsf{L}}_1} v \right)  \notag \\
& = \delta_0 \left( \sum_{k'' \in \mathsf{L}_3} \sum_{k' \in \mathsf{L}_2} \sum_{k \in \mathsf{L}_1} \rho(k'') \: \rho(k') \: \rho(k) \: v_0 \right) \\
& = \sum_{k \in \mathsf{S}_2} \sum_{l \in \mathsf{L}_{13}} \: \phi(p_{23}(k) - l, -k) \: \phi(-p_{21}(k)-l, p_{21}(k)+l) \notag \\
& = |\mathsf{L}_{13}| \sum_{k \in \mathsf{S}_2} \: \phi(k, p_{23}(k)) \;. \notag 
\end{align}
On the second line, we used the defining relations of the Heisenberg group and the fact that the triplets $(k,k',k'') \in \mathsf{L}_1 \times \mathsf{L}_2 \times \mathsf{L}_3$ such that $k + k' + k'' = 0$ are in bijection with $\mathsf{S}_2 \times \mathsf{L}_{13}$. The second factor of the third line is equal to 1, because both arguments are in $\mathsf{L}_1$, and $\phi$ is trivial when restricted to an admissible Lagrangian. Furthermore, using the 2-cocycle property \eqref{EqPhiCocycle} of $\phi$, for $k \in \mathsf{S}_2$, $l \in \mathsf{L}_{31}$,
\begin{align}
\phi(p_{23}(k) - l, -k) \: & = \phi(-l, p_{23}(k)-k) \phi(p_{23}(k),-k) \phi^{-1}(-l, p_{23}(k)) \notag \\
& = \phi(p_{23}(k),-k) \notag \;,
\end{align}
where the first (third) factor is equal to 1 because both of its arguments are in $\mathsf{L}_1$ ($\mathsf{L}_3$). We deduce that 
\be
a_{321} = \sum_{k \in S_2} \: q_{321}(k) = {\rm Gauss}(q_{321})
\ee
\end{proof}

\begin{proposition}
$q_{321}$ is a tame quadratic refinement of $B_{321}$, seen as a degenerate pairing on $\mathsf{S}_2$.
\end{proposition}
\begin{proof}
Using repeatedly the fact \eqref{EqPhiCocycle} that $\phi$ is a 2-cocycle and that it vanishes on admissible Lagrangians, we compute for $k,k' \in \mathsf{S}_2$
\begin{align}
q_{321}(k+k') \, & = \phi(p_{23}(k) + p_{23}(k'), -k-k') \notag \\
& = \phi(p_{23}(k) + p_{23}(k'), -k) \phi(-p_{21}(k) + p_{23}(k'), -k') \notag \\
& = \phi(p_{23}(k), p_{23}(k') - k) \phi(p_{23}(k'), -k) \phi^{-1}(-p_{21}(k), p_{23}(k')) \phi(p_{23}(k'), -k') \\
& = \phi(-p_{21}(k), p_{23}(k')) \phi(p_{23}(k), -k) \phi^{-1}(-k, p_{23}(k')) \phi(p_{23}(k'), -k) \notag \\
& \qquad \phi^{-1}(-p_{21}(k), p_{23}(k')) \phi(p_{23}(k'), -k') \notag \\
& = q_{321}(k) q_{321}(k') B(p_{23}(k'), -k) \notag \\
& = q_{321}(k) q_{321}(k') B_{321}(k, k') \;. \notag
\end{align}
In the last stages, we also used the skew-symmetry of $\phi$.

To show that $q_{321}$ is tame, we have to show that it vanishes on the radical $(\mathsf{L}_{12} \cup \mathsf{L}_{23}) \cap \mathsf{S}_2$ of $B_{321}$. Let $k = k_1 + k_3$ be an element of the radical, with $k_1 \in \mathsf{L}_{12} \cap \mathsf{S}_2$ and $k_3 \in \mathsf{L}_{23} \cap \mathsf{S}_2$. As $p_{23}(k_1) \in \mathsf{L}_{31} \cap \mathsf{S}_3$ and $\phi = 1$ on $\mathsf{L}_1$, $q_{321}(k_1) = 1$. Similarly, as $p_{23}(k_3) \in (\mathsf{L}_{23} \cap \mathsf{L}_{31}) \cap \mathsf{S}_3$ and $\phi = 1$ on $\mathsf{L}_3$, $q_{321}(k_3) = 1$. Finally $B_{321}(k_3, k_1) = B(k_3, p_{23}(k_1)) = 1$ from the isotropy of $\mathsf{L}_3$. The quadratic refinement property therefore ensures that $q_{321}(k) = 1$, so $q_{321}$ is tame.
\end{proof}

\subsection{Trivialization of the Maslov index}

\label{SecTrivMasInd}

The aim of this section is to prove Assertion \ref{PropTrivialCompCoc}. Given a decomposition $\Gamma_0 = \bigoplus_r \mathbb{Z}_r^{m_r}$, with $r$ running over powers of primes, we obtain a decomposition of $\mathsf{G}_M \simeq \bigoplus_r \mathsf{G}_{M,r} \simeq \bigoplus_r \mathbb{Z}_r^{2n_r}$, as explained in Section \ref{SecExStruct}. Any Lagrangian subgroup of $\mathsf{G}_M$ (in the sense explained in Section \ref{SecStrat}) projects onto Lagrangian subgroups of each term of the direct sum, for which the analysis of the previous section can be carried out. The total Maslov index is then expressed as a product with one factor coming from each term in the direct sum. We will be able to trivialize the Maslov index independently on each term of the direct sum, so we focus on a single term in the following. We set therefore $\mathsf{G}_M = \mathbb{Z}_r^{2n_r}$ and we will write $n$ for $n_r$. Our first task is to suitably diagonalize $q_{321}$.

\begin{proposition}
If $r$ is even, $B_{321}$ is even, in the sense that for $u \in \mathsf{U}_2$, $B_{321}(u,u) \in \exp 2\pi i \frac{2}{r} \mathbb{Z}$.
\end{proposition}
\begin{proof}
Let $k \in \mathsf{S}_2$ be a lift of $u \in \mathsf{U}_2$. We have
\be
(B_{321}(u,u))^{r/2} = (B(p_{21}(k), p_{23}(k)))^{r/2} = (q_r(p_{21}(k)))^{-1} (q_r(p_{23}(k)))^{-1} q_r(k) = 1 \;.
\ee
We used the fact that $B^{r/2}$ admits the quadratic refinement $q_r$ and that by the third admissibility condition, $q_r$ vanishes on $\mathsf{L}_i$, $i = 1,2,3$.
\end{proof}
We will now use Taylor's classification of finite abelian groups endowed with quadratic refinements (Theorem 3.5 of \cite{Taylor1984259}, also reviewed in Appendix \ref{SecPairFinAbGroups}). Together with the evenness of $B_{321}$ for $r$ even, Taylor's classification directly implies
\begin{proposition}
\label{PropFormQuadrFuncEven}
When $r$ is even, $(\mathsf{U}_2,q_{321})$, as a group endowed with a quadratic refinement, is a direct sum of groups $(\mathbb{Z}_r^2,q_{(r,a_0, a_1)})$. Writing $u_0$ and $u_1$ for the generators of $\mathbb{Z}_r^2$, the quadratic refinement $q_{(r,a_0, a_1)}$, $a_0, a_1 \in \mathbb{Z}_r$, satisfies
\be
q_{(r,a_0, a_1)}(u_0) = \exp \frac{2\pi i}{r} a_0 \;, \quad q_{(r,a_0, a_1)}(u_1) = \exp \frac{2\pi i}{r} a_1 \;,
\ee
and the associated symmetric bimultiplicative form $B_{a_0,a_1}$ is the following:
\be
\label{EqMatBlockEven}
B_{a_0,a_1}(m_iu_i, m'_ju_j) = \exp \frac{2\pi i}{r} m_i b_{ij} m'_j \;, \quad  (b_{ij}) = \left( \begin{array}{cc} 2a_0 & 1 \\ 1 & 2a_1 \end{array} \right)\;.
\ee
\end{proposition}
\begin{proposition}
\label{PropFormQuadrFuncOdd}
When $r$ is odd, $(\mathsf{U}_2,q_{321})$, as a group endowed with a quadratic refinement, is a direct sum of groups $(\mathbb{Z}_r,q_{(r,a)})$. $a$ is an integer prime to $r$, taken modulo $r$. Writing $u$ for the generator of $\mathbb{Z}_r$, the quadratic refinement is given by 
\be
\label{EqDefqra}
q_{(r,a)}(u) = \exp \frac{2\pi i}{r} a \;.
\ee
The associated bimultiplicative form $B_a$ is
\be
B_a(mu, m'u) = \exp \frac{4\pi i}{r} mam' \;.
\ee
\end{proposition}

Let $\omega_r = \exp \frac{2 \pi i} {r}$ if $r$ is odd and $\omega_r = \exp \frac{\pi i} {r}$ if $r$ is even. The following proposition expresses $q_{321}$ in terms of a bilinear form on $\mathsf{U}_2$.
\begin{proposition}
\label{PropQRInTermsOfBilin}
We have
\be
q_{321}(k) = \omega_r^{C(k,k)} \;,
\ee
where $C$ is a symmetric $\mathbb{Z}_r$-valued bilinear form on $\mathsf{U}_2 \simeq \mathbb{Z}_r^{n_\mathsf{U}}$ defined by 
\begin{align}
C = \frac{r}{2\pi i} \ln B_{321} \quad & r \mbox{ even,} \\
C = \frac{r}{4\pi i} \ln B_{321} \quad & r \mbox{ odd} \;.
\end{align}
\end{proposition}
\noindent
We used the fact that the division by $2$ is well-defined in $\mathbb{Z}_r$ for $r$ odd.
\begin{proof}
In the even case, Proposition \ref{PropFormQuadrFuncEven} implies that
\be
q_{(r,a_0, a_1)}(u_j) = \exp \frac{2\pi i}{r} a_j = \omega_r^{2a_j} = \omega_r^{\frac{r}{2\pi i} \ln B_{a_0, a_1}(a_j, a_j)} \;.
\ee
In the odd case, Proposition \ref{PropFormQuadrFuncOdd} implies that
\be
q_{(r,a)}(u) = \exp \frac{2\pi i}{r} a = \omega_r^a = \sqrt{\omega_r^{r\ln B_a(a,a)}} = \omega_r^{\frac{r}{4\pi i} \ln B_a(a,a)} \;,
\ee
where the square root is uniquely defined among $r$th roots of unity.
\end{proof}
Let us now recall the definition of the Kronecker symbol $(a|b)_K$, $a \in \mathbb{Z}$, $b \in \mathbb{N}$. Let $b = 2^{s_0} \prod_{i > 0} p_i^{s_i}$ be the prime decomposition of $b$. Then 
\be
\label{EqDefKronecker}
(a|b)_K := \binom{a}{2}^{s_0}_L \prod_{i > 0} \binom{a}{p_i}^{s_i}_L \;.
\ee
$\binom{a}{p_i}_L$ is the Legendre symbol of the odd prime $p_i$, namely the function on $\mathbb{Z}$ assigning $1$ to quadratic residue mod $p_i$, $-1$ to quadratic non-residue mod $p_i$ and $0$ to multiples of $p_i$. $\binom{a}{2}_L$ is defined to be $1$ if $a = \pm 1$ mod $8$, $-1$ if $a = \pm 3$ mod $8$, and zero if $a$ is even. The Kronecker symbol is multiplicative in its top argument as long as only strictly positive integers are involved (which will always be the case for us).

\begin{proposition}
\label{PropGaussSumDet}
For $r$ even, we have
\be
\label{EqGaussSumDetEven}
{\rm Gauss}(q_{321}) = \frac{|\mathsf{L}_{12}|^{1/2} |\mathsf{L}_{23}|^{1/2}|\mathsf{L}|^{1/2}}{|\mathsf{L}_{31}|^{1/2}} ({\rm det}C|r)_{K} \;.
\ee
where $|\mathsf{L}| := |\mathsf{L}_i|$.
For $r$ odd, we have
\be
\label{EqGaussSumDetOdd}
{\rm Gauss}(q_{321}) = \frac{|\mathsf{L}_{12}|^{1/2} |\mathsf{L}_{23}|^{1/2}|\mathsf{L}|^{1/2}}{|\mathsf{L}_{31}|^{1/2}} ({\rm Gauss}(q_{r,1}))^{n_\mathsf{U}} ({\rm det}C|r)_{K} \;.
\ee
\end{proposition}
\begin{proof}
We can decompose the computation of the Gauss sum into the computation of its complex modulus and the computation of its phase. 

Corollary 1.11 of \cite{Taylor} tells us that 
\be
\label{EqNormGaussSumq321}
|{\rm Gauss}(q_{321})| = \sqrt{|\mathsf{L}_{12} \cup \mathsf{L}_{23}| |\mathsf{S}_2|} \;,
\ee
because $|\mathsf{L}_{12} \cup \mathsf{L}_{23}|$ is the radical of $B_{321}$ on $\mathsf{S}_2$. We can rewrite
\be
\label{EqRelForNormGaussSum}
|\mathsf{L}_{12} \cup \mathsf{L}_{23}| = \frac{|\mathsf{L}_{12}| |\mathsf{L}_{23}|}{|\mathsf{L}_{123}|} \;, \quad |\mathsf{S}_2| = \frac{|\mathsf{L}|}{|\mathsf{R}_2|} = \frac{|\mathsf{L}| |\mathsf{L}_{123}|}{|\mathsf{L}_{13}|} \;,
\ee
where the third equality follows from the fact that $|\mathsf{R}_2| = |\mathsf{L}_{13}/\mathsf{L}_{123}|$, as will be shown shortly in Proposition \ref{PropFactDetC}. Combining \eqref{EqNormGaussSumq321} with \eqref{EqRelForNormGaussSum}, we obtain \eqref{EqGaussSumDetEven} and \eqref{EqGaussSumDetOdd} up to phases.

As $q_{321}$ is tame, the phase of its Gauss sum coincides with the phase of the Gauss sum of its reduction to $\mathsf{U}_2$. The computation is then a straightforward application of Proposition \ref{PropArfInvQRFromLat}, which is a special case of another result of Taylor, namely Proposition 2.8 of \cite{Taylor}. To apply Proposition \ref{PropArfInvQRFromLat} we lift $C$ to a integer valued matrix $\tilde{C}$ such that $\tilde{C} = C$ mod $r$. We see $\tilde{C}$ as a pairing on $\mathbb{Z}^{n_\mathsf{U}}$, and write $\mathbb{Z}_r^{n_\mathsf{U}} = \mathbb{Z}^{n_\mathsf{U}} \otimes_{\mathbb{Z}} \mathbb{Z}_r$. Proposition \ref{PropQRInTermsOfBilin} implies that $\frac{1}{2\pi i}\ln q_{321}(k) = \tilde{C}(k,k) \otimes q_{r,{\rm add}}(1)$, where $q_{r,{\rm add}}$ is the additive quadratic refinement on $\mathbb{Z}_r$ taking value $\frac{1}{r}$ on the generator for $r$ odd and $\frac{1}{2r}$ for $r$ even. 
Substituting in Proposition \ref{PropArfInvQRFromLat}, we obtain the phases of \eqref{EqGaussSumDetEven} and \eqref{EqGaussSumDetOdd}.
\end{proof}

Our aim is now to reexpress the determinants in the formulas above in terms of orientations, to obtain a trivialization of the Maslov index. $C$ induces a pairing on $\bigwedge^{n_\mathsf{U}} \mathsf{U}_2$, which we also write $C$. We have
\begin{lemma}
\label{LemDetCFromOrient}
Pick any orientation $o_{\mathsf{U}_2}$ of $\mathsf{U}_2$. Then
\be
{\rm det}(C) = C(o_{\mathsf{U}_2}, o_{\mathsf{U}_2}) 
\ee
\end{lemma}
\begin{proof}
Recall that $\mathsf{U}_2 = \mathbb{Z}_r^{n_\mathsf{U}}$ and pick an oriented system of generators $e_i$, $i = 1,...,n_\mathsf{U}$. Using the definition of the determinant,
\begin{align}
{\rm det}(C) & \: = \sum_{\sigma \in S_{n_\mathsf{U}}} (-1)^\sigma \prod_i C(e_i, e_{\sigma{i}}) \notag\\
& = C(e_1 \wedge ... \wedge e_{n_\mathsf{U}},e_1 \wedge ... \wedge e_{n_\mathsf{U}}) \\
& = C(o_{\mathsf{U}_2}, o_{\mathsf{U}_2}) \;. \notag
\end{align}
\end{proof}
Let us define the $\{1,0,-1\}$-valued pairings
\begin{align}
E = \left(  \frac{r}{2\pi i} \ln B | r \right)_K \quad & r \mbox{ even,} \\
E = \left( \frac{1}{2}\frac{r}{2\pi i}  \ln B | r \right)_K \quad & r \mbox{ odd} \;,
\end{align}
where the extra division by 2 is again well-defined in $\mathbb{Z}_r$ for $r$ odd. Then 
\be
\left( C(k,k') | r \right)_K = E(p_{21}(k), p_{23}(k')) \;.
\ee
Remark that the groups $\mathsf{L}_{ij}$, $\mathsf{L}_{123}$, $\mathsf{S}_i$, $\mathsf{R}_i$, $i,j = 1,2,3$ are all free $\mathbb{Z}_r$-modules. We pick isomorphisms
\begin{align}
\label{EqDecompKi}
\mathsf{L}_1 \: & \simeq \mathsf{L}_{123} \times \mathsf{L}_{31}/\mathsf{L}_{123} \times \mathsf{L}_{12}/\mathsf{L}_{123} \times \mathsf{R}_1 \times \mathsf{U}_1\;,  \notag \\
\mathsf{L}_2 \: & \simeq \mathsf{L}_{123} \times \mathsf{L}_{12}/\mathsf{L}_{123} \times \mathsf{L}_{23}/\mathsf{L}_{123} \times \mathsf{R}_2 \times \mathsf{U}_2 \;, \\
\mathsf{L}_3 \: & \simeq \mathsf{L}_{123} \times \mathsf{L}_{23}/\mathsf{L}_{123} \times \mathsf{L}_{31}/\mathsf{L}_{123} \times \mathsf{R}_3 \times \mathsf{U}_3 \;, \notag
\end{align}
as well as orientations $o_{123}$ on $\mathsf{L}_{123}$, $o_{ij}$ on $\mathsf{L}_{ij}$ and $o_{\mathsf{S}_i}$ on $\mathsf{S}_i$. Together with the orientations $o_i$ existing on $\mathsf{L}_i$ and the choice of ordering in \eqref{EqDecompKi}, they determine orientations $o_{\mathsf{R}_i}$ on $\mathsf{R}_i$, $o_{\mathsf{U}_i}$ on $\mathsf{U}_i$, $o_i/o_{ij}$ on $\mathsf{L}_i/\mathsf{L}_{ij}$ and $o_{ij}/o_{123}$ on $\mathsf{L}_{ij}/\mathsf{L}_{123}$. We take the convention that the orientations on quotients are induced by ordering first the generators of the subgroup to be quotiented followed by the remaining generators. We will write $\mathsf{L}_{ij} = \mathbb{Z}_r^{n_{ij}}$ and $\mathsf{L}_{123} =\mathbb{Z}_r^{n_{123}}$.
\begin{proposition}
\label{PropFactDetC}
We have
\be
\label{EqFactDetC}
({\rm det}(C)|r)_K = \left((-1)^{n_{31}-n_{123}}| r \right)_K E(o_1/o_{12}, o_2/o_{12}) E(o_2/o_{23}, o_3/o_{23}) \left(E(o_1/o_{13}, o_3/o_{13})\right)^{-1}
\ee
\end{proposition}
\begin{proof}
Consider the pairing determined by $E$ between $\mathsf{L}_2$ and $\mathsf{L}_3$. It has the following block structure with respect to the decomposition above:
\def\tmp{
  \begin{pmatrix}
  0 & 0 & 0 & 0 & 0 \\
  0 & 0 & 0 & * & 0 \\
  0 & 0 & 0 & 0 & 0 \\
  0 & 0 & * & * & * \\
  0 & 0 & 0 & * & *
 \end{pmatrix}
}
\be
\stackMath\def\stackalignment{r}
  \stackon
    {\mathsf{L}_2 \left\{\tmp\right.}
    {\overbrace{\phantom{\smash{\tmp\mkern -36mu}}}^{\mathlarger{\mathsf{L}_3}}\mkern 20mu}
\ee
For instance, $E(k_2, k_3) = 0$ for any $k_2 \in \mathsf{L}_{23}$, $k_3 \in \mathsf{L}_3$ because of the isotropy of $\mathsf{L}_3$. The pairing also vanishes when $k_2 \in \mathsf{U}_2$ and $k_3 \in \mathsf{L}_{31}/\mathsf{L}_{123}$, because $k_2$ can be written as the sum of an element in $\mathsf{L}_1$ and an element in $\mathsf{L}_3$. 

Remark also that $B$ identifies $\mathsf{L}_{12}/\mathsf{L}_{123}$ as the Pontryagin dual of $\mathsf{R}_3$. Indeed, any $U(1)$-valued character on $\mathsf{L}_{12}$ invariant under $\mathsf{L}_{123}$ can be extended to a character on $\mathsf{G}_M$ vanishing on $\mathsf{L}_3$. It has therefore to be of the form $B(\bullet, k_3)$ for some $k_3 \in \mathsf{L}_3$. Elements in $\mathsf{S}_3 \subset \mathsf{L}_3$ are associated to trivial characters on $\mathsf{L}_{12}/\mathsf{L}_{123}$. Using the fact that the pairing determined by $B$ between $\mathsf{L}_1/\mathsf{L}_{12}$ and $\mathsf{L}_2/\mathsf{L}_{12}$ is non-degenerate, counting generators in the decomposition \ref{EqDecompKi} proves the claim. Writing $\mathsf{R}_i \simeq \mathbb{Z}_r^{n_{\mathsf{R}_i}}$, we have
\be
\label{EqOrdKKR}
n_{12} - n_{123} = n_{\mathsf{R}_3} \;.
\ee
The same claims hold after permuting the indices cyclically.

We now use the block decomposition to compute the right-hand side of \eqref{EqFactDetC}.
\begin{align}
E(o_2/o_{23}, o_3/o_{23}) \: & = E((-1)^{n_{\mathsf{R}_2} n_{\mathsf{R}_3}} o_{12}/o_{123} \wedge o_{\mathsf{R}_2} \wedge o_{\mathsf{U}_2}, o_{31}/o_{123} \wedge o_{\mathsf{R}_3} \wedge o_{\mathsf{U}_3}) \notag \\
& = E(o_{\mathsf{U}_2}, o_{\mathsf{U}_3})E(o_{12}/o_{123}, o_{\mathsf{R}_3})E(o_{\mathsf{R}_2}, o_{31}/o_{123}) \\
& = E(o_{\mathsf{U}_2}, \tilde{p}_{23}(o_{\mathsf{U}_2})) \left( \frac{o_{\mathsf{U}_3}}{\tilde{p}_{23}(o_{\mathsf{U}_2})} | r \right)_K E(o_{12}/o_{123}, o_{\mathsf{R}_3})E(o_{\mathsf{R}_2}, o_{31}/o_{123}) \;, \notag
\end{align}
where $\frac{o_{\mathsf{U}_3}}{\tilde{p}_{23}(o_{\mathsf{U}_2})} \in \mathbb{Z}_r$. Similarly, we find that
\begin{align}
E(o_1/o_{12}, o_2/o_{12}) \: & = E(\tilde{p}_{21}(o_{\mathsf{U}_2}), o_{\mathsf{U}_2}) \left( \frac{o_{\mathsf{U}_1}}{\tilde{p}_{21}(o_{\mathsf{U}_2})} | r \right)_K E(o_{\mathsf{R}_1}, o_{23}/o_{123})E(o_{31}/o_{123}, o_{\mathsf{R}_2}) \;, \notag \\
E(o_1/o_{13}, o_3/o_{13}) \: & = E(\tilde{p}_{21}(o_{\mathsf{U}_2}), \tilde{p}_{23}(o_{\mathsf{U}_2})) \left( \frac{o_{\mathsf{U}_1}}{\tilde{p}_{21}(o_{\mathsf{U}_2})} \frac{o_{\mathsf{U}_3}}{\tilde{p}_{23}(o_{\mathsf{U}_2})} | r \right)_K 
\\
& \qquad \cdot E(o_{12}/o_{123}, o_{\mathsf{R}_3})E(o_{\mathsf{R}_1}, o_{23}/o_{123}) \notag
\end{align}
Therefore
\begin{align}
& E(o_3/o_{23}, o_2/o_{23}) E(o_2/o_{12}, o_1/o_{12}) \left(E(o_3/o_{13}, o_1/o_{13})\right)^{-1} \notag \\
= \: & \left( (-1)^{n_{\mathsf{R}_2}} | r \right)_K E(o_{\mathsf{U}_2}, \tilde{p}_{23}(o_{\mathsf{U}_2})) E(\tilde{p}_{21}(o_{\mathsf{U}_2}), o_{\mathsf{U}_2}) \left(E(\tilde{p}_{21}(o_{\mathsf{U}_2}), \tilde{p}_{23}(o_{\mathsf{U}_2}))\right)^{-1}  \\
= \: & \left( (-1)^{n_{31}-n_{123}} | r \right)_K \left(  C(o_{\mathsf{U}_2}, o_{\mathsf{U}_2}) | r \right)_K  \notag
\end{align}
The proposition now follows from Lemma \ref{LemDetCFromOrient}.
\end{proof}

\begin{proposition}
\label{PropFactorMaslovIndex}
The factorization property of the Maslov index described in Assertion \ref{PropTrivialCompCoc} holds if we set
\be
b(\tilde{\mathsf{L}}_2, \tilde{\mathsf{L}}_1) = |\mathsf{L}_{12}|^{1/2} |\mathsf{L}|^{1/2} E(o_2/o_{12}, o_1/o_{12})
\ee
for $r$ even or $r = 1$ mod $4$, and
\be
b(\tilde{\mathsf{L}}_2, \tilde{\mathsf{L}}_1) = |\mathsf{L}_{12}|^{1/2} |\mathsf{L}|^{1/2} i^{n-n_{12}} E(o_2/o_{12}, o_1/o_{12})
\ee
for $r = -1$ mod $4$.
\end{proposition}
Remark that $b(\tilde{\mathsf{L}}_2, \tilde{\mathsf{L}}_1)$ is independent of the choice of orientation $o_{12}$ on $\mathsf{L}_{12}$. Indeed, changing the orientation $o_{12}$ produces a square that is mapped to 1 by the Kronecker symbol. $b(\tilde{\mathsf{L}}_2, \tilde{\mathsf{L}}_1)$ depends only on the admissible Lagrangian subgroups $\mathsf{L}_i$ and on their orientations, as it should. 
\begin{proof}
Let us summarize the results obtained so far. Proposition \ref{PropaAsGaussSum} identifies $a_{321}$ as the Gauss sum of $q_{321}$. This Gauss sum was expressed in term of the determinant of a bilinear form $C$ in Proposition \ref{PropGaussSumDet}. The determinant was factorized in Proposition \ref{PropFactDetC}. We decompose the check of equation \eqref{EqTrivialCompCoc} into its modulus and its phase. Only the first two factors have a modulus different from 1 in the definition of $b(\tilde{\mathsf{L}}_2, \tilde{\mathsf{L}}_1)$. From the explicit expression for $a_{321}$ in \eqref{EqGaussSumDetEven} and \eqref{EqGaussSumDetOdd}, we see that \eqref{EqTrivialCompCoc} holds up to phases.

We now examine the phases. Let us start with $r$ even. By definition, $(-s|r)_K = (s|r)_K$, $s \in \mathbb{Z}_r$ so we can ignore all signs. Propositions \ref{PropaAsGaussSum}, \ref{PropGaussSumDet} and \ref{PropFactDetC} yield directly the desired factorization \eqref{EqTrivialCompCoc}. For the $r = p^m$ odd case, in comparing with the classification of Gauss sums in Appendix \ref{SecPairFinAbGroups}, it is useful to keep in mind that if $p = 1$ mod 4, $r = 1$ mod 4 and if $p = -1$ mod 4, $r = (-1)^m$ mod 4. 

If $r = 1$ mod $4$, the properties of the Kronecker symbol ensures that we still have $(-s|r)_K = (s|r)_K$. Moreover, Taylor's explicit computation of the elementary Gauss sums, Proposition 3.9 of \cite{Taylor1984259} reviewed in Appendix \ref{SecPairFinAbGroups}, shows that 
\be
{\rm Gauss}(q_{r,1}) = (1|r)_K = 1
\ee
The same argument as in the even case then yields the factorization.

If $r = -1$ mod $4$, then $(-s|r)_K = -(s|r)_K$. Replacing \eqref{EqFactDetC} into \eqref{EqGaussSumDetOdd}, we find that
\begin{align}
\label{Eq??}
{\rm Gauss}(q_{321}) & \: = ({\rm Gauss}(q_{r,1}))^{n_\mathsf{U}}  (-1)^{n_{31}-n_{123}} \cdot \\
& E(o_3/o_{23}, o_2/o_{23}) E(o_2/o_{12}, o_1/o_{12}) \left(E(o_3/o_{13}, o_1/o_{13})\right)^{-1} \;. \notag
\end{align}
Taylor's computation shows that 
\be
{\rm Gauss}(q_{r,1}) = i(1|r)_K = i \;.
\ee
Moreover, \eqref{EqDecompKi} and \eqref{EqOrdKKR} allow us to compute:
\be
n = n_{12} + n_{23} + n_{31} - 2n_{123} + n_\mathsf{U}
\ee
hence, modulo 4,
\begin{align}
n_\mathsf{U} + 2n_{31} - 2n_{123} \: & = n - n_{12} - n_{23} + n_{31} \\
& = (n - n_{12}) + (n - n_{23}) - (n - n_{31})
\end{align}
which proves the proposition.
\end{proof}

\begin{proposition}
The homomorphism $T_{\tilde{K}_2,\tilde{K}_1}$ defined in \eqref{EqCompCompT} is unitary.
\end{proposition}
\begin{proof}
Let $v_x$ be a unit norm vector in $\mathcal{S}^e(M,x)$ for some $x \in w + \mathsf{G}_M$. Then 
\be
v := \frac{1}{n^{1/2}} \sum_{k \in \mathsf{L}_1} \rho(k) v_x
\ee
is a unit norm vector in $\mathsf{I}(\mathsf{L}_1)$. Varying $x$, we can obtain an orthonormal basis of $\mathsf{I}(\mathsf{L}_1)$ composed of vectors of this form. It is therefore sufficient to check that $T_{\tilde{K}_2,\tilde{K}_1}$ preserves the norm of these vectors. We have
\begin{align}
\left|T_{\tilde{K}_2,\tilde{K}_1} v \right| \: & = |\mathsf{L}_{12}|^{-1/2} |\mathsf{L}|^{-1/2} \left|\sum_{k' \in \mathsf{L}_2} \rho(k') v \right| \notag \\
& = |\mathsf{L}_{12}|^{-1/2} |\mathsf{L}|^{-1} \left| \sum_{k' \in \mathsf{L}_2} \sum_{k \in \mathsf{L}_1} \rho(k') \rho(k) v_x \right| \notag \\
& = |\mathsf{L}_{12}|^{-1/2} |\mathsf{L}|^{-1} \left| \sum_{k' \in \mathsf{L}_2} \sum_{k \in \mathsf{L}_1} \phi(k',k) \rho(k'+k) v_x \right| \notag \\
& = |\mathsf{L}_{12}|^{1/2} |\mathsf{L}|^{-1} \left| \sum_{k' \in \mathsf{L}_2/\mathsf{L}_{12}} \sum_{k \in \mathsf{L}_1/\mathsf{L}_{12}} \sum_{k'' \in \mathsf{L}_{12}} \phi(k',k) \rho(k'+k+k'') v_x \right| \\
& = |\mathsf{L}_{12}|^{1/2} |\mathsf{L}|^{-1} \left( \sum_{k' \in \mathsf{L}_2/\mathsf{L}_{12}} \sum_{k \in \mathsf{L}_1/\mathsf{L}_{12}} \sum_{k'' \in \mathsf{L}_{12}} |v_x|^2 \right)^{1/2} \notag \\
& = |\mathsf{L}_{12}|^{1/2} |\mathsf{L}|^{-1} \frac{|\mathsf{L}|^{1/2}}{|\mathsf{L}_{12}|^{1/2}} \frac{|\mathsf{L}|^{1/2}}{|\mathsf{L}_{12}|^{1/2}} |\mathsf{L}_{12}|^{1/2} \notag \\
& = 1 \notag
\end{align}
We only used definitions up to the third line. On the fourth line, we decomposed the sum over $\mathsf{L}_1$ and $\mathsf{L}_2$ into a sum over $\mathsf{L}_1/\mathsf{L}_{12}$, $\mathsf{L}_2/\mathsf{L}_{12}$ and two sums on $\mathsf{L}_{12}$, making choices of representatives in $\mathsf{L}_i$ of the classes in $\mathsf{L}_i/\mathsf{L}_{12}$. The summand is independent of one of the sums over $L_{12}$, yielding a factor $|\mathsf{L}_{12}|$. On the fifth line, we used the fact that all the summands on the fourth line belong to distinct orthogonal Hermitian lines $\mathcal{S}^e(M, y)$, $y \in w + \mathsf{G}_M$. As $\phi$ has unit modulus and $\rho$ is unitary, they drop out of the norm computation. The sums can now be performed trivially and we can conclude the computation.
\end{proof}

\subsection{Metaplectic correction}

Our construction so far can be seen as a Hermitian bundle $\mathscr{I}_{o,\epsilon}$ of irreducible modules for the Heisenberg group $\mathsf{H}_M$ over the discrete set $L_{\rm a,o,\epsilon} = {\rm Obj}(\mathcal{C}_{o,\epsilon})$ of oriented $\epsilon$-admissible Lagrangian subgroups of $\mathsf{G}_M$. The maps $T_{\tilde{K}_2,\tilde{K}_1}$ provides a notion of flat unitary parallel transport on $\mathscr{I}_{o,\epsilon}$ that is compatible with the module structure on the fibers.

In the present section, our aim is to transfer these structures to the set $L_{\rm a,\epsilon}$ of unoriented admissible Lagrangians. $\mathscr{I}_{o,\epsilon}$ is not the pull-back of a bundle over $L_{\rm a,\epsilon}$, so it has to be suitably twisted before it can be pushed down to $L_{\rm a,\epsilon}$. This procedure can be seen as a discrete analogue of the "metaplectic correction" of geometric quantization \cite{MR1183739}.

We use again the fact that as a group endowed with an alternating pairing, $\mathsf{G}_M$ decomposes into a direct product of groups of the form $\mathbb{Z}_r^{2n}$ and focus on a single factor. Let $\mathsf{L} \subset \mathsf{G}_M$ be an admissible Lagrangian subgroup. As before, we write $\tilde{\mathsf{L}}$ for the admissible LL subgroup of $\mathsf{H}_M$ canonically associated to $\mathsf{L}$. Consider the bundle $\mathscr{I}_\epsilon$ over $L_{\rm a,\epsilon}$ whose fiber over $\mathsf{L}$ is the set of pairs $(v, o_\mathsf{L})$ composed of a vector $v \in \mathsf{I}(\tilde{\mathsf{L}})$ and an orientation $o_\mathsf{L}$ on $\mathsf{L}$, subject to the equivalence relation
\be
\left( (o'_\mathsf{L}/o_\mathsf{L}|r)_K v, o_\mathsf{L} \right) \sim \left( v, o'_\mathsf{L} \right) \;.
\ee 
Here $o'_\mathsf{L}/o_\mathsf{L}$ is the element $s \in \mathbb{Z}_r$ such that $o'_\mathsf{L} = s o_\mathsf{L}$.

\begin{proposition}
\label{PropDefParTranspUnor}
The flat unitary transport on $\mathscr{I}_{o,\epsilon}$ induces a flat unitary parallel transport on $\mathscr{I}_{\epsilon}$, given by
\be
\label{EqDefParTranspUnor}
S_{\tilde{\mathsf{L}}_2, \tilde{\mathsf{L}}_1}(v_1, o_1) = \left(T_{(\tilde{\mathsf{L}}_2,o_2), (\tilde{\mathsf{L}}_1,o_1)} v_1, o_2 \right) \;,
\ee
where we made explicit the dependence of the parallel transport $T$ on arbitrarily chosen orientations $o_i$ of $\mathsf{L}_i$.
\end{proposition}
\begin{proof}
Let us first check that \eqref{EqDefParTranspUnor} is independent of the orientations $o_1$ and $o_2$. Under a change of $o_1$, we have
\begin{align}
S_{\tilde{\mathsf{L}}_2, \tilde{\mathsf{L}}_1}((o_1/o'_1|r)_K v_1, o'_1) \: & = \left((o_1/o'_1|r)_K T_{(\tilde{\mathsf{L}}_2,o_2), (\tilde{\mathsf{L}}_1,o'_1)} v_1, o_2 \right) \notag \\
& = \left(T_{(\tilde{\mathsf{L}}_2,o_2), (\tilde{\mathsf{L}}_1,o_1)} v_1, o_2 \right) \\
& = S_{\tilde{\mathsf{L}}_2, \tilde{\mathsf{L}}_1}(v_1, o_1) \;. \notag
\end{align}
Under a change of $o_2$, we have
\begin{align}
S_{\tilde{\mathsf{L}}_2, \tilde{\mathsf{L}}_1}(v_1, o_1) \: & = \left(T_{(\tilde{\mathsf{L}}_2,o'_2), (\tilde{\mathsf{L}}_1,o_1)} v_1, o'_2 \right) \notag \\
& = \left((o_2/o_2'|r)_K T_{(\tilde{\mathsf{L}}_2,o_2), (\tilde{\mathsf{L}}_1,o_1)} v_1, o'_2 \right) \\
& \sim \left(T_{(\tilde{\mathsf{L}}_2,o_2), (\tilde{\mathsf{L}}_1,o_1)} v_1, o_2 \right) \notag \;. \notag
\end{align}
This shows that $S_{\tilde{\mathsf{L}}_2, \tilde{\mathsf{L}}_1}$ is well-defined. The corresponding parallel transport on $\mathscr{I}_\epsilon$ is flat because of \eqref{EqCompatCompT}. 
\end{proof}

\subsection{Definition of the state space}

\label{SecDefStSp}

We are finally ready to define the Hilbert space that the topological field theory assigns to a closed $4\ell+2$-dimensional $(\mathfrak{W},\mathsf{Z}^{\rm flat},\epsilon)$-manifold $(M,\check{u})$ of Arf invariant zero, with $\check{u} \in \mathsf{Z}^{\rm flat}_M$.

Define
\be
\label{EqDefVMu}
V'_{M,u} = \bigoplus_{x \in u + \mathsf{E}_M/\mathsf{T}_M} S^e(M,x) \simeq \bigoplus_{k \in \mathsf{K}_M} V_{M,u+w(k)}\;.
\ee
For the second isomorphism, we picked a set of representatives $w(k) \in \beta^{-1}(k)$, $k \in \mathsf{K}_M$ to express $V'_{M,u}$ as the direct sum of the vector spaces defined in \eqref{EqDefVecSpV}. $V'_{M,u}$ carries the right action of $\mathsf{H}_M$ induced by \eqref{EqDefRho}, which we also write $\rho$. Only the right hand side of \eqref{EqDefVMu} carries a natural left action \eqref{EqDefRhoT} of $\mathsf{H}_M$, which is dependent on the lift $w: \mathsf{K}_M \rightarrow \mathsf{Y}^{\rm flat}_M$. Of course, as a Hilbert space $V'_{M,u}$ is independent of $w$.

Let $\mathcal{C}_\epsilon$ be the category whose objects are the Lagrangian subgroups $\mathsf{L} \subset \mathsf{G}_M$ admissible with respect to $\epsilon$, and that admits one morphism for each pair of objects. We define a functor $\mathcal{F}_{u,\epsilon}$ from $\mathcal{C}_\epsilon$ into the category $\mathcal{H}$ of Hilbert spaces as follows:
\be
\mathcal{F}_{u,\epsilon}(\mathsf{L}) = V'_{M,u}/\rho(\tilde{\mathsf{L}}) \;, \quad \mathcal{F}_{u,\epsilon}(\mathsf{L}_1 \rightarrow \mathsf{L}_2) = \bigoplus_{k \in \mathsf{K}_M} S_{\tilde{\mathsf{L}}_2, \tilde{\mathsf{L}}_1}|_{V_{M,w(k)}} \;.
\ee
The morphisms are simply the sum of the morphisms $S_{\tilde{\mathsf{L}}_2, \tilde{\mathsf{L}}_1}$ for each of the Hilbert spaces $V_{M,w(k)}$. Note that as the definition of $S_{\tilde{\mathsf{L}}_2, \tilde{\mathsf{L}}_1}$ is independent of $w(k)$, $\mathcal{F}_{u,\epsilon}$ is independent of the choice of lift $w$. Proposition \ref{PropDefParTranspUnor} ensures that $\mathcal{F}_{u,\epsilon}$ satisfies the axioms of a functor. Moreover, all the automorphisms in $\mathcal{C}_\epsilon$ are trivial. This implies that $\mathcal{F}_{u,\epsilon}$ admits a non-trivial space of invariant sections $\mathsf{I}_{M,u,\epsilon}$.

We finally define the state space of the topological field theory to be
\be
\mathcal{R}(M,u,\epsilon) := \mu_M \cdot \mathsf{I}_{M,u,\epsilon} \;.
\ee
Here, the Hilbert space $\mu_M \cdot \mathsf{I}_{M,u,\epsilon}$ should be understood as the Hilbert space $\mathsf{I}_{M,u,\epsilon}$, whose inner product has been multiplied by the rational number $\mu_M$ defined in Section \ref{SecMeas}. There is no natural $\mathsf{H}_M$-module structure on $\mathcal{R}(M,u,\epsilon)$. We can obtain one, induced by the left action $\tilde{\rho}$, after choosing a lift $w$.

\subsection{Disjoint unions and orientation flips}

\label{SecDisUnOrFlipStSp}

In this section, we drop the label $\epsilon$. It is understood that all the manifolds are endowed with extra structures. 

Suppose that $M$ can be expressed as a disjoint union of two manifolds: $M = M_1 \sqcup M_2$.  Then the pairing on $\mathsf{G}_M = \mathsf{G}_{M_1} \times \mathsf{G}_{M_2}$ is block diagonal and $\mathsf{H}_M = \mathsf{H}_{M_1} \times_{U(1)} \mathsf{H}_{M_2}$. By $\times_{U(1)}$, we mean the quotient of the direct product $\mathsf{H}_{M_1} \times \mathsf{H}_{M_2}$ by the anti-diagonal action of $U(1)$, given in the model of \eqref{EqExplRealHeisGrp} by $(k_1, t_1, k_2, t_2) \rightarrow (k_1, tt_1, k_2, t^{-1}t_2)$, $k_i \in \mathsf{G}_{M_i}$, $t,t_1,t_2 \in U(1)$. More conceptually, when seeing the Heisenberg groups $\mathsf{H}_{M_1}$ and $\mathsf{H}_{M_2}$ as $U(1)$ bundles over $\mathsf{G}_{M_1}$ and $\mathsf{G}_{M_2}$, then $H_M$ is obtained by taking the direct product on the base and the tensor product in the fibers, using the correspondence between $U(1)$ bundles and Hermitian line bundles.

If $M_1$ and $M_2$ carry extra structures $\epsilon_1$ and $\epsilon_2$ as described in Section \ref{SecExStruct}, they determine a unique structure $\epsilon_{12}$ on $M_1 \sqcup M_2$. Given $\hat{u}_1 \in \mathsf{Z}^{\rm flat}_{M_1}$ and $\hat{u}_2 \in \mathsf{Z}^{\rm flat}_{M_2}$ and using Proposition \ref{PropCompPreqThMonoidStruct}, we have $V'_{M,u_1 + u_2} = V'_{M_1, u_1} \otimes V'_{M_2,u_2}$. The action of the Heisenberg group is compatible with this tensor structure, hence $\mathsf{I}_{M,u_1 + u_2, \epsilon_{12}} = \mathsf{I}_{M_1, u_1, \epsilon_1} \otimes \mathsf{I}_{M_2, u_2, \epsilon_2}$. Therefore, we have:
\begin{proposition}
If $M$ is a $4\ell+2$-dimensional $(\mathfrak{W},\mathsf{Z}^{\rm flat}, \epsilon)$-manifold that can be decomposed into the disjoint union $M_1 \sqcup M_2$ of two $(\mathfrak{W},\mathsf{Z}^{\rm flat}, \epsilon)$-manifolds, then $\mathcal{R}(M) = \mathcal{R}(M_1) \otimes \mathcal{R}(M_2)$.
\end{proposition}

We now analyse the behavior of the state space under an orientation flip $M \rightarrow -M$, where $(M,\hat{u},\epsilon)$ is now again a generic $4\ell+2$-dimensional $(\mathfrak{W},\mathsf{Z}^{\rm flat}, \epsilon)$-manifold. We have
\begin{proposition}
$\mathcal{R}(-M) = \overline{\mathcal{R}(M)}$, the complex conjugate Hilbert space to $\mathcal{R}(M)$.
\end{proposition}
\begin{proof}
From the corresponding property of the prequantum theory, Proposition \ref{PropCompPreqDagStruct}, we have $V_{-M} = \overline{V_M}$. Moreover, as groups $\mathsf{G}_M = \mathsf{G}_{-M}$, but the pairing on $\mathsf{G}_{-M}$ is the inverse of the pairing \eqref{DefBimultPairG} on $\mathsf{G}_M$. It follows that the Heisenberg group $\mathsf{H}_{-M}$, as a line bundle over $\mathsf{G}_M$, is the complex conjugate of $\mathsf{H}_M$, with the multiplication maps given by the complex conjugates of those of $\mathsf{H}_M$. It is now clear that the construction of the present section applied to $-M$ involves taking the complex conjugates of all the factors, and we obtain in the end $\mathsf{I}_{-M,u,\epsilon} = \overline{\mathsf{I}_{M,u,\epsilon}}$.
\end{proof}

\section{Vectors associated to manifolds with boundaries}

\label{SecVectManBound}

Let $(M,\hat{u})$ be a $4\ell+3$-dimensional $(\mathfrak{W},\mathsf{Z}^{\rm flat})$-manifold with boundary. Recall that given a flat differential cohomology class $x \in \mathsf{Y}^{\rm flat}_{M,\check{w}}$, with $\check{w}$ a flat differential cocycle on $\partial M$, the prequantum theory $\mathcal{S}^e$ associates to $(M,x)$ a vector $\mathcal{S}^e(M,x)$ in the Hermitian line $\mathcal{S}^e(\partial M, w)$, where $w$ is the equivalence class of $\check{w}$ in $\mathsf{Y}^{\rm flat}_{\partial M}$. In analogy with the definition of the partition function on closed manifolds in Section \ref{SecPartFunc}, we would like to define the vector associated to $M$ as
\be
\label{EqDefRManBoundTent}
\mu_M \sum_{w \in \mathsf{E}_{\partial M}} \sum_{x \in u + \mathsf{E}_{M,\check{w}}} \mathcal{S}^e(M,x) \;.
\ee
Note that the sums are finite. We need to make a choice of differential cocycle $\check{w}$ to define the second sum, but Proposition \ref{PropPQFTdDimWB} ensures that the result is independent of this choice. Because of the anomaly of the prequantum field theory, described in Proposition \ref{PropDimHilbSp}, the terms in the sum are non-vanishing only if $\beta(u_\partial) \in a_{\partial M} + \mathsf{K}_{\partial M}$, where we wrote $u_\partial := u|_{\partial M}$.

It remains however to associate the expression \eqref{EqDefRManBoundTent} canonically to a vector in the Hilbert space $\mathcal{R}(\partial M, u_\partial, \epsilon)$, for $\epsilon$ a suitable extra structure on $\partial M$. This is not completely straightforward, because of the counterpart for manifolds with boundaries of the partition function anomaly described in Section \ref{SecPartFuncAnom}.

\subsection{Action of the LL subgroup}

As shown in Appendix \ref{SecPerPair}, $M$ determines a pair of Lagrangian subgroups $(\mathsf{L}_M, \mathsf{L}_{M,2\Lambda})$ of $(\mathsf{G}_{\partial M}, \mathsf{G}_{\partial M, 2\Lambda})$. We can pick a lift $j: \mathsf{G}_{\partial M} \rightarrow \mathsf{G}_{\partial M, 2\Lambda}$ such that $j(\mathsf{L}_M) \subset \mathsf{L}_{M,2\Lambda}$. The second admissibility condition is automatically satisfied by Proposition \ref{PropArfInvCobInv}. Following the discussion in Section \ref{SecExStruct}, we can as well choose a set of quadratic refinements $\{q_r\}$ that turn $(\mathsf{L}_M, \mathsf{L}_{M,2\Lambda})$ into an admissible pair with respect to $\epsilon = (j,\{q_r\})$. The choice of $\{q_r\}$ is in fact irrelevant for the discussion below. There is a corresponding LL subgroup $\tilde{\mathsf{L}}_M \subset \mathsf{H}_{\partial M}$. The $\mathsf{H}_{\partial M}$-module $V'_{\partial M,u_\partial}$ defined in \eqref{EqDefVMu} is the direct sum of $|\mathsf{K}_{\partial M}|$ copies of the regular representation of $\mathsf{H}_{\partial M}$, with the right action denoted by $\rho$. The state space $\mathcal{R}(\partial M, u_\partial, \epsilon)$ is canonically isomorphic to $V'_{\partial M,u_\partial}/\rho(\tilde{\mathsf{L}}_M)$. Proposition \ref{PropPQFTdDimWB} ensures that \eqref{EqDefRManBoundTent} is a vector in $V'_{\partial M,u}$. 
\begin{lemma}
\label{CorVectManBoundInStSp}
If \eqref{EqDefRManBoundTent} is invariant under $\rho(\tilde{k})$ for all $k \in \mathsf{L}_M$, it determines canonically a vector in $\mathcal{R}(\partial M,u_\partial, \epsilon)$.
\end{lemma}
\begin{proof}
This follows from the remarks above and from the canonical identification of the subspace of $V'_{\partial M, u_\partial}$ consisting of vectors invariant under the action of $\rho(\tilde{\mathsf{L}}_M)$ with the quotient $V'_{\partial M, u_\partial}/\rho(\tilde{\mathsf{L}}_M)$.
\end{proof}
\begin{proposition}
\label{PropInvStateActLLSubGrp}
Assume that $k \in \mathsf{L}_M$ and that its extension to $M$ lies in $\mathsf{G}_M$, i.e. can be represented by a differential cocycle with vanishing characteristic. Let $\tilde{k} $ be the lift of $k$ to $\tilde{\mathsf{L}}_M$. Then the vector \eqref{EqDefRManBoundTent} is invariant under $\rho(\tilde{k})$.
\end{proposition}
\begin{proof}
We first show that for each $k$ as above
\be
\label{EqPropUniLLSubG1}
Q(x,j(k)) \mathcal{S}^e(M,x) = \mathcal{S}^e(M, x+k) \;.
\ee
(We abuse slightly the notation and write $x$ and $k$ both for the differential cohomology classes on $M$ and for their restrictions to $\partial M$.) To see that \eqref{EqPropUniLLSubG1} holds, we pick a differential form lift $\hat{j}(k)$ of $j(k)$ and construct the associated flat differential cocycle $\check{j}(k) = (0,\hat{j}(k), 0)$. We construct a closed $4\ell+3$-dimensional $\mathfrak{W}$-manifold $N$ endowed with a non-flat differential cocycle $\check{x}_N$ by gluing together the following three manifolds:
\be
(-M, \check{x}) \;, \quad (\partial M \times I, \check{y}_{\check{x}, \hat{j}(k)}) \;, \quad (M, \check{x} + \check{j}(k)) \;,
\ee
where we used the notation of \eqref{EqDefCocYOnCyl}. $S^e(N, x_N)$ computes the scalar product of the right-hand side of \eqref{EqPropUniLLSubG1} with the left-hand side. As the vectors have unit norm and belong to the same complex line, if we prove that $S^e(N, x_N) = 1$, we prove \eqref{EqPropUniLLSubG1}.

To prove that $S^e(N, x_N) = 1$, remark that $N$ is the boundary of a manifold homeomorphic to $M \times I$. We can use Proposition \ref{PropActFromBoundedMan} to compute $S^e(N, x_N)$ on $M \times I$. The Wu class of $M \times I$ vanishes and $x_N$ extends to $M \times I$ as a differential cocycle whose curvature is of the form $\hat{j}(k) \wedge d\hat{t}$. Therefore the integrand of \eqref{EqACtFromBoundedMan} vanishes identically, and $S^e(N, x_N) = 1$. (There is a slight subtlety here due to the fact that the manifold bounded by $N$ is homeomorphic but not diffeomorphic to $M \times I$, the failure of smoothness occuring at the corners $\partial M \times \partial I$. However, if we pick $\hat{t}$ to be constant in neighborhoods of $\partial I$, the curvature pulled back from $M \times I$ through this homeomorphism is smooth.)

\eqref{EqPropUniLLSubG1} implies that
\be
\rho(k, 1) \mathcal{S}^e(M, x) = \mathcal{S}^e(M, x+k) \;,
\ee
showing that \eqref{EqDefRManBoundTent} is invariant under $\rho(\tilde{\mathsf{L}}_M)$.
\end{proof}
Remark that if the extension of $k$ has a non-vanishing characteristic, the proof above fails because we cannot endow $M \times I$ with a suitable differential cocycle that restrict to $\check{x}_N$ on $N$.

\subsection{Anomaly}

\label{SecAnOnManWithBound}

The following proposition characterizes the failure of \eqref{EqPropUniLLSubG1} for generic $k \in \mathsf{L}_M$. Pick a cocycle lift $\hat{j}(k)$ of $j(k)$ and construct the associated flat differential cocycle $\check{j}(k) = (0,\hat{j}(k), 0)$. Extend it to $M$. We construct a $4\ell+3$-dimensional $\mathfrak{W}$-manifold $N$ endowed with a non-flat differential cocycle $\check{x}_N$ by gluing together the cylinder $(\partial M \times I, \check{y}_{0, \hat{j}(k)})$ and $(M, \check{j}(k))$, where we used again the notation of \eqref{EqDefCocYOnCyl}. Remark that $\check{x}_N$ vanishes on $\partial N \simeq M \times \{0\}$, so the exponentiated action $S^e(M,x_N)$ is a well-defined complex number of unit modulus.
\begin{proposition}
We have
\be
\label{EqSQSEqS}
\mathcal{S}^e(-M, x+k) Q(x,j(k)) \mathcal{S}^e(M,x) = \left( S^e(N,x_N) \right)^{-1} \;.
\ee
\end{proposition}
\begin{proof}
Using the fact that $N$ is diffeomorphic to $M$, the left-hand side can be rewritten
\be
\label{EqRatSxSxpxn}
\frac{\mathcal{S}^e(N,x)}{\mathcal{S}^e(N, x+x_N)} \;,
\ee
where $x$ denotes here the class on $N$ obtained from $x$ on $M$ by pulling back its boundary value to $\partial M \times I \subset N$. \eqref{EqRatSxSxpxn} can be computed by evaluating the exponentiated action on the closed manifold obtained by gluing $(N,x)$ and $(-N,x+x_N)$ along their boundary. 

Let $\check{x} = (\hat{a}, \hat{h}, 0)$ be a cocycle representative of $x$ on $N$, and write also $\check{x}_N = (\hat{a}_N, \hat{h}_N, \omega_N)$. Revisiting the proof of Proposition \ref{PropActQuadrRef} and taking into account that $\check{x}_N$ is not flat, we have
\be
\bar{l}(\check{x} + \check{x}_N) = \left(\hat{l}(\check{x}) + \hat{l}(\check{x}_N) + \frac{1}{2} \hat{a} \cup \hat{h}_N + \frac{1}{2} \hat{a}_N \cup \hat{h} + \frac{1}{2} \hat{h} \cup \hat{\omega}_N, (\hat{a} + \hat{a}_N)_2 \right )
\ee
\be
\bar{l}(\check{x}) \boxplus \bar{l}(\check{x}_N) = \left (\hat{l}(\check{x}) + \hat{l}(\check{x}_N) + \frac{1}{2} \hat{a} \cup_1 \hat{a}_N, (\hat{a} + \hat{a}_N)_2 \right )
\ee
\begin{align}
\label{EqWEN1}
\bar{l}(\check{x} + \check{x}_N) \boxminus \bar{l}(\check{x}) \boxminus \bar{l}(\check{x}_N)  \: & = \left( \frac{1}{2} \hat{a} \cup \hat{h}_N + \frac{1}{2} \hat{h} \cup \hat{a}_N + \frac{1}{2} \hat{h} \cup \hat{\omega}_N - \frac{1}{2} d(\hat{h} \cup_1 \hat{a}_N), 0 \right ) \notag \\
& = \left( \hat{a} \cup \hat{h}_N + \frac{1}{2} d\left(\hat{h} \cup \hat{h}_N - \hat{h} \cup_1 \hat{a}_N \right) , 0 \right) 
\end{align}
Making an arbitrary choice of E-homology class $\bar{n}$ of $N$, $1/2\pi i$ times the log of \eqref{EqRatSxSxpxn} can be computed as
\be
\label{EqWEN2}
\langle \bar{l}(\check{x}) - \bar{l}(\check{x}+ \check{x}_N), \bar{n} \rangle = -\langle \bar{l}(\check{x}_N), \bar{n} \rangle - \langle \hat{a} \cup \hat{h}_N, [N] \rangle \;.
\ee
The exact term in \eqref{EqWEN1} does not contribute because $\check{x}_N = 0$ on $\partial N$.
The first term on the right hand side of \eqref{EqWEN2} is $-S(N,x_N)$. It is independent of the choice of $\bar{n}$ because $\check{x}_N$ vanishes on $\partial N$. The second term vanishes modulo 1 on $M \subset N$ because $\hat{a}$ is $\Lambda$-valued and $\hat{h}_N$ is $\Lambda^\ast$-valued there. On $\partial M \times I \subset N$, it vanishes for dimensional reason, as the cup product can be expressed as a function times a cocycle pulled back from $\partial M$. We find therefore that \eqref{EqRatSxSxpxn} is given by $S^e(N,x_N)$ proving the proposition.
\end{proof}

So if $S^e(N,x_N) \neq 0$ for some $k \in \mathsf{L}_M$, \eqref{EqDefRManBoundTent} is not invariant under the action $\rho(\tilde{\mathsf{L}}_M)$ and the only vector in $\mathcal{R}(\partial M,u_\partial, \epsilon)$ that can be canonically associated to \eqref{EqDefRManBoundTent} is the zero vector. We can see this phenomenon as an analogue of the partition function anomaly described in Section \ref{SecPartFuncAnom}. The analogy is obvious in case $k = 0$. In this case, $x_N$ is an element of $\mathsf{E}_{M,\partial M}$ and \eqref{EqSQSEqS} reduces to
\be
\frac{\mathcal{S}^e(M,x)}{\mathcal{S}^e(M, x + x_N)} = \left( S^e(M,x_N) \right)^{-1} \;.
\ee
The remarks under Proposition \ref{PropActQRLP} imply that $S^e(M; x_N) = \pm 1$. If $S^e(M; x_N) = - 1$, then the terms in the sum \eqref{EqDefRManBoundTent} associated to $x$ and $x+x_N$ cancel in pairs, so \eqref{EqDefRManBoundTent} vanishes.

\subsection{Definition}

\label{SecDefVectManBound}

In view of the discussion in the previous section, we define
\be
\label{EqDefRManBound}
\mathcal{R}(M,u) := \mu_M \pi_{\tilde{\mathsf{L}}_M}\sum_{w \in \mathsf{E}_{\partial M}} \sum_{x \in u + \mathsf{E}_{M,\check{w}}} \mathcal{S}^e(M,x) \;,
\ee
where $\pi_{\tilde{\mathsf{L}}_M}$ is the projection onto the subspace of $V'_{\partial M,u}$ consisting of vectors invariant under $\rho(\tilde{\mathsf{L}}_M)$. By Lemma \ref{CorVectManBoundInStSp}, $\mathcal{R}(M,u) \in \mathcal{R}(\partial M,u_\partial, \epsilon)$.

From the discussion in Section \ref{SecDisUnOrFlipStSp}, the definition \eqref{EqDefRManBound} is compatible with the monoidal and dagger structures of $\mathcal{B}_{\mathfrak{W}, {\rm flat}, \epsilon}^{4\ell+3,1}$ and $\mathcal{H}$.

\section{Gluing}

\label{SecGluing}

We now want to prove that $\mathcal{R}$ is really a functor from the cobordism category of $(\mathfrak{W}, \mathsf{Z}^{\rm flat},\epsilon)$-manifolds to the category of Hilbert spaces. In view of Proposition \ref{PropAltCondCheckFunc} and the discussions of Sections \ref{SecDisUnOrFlipStSp} and \ref{SecDefVectManBound}, it remains to show that the gluing condition \eqref{EqConstrGluing} holds.

\subsection{Setup}

\label{SecSetupGluing}

We take $(M,\check{u})$ to be a $4\ell+3$-dimensional $(\mathfrak{W}, \mathsf{Z}^{\rm flat})$-manifold, possibly with boundary. We write $(M_N, \check{u}')$ for the $4\ell+3$-dimensional $(\mathfrak{W}, \mathsf{Z}^{\rm flat})$-manifold obtained from $M$ by cutting along a codimension 1 closed submanifold $N \subset M$ of Arf invariant zero, see Figure \ref{Fig-Prop51-1}. We will write $\pi: M_N \rightarrow M$ for the gluing map and $u' := \pi^\ast(u)$. We endow $N$ and $\partial M$ with extra structures $\epsilon_N$ and $\epsilon_{\partial M}$ necessary for the construction of the state space of the theory on these manifolds, as described in Section \ref{SecExStruct}. Note that the existence of $\epsilon_N$ implies that $N$ has Arf invariant zero (see Section \ref{SecCobInv}). They induce a structure $\epsilon_{\partial M_N}$ on $\partial M_N = N \sqcup -N \sqcup \partial M$. The pair of Lagrangian subgroups $(\mathsf{L}_{M_N}, \mathsf{L}_{M_N, 2\Lambda}) \subset (\mathsf{G}_{M_N}, \mathsf{G}_{M_N, 2\Lambda})$ naturally determined by $M_N$ is not necessarily admissible with respect to the induced structure $\epsilon_{\partial M_N}$. As will be discussed in Section \ref{SecProofGluinFormAnomCase}, this is another incarnation of the anomaly of Sections \ref{SecPartFuncAnom} and \ref{SecAnOnManWithBound}.

The gluing formula of Proposition \ref{PropAltCondCheckFunc} reads
\be
\label{EqGluingConstr}
\mathcal{R}(M,u) \stackrel{?}{=} {\rm Tr}_{\mathcal{R}(N,u|_N, \epsilon_N )} \mathcal{R}(M_N,u') \;.
\ee
To understand this equation, remark that $\mathcal{R}(M_N,u') \in \mathcal{R}(\partial M_N,u'|_{\partial M_N}, \epsilon_{\partial M_N}) \simeq \mathcal{R}(N, u|_N, \epsilon_N) \otimes (\mathcal{R}(N,u|_N, \epsilon_N))^\ast \otimes \mathcal{R}(\partial M,u|_{\partial M}, \epsilon_{\partial M})$. ${\rm Tr}_{\mathcal{R}(N,u|_N, \epsilon_N)}$ is then the canonical pairing 
\be
{\rm Tr}_{\mathcal{R}(N,u|_N, \epsilon_N)} : \mathcal{R}(N,u|_N, \epsilon_N) \otimes (\mathcal{R}(N,u|_N, \epsilon_N))^\ast \otimes \mathcal{R}(\partial M,u|_{\partial M}, \epsilon_{\partial M}) \rightarrow \mathcal{R}(\partial M,u|_{\partial M}, \epsilon_{\partial M})
\ee 
We also recall that $\mathcal{R}(N,u|_N, \epsilon_N) = \mu_N \cdot\mathsf{I}_{N,u|_N, \epsilon_N}$, with the product understood as a rescaling of the inner product of $\mathsf{I}_{N,u|_N, \epsilon_N}$ by $\mu_N$.

\subsection{Trace homomorphism}

\label{SecTracHom}

We will need the following characterization of the trace homomorphism. In the model of $\mathsf{H}_N$ provided by a trivialization of the lines $L_k$, described around \eqref{EqExplRealHeisGrp}, we call \emph{diagonal} the elements of $\mathsf{H}_N \times_{U(1)} \mathsf{H}_{-N}$ of the form $(k,t,k,\bar{t}) \sim (k,1,k,1)$, $k \in \mathsf{G}_N$, $t \in U(1)$. The subgroup $\mathsf{D}_N$ of diagonal elements is commutative and isomorphic to $\mathsf{G}_N$. From now on, we will see $\mathsf{D}_N$ as a subgroup of $\mathsf{H}_{\partial M_N}$.

Recall from Section \ref{SecDefStSp} that $\mathsf{I}_{\partial M_N, u'|_{\partial M_N}, \epsilon_{\partial M_N}}$ does not carry a natural action of $\mathsf{H}_{\partial M_N}$. However, we do obtain an action $\tilde{\rho}_w$ after we choose a lift $w: \mathsf{K}_{\partial M_N} \rightarrow \mathsf{Y}^{\rm flat}_{\partial M_N}$. We have $\mathsf{K}_{\partial M_N} \simeq \mathsf{K}_{N} \times \mathsf{K}_{-N} \times \mathsf{K}_{\partial M}$ and $\mathsf{Y}^{\rm flat}_{\partial M_N} \simeq \mathsf{Y}^{\rm flat}_{N} \times \mathsf{Y}^{\rm flat}_{-N} \times \mathsf{Y}^{\rm flat}_{\partial M}$. We choose the lift $w$ so that it sends elements of the diagonal subgroup in $\mathsf{K}_{\partial M_N}$ to diagonal elements in $\mathsf{Y}^{\rm flat}_{\partial M_N}$.

\begin{proposition}
\label{PropCharTrAsProj}
${\rm Tr}_{\mathsf{I}_{N, u|_N, \epsilon_N}}$ is the projection onto the subspace of $\mathsf{I}_{\partial M_N, u'|_{\partial M_N}, \epsilon_{\partial M_N}}$ invariant under the action $\tilde{\rho}_w(\mathsf{D}_N)$, composed with the trace ${\rm Tr}_{V'_{N,u|_N}}$.
\end{proposition}
\noindent Note that the action $\tilde{\rho}(\mathsf{D}_N)$, and hence the associated expression for the trace, is independent of the choice of lift $w$, as long as the latter sends the diagonal of $\mathsf{K}_{\partial M_N}$ into the diagonal subgroup of $\mathsf{Y}^{\rm flat}_{\partial M_N}$.
\begin{proof}
Let us pick admissible Lagrangian subgroups $\mathsf{L}_1 \subset \mathsf{G}_{\partial M}$, $\mathsf{L}_2 \subset \mathsf{G}_N \simeq \mathsf{G}_{-N}$ with respect to $\epsilon_{\partial M}$ and $\epsilon_N$. Then $\mathsf{L}_3 := \mathsf{L}_1 \times \mathsf{L}_2 \times \mathsf{L}_2$ is an admissible Lagrangian subgroup of $\mathsf{G}_{\partial M_N}$ with respect to $\epsilon_{\partial M_N}$. Let $\tilde{\mathsf{L}}_1$, $\tilde{\mathsf{L}}_2$ and $\tilde{\mathsf{L}}_3$ be the associated admissible LL subgroups of $\mathsf{H}_{\partial M}$, $\mathsf{H}_{N}$, and $\mathsf{H}_{\partial M_N}$. Let us write $\mathsf{I}(\tilde{\mathsf{L}}_1) \subset V'_{\partial M, u|_{\partial M}}$, $\mathsf{I}(\tilde{\mathsf{L}}_2) \subset V'_{N, u|_{N}}$ and $\mathsf{I}(\tilde{\mathsf{L}}_3) \subset V'_{\partial M_N, u'|_{\partial M_N}}$ for the subspaces invariants under the actions of $\rho(\tilde{\mathsf{L}}_1)$, $\rho(\tilde{\mathsf{L}}_2)$ and $\rho(\tilde{\mathsf{L}}_3)$. These three subspaces are canonically isomorphic to $\mathsf{I}_{\partial M, u|_{\partial M}, \epsilon_{\partial M}}$, $\mathsf{I}_{N, u|_{N}, \epsilon_{N}}$ and $\mathsf{I}_{\partial M_N, u'|_{\partial M_N}, \epsilon_{\partial M_N}}$ respectively. Moreover, we have
\be
\label{EqRealIrrModRPartMN}
\mathsf{I}(\tilde{\mathsf{L}}_3) \simeq \mathsf{I}(\tilde{\mathsf{L}}_1) \otimes \mathsf{I}(\tilde{\mathsf{L}}_2) \otimes \overline{\mathsf{I}(\tilde{\mathsf{L}}_2)} \;.
\ee
In under this isomorphism, ${\rm Tr}_{\mathsf{I}_{N, u|_N, \epsilon_N}}$ is the canonical pairing between $\mathsf{I}(\tilde{\mathsf{L}}_2)$ and $\overline{\mathsf{I}(\tilde{\mathsf{L}}_2)}$. 

Let us now focus on the summand $\mathsf{I}_k(\tilde{\mathsf{L}}_2) := V_{N, u + w(k)} \cap \mathsf{I}(\tilde{\mathsf{L}}_2) \subset \mathsf{I}(\tilde{\mathsf{L}}_2)$ associated to $k \in \mathsf{K}_N$. We pick a Lagrangian subgroup $\mathsf{L}_2^\perp$ complementary to $\mathsf{L}_2$ in $\mathsf{G}_N$.  Moreover, we choose a unit vector $c_0 \in S^e(N, u + w(k))$. We obtain a unit vector 
\be
v = \frac{1}{|\mathsf{L}_2|^{1/2}} \sum_{l \in \mathsf{L}_2} \rho(l) v_0 \in \mathsf{I}_k(\tilde{\mathsf{L}}_2) \;.
\ee
The action $\rho(\mathsf{L}_2^\perp)$ generates an orthonormal basis $\{e_i\}$ of $\mathsf{I}_k(\tilde{\mathsf{L}}_2)$ from $v$. We obtain in this way an orthonormal basis 
\be
\label{EqOrthBasINI-N}
\{e_i \otimes \bar{e}_j\}_{i,j = 1,...,n_N}
\ee
of $\mathsf{I}_k(\tilde{\mathsf{L}}_2) \otimes \overline{\mathsf{I}_k(\tilde{\mathsf{L}}_2)}$. ${\rm Tr}_{\mathsf{I}_{N, u|_N, \epsilon_N}}$, restricted to $\mathsf{I}_k(\tilde{\mathsf{L}}_2)$, is the projection onto the 1-dimensional subspace generated by $\sum_i e_i \otimes \bar{e}_i$.

Let us now understand the action $\tilde{\rho}(\mathsf{D}_N)$ on the orthonormal basis \eqref{EqOrthBasINI-N}. Using the isomorphism $\mathsf{D}_N \simeq \mathsf{G}_N \simeq \mathsf{L}_2 \times \mathsf{L}_2^\perp$, write $\mathsf{D}_N = \mathsf{D}_2 \times \mathsf{D}_2^\perp$. Suppose that $e_i = \rho(l_i)v$, $l_i \in \mathsf{L}^\perp_2$. Then, for each $l \in \mathsf{L}_2$,
\be
\tilde{\rho}(l) e_i = B(l, l_i)e_i
\ee
As the pairing identifies $\mathsf{L}^\perp_2$ with the dual of $\mathsf{L}_2$, we see that the invariant subspace of $\tilde{\rho}(\mathsf{D}_2)$ is generated by the vectors $\{e_i \otimes \bar{e}_i\}_{i = 1,...,n_N}$. Moreover, by the definition of the $e_i$, the action $\tilde{\rho}(\mathsf{D}_2^\perp)$ permutes them transitively. This implies that up to a phase, the only invariant unit vector under $\tilde{\rho}(\mathsf{D}_N)$ is 
\be
e_{\rm Tr} := \frac{1}{|\mathsf{L}_2|^{1/2}} \sum_i e_i \otimes \bar{e}_i \;.
\ee
This argument is valid for each of the direct summands $\mathsf{I}_k(\tilde{\mathsf{L}}_2)$, which proves the proposition.
\end{proof}

\subsection{Proof of the gluing formula in the anomaly-free case}

\label{SecProofGluing}

Let $\mathsf{L}^\mathsf{E}_{M_N}$ be the Lagrangian subgroup of $\mathsf{E}_{\partial M_N}$ determined by $M_N$, and let $\mathsf{L}^\mathsf{E}_M$ be the Lagrangian subgroup of $\mathsf{E}_{\partial M}$ determined by $M$, see Lemma \ref{LemRestLagr}. Let $\mathsf{E}_N^{\rm ext}$ be the subgroup of $\mathsf{E}_N$ consisting of cohomology classes that admit an extension to $M$ vanishing on $\partial M$.

We describe now the structure of $\mathsf{L}^\mathsf{E}_{M_N}$. $\mathsf{L}^\mathsf{E}_{M_N} \subset \mathsf{E}_{\partial M_N} = \mathsf{E}_{\partial M} \times \mathsf{E}_{N} \times \mathsf{E}_{-N}$ contains a subgroup $\mathsf{P}$ consisting of elements of the form $(k,l,l)$, $k \in \mathsf{E}_{\partial M}$, $l \in \mathsf{E}_{N} \simeq \mathsf{E}_{-N}$. Via gluing, an element of $\mathsf{P}$ can be pushed-forward to a cohomology class on $M$ restricting to $l$ on $N$. Hence we actually have $k \in \mathsf{L}^\mathsf{E}_M$. The inverse of the push-forward map is $\pi^\ast$, so all $k \in \mathsf{L}^\mathsf{E}_M$ appear in this way.
Moreover, $\mathsf{E}^{\rm ext}_N$ can be embedded in $\mathsf{P}$ as elements of the form $(0,l,l)$, $l \in \mathsf{E}^{\rm ext}_N$. We have therefore
\be
\mathsf{P}/\mathsf{E}^{\rm ext}_N \simeq \mathsf{L}^\mathsf{E}_M \;.
\ee

We now assume that the pair of Lagrangian subgroups $(\mathsf{L}_{M_N}, \mathsf{L}_{M_N, 2\Lambda})$ of $(\mathsf{G}_{M_N}, \mathsf{G}_{M_N, 2\Lambda})$ is admissible with respect to $\epsilon_{\partial M_N}$. Equivalently, we require that there is a canonical isomorphism between the realizations $\mathsf{I}(\tilde{\mathsf{L}}_3)$ and $\mathsf{I}(\tilde{\mathsf{L}}_{M_N})$ of the state space. This ensures that we can apply Proposition \ref{PropCharTrAsProj} in order to compute the trace of $\mathcal{R}(M_N,u')$. 

\begin{proposition}
\label{PropProofGluing}
If $(\mathsf{L}_{M_N}, \mathsf{L}_{M_N, 2\Lambda})$ is an admissible pair with respect to $\epsilon_{\partial M_N}$, \eqref{EqGluingConstr} holds.
\end{proposition}
\begin{proof}

Following Proposition \ref{PropCharTrAsProj}, we express ${\rm Tr}_{\mathcal{R}(N,u|_N)} = \mu_N {\rm Tr}_{\mathsf{I}_{N,u|_N, \epsilon_N}}$ as the projection on the subspace invariant under the action $\tilde{\rho}(\mathsf{D}_N)$, composed with the trace ${\rm Tr}_{V'_{N,u|_N}}$. We also use the definition of $\mathcal{R}(M_N,u')$ to obtain
\be
{\rm Tr}_{\mathcal{R}(N,u|_N)} \mathcal{R}(M_N,u') = \frac{\mu_N \mu_{M_N}}{|\mathsf{D}_N|} {\rm Tr}_{V'_{N,u|_N}} \sum_{a \in \mathsf{D}_N} \tilde{\rho}(a) \sum_{v' \in \mathsf{E}_{\partial M_N}} \sum_{x' \in u' + \mathsf{E}_{M_N,\check{v}'}} \mathcal{S}^e(M_N, x') \;,
\ee
where $\check{v}'$ a differential cocycle representing the class $v'$. The sum over $\mathsf{E}_{\partial M_N}$ reduces to a sum over $\mathsf{L}^\mathsf{E}_{M_N}$, so the right-hand side becomes 
\be
\frac{\mu_N \mu_{M_N}}{|\mathsf{D}_N|} {\rm Tr}_{V'_{N,u|_N}} \sum_{a \in \mathsf{D}_N} \sum_{k' \in \mathsf{L}^\mathsf{E}_{M_N}} \sum_{x' \in u' + \mathsf{E}_{M_N, \check{k}'}} \tilde{\rho}(a)  \mathcal{S}^e(M_N, x')\;,
\ee
where $\check{k}'$ is a differential cocycle representing the class $k'$. We now use the fact that the sum over the elements of $\mathsf{L}^\mathsf{E}_{M_N}$ outside $\mathsf{P}$ produce a vector orthogonal to the invariant subspace of the action $\tilde{\rho}(\mathsf{D}_N)$. Moreover, the action of $\tilde{\rho}(\mathsf{D}_N)$ on the sum over $\mathsf{P}$ followed by the trace coincides with the trace. The sum over $\mathsf{D}_N$ yields a factor $|\mathsf{D}_N|$ cancelling the denominator of the prefactor, yielding
\be
\mu_N \mu_{M_N} {\rm Tr}_{V'_{N,u|_N}} \sum_{k' \in \mathsf{P}} \sum_{x' \in u' + \mathsf{E}_{M_N, \check{k}'}} \mathcal{S}^e(M_N, x') \;.
\ee
We use the isomorphism $\mathsf{P} \simeq \mathsf{L}^\mathsf{E}_M \times \mathsf{E}_N^{\rm ext}$. This isomorphism is not canonical, so we need to choose a lift from $\mathsf{L}^\mathsf{E}_M$ into $\mathsf{P}$. We denote the lift of $k \in \mathsf{L}^\mathsf{E}_M$ by $k'$. Let $\check{j}$ be a cocycle lift of $j \in \mathsf{E}^{\rm ext}_N$ and by a slight abuse of notation, let us also write $\check{j}$ for the corresponding cocycle valued in the diagonal part of $\mathsf{E}_N \times \mathsf{E}_{-N} \subset \mathsf{E}_{\partial M_N}$. We obtain
\be
\mu_N \mu_{M_N} {\rm Tr}_{V'_{N,u|_N}} \sum_{k \in \mathsf{L}^\mathsf{E}_M} 
\sum_{j \in \mathsf{E}_N^{\rm ext}} 
\sum_{x' \in u' + \mathsf{E}_{M_N,\check{k}' + \check{j}}}  \mathcal{S}^e(M_N, x') \;.
\ee
We perform the trace, using in particular the fact that the prequantum theory $\mathcal{S}^e$ satisfies the gluing relation. In the process, we replace the sum over $\mathsf{E}^{\rm ext}_N \times \mathsf{E}_{M_N,\check{v}}$ by a sum over $\mathsf{E}_{M,\check{v}'}$. This produces a factor $|{\rm ker}(\iota_{N,M})|$, where $\iota_{N,M}$ is is the homomorphism $\iota_{N,M}: \mathsf{E}_{M,N \sqcup \partial M} \rightarrow \mathsf{E}_{M,\partial M}$ induced by the inclusion $\mathsf{F}_{M, N \sqcup \partial M} \subset \mathsf{F}_{M,\partial M}$. This step is justified in Lemma \ref{LemCompTracGluing} below.
\be
\mu_N \mu_{M_N} |{\rm ker}(\iota_{N,M})| \sum_{k \in \mathsf{L}^\mathsf{E}_M} \sum_{x \in u + \mathsf{E}_{M,\check{k}}} \mathcal{S}^e(M, x) \;.
\ee
We convert the sum over $\mathsf{L}^\mathsf{E}_M$ into a sum over $\mathsf{E}_{\partial M}$ and use Proposition \ref{PropRelMeasures} to rewrite the prefactor. Finally, we use the definition of $\mathcal{R}(M,u)$ to obtain
\be
\mu_M \sum_{v \in \mathsf{E}_{\partial M}} \sum_{x \in u + \mathsf{E}_{M,\check{v}}} 
\mathcal{S}^e(M, x) = \mathcal{R}(M,u)
\ee
\end{proof}

\begin{lemma}
\label{LemCompTracGluing}
In the notation of Proposition \ref{PropProofGluing}, we have
\be
\label{EqCompTracGluing}
{\rm Tr}_{V'_{N,u|_N}} \sum_{k \in \mathsf{L}^\mathsf{E}_M} 
\sum_{j \in \mathsf{E}_N^{\rm ext}} 
\sum_{x' \in u' + \mathsf{E}_{M_N, \check{k}' + \check{j}}} \mathcal{S}^e(M_N, x') =  |{\rm ker}(\iota_{N,M})| \sum_{k \in \mathsf{L}^\mathsf{E}_M} \sum_{x \in u + \mathsf{E}_{M, \check{k}}} \mathcal{S}^e(M, x)
\ee
\end{lemma}
\begin{proof}
To prove this Lemma, we will simply compare the terms in the sums on both sides and see that they coincide up to a multiplicity $|{\rm ker}(\iota_{N,M})|$.

Let us write $\mathsf{E}'_{M_N} := \bigcup_{\check{w}' | \check{w}'|_N = \check{w}'|_{-N}} \mathsf{E}_{M_N,\check{w}'}$, the subgroup of $\mathsf{E}_{M_N}$ consisting of elements whose restriction to $N, -N \subset \partial M_N$ coincide. We see that the third sum on the left-hand side involves only elements in $\mathsf{E}'_{M_N}$, because $k' \in \mathsf{P}$ and $j \in \mathsf{E}^{\rm ext}_N$. 

Gluing induces a homomorphism $\pi_\ast$ from $\mathsf{F}_{M_N,\check{w}'}$ to $\mathsf{F}_{M,\check{w}}$, where $\check{w} = \check{w}'|_{\partial M}$, $\check{w}'|_N = \check{w}'|_{-N}$. If $\check{x} \in \mathsf{F}_{M_N,\check{0}}$ is exact, then $\pi_\ast(\check{x})$ is exact as well, so we obtain a homomorphism $\pi^{\mathsf{E}}_\ast$ from $\mathsf{E}'_{M_N}$ to $\mathsf{E}_{M} := \bigcup_{\check{w} \in \mathsf{F}_{\partial M}} \mathsf{E}_{M,\check{w}}$. The gluing relation for the prequantum theory ensures that 
\be
\label{EqValActPushForw}
{\rm Tr}_{V'_{N,u|_N}} \mathcal{S}^e(M_N,x) = \mathcal{S}^e(M,\pi^{\mathsf{E}}_\ast(x)) \;,
\ee 
so we can use $\pi^{\mathsf{E}}_\ast$ to compare the terms on both sides. 

$\pi^\mathsf{E}_\ast$ is surjective, because for all $x$, $\pi_\ast^\mathsf{E}(\pi^\ast(x)) = x$. We now compute its kernel. Consider the homomorphism $\iota_{N,M}: \mathsf{E}_{M,N \sqcup \partial M} \rightarrow \mathsf{E}_{M,\partial M}$ induced by the inclusion $\mathsf{F}_{M, N \sqcup \partial M} \subset \mathsf{F}_{M,\partial M}$. We have $\mathsf{E}_{M,N \sqcup \partial M} \simeq \mathsf{E}_{M_N,\partial M_N} \simeq \mathsf{E}_{M_N,\hat{0}}$ and 
\be
\label{EqEqOrdKerPiAstIota}
\iota_{N,M} = \pi^\mathsf{E}_\ast|_{\mathsf{E}_{M_N,\hat{0}}} \;.
\ee 
As any element of the kernel of $\pi_\ast^{\mathsf{E}}$ has to vanish on $\partial M_N$, 
\be
\label{EqEqOrdKerPiAstIota2}
|{\rm ker}(\pi^{\mathsf{E}}_\ast)| = |{\rm ker}(\iota_{N,M})| \;.
\ee
We conclude that the two sides coincide term by term up to a multiplicity factor $|{\rm ker}(\iota_{N,M})|$.
\end{proof}

\subsection{Proof of the gluing formula in the anomalous case}

\label{SecProofGluinFormAnomCase}

We now consider the case where the pair of Lagrangian subgroups $(\mathsf{L}_{M_N}, \mathsf{L}_{M_N, 2\Lambda})$ determined by $M_N$ is not admissible with respect to $\epsilon_{\partial M_N} = (j,\{q_r\})$. 

Let us first consider the third admissibility condition of Section \ref{SecAdmLagSubg}. We will use the notations of the proof of Proposition \ref{PropCharTrAsProj}. By definition, the quadratic refinement $q_{r,\partial M_N}$ determined by $\epsilon_{\partial M_N}$ vanishes on $\mathsf{L}_{3, r}/2\mathsf{G}_{M_N,r} \subset \mathsf{G}_{M_N,r}/2\mathsf{G}_{M_N,r}$. In order for $\mathsf{L}_{M_N}$ to satisfy the third admissibility condition, we need that $q_{r,\partial M_N}$ vanishes on $\mathsf{L}_{M_N,r}/2\mathsf{G}_{M_N,r}$ as well. But we can always choose $q_{r,\partial M_N}$ so that this is the case. Indeed, given $\mathbb{Z}_2^{2n}$ endowed with a (skew-)symmetric perfect pairing and any two Lagrangian subgroups, there is always a quadratic refinement of Arf invariant zero vanishing on the two Lagrangian subgroups. Decomposing $q_{r,\partial M_N} = q_{r,N} + q_{r,-N} + q_{r,\partial M}$, the fact that $q_{r,\partial M_N}$ vanishes on $\mathsf{L}_{3, r}/2\mathsf{G}_{M_N,r}$ implies that $q_{r,N} = q_{r,-N}$, so $q_{r,\partial M_N}$ has the form assumed in Section \ref{SecSetupGluing}. Therefore, we can always choose $\{q_r\}$ so that the third admissibility condition is satisfied.

The second admissibility condition requires $Q(x,b) = 1$ for $b \in \mathsf{L}_{M_N,2\Lambda} \cap H^{2\ell+1}(\partial M_N;\Lambda)/H^{2\ell+1}(\partial M_N;2\Lambda)$. This is automatic by Proposition \ref{PropArfInvCobInv}.

Let us now turn to the first admissibility condition. Recall that in Section \ref{SecSetupGluing}, we made the assumption the lift determined by $\epsilon_{M_N}$ is the sum of lifts on $\mathsf{G}_N$, $\mathsf{G}_{-N}$ and $\mathsf{G}_{\partial M}$. Moreover, we assumed that the lifts on $\mathsf{G}_N$ and $\mathsf{G}_{-N}$ coincide under the canonical isomorphism $\mathsf{G}_{N} \simeq \mathsf{G}_{-N}$ and $\mathsf{G}_{N,2\Lambda} \simeq \mathsf{G}_{-N,2\Lambda}$. The proof of the gluing in Section \ref{SecProofGluing} does not apply if the pair of Lagrangians $(\mathsf{L}_{M_N}, \mathsf{L}_{M_N, 2\Lambda}) \subset (\mathsf{G}_{\partial M_N}, \mathsf{G}_{\partial M_N, 2\Lambda})$ is incompatible with any lift equivalent to one of the form above. More precisely, we do not have a canonical way to define the trace ${\rm Tr}_{\mathsf{I}_{N, u|_N, \epsilon_N}}$. We will show that in this case, the value of the field theory on the glued manifold suffers from the anomaly described in Section \ref{SecAnOnManWithBound} and vanishes identically.

Let us therefore pick a lift $j$ compatible with $(\mathsf{L}_{M_N}, \mathsf{L}_{M_N, 2\Lambda})$. Let us also write $\mathsf{G}_N^{\rm diag}$ and $\mathsf{G}_{N,2\Lambda}^{\rm diag}$ for the diagonal subgroups of $\mathsf{G}_N \times \mathsf{G}_{-N}$ and $\mathsf{G}_{N,2\Lambda} \times \mathsf{G}_{-N,2\Lambda}$. From now on, we identify canonically $\mathsf{G}_N$ and $\mathsf{G}_{-N}$, as well as $\mathsf{G}_{N,2\Lambda}$ and $\mathsf{G}_{-N,2\Lambda}$. The fact that $j$ is not equivalent to a lift of the type assumed in Section \ref{SecSetupGluing} means that there are elements $k$ in  $\mathsf{L}_{M_N} \cap \mathsf{G}^{\rm diag}_N \times \mathsf{G}_{\partial M}$ such that $j(k) \notin \mathsf{G}^{\rm diag}_{N,2\Lambda} \times \mathsf{G}_{\partial M,2\Lambda}$. Writing $j(k) = j_N(k) + j_{-N}(k) + j_{\partial M}(k)$, we have moreover $Q(y,j_N(k)) \neq Q(y,j_{-N}(k))$ for some $y \in \mathsf{G}_N$.

Let us write $b_N(k) := j_N(k) - j_{-N}(k) \neq 0$. Let $\chi_N \in H^1(M,M-N;\mathbb{Z})$ be the Poincaré dual of the homology class of $N$ in $M$. As $j(k)$ extends to $M_N$, we have the following differential cocycle on $M$
\be
\check{w}_k := (\hat{b}_N(k) \cup \hat{\chi}_N, -\pi_\ast(\hat{j}(k)) \wedge \hat{\rho} , \pi_\ast(\hat{j}(k)) \wedge d\hat{\rho}) \;,
\ee
where we picked differential form representatives $\hat{j}(k)$, $\hat{b}_N(k)$ and $\hat{\chi}_N$ of $j(k)$, $b_N(k)$ and $\chi_N$. $\hat{\chi}_N$ is supported on a tubular neighborhood of $N$. $\hat{\rho}$ is a function vanishing on $\partial M$ and equal to 1 outside a tubular neighborhood of $\partial M$. Remark that $\pi_\ast(j(k))$ is discontinuous at $N$, with $d\pi_\ast(j(k)) = \hat{b}_N(k) \cup \hat{\chi}_N$, which ensures that $d\hat{w}_k = 0$. To lighten the notation, we will write $\hat{j}(k)$ for $\pi_\ast(\hat{j}(k))$. According to the analysis of Section \ref{SecAnOnManWithBound}, the theory is anomalous on $M$ if for some $k$, we have
\be
S^e(M,w_k) \neq 1 \;.
\ee
\begin{proposition}
With the notations introduced above, we have 
\be
\label{EqRelHamAnPartAnLift}
\frac{Q(x, j_N(k))}{Q(x,j_{-N}(k))} = S^e(M,w_k) \;.
\ee
\end{proposition}
\noindent Assuming no anomaly on $M_N$, this proposition says that the field theory on the glued manifold vanishes because of the anomaly of Section \ref{SecAnOnManWithBound} precisely when the trace homomorphism is ill-defined. The gluing formula \eqref{EqGluingConstr} can therefore be satisfied in every case by defining the trace homomorphism to be zero whenever it is not well-defined by the procedure of Section \ref{SecTracHom}. 
\begin{proof}
We start by computing the left hand side. \eqref{EqAlphIndepX} and the fact that $Q(x,b_N(k))$ is independent of $x$ imply that
\be
\frac{Q(x, j_N(k))}{Q(x,j_{-N}(k))} = Q(x,b_N(k)) \psi(j_{N}(k), j_{-N}(k)) \;.
\ee
$Q(x,b_{N}(k))$ can be computed by evaluating the action on $N \times I$ endowed with the differential cocycle $\check{y}_{\check{x}, \hat{b}_{N}(k)}$, in the notation of \eqref{EqDefCocYOnCyl}. Using the arguments and notations of the proof of Proposition \ref{PropQ2beq1}, this differential cocycle is equivalent to the differential cocycle
\be
\check{y}' := (\hat{a} - \hat{b}_{N}(k) \cup d\hat{\theta}, \hat{h} - \hat{b}_N(k) \cup \hat{\theta} + \hat{b}_N(k) \wedge \hat{t}, -\hat{b}_{N}(k) \wedge d\hat{t}) \;,
\ee
which passes to a differential cocycle on $N \times S^1$. So $Q(x,b_{N}(k))$ can also be computed by evaluating the action on $N \times S^1$ endowed with the differential cocycle $\check{y}'$. After some calculations similar to the ones of Section \ref{SecWilHooftOp}, we obtain :
\be
\frac{1}{2\pi i} \ln Q(x,b_{N}(k)) = \left \langle \bar{l}(\check{y}'), [N \times S^1]_E \right \rangle \;,
\ee
\be
\bar{l}(\check{y}') = \left(\hat{l}(\check{y}'), (\hat{a} - \hat{b}_{N}(k) \cup d\hat{\theta})_2 \right) \;,
\ee
\begin{align}
\hat{l}(\check{y}') = & \: \frac{1}{2} \hat{a} \cup (\hat{h} - \hat{\eta}_\Lambda) - \frac{1}{2} \hat{b}_N(k) \cup d\hat{\theta} \cup \hat{h} - \frac{1}{2} \hat{h} \cup \hat{b}_N(k) \cup d\hat{\theta} \\
& + \frac{1}{2} \hat{b}_N(k) \cup d\hat{\theta} \cup \hat{\eta}_\Lambda + \mbox{exact}
\;.
\end{align}
We can therefore write
\be
\bar{l}(\check{y}') = \left(\frac{1}{2} \hat{a} \cup (\hat{h} - \hat{\eta}_\Lambda), (\hat{a})_2 \right)  \, \boxplus \, \bar{l}_Q(\check{y}') \;,
\ee
\be
\bar{l}_Q(\check{y}') := \left( \frac{1}{2} \hat{b}_{N}(k) \cup d\hat{\theta} \cup  \hat{\eta}_\Lambda, (\hat{b}_{N}(k) \wedge d\hat{t})_2 \right) \;,
\ee
so Proposition \ref{PropActVanishPullBackTorus} guarantees that 
\be
\frac{1}{2\pi i} \ln Q(x,b_{N}(k)) = \left \langle \bar{l}_Q(\check{y}'), [N \times S^1]_E \right \rangle \;,
\ee

$\psi(j_N(k), j_{-N}(k))$ can be computed on $N \times S^1$ as well:
\be
\frac{1}{2\pi i} \ln \psi(j_N(k), j_{-N}(k)) = \langle \bar{l}_\psi, [N \times S^1]_E \rangle \;,
\ee
\be
\bar{l}_\psi = \left( - \frac{1}{2} (\hat{j}_N(k)  \cup d\hat{\theta}) \cup \hat{j}_{-N}(k), 0 \right) \;.
\ee
We deduce that
\be
\label{EqActQpsi}
\frac{1}{2\pi i} \ln \left( Q(x,b_{N}(k))\psi(j_N(k), j_{-N}(k)) \right) = \langle \bar{l}_{Q\psi}, [N \times S^1]_E \rangle \;,
\ee
\be
\label{EqEthClassQpsi}
\bar{l}_{Q\psi} := (\hat{l}_{Q\psi}, (\hat{b}_{N}(k) \wedge d\hat{t})_2) \;,
\ee
\be
\label{EqLagQpsi}
\hat{l}_{Q\psi} := \frac{1}{2} \hat{b}_{N}(k) \cup d\hat{\theta} \cup \hat{\eta}_\Lambda - \frac{1}{2} (\hat{j}_N(k) \cup d\hat{\theta}) \cup \hat{j}_{-N}(k) \;.
\ee

We now turn to the right hand side of \eqref{EqRelHamAnPartAnLift}. We have
\be
\label{EqActwk}
S(M, w_k) = \left \langle \bar{l}(\check{w}_k), [M]_E \right \rangle \;,
\ee
\be
\label{EqEthClasslwk}
\bar{l}(\check{w}_k) = \left ( \hat{l}(\check{w}_k), (\hat{b}_{N}(k) \cup \hat{\chi}_N)_2 \right ) \;,
\ee
\begin{align}
\label{EqLaglwk}
\hat{l}(\check{w}_k)  = & \: \frac{1}{2}(\hat{b}_{N}(k) \cup \hat{\chi}_N) \cup (-\hat{j}(k) \wedge \hat{\rho} - \hat{\eta}_\Lambda) - \frac{1}{2} \hat{j}(k) \wedge \hat{\rho} \wedge \hat{j}(k) \wedge d\hat{\rho} \notag \\
= & \: - \frac{1}{2}(\hat{b}_N(k) \cup \hat{\chi}_N) \cup  \hat{\eta}_\Lambda - \frac{1}{2}(\hat{j}_N(k) \cup \hat{\chi}_N) \cup \hat{j}_{-N}(k) + \mbox{exact} \;.
\end{align}
The last term on the first line vanishes by the alternating property of the wedge product. As $\hat{\rho}$ is equal to 1 on the support of $\hat{\chi}_N$ we can drop it in the first product. The remaining two terms yield the final result, after using the definition of $\hat{b}_N$. 

Up to signs irrelevant modulo 1, we see that \eqref{EqActQpsi}, \eqref{EqEthClassQpsi} and \eqref{EqLagQpsi} have the same form as \eqref{EqActwk}, \eqref{EqEthClasslwk} and \eqref{EqLaglwk}. To prove that \eqref{EqActQpsi} and \eqref{EqActwk} are equal, we remark that $\bar{l}_{Q\psi}$ is supported on $N \subset N \times S^1$ and $\bar{l}(\check{w}_k)$ is supported on $N \subset M$. By excision, \eqref{EqActQpsi} and \eqref{EqActwk} can be evaluated on $N \times S^1$, on which $\hat{\chi}_N$ can be taken to be equal to $d\hat{\theta}$. 
\end{proof}

\subsection{Main theorem}

Combining the results of Sections \ref{SecStateSpace},  \ref{SecVectManBound} and \ref{SecGluing}, we see that $\mathcal{R}$ satisfies the hypotheses of Proposition \ref{PropAltCondCheckFunc}, which proves 
\begin{theorem}
\label{ThMainTheorem}
$\mathcal{R}$ is a field theory functor from the bordism category of $(\mathfrak{W}, \mathsf{Z}^{\rm flat},\epsilon)$ manifolds to the category of finite dimensional Hilbert spaces.
\end{theorem}

\section{Hamiltonian anomaly}

\label{SecHamAn}

The construction of the state space on a $4\ell+2$-dimensional $(\mathfrak{W}, \mathsf{Z}^{\rm flat})$-manifold $M$ in Section \ref{SecStateSpace} required a choice of extra structure $\epsilon_M$ when $\Gamma_0$ has even order. This is why the domain of the field theory functor $\mathcal{R}$ is in this case the bordism category of $(\mathfrak{W}, \mathsf{Z}^{\rm flat},\epsilon)$-manifolds. 

If we insist on seeing this field theory as defined on $(\mathfrak{W}, \mathsf{Z}^{\rm flat})$-manifolds, then it has to be an anomalous field theory. Technically, $\mathcal{R}$ is not a functor anymore, but a natural transformation from the functor associated to the trivial $4\ell+4$-dimensional field theory to the functor associated to a certain $4\ell+4$-dimensional field theory $\mathcal{A}$ characterizing the anomaly, the anomaly field theory \cite{Freed:2014iua, Monnierd}. The partition function of the anomalous field theory on a $4\ell+3$-dimensional manifold $W$ is no longer a complex number, but a vector in the Hilbert space $\mathcal{A}(W)$. Analogously, the state space of the anomalous field theory on a $4\ell+2$-dimensional manifold $M$ is not a Hilbert space, but an object in the category (more precisely the 2-Hilbert space) $\mathcal{A}(M)$. We refer the reader to \cite{Monnierd} for definitions of the higher categorical concepts appearing in this section, as well as for details about the functorial picture of anomalous field theories. 

Our aim in this section is to show how the field theory $\mathcal{R}$ unambiguously assigns an object $\mathcal{R}(M,u)$ in a 2-Hilbert space to each $4\ell+2$-dimensional $(\mathfrak{W}, \mathsf{Y}^{\rm flat}/\mathsf{E})$-manifold $M$. Unlike the rest of the paper, we keep this discussion informal and do not prove our claims. We start by reinterpreting the constructions of Sections \ref{SecWilHooftOp} and \ref{SecStateSpace} in a categorical setting. 

Recall that the Heisenberg group $\mathsf{H}_M$, before any choice of trivialization, consists of a set of Hermitian lines $\{L_k\}$, $k \in \mathsf{G}_M$, together with isomorphisms 
\be
\mu_{k_1,k_2}: L_{k_1} \otimes L_{k_2} \simeq L_{k_1+k_2}
\ee
describing the group composition law. We can see it as a very simple instance of a 2-group, see Appendix \ref{App2Groups} for a definition. Indeed, consider the category $\mathcal{H}_M$ whose objects are the Hermitian lines generated by the $\{L_k\}_{k \in \mathsf{G}_M}$ via the tensor product. We take the morphisms to be isomorphisms of Hermitian lines. The isomorphisms $\{\mu_{k_1,k_2} \circ (\bullet \otimes \bullet) \}_{k_1, k_2 \in \mathsf{G}}$ provide a functor $\mathcal{H}_M \times \mathcal{H}_M$ to $\mathcal{H}_M$, the higher equivalent of the group composition law. The identity object is $L_0 \simeq \mathbb{C}$, the inverse of $L_k$ is $L_{-k}$ and the other axioms of a 2-group given in Appendix \ref{App2Groups} are easily checked. We can therefore see the Heisenberg group as a categorified version of $\mathsf{G}_M$.

Representations of 2-groups are defined in Appendix \ref{App2Groups}. The object underlying a representation is a category, in fact a 2-vector space, endowed with a suitable action of the 2-group. The representations form a 2-category \cite{2008arXiv0812.4969B}. A class of representations of the Heisenberg 2-group $\mathcal{H}_M$ can be constructed as follows. Let $\mathsf{A}$ be a finite set endowed with an action of $\mathsf{G}_M$, together with vector spaces $\{V_a\}_{a \in \mathsf{A}}$. Let $\mathcal{H}(V_a)$ be the category of Hilbert spaces whose objects are Hilbert spaces of the form $W \otimes V_a$, $W \in \mathcal{H}_1$ and whose morphisms between $W_1 \otimes V_a$ and $W_2 \otimes V_a$ are the homomorphisms from $W_1$ to $W_2$. $\mathcal{H}(V_a)$ is linearly equivalent to the category $\mathcal{H}$ of Hilbert spaces. Let $\mathcal{A}$ be the Cartesian product of the categories $\mathcal{H}(V_a)$ for $a \in \mathsf{A}$. $\mathcal{A}$ is a category linearly equivalent to $(\mathcal{H})^{\times |A|}$, i.e a 2-vector space of dimension $|A|$. A representation of $\mathcal{H}_M$ on $\mathcal{A}$ consists of isomorphisms
\be
\alpha_{a,k}: V_a \otimes L_k \simeq V_{a+k}
\ee
satisfying the obvious relations 
\be
\label{EqConstIsomAct2Grp}
\alpha_{a, k_1 + k_2} \circ (1 \otimes \mu_{k_1, k_2}) = \alpha_{a + k_1, k_2} \circ (\alpha_{a,k_1} \otimes 1)
\ee

As an example, consider the category $\mathcal{V}'_{M,u}$ constructed on the model above with $\mathsf{A}$ being the coset $u + \mathsf{E}_M/\mathsf{T}_M$, endowed with the natural action of $\mathsf{G}_M$, and $\{V_a\}$ given by $\{S^e(M,x)\}_{x \in u + \mathsf{E}_M/\mathsf{T}_M}$. The set of isomorphisms $\{\alpha_{a,k}\}$ turning $\mathcal{V}'_{M,u}$ into a 2-representation is provided by $\{P(x,k,\bullet)\}_{x \in u + \mathsf{E}_M/\mathsf{T}_M, k \in \mathsf{G}_M}$ (see Section \ref{SecPushDown}). The relations \eqref{EqConstIsomAct2Grp} hold because this is precisely how the product of the Heisenberg group has been defined. In constructing $V'_{M,u}$ in \eqref{EqDefVMu}, we took the direct sum of the Hermitian lines $\{S^e(M,x)\}_{x \in u + \mathsf{E}_M/\mathsf{T}_M}$. We see now that this collection of Hermitian lines can be seen more naturally as an object in $\mathcal{V}'_{M,u}$. In fact, in a suitable sense, this object generates the whole category $\mathcal{V}'_{M,u}$ over $(\mathcal{H})^{\times |A|}$.

The object that the field theory $\mathcal{R}$ naturally associates to a $4\ell+2$-dimensional $(\mathfrak{W}, \mathsf{Z}^{\rm flat})$-manifold $M$ is an object in a representation obtained from $\mathcal{V}'_{M,u}$ by a higher categorical version of the invariant section construction, as we explain now. Let us start by picking a Lagrangian subgroup $\mathsf{L} \subset \mathsf{G}_M$ admissible with respect to some extra structure $\epsilon = (j, \{q_r\})$ in the sense of Section \ref{SecExStruct}. Let us also pick $x_0 \in u + \mathsf{E}_M/\mathsf{T}_M$ and a Lagrangian subgroup $\mathsf{L}^\perp$ complementary to $\mathsf{L}$. The lift $j$ determines an LL subgroup $\tilde{\mathsf{L}}$, which provides trivializations of the lines $L_k$, $k \in \mathsf{L}$, in which $\mu_{k_1,k_2} = 1 \in \mathbb{C}$ for $k_1, k_2 \in \mathsf{L}$. Given $x \in u + \mathsf{E}_M/\mathsf{T}_M$, this allows us to canonically identify the lines $S^e(M,x+k)$ for all $k \in \mathsf{L}$, yielding Hermitian lines $R(x+\mathsf{L},\tilde{\mathsf{L}})$, where we emphasize that these lines depend on the LL subgroup $\tilde{\mathsf{L}}$ lifting $\mathsf{L}$. The orbits $x+\mathsf{L}$ can be identified with the elements of $x_0 + \mathsf{L}^\perp$. Let $I(\tilde{\mathsf{L}}, x_0, \mathsf{L}^\perp)$ be the object of $\mathcal{V}'_{M,u}$ having $W_x = R(x+\mathsf{L},\tilde{\mathsf{L}})$ for each $x \in x_0 + \mathsf{L}^\perp$ and $W_x = 0$ otherwise. 

$I(\tilde{\mathsf{L}}, x_0, \mathsf{L}^\perp)$ depends on $\mathsf{L}$ and $j$ through $\tilde{\mathsf{L}}$, as well as on the base point $x_0$ and the complementary Lagrangian $\mathsf{L}^\perp$. To obtain a canonical object in a canonical 2-representation, we use a higher categorical version of the invariant section construction.

To this end, consider the category $\mathcal{C}$ whose objects are quadruplets $(\mathsf{L}, \epsilon, x_0, \mathsf{L}^\perp)$ defined as above. The morphisms in $\mathcal{C}$ are generated via composition by the following three types of morphisms:
\begin{enumerate}
\item Shift of the base point: $\mu_{C_1, x_2}: C_1 \rightarrow C_2$, with $C_1 = (\mathsf{L}, \epsilon, x_1, \mathsf{L}^\perp)$, $C_2 = (\mathsf{L}, \epsilon, x_2, \mathsf{L}^\perp)$.
\item Change of the complementary Lagrangian subgroup: $\mu_{C_1, \mathsf{L}^\perp_2}: C_1 \rightarrow C_2$, with $C_1 = (\mathsf{L}, \epsilon, x_0, \mathsf{L}^\perp_1)$ and $C_2 = (\mathsf{L}, \epsilon, x_0, \mathsf{L}^\perp_2)$.
\item Change of the Lagrangian subgroup and the extra structure: $\mu_{C_1, \mathsf{L}_2, \epsilon_2}: C_1 \rightarrow C_2$, with $C_1 = (\mathsf{L}_1, \epsilon_1, x_0, \mathsf{L}^\perp)$ and $C_2 = (\mathsf{L}_2, \epsilon_2, x_0, \mathsf{L}^\perp)$. We implicitly assumed that $\mathsf{L}^\perp$ is transversal to both $\mathsf{L}_1$ and $\mathsf{L}_2$.
\end{enumerate}
$\mathcal{C}$ is a transitive groupoid. We will consider a functor $\mathcal{F}$ from $\mathcal{C}$ into the 2-category of representations of $\mathcal{H}_M$. Such a functor should assign a representation to each object of $\mathcal{C}$ and a 1-morphism of representations to each morphism of $\mathcal{C}$, subject to the usual axioms for functors. $\mathcal{F}$ will assign to each object of $\mathcal{C}$ the representation $\mathcal{V}'_{M,u}$, hence it should assign 1-morphisms from $\mathcal{V}'_{M,u}$ to itself to each morphism of $\mathcal{C}$. For a 2-vector space that is canonically equivalent to $\mathcal{H}_1^n$ for some $n$, like $\mathcal{V}'_{M,u}$, 1-morphisms can be simply described as $n \times n$ matrices of vector spaces. The objects of the 2-vector space are pictured as vectors of vector spaces, and the 1-morphisms act on them as regular matrices, with the addition replaced by the direct sum and the multiplication by the tensor product \cite{2008arXiv0812.4969B}. 

We now describe the action of $\mathcal{F}$ on three types of morphisms of $\mathcal{C}$, which determine $\mathcal{F}$ completely via the compatibility condition with the compositions.
\begin{enumerate}
\item Let $C_1 = (\mathsf{L}, \epsilon, x_1, \mathsf{L}^\perp)$ and $C_2 = (\mathsf{L}, \epsilon, x_2, \mathsf{L}^\perp)$. Writing $k = x_2 - x_1 \in \mathsf{G}_M$, we set
\be
\mathcal{F}(\mu_{C_1,x_2}) := \sum_{x \in u + \mathsf{E}_M/\mathsf{T}_M} e_{x + k,x} \otimes \mathbb{C}
\;.
\ee
$e_{x + k, x}$ is the $|\mathsf{G}_M| \times |\mathsf{G}_M|$ matrix with zero matrix elements except for a $1$ in position $(x+k, x)$. By $e_{x + k, x} \otimes \mathbb{C}$, we mean the matrix of vector spaces whose matrix elements are the zero vector space everywhere, except in position $(x+k, x)$, where it is $\mathbb{C}$.

$\mathcal{F}(\mu_{C_1,x_2})$ applied to the object $I(\tilde{\mathsf{L}}, x_1, \mathsf{L}^\perp)$ yields $I(\tilde{\mathsf{L}}, x_2, \mathsf{L}^\perp)$.

\item Let $C_1 = (\mathsf{L}, \epsilon, x_0, \mathsf{L}^\perp_1)$ and $C_2 = (\mathsf{L}, \epsilon, x_0, \mathsf{L}^\perp_2)$. The automorphism group of $\mathsf{G}_M$ contains a transvection parallel to $\mathsf{L}$ leaving $\mathsf{L}$ fixed and mapping $\mathsf{L}^\perp_1$ to $\mathsf{L}^\perp_2$. This automorphism can be pictured as a $|\mathsf{G}_M| \times |\mathsf{G}_M|$ permutation matrix $v$. We define
\be
\mathcal{F}(\mu_{C_1, \mathsf{L}^\perp_2}) := v \otimes \mathbb{C} \;.
\ee

$\mathcal{F}(\mu_{C_1, \mathsf{L}^\perp_2})$ applied to $I(\tilde{\mathsf{L}}, x_0, \mathsf{L}^\perp_1)$ yields $I(\tilde{\mathsf{L}}, x_0, \mathsf{L}^\perp_2)$.

\item Let $C_1 = (\mathsf{L}_1, \epsilon_1, x_0, \mathsf{L}^\perp)$ and $C_2 = (\mathsf{L}_2, \epsilon_2, x_0, \mathsf{L}^\perp)$. We define
\be
\mathcal{F}(\mu_{C_1, \mathsf{L}_2, \epsilon_2}) := \sum_{x - x_0 \in \mathsf{L}^\perp} e_{x,x} \otimes R(x+\mathsf{L}_2,\tilde{\mathsf{L}}_2) \otimes (R(x+\mathsf{L}_1,\tilde{\mathsf{L}}_1))^{-1} \;.
\ee

Again, $\mathcal{F}(\mu_{C_1, \mathsf{L}_2, \epsilon_2})$ applied to $I(\tilde{\mathsf{L}}_1, x_0, \mathsf{L}^\perp)$ yields $I(\tilde{\mathsf{L}}_2, x_0, \mathsf{L}^\perp)$.
\end{enumerate}

We now describe the invariant sections of $\mathcal{F}$. Let $|\mathcal{C}|$ denote the (finite) number of objects of $\mathcal{C}$. Consider the category $(\mathcal{V}'_M)^{\times |\mathcal{C}|}$, the Cartesian product of $|\mathcal{C}|$ copies of $\mathcal{V}'_M$. An object $O \in (\mathcal{V}'_M)^{\times |\mathcal{C}|}$ is a size $|\mathcal{C}|$ vector of objects of $\mathcal{V}'_M$, and we write $O_C$ for the component labeled by $C \in \mathcal{C}$. We say that $O$ is an invariant section of $\mathcal{F}$ if
\be
\mathcal{F}(\mu) O_C = O_{\mu C}
\ee
for each morphism $\mu$ of $\mathcal{C}$. The category of invariant sections of $\mathcal{F}$ is the full subcategory $\mathcal{I}_{M,u}$ of $(\mathcal{V}'_M)^{\times |\mathcal{C}|}$ whose objects are the invariant sections of $\mathcal{F}$. By the remarks above, we know that $\mathcal{I}_{M,u}$ is not empty: it contains at least the object $I_{M,u}$ defined by $(I_{M,u})_C := I(\tilde{\mathsf{L}}, x_0, \mathsf{L}^\perp)$ for $C = (\mathsf{L}, \epsilon, x_0, \mathsf{L}^\perp)$. 

$(I_{M,u})_C$ is the state space that the theory $\mathcal{R}$ naturally assigns to a $4\ell+2$-dimensional $(\mathfrak{W}, \mathsf{Z}^{\rm flat})$-manifold $M$. Note that it can be seen as a vector space only after an object $C$ has been chosen. Given $C = (\mathsf{L}, \epsilon, x_0, \mathsf{L}^\perp)$, we can take the direct sum of the vector components of $I(\tilde{\mathsf{L}}, x_0, \mathsf{L}^\perp)$ and obtain the vector space $V'_{M,u}/\rho(\tilde{\mathsf{L}})$, which is canonically isomorphic to $\mathcal{R}(M,u,\epsilon)$ (see Section \ref{SecDefStSp}). After the direct sum is taken, the dependence on $x_0$ and $\mathsf{L}^\perp$ is easily seen to disappear. 

The category $\mathcal{I}_{M,u}$ should be the 2-Hilbert space associated by the $4\ell+4$-dimensional anomaly field theory $\mathcal{A}$ to the $4\ell+2$-dimensional manifold $M$. However we will not discuss the anomaly field theory further here.

\section{Examples}

\label{SecEx}

We illustrate the general construction above in various particular cases, in which simplifications occur.

\subsection{Self-dual lattice}

The simplest case arises when the lattice $\Lambda$ is self-dual: $\Lambda = \Lambda^\ast$. Then $\Gamma$ is the trivial group, and $\mathsf{E}_{M,P} = 1$. There is no discrete gauging, hence no anomaly either. The theory obtained is simply the prequantum field theory of Section \ref{SecSubPreqTh}. Note that the constraint described by Proposition \ref{PropDimHilbSp} is still relevant to this case. As $\mathsf{K}_M = 1$, the non-triviality of the state space of the prequantum theory fixes the background flux $a_M$ uniquely.

\subsection{Even lattice}

\label{SecEvLat}

We examine now the case where $\Lambda$ is even. A stronger condition that we will examine in parallel is $\Lambda$ being "totally even", i.e. a lattice whose pairing is $2\mathbb{Z}$-valued. The latter case includes the case of ordinary $U(1)$ level $k$ Chern-Simons theory, with $\Lambda = \sqrt{2k} \mathbb{Z}$. 

\paragraph{Induced pairing on cochains and cohomology} The pairing induced on $\bar{\Lambda}$-valued cohomology is even. However, the pairing induced on $\bar{\Lambda}$-valued cochains is not necessarily even. This is because while the cup product pairing on cohomology is graded symmetric, it is not symmetric at the level of cochains. If the pairing on $\Lambda$ is totally even, then we obtain a totally even pairing on both the cochain and cohomology groups.

\paragraph{Wu structure} If $\Lambda$ is even, $\Gamma^{(2)} = \Gamma_2$ in the notation of \eqref{EqDecompGamma2}, so $d\hat{\eta} = \hat{\nu} = 0$  and there is a preferred Wu structure $\eta = 0$, which we choose. If $\Lambda$ is totally even, then $\Gamma^{(2)} = 0$ and the preferred Wu structure is the only one.

\paragraph{E-theory} The twisted E-theory coincides with the ordinary E-theory. Moreover, the long exact sequence \eqref{EqLongExSeqETh} splits into short exact sequences
\be
0 \rightarrow H^p(M;\mathbb{R}/\mathbb{Z}) \stackrel{i}{\rightarrow} E[\Gamma^{(2)},q]^p(M) \stackrel{j}{\rightarrow} H^{p-q}(M;\bar{\Gamma}^{(2)}) \rightarrow 0
\ee
because the evenness of the pairing ensures that the Steenrod square ${\rm Sq}^{q+1}$ vanishes modulo 1. The degree $p$ E-theory classes on $M$ are pairs $(s,y) \in H^p(M;\mathbb{R}/\mathbb{Z}) \rtimes H^{p-q}(M;\bar{\Gamma}^{(2)})$, where the semi-direct structure is induced by the non-commutative sum \eqref{EqGrpLawCocModETh} on E-cochains. 

If $\Lambda$ is totally even, as $\Gamma^{(2)} = 0$, the E-theory coincides with $\mathbb{R}/\mathbb{Z}$-valued cohomology. The integration homomorphism of Proposition \ref{PropConstrFundHomEClass} coincides with the pairing of $\mathbb{R}/\mathbb{Z}$-valued cohomology with the fundamental $\mathbb{Z}$-homology class of the underlying manifold.

\paragraph{Action} The Lagrangian simplifies to 
\be
\hat{l}(\check{x}) = \left[ \check{x} \cup \check{x} \right]_h = \frac{1}{2} \hat{a} \cup \hat{h} + \frac{1}{2} \hat{h} \cup \hat{\omega} + \frac{1}{2}H^\cup_\wedge(\hat{\omega}, \hat{\omega}) \quad {\rm mod} \; 1 \;.
\ee
The equality above holds exactly in the totally even case and up to exact terms in the even case. For $\hat{a} = 0$, this Lagrangian reduces to the familiar Lagrangian of abelian Chern-Simons theory, of the form $\hat{h} \wedge_2 d\hat{h}$ up to exact terms. $\wedge_2$ is here the wedge product constructed from to half the pairing on $\Lambda$.

The action obtained from the integration of this Lagrangian in ordinary cohomology is not necessarily gauge invariant unless $\Lambda$ is totally even, see Remark \ref{RemLagNotGaugeInv}. This is directly linked to the fact that the E-theory integration map coincides with the ordinary one only in the totally even case. 

On boundaries, the formulas \eqref{EqACtFromBoundedMan} and  \eqref{EqActBManField} have $\hat{\lambda}$ and $\lambda$ equal zero if the bounded manifold $W$ admits a Wu structure (hence in particular if $\Lambda$ is totally even). They coincide then with the well-known formulas for the abelian Chern-Simons action on boundaries. In the even case, the subtleties involving E-theory may therefore only appear on spacetimes that do not bound a Wu manifold.

Recall that we required the Wu structure on $4\ell+3$-dimensional manifolds to be \emph{admissible}, in the sense that the quadratic refinement of the linking pairing defined by the action is tame. We showed in Proposition \ref{PropExWuStructqMTame} that admissible Wu structures always exist. However, we made a very specific choice of Wu structure in the present case, namely the one corresponding to $\eta = 0$. Using \eqref{EqACtFromBoundedMan}, it is easy to show that it is always admissible on boundaries of Wu manifolds, but on generic manifolds with 2-torsion, nothing ensures that this Wu structure is admissible. The classical Chern-Simons theory associated to an even lattice may therefore require a non-trivial Wu structure on certain spacetimes. Note that this problem does not appear if $\Lambda$ is totally even, as in this case the group $\Gamma^{(2)}$ is trivial. It does not appear either in dimension 3 ($\ell = 0$), as in this case every 3-manifold is a spin boundary.

The same remarks apply to $4\ell+2$-dimensional manifolds. The zero Wu structure may not be admissible, see Proposition \ref{PropDefAdmWuStruct}. This complication is absent in the totally even case and in the three-dimensional case.

There is no simplification worth pointing out in the construction of the prequantum theory and the partition function when $\Lambda$ is an even or totally even lattice. From now on, we focus exclusively on the totally even case.

\paragraph{Wilson operators} In the totally even case, the proof of proposition \eqref{PropQ2beq1} shows that $Q(x,b) = 1$, $b \in H^{2\ell+1}_{\rm free}(M;\bar{\Lambda})$, hence the quadratic refinement $\mathfrak{q}_M$ is trivial. The Arf invariant $\mathfrak{a}_M$ is automatically zero. This also implies that the skew-symmetric refinement $\psi$ restricted to $H^{2\ell+1}_{\rm free}(M;\bar{\Lambda})$ is trivially equal to 1. 

The Hermitian lines $L_k$ are all canonically trivial. This is linked with the fact that the pairing $B$ on $\mathsf{G}_M$ \eqref{DefBimultPairG} has a canonical square root given by 
\be
\label{DefCanSqRootPairG}
\phi_0(k_1, k_2) := \exp -\pi i \langle j(k_1) \wedge j(k_2), [M] \rangle \;.
\ee
This pairing is independent of $j$ thanks to the total evenness. Unlike $\phi$, which is in general only a 2-cocycle, $\phi_0$ is bimultiplicative. One can also check that \eqref{EqAlphIndepX} ensures that the operators $Q(x,j(k))$ are independent of the lift $j$ as well.

\paragraph{State space} We just saw that the first component of the extra structure required to construct the state space, the lift $j$, is irrelevant in the totally even case. The second component, a quadratic refinement of the pairing $(B_r)^{r/2}$ for each even $r$ appearing in the decomposition \eqref{EqDecompGM}, can be chosen canonically. Indeed, we have $(B_r)^{r/2} = ((\phi_0)_r)^r = 1$ is the trivial pairing, so we can take $q_r = 1$. Therefore, effectively no extra structure is needed in the totally even case.

Any pair of Lagrangian subgroups $(\mathsf{L}_{2\Lambda}, \mathsf{L})$ such that $\mathsf{L}_{2\Lambda}$ reduces modulo $H^{2\ell+1}_{\rm free}(M;\bar{\Lambda})$ to $\mathsf{L}$ is admissible in the sense of Section \ref{SecAdmLagSubg}. Indeed, there is always a lift $j$ satisfying first admissibility condition, and by the discussion above, the choice of such a $j$ has no influence on the rest of the construction. As the quadratic refinements $\mathfrak{q}_M$ and $q_r$ vanish, the second and third admissibility conditions are automatically satisfied. Moreover, given any number of Lagrangian subgroups $\mathsf{L}_i$ of $\mathsf{G}_M$. We can always complete them to admissible pairs. For that, pick a lift $j$ and choose Lagrangian subgroups of $\mathsf{G}_{M,2\Lambda}$ lifting the $\mathsf{L}_i$ with respect to $j$. This means that there are no constraints on the triplets of Lagrangian subgroups relevant to the discussion of the Maslov index.

\paragraph{Summary} In the even case, we have a preferred Wu structure as well as minor simplifications in the form of the Lagrangian. The preferred Wu structure may however not be admissible, in which case gauging requires a non-trivial Wu structure and the original action. 

In the totally even case, there is a unique Wu structure, which is admissible. The action coincides with the one of ordinary abelian Chern-Simons theory. All the anomalies are absent. No extra structure is required to construct the state space.

\subsection{Odd order}

\label{SecExOddOrd}

We focus here on the case where $\Gamma = \Lambda^\ast/\Lambda$ has odd order.

In the odd case, there is no preferred Wu structure. No special simplification occurs in the construction of the action. However, the quadratic refinement of the linking pairing defined by the action is automatically tame. This happens because its restriction to $\mathsf{K}_{M,\partial M}$ is a $\mathbb{Z}_2$-character and $|\mathsf{E}_{M,\partial M}| = |H^{2\ell+1}(M,\partial M; \bar{\Gamma})|$ odd implies $|\mathsf{K}_{M,\partial M}|$ odd. We conclude that every Wu structure on a $4\ell+3$-dimensional manifold is admissible. Similarly, every Wu structure on a $4\ell+2$-dimensional manifold is admissible.

The discussion above shows that there is no partition function anomaly. As was already discussed in Section \ref{SecStateSpace}, there is no extra structure needed to be chosen for the construction of the state space in the odd case. Therefore there is no Hamiltonian anomaly either.

\subsection{Rank 1 case}

A general rank 1 integral lattice is of the form $\Lambda = \sqrt{n} \mathbb{Z}$. We see that either $\Lambda$ is totally even if $n$ is even or $\Gamma = \mathbb{Z}_n$ has odd order if $n$ is odd. We deduce from the discussion above that no anomaly is present in the rank 1 case. 

For $n$ odd, the Wu structures on a closed manifold $M$ form a torsor for $H^{2\ell+1}(M; \mathbb{Z}_2)$ and are spin structures in the case $\ell = 0$. There is a unique trivial Wu structure when $n$ is even. This is consistent with the fact that the theory with $n = 2k$ is nothing but the level $k$ $U(1)$ (ordinary) Chern-Simons theory.

\subsection*{Acknowledgments}

This research has been supported in part by Forschungskredit FK-14-108 of the University of Zürich, SNF Grant No. 200020-149150/1, 152812 and 165666, and by NCCR SwissMAP, funded by the Swiss National Science Foundation.

\appendix

\section{Pairings on finite abelian groups and quadratic refinements}

\label{SecPairFinAbGroups}

\paragraph{Multiplicative notation} In the main text, we sometimes find it more practical to use a multiplicative notation for abelian groups and see the pairings and quadratic refinements as valued in $U(1)$ rather than $\mathbb{R}/\mathbb{Z}$. We call the pairings bilinear or bimultiplicative to emphasize which notation we are using. We can freely switch between the two pictures using the exponential/logarithmic maps relating $\mathbb{R}/\mathbb{Z}$ to $U(1)$. The definitions below are given in the additive notation and have obvious analogues in the multiplicative case that we do not repeat.

\paragraph{Pairings} Let $\Xi_i$, $i = 1,2,3$ be three abelian groups, with the group operations written additively. A pairing between $\Xi_1$ and $\Xi_2$ with value in $\Xi_3$ is a group homomorphism 
\be
\label{EqPairing3Groups2}
(\bullet, \bullet): \Xi_1 \times \Xi_2 \rightarrow \Xi_3 \;.
\ee

The \emph{left radical} of $(\bullet, \bullet)$ is the subgroup of $\Xi_1$ composed of elements $x$ satisfying $(x,y) = 0$ for all $y \in \Xi_2$. The right radical is defined analogously. We call the pairing \emph{non-degenerate} if both radicals are trivial. The pairing is called $\emph{perfect}$ if the homomorphism from $\Xi_1$ to ${\rm Hom}(\Xi_2, \Xi_3)$ sending $x \in \Xi_1$ to $(x,\bullet)$ is an isomorphism. A perfect pairing has vanishing left and right radicals. For example, let $\Xi$ be a finite group and $\Xi^\ast$ be its \emph{Pontryagin dual}, i.e. the group of $\mathbb{R}/\mathbb{Z}$-valued characters on $\Xi$. Then the evaluation yields a perfect pairing between $\Xi$ and $\Xi^\ast$ with value in $\mathbb{R}/\mathbb{Z}$. Any $\mathbb{R}/\mathbb{Z}$-valued perfect pairing on $\Xi$ (i.e. between $\Xi$ and itself) provides an isomorphism with $\Xi^\ast$.

A pairing on $\Xi$ is \emph{symmetric} if $(x,y) = (y,x)$, \emph{alternating} if $(x,x) = 0$ and \emph{skew-symmetric} if $(x,y) = -(y,x)$, for all $x,y \in \Xi$.
Alternating pairings are skew-symmetric, but not vice versa. For a pairing with one of the symmetry properties above, the left and right radicals coincide and we speak about the radical of the pairing.

\begin{lemma}
A (skew-)symmetric $\mathbb{R}/\mathbb{Z}$-valued pairing on a finite abelian group $\Xi$ with radical $\Xi^\perp$ induces a perfect pairing on $\Xi/\Xi^\perp$.
\end{lemma}
\begin{proof}
The pairing provides a homomorphism $p: \Xi \rightarrow {\rm Hom}(\Xi, \mathbb{R}/\mathbb{Z})$, $x \rightarrow (x, \bullet)$. The elements of ${\rm Hom}(\Xi, \mathbb{R}/\mathbb{Z})$ obtained in this way are clearly constant along the orbits of $\Xi^\perp$, hence are pullbacks of elements of ${\rm Hom}(\Xi/\Xi^\perp, \mathbb{R}/\mathbb{Z})$. We obtain therefore a homomorphism $p': \Xi \rightarrow {\rm Hom}(\Xi/\Xi^\perp, \mathbb{R}/\mathbb{Z})$. 

The kernel of this homomorphism is $\Xi^\perp$. Quotienting by $\Xi^\perp$ we obtain an injective homomorphism $\Xi/\Xi^\perp \rightarrow {\rm Hom}(\Xi/\Xi^\perp, \mathbb{R}/\mathbb{Z})$. As $\Xi/\Xi^\perp$ and its Pontryagin dual have the same order, this is an isomorphism. Hence the induced pairing is perfect.
\end{proof}

Let $\Xi$ be an abelian group endowed with a symmetric pairing valued in a group admitting $\mathbb{Z}_2$ as a quotient. A \emph{characteristic element} is an element $\lambda$ of $\Xi$ such that 
\be
(x,x) = (\lambda,x) \mbox{ mod } 2 \;, \quad \forall x \in \Xi \;.
\ee

\paragraph{Quadratic refinements} A thorough reference on quadratic refinements is \cite{Taylor}, in which the proofs of the various assertions below can be found. Let $\Xi$ be a finite abelian group. We will see it as a $\mathbb{Z}$-module and hence freely write expressions of the form 
\be
nx = \underbrace{x + ... + x}_{n \; {\rm times}}\;, \quad n \in \mathbb{Z} \;, \; \; x \in \Xi \;.
\ee
\begin{definition}
\label{DefQuadrRef}
A \emph{refinement} (\emph{enhancement} in \cite{Taylor}) is a function $q: \Xi \rightarrow \mathbb{Q}/\mathbb{Z}$ such that the first condition below holds. If the second condition holds as well, the refinement is said to be \emph{quadratic}:
\begin{enumerate}
\item For $x,y \in \Xi$, $B(x,y) := q(x+y) - q(x) - q(y)$ is an $\mathbb{R}/\mathbb{Z}$-valued pairing on $\Xi$.
\item $q(nx) = n^2 q(x)$ for all $n \in \mathbb{N}$. 
\end{enumerate}
\end{definition}
A quadratic refinement $q$ "refines" the pairing $B$ in the sense that the axioms above imply that $2q(x) = B(x,x)$, which is equivalent to the second condition in the presence of the first one. By definition, the pairing $B$ is symmetric. Let $\Xi^\perp := \{ x \in \Xi \,|\, B(x,y) = 0 \: \forall y \in \Xi\}$ be its radical. Then $q|_{\Xi^\perp}$ is linear, i.e. a homomorphism. A quadratic refinement is \emph{tame} if $q|_{\Xi^\perp} = 0$. A tame quadratic refinement passes to a quadratic refinement on $\Xi/\Xi^\perp$, on which it refines the perfect pairing induced by $B$.

If $q_1$ and $q_2$ are both quadratic refinements of $B$, then $q_2 - q_1$ is a $\mathbb{Z}_2$-valued character on $\Xi$.

\paragraph{Gauss sums and Arf invariants} Given a quadratic refinement $q$ on $\Xi$, we can compute its Gauss sum
\be
{\rm Gauss}(q) := \sum_{x \in \Xi} \exp 2 \pi i q(x) \;.
\ee
The norm $|{\rm Gauss}(q)|$ is equal to zero if $q$ is not tame, and to $\sqrt{|\Xi^\perp| |\Xi|}$ is $q$ is tame. In particular, it is equal to $\sqrt{|\Xi|}$ if $B$ is perfect. Here $|\Xi|$ denotes the order of the finite group $\Xi$. 

The argument of ${\rm Gauss}(q)$ is the \emph{Arf invariant} of $q$:
\be
{\rm Arf}(q) := {\rm Arg}({\rm Gauss}(q)) = \frac{1}{2\pi} {\rm Im} \ln {\rm Gauss}(q) \in \mathbb{R}/\mathbb{Z} \;.
\ee
It turns out that the Arf invariant is valued in $\frac{1}{8}\mathbb{Z}/\mathbb{Z}$, as the following classification shows.

\paragraph{Classification of quadratic refinements} Finite abelian groups endowed with quadratic refinements inducing perfect pairings were classified up to isomorphism in \cite{2003math......1040D}, Theorem 3.5. Their Arf invariants were computed in Theorem 3.9 of the same reference. We record here this classification, starting by enumerating the elementary building blocks. In the following $m$ and $a$ (respectively $a_1$ and $a_2$) are positive integers, $p$ is a prime different from 2, $x$ (respectively $x_1$ and $x_2$) is the generator of $\Xi$, $\binom{a}{p}_L$ is the Legendre symbol of $p$, taking value $1$ if $a$ is a quadratic residue mod $p$ and $-1$ if $a$ is a quadratic non-residue mod $p$.
\begin{enumerate}
\item $\Xi_{r,a} := \mathbb{Z}_r$, $r = p^m$, $q(x) = \frac{a}{r}$, $a$ prime to $p$, $a < r$.
\be
\exp 2\pi i {\rm Arf}(q) = \left \{ \begin{array}{ll}
1 & \mbox{ if } m \mbox{ is even} \\
\binom{a}{p}_L & \mbox{ if } m \mbox{ is odd and } p = 1 \mbox{ mod } 4 \\
i\binom{a}{p}_L & \mbox{ if } m \mbox{ is odd and } p = 3 \mbox{ mod } 4
 \end{array} \right.
\ee
\item $\Xi_{r,a} := \mathbb{Z}_r$, $r = 2^m$, $q(x) = \frac{a}{2r}$, $a$ prime to $2$, $a < 2r$.
\be
\exp 2\pi i {\rm Arf}(q) = \left \{ \begin{array}{ll}
\exp \frac{2\pi i}{8} a & \mbox{ if } m \mbox{ is even} \\
\exp \pm \frac{2\pi i}{8} & \mbox{ if } m \mbox{ is odd and } a = \pm 1 \mbox{ mod } 4
 \end{array} \right.
\ee
\item $\Xi_{r,a_1,a_2} := \mathbb{Z}_r \oplus \mathbb{Z}_r$, $r = 2^m$, $q(x_1) = \frac{a_1}{r}$, $q(x_2) = \frac{a_2}{r}$, $q(x_1+x_2) = \frac{a_1 + a_2 + 1}{r}$, $a_1,a_2 < r$.
\be
\exp 2\pi i {\rm Arf}(q) = \left \{ \begin{array}{ll}
1 & \mbox{ if } m \mbox{ is even} \\
1 & \mbox{ if } m \mbox{ is odd and } a_1a_2 \mbox{ is even} \\
-1 & \mbox{ if } m \mbox{ is odd and } a_1a_2 \mbox{ is odd}
 \end{array} \right.
\ee
\end{enumerate}
The classification theorem says that up to isomorphism, any finite group endowed with a perfect pairing and a quadratic refinement is given by a direct sum of the groups above, with quadratic refinement given by the sum of the quadratic refinements of the summands. Note that this decomposition is not necessarily unique. In fact some blocks are isomorphic. The Arf invariant of the direct sum is given by the sum modulo 1 of the Arf invariants of the summands.

\paragraph{$\mathbb{Z}_2$-valued quadratic refinements} If a quadratic refinement is $\mathbb{Z}_2$-valued, then $B(x,x) = 0$ for all $x \in \Xi$. This means that $B$ is both symmetric and alternating. Unless $B$ vanishes trivially, this is possible only if $\Xi = \mathbb{Z}_2^n$. Moreover if $B$ is perfect, as we will assume, then $n$ is even.

Using the classification above, we see that $\Xi$ has to be a direct sum of groups of the third type with $r = 2$. It is easy to check that there are only two such groups up to isomorphism, $\Xi_{2,0,0}$ and $\Xi_{2,1,1}$. $\Xi_{2,0,0}$ and $\Xi_{2,1,1}$ are sometimes called respectively the hyperbolic and elliptic $\mathbb{Z}_2$-plane. The Arf invariant is 0 in the hyperbolic case, and $\frac{1}{2}$ in the elliptic case.
\begin{proposition}
\label{PropDecompZ2QRArf0Arf1}
In the Arf invariant 0 case, 
\be
\Xi \simeq \bigoplus_{j = 0}^{n/2} \Xi_{2,0,0} \;.
\ee
In the Arf invariant 1 case,
\be
\Xi \simeq \left(\bigoplus_{j = 0}^{n/2-1} \Xi_{2,0,0}\right) \oplus \Xi_{2,1,1}\;.
\ee
\end{proposition}
\begin{proof}
To prove the proposition, it is sufficient to show that 
\be
\Xi_{2,1,1} \oplus \Xi_{2,1,1} \simeq \Xi_{2,0,0} \oplus \Xi_{2,0,0}
\ee
Let $(x_1,x_2)$ and $(x_3, x_4)$ be generators of the summands on the left-hand side. Then one easily check that $(x_1 + x_3, x_1 + x_2 + x_3)$ and $(x_2 + x_4, x_2 + x_3 + x_4)$ generate two copies of $\Xi_{2,0,0}$.
\end{proof}

As $B$ is alternating, we can speak of Lagrangian subgroups, i.e. subgroups $L \subset \Xi$ on which $B$ vanishes identically. A Lagrangian subgroup is said to be $q$-isotropic if $q|_L = 0$. 
\begin{proposition}
The Arf invariant of $\Xi$ is $0$ if and only if there exists a $q$-isotropic Lagrangian subgroup.
\end{proposition}
\begin{proof}
$\Xi_{2,0,0}$ admits two $q$-isotropic Lagrangian subgroups and $\Xi_{2,1,1}$ admits none. If the Arf invariant is zero, we can easily construct a $q$-isotropic Lagrangian subgroup in $\Xi$ using the decompositions of Proposition \ref{PropDecompZ2QRArf0Arf1}. 

Conversely, suppose that there is a $q$-isotropic Lagrangian subgroup $L$. Pick a basis $\{x_i\}$ of generators of $L$. As the pairing $B$ is perfect, there is a unique $y_1 \in \Xi$ such that $B(x_1,y_1) = \frac{1}{2}$, $B(x_i, y_1) = 0$, $i \neq 1$. We can therefore split off a summand $\Xi_1$ generated by $\{x_1,y_1\}$. $\Xi_1$ is isomorphic to $\Xi_{2,0,0}$ as $x_1$ obviously generates a $q$-isotropic Lagrangian subgroup. $L/\Xi_1$ is a Lagrangian subgroup of $\Xi/\Xi_1$. Repeating the argument, we see that $\Xi \simeq \bigoplus_{j = 1}^{n/2} \Xi_{2,0,0}$, so the Arf invariant vanishes.
\end{proof}

\paragraph{Tensor product with a lattice} Suppose that $\Lambda$ is a lattice with a $\mathbb{Z}$-valued pairing $B_\Lambda$ and $\Xi$ a finite group endowed with a tame quadratic refinement $q$. Assume that $\Xi/\Xi^\perp$ has order $p^m$ for some prime $p$ and that ${\rm det} B_\Lambda \neq 0 \; {\rm mod} \; p$. We can construct the finite group $\Xi_\Lambda := \Lambda \otimes \Xi$ and endow it with the quadratic refinement $q_\Lambda := Q_\Lambda \otimes q$, where $Q_\Lambda$ is the quadratic form associated to the pairing $B_\Lambda$. Recall the definition of the Kronecker symbol in \eqref{EqDefKronecker} and define $(a | b)'_K := \frac{1}{2\pi i} \ln (a | b)_K \in \{0,\frac{1}{2}\}$ for positive integers $a$, $b$ such that $a $ is not a multiple of $b$. 

The following proposition relates the Arf invariant of $q_\Lambda$ to the properties of $\Lambda$ and the Arf invariant of $q$, for the two cases of interest to us in the main text.
\begin{proposition}
\label{PropArfInvQRFromLat}
If $p$ is even and $B_\Lambda$ is an even pairing, then
\be
{\rm Arf}(q_\Lambda) = ({\rm det} B_\Lambda | 2)'_K \;. 
\ee
If $p$ is odd,
\be
{\rm Arf}(q_\Lambda) = {\rm rank}(\Lambda) {\rm Arf}(q) + ({\rm det} B_\Lambda | p)'_K \;.
\ee
\end{proposition}
\begin{proof}
These claims are special cases of Proposition 2.8 of \cite{Taylor}.
\end{proof}

\section{Pairings on cohomology}

\label{SecPerPair}

\paragraph{Local systems of groups} Let $\Xi$ be an abelian group and $M$ be a $d$-dimensional compact oriented manifold, possibly with boundary. A \emph{local system of groups} over $M$ is a fiber bundle over $M$ with fiber $\Xi$ carrying an action of the fundamental groupoid of $M$: given $x,y \in M$ and a homotopy class of path from $x$ to $y$ relative to $\{x,y\}$, there is an isomorphism from the fiber $\Xi_x$ over $x$ to the fiber $\Xi_y$ over $y$. In particular, given a point $x \in M$, a local system determines a homomorphism from $\pi_1(M,x)$ to ${\rm Aut}(\Xi)$. We call the image of this homomorphism the \emph{structure group} of the local system, which is independent of $x$. We denote a local system over $M$ with fiber $\Xi$ with a bar: $\bar{\Xi}$.

In the present paper, we only consider local systems of groups endowed with a pairing, and we require all the isomorphisms to be compatible with the pairings on the fibers. Their structure group is a subgroup of $O(\Xi) \subset {\rm Aut}(\Xi)$, the subgroup of automorphisms of $\Xi$ preserving the pairing. We call such local systems \emph{orthogonal}.

Isomorphisms of orthogonal local systems are bundle maps commuting with the action of the fundamental groupoids and preserving the pairings on the fibers.

There is a complex $\tilde{C}^\bullet(M,\bar{\Xi})$ of cochains valued in the local system $\bar{\Xi}$, see for instance Section 2 of \cite{10.2307/1969364}, with corresponding cohomology groups $H^\bullet(M,\bar{\Xi})$ .

\paragraph{Cup product} Let $\Xi_i$, $i = 1,2,3$ be abelian groups and suppose we have a pairing 
\be
\label{EqPairing3Groups}
(\bullet, \bullet): \Xi_1 \times \Xi_2 \rightarrow \Xi_3 \;.
\ee
Let $\bar{\Xi}_i$, $i = 1,2,3$ be local systems over $M$ with fibers $\Xi_i$ such that \eqref{EqPairing3Groups} can be promoted to a bundle map lifting the identity on $M$. We can define cup products 
\be
\cup: C^p(M; \bar{\Xi}_1) \times C^q(M; \bar{\Xi}_2) \rightarrow C^{p+q}(M; \bar{\Xi}_3) \;.
\ee
\be
\cup: C^p(M; \bar{\Xi}_1) \times C^q(M,\partial M; \bar{\Xi}_2) \rightarrow C^{p+q}(M,\partial M; \bar{\Xi}_3) \;.
\ee
$\cup$ induces a graded symmetric cup product on the corresponding cohomology groups.

\paragraph{Cup product pairing} Suppose that $\bar{\Xi}_3$ is a trivial local system. This occurs for instance if $\Xi_1 = \Xi_2 = \Xi$ is endowed with a pairing and $\bar{\Xi}$ is an orthogonal local system. Assuming that $p + q = d$ and that $M$ carries a fundamental homology class $[M,\partial M] \in H_d(M,\partial M; \Xi_3)$, we can define the \emph{cup product pairing} $I(\bullet, \bullet): C^p(M; \bar{\Xi}_1) \times C^q(M, \partial M; \bar{\Xi}_2) \rightarrow \Xi_3$:
\be
\label{EqDefIntPair}
I(\hat{x}_1, \hat{x}_2) := \langle \hat{x}_1 \cup \hat{x}_2, [M,\partial M] \rangle \in \Xi_3 \;.
\ee
$I$ induces a pairing on the corresponding cohomology groups.

\begin{lemma}
\label{LemPrairPer}
\cite{steenrod1943homology}, Section 15. Assume that $\Xi_1$ and $\Xi_2$ are Pontryagin duals of each other, $\Xi_3 = \mathbb{R}/\mathbb{Z}$. Let $\bar{\Xi}_3$ be the trivial local system, and $\bar{\Xi}_1$, $\bar{\Xi}_2$ be local systems such that the natural pairing lifts to a bundle map. Then $H^p(M; \bar{\Xi}_1)$ is the Pontryagin dual of $H^{d-p}(M,\partial M; \bar{\Xi}_2)$ with respect to the cup product pairing. 
\end{lemma}

Here are some examples relevant to the present work.
\begin{enumerate}
\item Let $\Xi_1 = \Lambda$ be a lattice, $\Xi_2$ its dual $\Lambda^\ast$. Let $\bar{\Lambda}$ be an orthogonal local system over $M$ with fiber $\Lambda$. Let $\Xi_3 = \mathbb{Z}$ and let $\bar{\Xi}_2 := \bar{\Lambda}^\ast$ be the local system dual to $\bar{\Lambda}$. We obtain a $\mathbb{Z}$-valued pairing between $H^\bullet(M; \bar{\Lambda})$ and $H^\bullet(M,\partial M; \bar{\Lambda}^\ast)$. This pairing is degenerate in general.
\item Recall that $\Gamma^{(2)}$ is the quotient of $\Lambda/2\Lambda$ by the radical of the induced pairing. $\bar{\Lambda}$ induces a local system $\bar{\Gamma}^{(2)}$ with fiber $\Gamma^{(2)}$. $\Gamma^{(2)}$ carries a perfect pairing and we obtain a $\mathbb{Z}_2$-valued perfect cup product pairing between $H^\bullet(M;\bar{\Gamma}^{(2)})$ and $H^\bullet(M,\partial M;\bar{\Gamma}^{(2)})$
\item The lattice pairing can be extended to an $\mathbb{R}$-valued pairing on $\mathbb{R}^n$. There is a corresponding perfect $\mathbb{R}$-valued cup product pairing between $H^\bullet(M; \mathscr{V})$ and $H^\bullet(M,\partial M; \mathscr{V})$. It induces a $\mathbb{Q}$-valued pairing between $H^\bullet(M; \bar{\Lambda}^\ast)$ and $H^\bullet(M,\partial M; \bar{\Lambda}^\ast)$, whose left and right radicals are $H_{\rm tors}^\bullet(M; \bar{\Lambda}^\ast)$ and $H_{\rm tors}^\bullet(M,\partial M; \bar{\Lambda}^\ast)$. We therefore obtain a perfect $\mathbb{Q}$-valued pairing $B_{\Lambda^\ast}$ between $H_{\rm free}^\bullet(M; \bar{\Lambda}^\ast)$ and $H_{\rm free}^\bullet(M,\partial M; \bar{\Lambda}^\ast)$.
\item Taking the latter pairing modulo 1, we obtain a $\mathbb{Q}/\mathbb{Z}$-valued pairing $B$ between $\mathsf{G}_{M} :=  H_{\rm free}^\bullet(M; \bar{\Lambda}^\ast)/H_{\rm free}^\bullet(M; \bar{\Lambda})$ and $\mathsf{G}_{M,\partial M} :=  H_{\rm free}^\bullet(M,\partial M; \bar{\Lambda}^\ast)/H_{\rm free}^\bullet(M,\partial M; \bar{\Lambda})$. As $B_{\Lambda^\ast}(x,y)$ takes integral values if and only if $x \in H_{\rm free}^\bullet(M; \bar{\Lambda})$ or $y \in H_{\rm free}^\bullet(M,\partial M; \bar{\Lambda})$, $B$ is perfect. In the main text, we find it more useful to write $B$ multiplicatively.
\item The $\mathbb{Q}$-valued pairing on $\Lambda^\ast$ induces a perfect $\mathbb{Q}/\mathbb{Z}$-valued pairing on $\Gamma = \Lambda^\ast/\Lambda$. This induces a perfect cup product pairing between $H^\bullet(M;\bar{\Gamma})$ and $H^\bullet(M,\partial M; \bar{\Gamma})$. 
\end{enumerate}

\paragraph{Linking pairing} We use the notations of Section \ref{SecBackStruct}. Let $M$ be a $4\ell+3$-dimensional manifold. Let $H^{2\ell+2}_{\rm tors}(M; \bar{\Lambda})$ be the torsion subgroup of $H^{2\ell+2}(M; \bar{\Lambda})$. Let $x_1, x_2 \in H^{2\ell+2}_{\rm tors}(M; \bar{\Lambda})$, $\hat{x}_1, \hat{x}_2$ be cocycle representatives, and assume that $x_2$ is of order $k$. Then we can find a cochain $\hat{y}$ such that $d\hat{y} = k\hat{x}_2$. Consider
\be
\label{EqDefLinkPair}
L(x_1, x_2) := -\frac{1}{k} \langle \hat{x}_1 \cup \hat{y}, [M] \rangle \in \mathbb{R}/\mathbb{Z} \;,
\ee
where we used the cup product defined above. Then standard arguments \cite{Blanchfield1957} show that $L$ is a well-defined symmetric pairing on $H^{2\ell+2}_{\rm tors}(M; \bar{\Lambda})$, the \emph{linking pairing}.  The sign is introduced for convenience.

If $x_2$ is $k$-torsion, then $L(x_1,x_2)$ is valued in $\frac{1}{k}\mathbb{Z}/\mathbb{Z}$. Together with the symmetry of the pairing, this implies that if $x_1$ is $l$-torsion with $l$ prime to $k$, then $L(x_1, x_2) = 0$. The radical of $L$ is $\mathsf{K}_M \subset H^{2\ell+2}_{\rm tors}(M;\bar{\Lambda})$, namely the group of classes admitting a cocycle representative having a trivialization in $C^{2\ell+2}(M;\bar{\Lambda}^\ast)$.

\paragraph{Lagrangian subgroups of $\mathsf{G}_{\partial M}$} We now show that a $4\ell+3$-dimensional $\mathfrak{W}$-manifold $M$ determines a Lagrangian subgroup of $\mathsf{G}_{\partial M}$.

We have a long exact sequence for the $\bar{\Gamma}$-valued cohomology of the pair $(M,\partial M)$:
\be
\label{EqLongExSeqRelCohomGamma}
... \rightarrow H^p(M,\partial M;\bar{\Gamma}) \stackrel{\iota}{\rightarrow} H^p(M;\bar{\Gamma}) \stackrel{\rho}{\rightarrow} H^p(\partial M;\bar{\Gamma}) \stackrel{\epsilon}{\rightarrow} H^{p+1}(M,\partial M;\bar{\Gamma}) \rightarrow ...
\ee
$\iota$ takes the absolute cohomology class of any cocycle representative of a relative cohomology class. $\rho$ is the restriction to the boundary. $\epsilon$ is the connecting homomorphism, obtained by extending any representative cocycle to $M$, taking its differential and finally its relative cohomology class. 

\begin{lemma}
Let $x \in H^p(M;\bar{\Gamma})$ and $y \in H^p(\partial M;\bar{\Gamma})$. The following relation holds: 
\be
\label{EqRelPairingBulkBound}
I_{\partial M}(\rho(x), y) = I_M(x, \epsilon(y)) \;, 
\ee
\end{lemma}
\noindent The two pairings in \eqref{EqRelPairingBulkBound} are the cup product pairings associated respectively to $\partial M$ and to $M$.
\begin{proof}
Given representative cocycles $\hat{x}$ and $\hat{y}$, we have 
\be
I_M(x, \epsilon(y)) = \langle \hat{x} \cup d\hat{y}_{\rm ext}, [M,\partial M] \rangle = \langle \rho(\hat{x}) \cup \hat{y}, [\partial M] \rangle = I_{\partial M}(\rho(x), y)\;.
\ee
where $\hat{y}_{\rm ext}$ is any extension of $\hat{y}$ to $M$ as a cochain, and we used the definition of the differential: $\langle d\hat{z}, \hat{\sigma} \rangle := \langle \hat{z}, \partial \hat{\sigma} \rangle$ for any chain $\hat{\sigma}$.
\end{proof}

\begin{lemma}
\label{LemRestLagr}
The subgroup $\mathsf{L}$ of elements of $H^{2\ell+1}(\partial M;\bar{\Gamma})$ that are restrictions of elements of $H^{2\ell+1}(M;\bar{\Gamma})$ is Lagrangian with respect to the alternating cup product pairing on $H^{2\ell+1}(\partial M;\bar{\Gamma})$.
\end{lemma}
\begin{proof}
We have $\epsilon \circ \rho = 0$ and hence \eqref{EqRelPairingBulkBound} implies that $\mathsf{L}$ is isotropic. The fact that $\mathsf{L}$ is also coisotropic follows from \eqref{EqRelPairingBulkBound} as well. If $y \in H^{2\ell+1}(\partial M;\bar{\Gamma})$ is such that $I_{\partial M}(\rho(x), y) = 0$ for all $x \in H^{2\ell+1}(M;\bar{\Gamma})$, then $\epsilon(y) = 0$. Hence, by exactness, $y = \rho(x') \in \mathsf{L}$. 
\end{proof}

Recall the notation 
\be
\label{EqDefTM2}
\mathsf{T}_{\partial M} := H^{2\ell+1}_{\rm tors}(\partial M; \bar{\Lambda}^\ast)/H^{2\ell+1}_{\rm tors}(\partial M; \bar{\Lambda})
\ee
\be
\label{EqDefCEM2}
\mathsf{C}_{\partial M} := H^{2\ell+1}(\partial M; \bar{\Lambda}^\ast)/H^{2\ell+1}(\partial M; \bar{\Lambda}) \;, \quad \mathsf{E}_{\partial M} := H^{2\ell+1}(\partial M;\bar{\Gamma}) \;.
\ee
We have the filtration
\be
\label{EqFiltrGroupCohom}
\mathsf{T}_{\partial M} \subset \mathsf{C}_{\partial M} \subset \mathsf{E}_{\partial M} \;, \quad \mathsf{G}_{\partial M} := \mathsf{C}_{\partial M}/\mathsf{T}_{\partial M} \;, \quad \mathsf{K}_{\partial M} := \mathsf{E}_{\partial M}/\mathsf{C}_{\partial M} \;.
\ee
With respect to the cup product pairing $I_{\partial M}$, $\mathsf{T}_{\partial M}$ is isotropic and its orthogonal is $\mathsf{C}_{\partial M}$. This means that the pairing identifies the elements of $\mathsf{T}_{\partial M}$ with the characters on $\mathsf{E}_{\partial M}$ that are constant on $\mathsf{C}_{\partial M}$. The above implies that 
\be
|\mathsf{E}_{\partial M}| = |\mathsf{G}_{\partial M}| |\mathsf{K}_{\partial M}|^2
\ee
Let $\mathsf{L}$ be the Lagrangian subgroup of Lemma \ref{LemRestLagr}. As a Lagrangian subgroup of $\mathsf{E}_{\partial M}$, it has order $|\mathsf{G}_{\partial M}|^{1/2} |\mathsf{K}_{\partial M}|$. Its intersection $\mathsf{L}'$ with $\mathsf{C}_{\partial M}$ satisfies $|\mathsf{L}'| \geq |\mathsf{G}_{\partial M}|^{1/2}$. Define $\mathsf{L}_M$ to be the projection of $\mathsf{L}'$ along the orbits of the action of $\mathsf{T}_M$.
\begin{proposition}
\label{PropLMLagSub}
$\mathsf{L}_M$ is a Lagrangian subgroup of $\mathsf{G}_{\partial M}$.
\end{proposition}
\begin{proof}
As $\mathsf{L}$ is Lagrangian, $\mathsf{L}'$ and $\mathsf{L}_M$ are isotropic. To check the coisotropy, assume that we have $x \in \mathsf{G}_{\partial M}$ that has vanishing pairing with every element of $\mathsf{L}_M$. Let $y$ be any lift of $x$ to $\mathsf{C}_M$. Then the pairing with $y$ defines a character of $\mathsf{E}_M$ that vanishes on $\mathsf{L}'$. By the remarks above, we can find a $y'$ such that $y' - y \in \mathsf{T}_M$ and whose associated character vanishes on all of $\mathsf{L}$. As $\mathsf{L}$ is Lagrangian, $y' \in \mathsf{L}'$. $y'$ is also a lift of $x$, so $x$ belongs $\mathsf{L}_{M}$, showing that the latter is a Lagrangian subgroup of $\mathsf{G}_{\partial M}$.
\end{proof}

\paragraph{Lagrangian subgroups of $\mathsf{G}_{\partial M, 2\Lambda}$} The construction above can be repeated, starting from the lattice $\Lambda' := \sqrt{2}\Lambda$. We write $\mathsf{X}_{\partial M,2\Lambda}$ the groups defined analogously to \eqref{EqDefTM2} - \eqref{EqFiltrGroupCohom}, for $\mathsf{X} = \mathsf{T}, \mathsf{C}, \mathsf{E}, \mathsf{G}, \mathsf{K}$. $\mathsf{G}_{\partial M, 2\Lambda}$ is isomorphic to $H^{2\ell+1}_{\rm free}(M;\bar{\Lambda}^\ast)/2H^{2\ell+1}_{\rm free}(M;\bar{\Lambda})$, consistently with the definition in Section \ref{SecExStruct}. Repeating the construction above Proposition \ref{PropLMLagSub}, we obtain 
\begin{proposition}
Any $\mathfrak{W}$-manifold $M$ with boundary determines a Lagrangian subgroup $\mathsf{L}_{M,2\Lambda} \subset \mathsf{G}_{\partial M, 2\Lambda}$.
\end{proposition}

\paragraph{Lagrangian subgroups of $\mathsf{G}^{(2)}_{\partial M}$} Construct a lattice $\Lambda''$ as follows. Quotient $\Lambda$ with the sublattice $\Lambda_{\rm even}$ generated by primitive elements having even scalar product with every other element of $\Lambda$. Pick a basis of $\Lambda'' := \Lambda/\Lambda_{\rm even}$ and define the scalar product between two basis elements to be $2$ if the scalar product of their representative in $\Lambda$ is odd, and zero otherwise. This endows $\Lambda''$ with the structure of a lattice, with the property that $(\Lambda'')^\ast/\Lambda'' = \Gamma^{(2)}$.

Now repeat the construction above with $\Lambda''$, writing $\mathsf{X}_{\partial M}^{(2)}$, $\mathsf{X} = \mathsf{T}, \mathsf{C}, \mathsf{E}, \mathsf{G}, \mathsf{K}$, for the groups defined analogously to \eqref{EqDefTM2} - \eqref{EqFiltrGroupCohom}. In particular, the definition of $\mathsf{G}^{(2)}$ coincides with the definition given above Proposition \ref{PropqMCobTame}. Repeating the construction above Proposition \ref{PropLMLagSub}, we obtain 
\begin{proposition}
\label{PropLMLagSubMod2}
Any $\mathfrak{W}$-manifold $M$ with boundary determines a Lagrangian subgroup $\mathsf{L}_{M}^{(2)} \subset \mathsf{G}_{\partial M}^{(2)}$.
\end{proposition}

\section{Wu classes and Wu structures}

\label{SecWuStruct}

\paragraph{Higher cup products and Steenrod squares} Let $M$ be an oriented manifold, possibly with boundary. Given any pairing $\bar{\Xi}_1 \times \bar{\Xi}_2 \rightarrow \bar{\Xi}_3$ between local systems of abelian groups, there are higher cousins of the cup product for each non-negative integer $i$:
\be 
\cup_i: C^p(M;\bar{\Xi}_1) \times C^q(M;\bar{\Xi}_2) \rightarrow C^{p+q-i}(M;\bar{\Xi}_3) \;.
\ee
The ordinary cup product coincides with $\cup_0$. We also formally set $\cup_i = 0$ for $i < 0$. The higher cup products satisfy the fundamental relations (\cite{1947}, Theorem 5.1)
\be
\label{EqRelHighCupProd}
d(\hat{u} \cup_i \hat{v}) - d\hat{u} \cup_i \hat{v} - (-1)^p \hat{u} \cup_i d\hat{v} = (-1)^{p+q-i} \hat{u} \cup_{i-1} \hat{v} + (-1)^{pq + p + q} \hat{v} \cup_{i-1} \hat{u} \;,
\ee
where $\hat{u}$ and $\hat{v}$ are cochains of degree $p$ and $q$ respectively. In words, the failure of the Leibniz rule for the product $\cup_i$ is equal to the failure of $\cup_{i-1}$ to be graded symmetric. 

If $\bar{\Xi}_1 = \bar{\Xi}_2 = \bar{\Xi}$, we define the \emph{Steenrod square operations} on cochains by
\begin{align}
{\rm Sq}^i: & \: C^p(M;\bar{\Xi}) \rightarrow C^{p+i} (M;\bar{\Xi}_3) \\
& \: {\rm Sq}^i(\hat{x}) = \hat{x} \cup_{p-i} \hat{x} \;.
\end{align}
(Our notation follows the standard convention and differs from \cite{1947} by the substitution $i \leftrightarrow p-i$.) An easy consequence of \eqref{EqRelHighCupProd} is that when $\hat{x}$ is closed, $d({\rm Sq}^i(\hat{x})) = 0$ when $i$ is odd and $d({\rm Sq}^i(\hat{x})) = 0$ mod 2 when $i$ is even. On a closed manifold, the operators ${\rm Sq}^i$ induce Steenrod operations on the cohomology, which will be denoted by the same symbols. On a manifold with boundary, they induce Steenrod operations on the cohomology relative to the boundary.

\paragraph{Wu classes} We now restrict ourselves to smooth compact oriented manifolds. A universal way of defining the Wu class of degree $p$ is the following. Consider the Thom spectrum $MSO$ of $BSO$, carrying the universal Thom class $u$. $\nu_p$ is the class in $H^p(BSO;\mathbb{Z}_2)$ 
satisfying
\be
\label{EqUnivDefWuClass}
\nu_p \cup u = \chi{\rm Sq}^p(u) \;,
\ee
where $\chi$ is the antipode of the Steenrod algebra. The cup product and Steenrod square in \eqref{EqUnivDefWuClass} are defined with the standard perfect pairing $\mathbb{Z}_2 \times \mathbb{Z}_2 \rightarrow \mathbb{Z}_2$. As defined above, $\nu_p$ is a stable characteristic class of oriented vector bundles. We will consider only Wu classes associated to tangent bundles. In the case of a $d$-dimensional manifold $M$, assume that the pulled back Wu class $\nu_p(M)$ is trivialized on the boundary, thereby defining a class in $H^p(M,\partial M; \mathbb{Z}_2)$. \eqref{EqUnivDefWuClass} implies that
\be
\label{EqDefWuClass}
{\rm Sq}^p(x) = x \cup \nu_p(M)
\ee
for all $x \in H^{d-p}(M, \partial M; \mathbb{Z}_2)$. $\nu_p(M) = 0$ for $p > d/2$. When the dimension of $M$ is even and $p = d/2$, \eqref{EqDefWuClass} can be written
\be
x \cup x = x \cup \nu_p(M)
\ee
for all $x \in H^{d/2}(M,\partial M; \mathbb{Z}_2)$. This means that the Wu class is a characteristic element for the cup product pairing on $H^{p}(M,\partial M; \mathbb{Z}_2)$.

\paragraph{Wu structures}  Let us consider the homotopy fiber $BSO[\nu_p]$ of the map $\nu_p : BSO \rightarrow \Sigma^p H\mathbb{Z}_2$, where $H\mathbb{Z}_2$ is the $\mathbb{Z}_2$-cohomology spectrum. 
\begin{definition}
\label{DefWuStruct}
A degree $p$ \emph{Wu structure} on (the stable tangent bundle of) a manifold $M$ is a lift of the classifying map of the tangent bundle of $M$ from $BSO$ to $BSO[\nu_p]$.
\end{definition}

Wu structures exist on $M$ if and only if $\nu_p(M) = 0$. When this is the case, the set of Wu structures on $M$ forms a torsor for $H^{p-1}(M,\partial M; \mathbb{Z}_2)$. A degree $2$ Wu structure is a spin structure, and one can show that spin manifolds admit Wu structures for all $p$ \cite{Miller_Lee_1987}. We also deduce from the above that a $d$-dimensional manifold $M$ admits degree $p$ Wu structures for all $p > d/2$.

We will now generalize the notion of Wu class and Wu structure defined above to cohomology with coefficients in a local system. We follow the notations of Section \ref{SecBackStruct}. 

\paragraph{Characteristic element of the coefficient group} Let $\Gamma^{(2)}_{\rm deg} := \Lambda/2\Lambda$. $\Gamma^{(2)}_{\rm deg}$ inherits a $\frac{1}{2}\mathbb{Z}/\mathbb{Z}$-valued pairing from the pairing on $\Lambda$ (divided by two). This pairing has a radical $\Gamma_{\rm rad} \subset \Gamma^{(2)}_{\rm deg}$. $\Gamma^{(2)} := \Gamma^{(2)}_{\rm deg}/\Gamma_{\rm rad}$ carries a perfect pairing. According to the classification reviewed in Appendix \ref{SecPairFinAbGroups} (Theorem 3.5 of \cite{Taylor1984259}), $\Gamma^{(2)}$ can be diagonalized into a direct sum of two types of elementary terms:
\begin{itemize}
\item $\mathbb{Z}_2$ with generator $a$ and pairing $(a,a) = \frac{1}{2}$. We write the sum of these terms $\Gamma_1$.
\item $\mathbb{Z}_2 \times \mathbb{Z}_2$ with generators $a_1$, $a_2$ and pairing $(a_i, a_i) = 0$, $(a_1, a_2) = \frac{1}{2}$. We write the sum of these terms $\Gamma_2$
\end{itemize}
We have therefore a direct sum decomposition
\be
\label{EqDecompGamma2}
\Gamma^{(2)} = \Gamma_1 \oplus \Gamma_2 \;.
\ee
\begin{lemma}
\label{LemCharElGamma2}
$\Gamma^{(2)}$ has a unique characteristic element $\gamma_1$
\end{lemma}
\begin{proof}
If $\Xi_1$ and $\Xi_2$ are two abelian groups endowed with pairings. Any characteristic element of $\Xi_1 \oplus \Xi_2$ is the sum of a characteristic element of $\Xi_1$ and a characteristic element of $\Xi_2$. But the elementary summands composing $\Gamma_1$ and $\Gamma_2$ each admit a unique characteristic element, respectively the non-trivial element and zero.
\end{proof}

\paragraph{Wu class relative to a local system} Consider the class $\nu^{\Gamma}_p := \nu_p \otimes \gamma_1$ in $H^p(BSO; \Gamma^{(2)})$. From \eqref{EqUnivDefWuClass}, it satisfies
\be
\label{EqUnivDefWuClassLocSys}
\nu_p^{\Gamma} \cup u = \chi{\rm Sq}^p(u) \otimes \gamma_1 \;.
\ee
Its pull-back $\nu^{\Gamma}_p(M)$ to a manifold $M$ can be seen as an element of $H^p(M; \bar{\Gamma}^{(2)})$ for any local system $\bar{\Gamma}^{(2)}$ over $M$ with structure group $O(\Gamma^{(2)})$, because $\gamma_1$, being the unique characteristic element, is invariant under $O(\Gamma^{(2)})$.

Let $M$ be a manifold of dimension $d$. Embed $M$ into $\mathbb{R}^{2d}$, let $\mathscr{N}_M$ be the normal bundle and $M\mathscr{N}_M$ the corresponding Thom space. Assume that $\nu^{\Gamma}_p(M)$ is trivialized on $\partial M$ and is therefore refined to a class in $H^p(M,\partial M; \bar{\Gamma}
^{(2)})$. In the main text, this is the case because $\nu^{\Gamma}_p(M)$ comes with a distinguished cocycle representative.
\begin{lemma}
For all $x \in H^{d-p}(M,\partial M;\bar{\Gamma}^{(2)})$, we have
\be
\label{EqProSteenSqLocSys}
{\rm Sq}^p(x) = x \cup \nu^{\Gamma}_p(M) \;,
\ee
where the Steenrod square and the cup product are those associated to the perfect pairing $\Gamma^{(2)} \times \Gamma^{(2)} \rightarrow \mathbb{Z}_2$.
\end{lemma}
\begin{proof}
\eqref{EqUnivDefWuClassLocSys} implies that after choosing representative cocycles, 
\be
\label{EqUnivDefWuClassLocSysCochains}
\hat{\nu}_p^{\Gamma}(M) \cup \hat{u}(M) = ({\rm Sq}^p)^\ast(\hat{u}(M)) \otimes \gamma_1 + d\hat{c} \otimes \gamma_1
\ee
for some $\mathbb{Z}_2$-valued cochain $\hat{c}$. $({\rm Sq}^p)^\ast$ is the dual of ${\rm Sq}^p$ with respect to the pairing between the cochain groups $C^p(M,\partial M; \mathbb{Z}_2)$ and $C^{2d-p}(M\mathscr{N}_M, \mathbb{Z}_2)$, which coincides with $\chi {\rm Sq}^p$ on the cohomology. Let $\hat{x}$ be a cocycle in $C^p(M,\partial M; \bar{\Gamma}^{(2)})$. Using a trivialization of $\bar{\Gamma}^{(2)}$ on some open cover of $M$, we can write $\hat{x} = \sum_i \hat{y}_i \otimes \alpha_i$, where $\hat{y}_i$ are $\mathbb{Z}_2$-valued cochains and  $\alpha_i \in \Gamma^{(2)}$. We have
\begin{align}
\hat{x} \cup \hat{\nu}_p^{\Gamma}(M) \cup \hat{u}(M) \: & = \sum_i \hat{y}_i \cup ({\rm Sq}^p)^\ast(\hat{u}(M)) \otimes (\alpha_i, \gamma_1) + \hat{y}_i \cup d\hat{c} \otimes (\alpha_i,\gamma_1) \notag \\
& = \sum_i {\rm Sq}^p(\hat{y}_i) \cup \hat{u}(M) \otimes (\alpha_i, \alpha_i) + \sum_i d(\hat{y}_i \cup \hat{c}) \otimes (\alpha_i, \gamma_1) \\
& = {\rm Sq}^p(\hat{x}) \cup \hat{u}(M) + d(\hat{x} \cup (\hat{c} \otimes \gamma_1)) \notag
\end{align}
The Thom isomorphism yields \eqref{EqProSteenSqLocSys}.
\end{proof}

\paragraph{Classifying space and spectrum for local systems} The classifying space for the degree $p$ cohomology with local coefficients has been described in \cite{gitler1963cohomology}, Section 7. The classifying space of interest to us is, in the notation of \cite{gitler1963cohomology}, $L_{O(\Gamma^{(2)})}(\Gamma^{(2)},p)$, which we will write $L(\Gamma^{(2)},p)$ to lighten the notation. Homotopy classes of maps from a manifold $M$ into $L(\Gamma^{(2)},p)$ are in bijection with pairs composed of a local system of groups $\bar{\Gamma}^{(2)}$ with fiber $\Gamma^{(2)}$ and a cohomology class in $H^p(M; \bar{\Gamma}^{(2)})$. $L(\Gamma^{(2)},p)$ is constructed as the quotient of $EO(\Gamma^{(2)}) \times K(\Gamma^{(2)},p)$ by the natural diagonal action of $O(\Gamma^{(2)})$ on these spaces. As the action is free on $EO(\Gamma^{(2)})$, $L(\Gamma^{(2)},p)$ is a fibration
\be
K(\Gamma^{(2)},p)  \stackrel{\iota}{\hookrightarrow} L(\Gamma^{(2)},p) \stackrel{\pi}{\rightarrow} BO(\Gamma^{(2)}) \;.
\ee
$\iota$ is an inclusion corresponding to the trivial local system and $\pi$ is a projection obtained by discarding the cohomology class. 

Unlike $K(\Gamma^{(2)},p)$, the classifying space $L(\Gamma^{(2)},p)$ cannot be straightforwardly promoted to a spectrum. However, from our partial understanding of these matters, it can be promoted to a parametrized spectrum over $BO(\Gamma^{(2)})$, see \cite{May2006}, Section 22.1. (One can think of parametrized spectra as fibrations whose base is a topological space and whose fiber is a spectrum.) We will write $\bar{H}\Gamma^{(2)}$ for this paramterized spectrum.

\paragraph{Wu structure relative to a local system} $\nu^{\Gamma}_p$ defines a map
\be
\label{EqMapWuClassRelGamma}
\nu^{\Gamma}_p: BSO \wedge BO(\Gamma^{(2)}) \rightarrow \Sigma^p \bar{H}\Gamma^{(2)} 
\ee
covering the identity on $BO(\Gamma^{(2)})$. We will write $BSO[\Gamma^{(2)},p]$ for the homotopy fiber of this map, taken fiberwise over $BO(\Gamma^{(2)})$. $BSO[\Gamma^{(2)},p]$ is a parametrized spectrum over $BO(\Gamma^{(2)})$. Let $M$ be a smooth manifold endowed with a local system $\bar{\Gamma}^{(2)}$ with fiber $\Gamma^{(2)}$. There is a map from $M$ into $BSO \wedge BO(\Gamma^{(2)})$ classifying the tangent bundle and the local system $\bar{\Gamma}^{(2)}$. (More precisely, assuming these structures fixed on the boundary, they are classified by a homotopy class of such maps relative to the boundary.)
\begin{definition}
A degree $p$ \emph{Wu structure} relative to $\bar{\Gamma}^{(2)}$ on the stable tangent bundle of $M$ is a lift of the classifying map above to $BSO[\Gamma^{(2)},p]$.
\end{definition}

Wu structures relative to $\bar{\Gamma}^{(2)}$ exist on $M$ if and only if $\nu^{\Gamma}_p(M) = 0$. There are two cases leading to the existence of a Wu structure relative to $\bar{\Gamma}^{(2)}$. First, if the pairing on $\Lambda$ is even, $\gamma_1 = 0$ and Wu structures always exist. The map \eqref{EqMapWuClassRelGamma} is the constant map, so there is a canonical Wu structure relative to $\bar{\Gamma}^{(2)}$. Second, if $\Lambda$ is not even, $\nu^{\Gamma}_p(M) = 0$ iff $\nu_p(M) = 0$, so a manifolds admits a Wu structure relative to $\bar{\Gamma}^{(2)}$ iff it admits a Wu structure in the sense of Definition \ref{DefWuStruct}. The same remarks concerning existence therefore apply. The Wu structures relative to $\bar{\Gamma}^{(2)}$ on $M$ form a torsor for $H^{p-1}(M,\partial M;\bar{\Gamma}^{(2)})$.

We will pick a cocycle representative $\hat{\nu}_p$ of the universal Wu class $\nu_p$ on $BSO$, yielding a cocycle $\hat{\nu}_p^\Gamma := \hat{\nu}_p \otimes \gamma_1$ representing the Wu class relative to $\Gamma^{(2)}$. The pull-back of $\hat{\nu}_p^\Gamma$ to $BSO[\Gamma^{(2)},p]$ is exact by definition. Let $\hat{\eta}$ be a trivializing cochain: $d\hat{\eta} = \hat{\nu}^\Gamma_p$. In the present work, we assume that a choice of cocycle $\hat{\nu}_{p}$ and of trivializing cochain $\hat{\eta}$ has been made once and for all for $p = 2\ell+2$. We also assume that $\hat{\eta}$ has been chosen to restrict to a cocycle representing the canonical generator of the cohomology of the fibers of $BSO[\Gamma^{(2)},p] \rightarrow BSO \wedge BO(\Gamma^{(2)})$. Given such a choice, any classifying map $M \rightarrow BSO[\Gamma^{(2)},p]$ provides by pull-back a pair $(\hat{\nu}^\Gamma_p(M), \hat{\eta}(M))$ on $M$ such that $\hat{\nu}(M) = d\hat{\eta}(M)$. Under a change of the Wu structure on $M$, associated to an element $\delta \in H^{p-1}(M,\partial M; \mathbb{Z}_2)$, $\hat{\eta}(M)$ changes by a cocycle $\hat{\delta}$ representing $\delta$. The cochain $\hat{\eta}(M)$ is therefore a useful bookkeeping device for the Wu structure on $M$. 

\paragraph{Notation} In the main text, to lighten the notation, we simply write $\nu_p$ for $\nu_p^\Gamma(M)$.  Similarly, we write $\hat{\eta}$ for $\hat{\eta}(M)$.

\section{E-theory}

\label{SecEth}

We describe here a generalized cohomology theory known as E-theory \cite{2005math......4524J, Freed:2006mx}. More precisely, we define here the $\mathbb{R}/\mathbb{Z}$-valued version of E-theory, as well as certain natural generalizations. We also describe a cochain model for these generalized cohomology theories. 

Unless noted otherwise, all the operations on parametrized spectra, such as suspensions, smash products or homotopy fiber, should be understood fiberwise. When writing maps of parametrized spectra, we always implicitly assume that they cover the identity on the base.

\paragraph{The E-theory spectrum} Let $q$ be a positive integer. The fiberwise homotopy fiber over $BO(\Gamma^{(2)})$ of the map
\be
{\rm Sq}^{q+1}: \Sigma^{-q} \bar{H} \Gamma^{(2)} \rightarrow \Sigma H\mathbb{R}/\mathbb{Z} \wedge BO(\Gamma^{(2)})
\ee
is a parametrized spectrum $E[\Gamma^{(2)},q]$ over $BO(\Gamma^{(2)})$. Given a $\mathfrak{W}$-manifold $M$, which includes the data of a classifying map $M \rightarrow BO(\Gamma^{(2)})$, we obtain generalized cohomology groups fitting in the long exact sequence
\begin{align}
\label{EqLongExSeqETh}
...H^p(M;\mathbb{R}/\mathbb{Z}) & \stackrel{i}{\rightarrow} E[\Gamma^{(2)},q]^p(M) \stackrel{j}{\rightarrow} H^{p-q}(M;\bar{\Gamma}^{(2)}) \\
& \stackrel{{\rm Sq}^{q+1}}{\rightarrow} H^{p+1} (M;\mathbb{R}/\mathbb{Z}) \stackrel{i}{\rightarrow} E[\Gamma^{(2)},q]^{p+1}(M) \stackrel{j}{\rightarrow} H^{p+1-q}(M;\bar{\Gamma}^{(2)})... \notag
\end{align}
We are focusing here on $\mathbb{R}/\mathbb{Z}$-valued E-theory. For $q = 1$ and $\bar{\Gamma}^{(2)}$ the trivial local system with fiber $\mathbb{Z}_2$ and non-degenerate pairing, the corresponding $\mathbb{Z}$-valued E-theory groups have been described in \cite{2005math......4524J, Freed:2006mx}.

\paragraph{A cochain model for E-cohomology} Consider the non-abelian group $\bar{C}^p(M)$ defined as follows. Its elements, which we call E-cochains of degree $p$, are pairs of cochains $\bar{x} = (\hat{s}, \hat{y}) \in C^p(M;\mathbb{R}/\mathbb{Z})\times C^{p-q}(M;\bar{\Gamma}^{(2)})$. We write them with a bar. The group law is given by
\be 
\label{EqGrpLawCocModETh}
(\hat{s}_1, \hat{y}_1) \boxplus (\hat{s}_2, \hat{y}_2) = (\hat{s}_1 + \hat{s}_2 + d\hat{y}_1 \cup_{p-2q+1} \hat{y}_2 + \hat{y}_1 \cup_{p-2q} \hat{y}_2, \hat{y}_1 + \hat{y}_2) \;.
\ee
$\cup_i$ denotes the higher cup products reviewed above, for $G_1 = G_2 = \Gamma^{(2)}$, $G_3 = \frac{1}{2}\mathbb{Z}/\mathbb{Z}$, with the pairing induced from $\Gamma$. Easy computations show
\begin{lemma}
$\boxplus$ is associative. The inverse of $(\hat{s}, \hat{y})$ is
\be
\boxminus (\hat{s}, \hat{y}) = (-\hat{s} + d\hat{y} \cup_{p-2q+1} \hat{y} + \hat{y} \cup_{p-2q} \hat{y}, \hat{y}) \;.
\ee
\end{lemma}
We define the following differential on $\hat{E}^p$:
\be
\label{EqDefDiffEThCoch}
d_0(\hat{s}, \hat{y}) = (d\hat{s} + \hat{y} \cup_{p-2q} d\hat{y} + \hat{y} \cup_{p-2q-1} \hat{y}, d\hat{y})
\ee
The kernel $\bar{Z}^p(M) \subset \bar{C}^p(M)$ of $d_0$ is the set of degree $p$ E-cocycles. Remark that $\hat{y} \cup_{p-2q-1} \hat{y} = {\rm Sq}^{q+1}\hat{y}$, so E-cocycles satisfy the relations
\be
\label{EqConstCocModETh}
d\hat{s} = {\rm Sq}^{q+1}\hat{y} \;, \quad d\hat{y} = 0 \;.
\ee 
The image $\bar{B}^p(M) \subset \bar{C}^p(M)$ of $d_0$ is the subset of exact E-cochains. We record some properties of the differential on $E$-cochains.
\begin{lemma}
\label{LemPropUnTwistDiffETh}
\begin{enumerate}
\item $d_0$ is a group homomorphism. Hence $\bar{Z}^p(M)$ and $\bar{B}^p(M)$ are subgroups of  $\bar{C}^p(M)$.
\item $d_0^2 = 0$, hence $\bar{B}^p(X) \subset \bar{Z}^p(M)$.
\item $\bar{B}^p(M)$ is a normal subgroup of $\bar{Z}^p(M)$.
\end{enumerate}

\end{lemma}
\begin{proof}
For point 1, we need to show that for $\bar{s}_1, \bar{s}_2 \in \bar{C}^p(M)$, 
\be
\label{EqDiffEThHomGrp}
d_0(\bar{s}_1 \boxplus \bar{s}_2) = d_0\bar{s}_1 \boxplus d_0\bar{s}_2
\ee
Writing $\bar{s}_i = (\hat{s}_i, \hat{y}_i)$, $i = 1,2$, we compute
\begin{align}
d_0(\bar{s}_1 \boxplus \bar{s}_2) \: & = (\hat{s}_3, d\hat{y}_1 + d\hat{y}_2) \\
\hat{s}_3 \: & = d\hat{s}_1 + d\hat{s}_2 + d(d\hat{y}_1 \cup_{p-2q+1} \hat{y}_2) + d(\hat{y}_1 \cup_{p-2q} \hat{y}_2) \\
& \quad + (\hat{y}_1 + \hat{y}_2) \cup_{p-2q} d(\hat{y}_1 + \hat{y}_2) + (\hat{y}_1 + \hat{y}_2) \cup_{p-2q-1} (\hat{y}_1 + \hat{y}_2) \notag \\
d_0\bar{s}_1 \boxplus d_0\bar{s}_2 \: & = (\hat{s}_4, d\hat{y}_1 + d\hat{y}_2) \\
\hat{s}_4 \: & = d\hat{s}_1 + d\hat{s}_2 + \hat{y}_1 \cup_{p-2q} d\hat{y}_1 + \hat{y}_1 \cup_{p-2q-1} \hat{y}_1  \\
& \quad + \hat{y}_2 \cup_{p-2q} d\hat{y}_2 + \hat{y}_2 \cup_{p-2q-1} \hat{y}_2 + d\hat{y}_1 \cup_{p-2q+1} d\hat{y}_2 \;.   \notag
\end{align}
After obvious cancellations, we obtain
\begin{align}
\hat{s}_3 - \hat{s}_4 \: & =  d(d\hat{y}_1 \cup_{p-2q+1} \hat{y}_2) + d(\hat{y}_1 \cup_{p-2q} \hat{y}_2) + \hat{y}_1 \cup_{p-2q} d\hat{y}_2 + \hat{y}_2 \cup_{p-2q} d\hat{y}_1 \notag \\
& \quad + \hat{y}_1 \cup_{p-2q-1} \hat{y}_2 + \hat{y}_2 \cup_{p-2q-1} \hat{y}_1 + d\hat{y}_1 \cup_{p-2q+1} d\hat{y}_2 = 0 
\end{align}
On the second line, we used \eqref{EqRelHighCupProd} to replace the differentials of the cup products and performed the cancellations.

For point 2, we have:
\be
d_0^2(\hat{s}, \hat{y}) = (d(\hat{y} \cup_{p-2q} d\hat{y}) + d(\hat{y} \cup_{p-2q-1} \hat{y}) + d\hat{y} \cup_{p-2q} d\hat{y} , 0 )
\ee
Using \eqref{EqRelHighCupProd} twice, this expression is easily seen to vanish.

For point 3, we need to show that if $\bar{s}$ is a degree $p$ E-cocycle and $\hat{r}$ is a degree $p-1$ E-cochain, then there is a degree $p-1$ E-cochain $\hat{r}'$ such that 
\be
\label{EqExECochNormSubGrp}
\bar{s} \boxplus d_0\bar{r} = d_0\bar{r}' \boxplus \bar{s} \;.
\ee
Writing $\bar{s} = (\hat{s}, \hat{y})$, $\bar{r} = (\hat{r}, \hat{x})$, we compute
\be
\bar{s} \boxplus d_0\bar{r} = (\hat{s} + d\hat{r} + \hat{x} \cup_{p-2q-1} d\hat{x} + \hat{x} \cup_{p-2q-2} \hat{x} + \hat{y} \cup_{p-2q} d\hat{x}, \hat{y} + d\hat{x}) \;.
\ee
Remark now that $\hat{y} \cup_{p-2q} d\hat{x} = d\hat{x} \cup_{p-2q} \hat{y} + d\hat{c}$ for some cochain $\hat{c}$. Defining $\bar{r}' = (\hat{r} + \hat{c}, \hat{x})$, we obtain \eqref{EqExECochNormSubGrp}.
\end{proof}

Point 3 of the Lemma ensures that the "cohomology" $\bar{Z}^p(M)/\bar{B}^{p-1}(M)$ is a group. It is easy to see that it is an abelian group.
\begin{proposition}
\label{PropCochModETh}
The quotient $\bar{Z}^p(M)/\bar{B}^{p-1}(M)$ is isomorphic to the E-theory group $E[\Gamma^{(2)},q]^p(M)$.
\end{proposition}
\begin{proof}
One can check that the quotient fits in the long exact sequence \eqref{EqLongExSeqETh}. $i$ sends an element $s \in H^p(M;\mathbb{R}/\mathbb{Z})$ on the equivalence class of $(\hat{s}, 0)$, where $\hat{s}$ is any cocycle representing $s$. $j$ sends the equivalence class of $(\hat{s}, \hat{y})$ onto the cohomology class $y$. Exactness at $E[\bar{\Gamma}^{(2)},q]^p(M)$ is obvious. \eqref{EqConstCocModETh} ensures exactness at $H^{p-q}(M; \bar{\Gamma}^{(2)})$, and exactness at $H^p(M;\mathbb{R}/\mathbb{Z})$ follows from 
\be
({\rm Sq}^{q+1}(\hat{y}), 0) = (\hat{y} \cup_{p-2q} \hat{y}, 0) = d_0(0, \hat{y}) 
\ee
for any $\hat{y} \in Z^{p-q-1}(M; \bar{\Gamma}^{(2)})$.
\end{proof}

\paragraph{Twisted E-theory} We need a twisted version of the generalized cohomology theory defined above. Define the parametrized spectrum $E_\nu[\Gamma^{(2)},q]$ over $BSO[\Gamma^{(2)}, q+1]$ as the fiberwise homotopy fiber of the map 
\begin{align}
\alpha := \; & {\rm Sq}^{q+1}(\bullet) + \bullet \cup \nu_{q+1} : \\
& \Sigma^{-q} H \bar{\Gamma}^{(2)} \wedge_{BO(\Gamma^{(2)})} BSO[\Gamma^{(2)}, q+1] \rightarrow \Sigma H\mathbb{R}/\mathbb{Z} \wedge BSO[\Gamma^{(2)}, q+1] \;. \notag
\end{align}
We used the notation $\wedge_{BO(\Gamma^{(2)})}$ to distinguish the fiberwise smash product from the ordinary one. A $\mathfrak{W}$-manifold $M$ comes with a natural classifying map into $BSO[\Gamma^{(2)}, q+1]$. The corresponding cohomology groups fit in a long exact sequence similar to \eqref{EqLongExSeqETh}:
\begin{align}
\label{EqLongExSeqETwTh}
...H^p(M;\mathbb{R}/\mathbb{Z}) & \stackrel{i}{\rightarrow} E_\nu[\Gamma^{(2)},q]^p(M) \stackrel{j}{\rightarrow} H^{p-q}(M;\bar{\Gamma}^{(2)}) \\
& \stackrel{\alpha}{\rightarrow} H^{p+1} (M;\mathbb{R}/\mathbb{Z}) \stackrel{i}{\rightarrow} E_\nu[\Gamma^{(2)},q]^{p+1}(M) \stackrel{j}{\rightarrow} H^{p+1-q}(M;\bar{\Gamma}^{(2)})... \notag
\end{align}

Let $\hat{\nu}$ the Wu cocycle on $M$. We have a cochain model for $E_\nu[\Gamma,q]^p(M)$ given by the group of E-cochains $\bar{C}^p(M)$, with a twisted differential
\be
\label{EqDefTwDiffEThCoch}
d(\hat{s}, \hat{y}) = (d\hat{s} + \hat{y} \cup_{p-2q} d\hat{y} + \hat{y} \cup_{p-2q-1} \hat{y} + \hat{y} \cup \hat{\nu}, d\hat{y}) \;.
\ee
We write $\bar{Z}_\nu^p(M) := {\rm ker}(d|_{\bar{C}^p(M)})$ for the group of degree p twisted E-cocycles and $\bar{B}_\nu^p(M) := {\rm im}(d|_{\bar{C}^{p-1}(M)})$ for the group of degree $p$ twisted exact E-cochains. The twisted differential has the same properties as the differential $d_0$:
\begin{lemma}
\begin{enumerate}
\item $d$ is a group homomorphism.
\item $d^2 = 0$
\item $\bar{B}_\nu^p(M)$ is a normal subgroup of $\bar{Z}_\nu^p(M)$.
\end{enumerate}
\end{lemma}
\begin{proof}
All three points are easily deduced from Lemma \ref{LemPropUnTwistDiffETh} by analysing the extra terms induced by the twisted differential.

For point 1, notice that in the twisted case, both sides of \eqref{EqDiffEThHomGrp} get an extra contribution of $(\hat{y}_1 + \hat{y}_2) \cup \hat{\nu}$. For point 2, notice that the two extra terms generated by the twisted differentials are both equal to $d\hat{y} \cup \hat{\nu}$ and cancel each other. For point 3, the extra terms induced by $d\bar{r}$ and $d\bar{r}'$ are both $\hat{x} \cup \hat{\nu}$.
\end{proof}
\begin{proposition}
\label{PropCochModTwETh}
The quotient $\bar{Z}^p_\nu(M)/d \bar{C}^p(M)$ is isomorphic to the E-theory group $E_\nu[\Gamma^{(2)},q]^p(M)$.
\end{proposition}
\begin{proof}
The quotients fit in the long exact sequence \eqref{EqLongExSeqETwTh}, see the proof of Proposition \ref{PropCochModETh}.
\end{proof}

\paragraph{Integration on manifolds} An integration theory exists for $E[\mathbb{Z}_2,1]$-theory with trivial local systems \cite{Freed:2006mx} (Proposition 4.4), however the proof there does not generalize because of the lack of multiplicative structure on the Brown-Comenetz dual of $E_\nu[\Gamma^{(2)},q]$ for $q > 1$. We describe here a construction that provides a partial substitute for integration.

\begin{proposition}
\label{PropConstrFundHomEClass}
Let $M$ be a $\mathfrak{W}$-manifold of dimension $m < 2q+2$, possibly with boundary. Then there is a canonical element $u_E \in {\rm Hom}(E_\nu[\Gamma^{(2)},q]^m(M,\partial M); \mathbb{R}/\mathbb{Z})$. 
\end{proposition}
Canonical should be understood as "canonical up to universal choices on the classifying spaces/spectra". By a slight abuse of language, we will call $u_E$ the fundamental E-homology class of the $\mathfrak{W}$-manifold $M$ and write it $[M,\partial M]_E$, or simply $[M]_E$ when $M$ has an empty boundary. Given any E-cohomology class $\bar{s}$ of degree $m$, we define its integral over $M$ by $\langle \bar{s}, [M]_E \rangle \in \mathbb{R}/\mathbb{Z}$.
\begin{proof}
Consider the parametrized spectrum $F_\nu[\Gamma^{(2)},q]$ over $BSO[\Gamma^{(2)}, q+1]$ defined as the fiberwise homotopy fiber of the map
\begin{align}
\label{DefSpecFHomFiber}
\tilde{\alpha} := \: & \chi {\rm Sq}^{q+1}(\bullet) \otimes \gamma_1 + \nu_{q+1} \cap \bullet : \\ 
& H\mathbb{Z} \wedge BSO[\Gamma^{(2)}, q+1] \rightarrow  \Sigma^{q+1} H \bar{\Gamma}^{(2)} \wedge_{BO(\Gamma^{(2)})} BSO[\Gamma^{(2)}, q+1] \;, \notag
\end{align}
where $\chi$ is the antipode of the Steenrod algebra and $\gamma_1$ is the characteristic element of $\Gamma^{(2)}$ (see Lemma \ref{LemCharElGamma2}). The $F_\nu[\Gamma^{(2)},q]$-homology of the pair $(M,\partial M)$ fits in the long exact sequence
\begin{align}
... \rightarrow H_{p+q}(M,\partial M; \bar{\Gamma}^{(2)}) \rightarrow & \: F_\nu[\Gamma^{(2)},q]_{p}(M,\partial M) \notag \\
\rightarrow & \: H_{p}(M,\partial M; \mathbb{Z}) \stackrel{\tilde{\alpha}}{\rightarrow} H_{p+q+1}(M,\partial M; \bar{\Gamma}^{(2)}) \rightarrow ... \;,
\end{align}

We have canonical isomorphisms
\be
H_p(M,\partial M; \mathbb{Z}) \simeq {\rm Hom}(H^p(M,\partial M; \mathbb{R}/\mathbb{Z}); \mathbb{R}/\mathbb{Z}) \;,
\ee
\be
H_{p+q}(M,\partial M; \bar{\Gamma}^{(2)}) \simeq {\rm Hom}(H^{p+q}(M, \partial M; \bar{\Gamma}^{(2)}); \mathbb{R}/\mathbb{Z}) \;.
\ee 
We deduce that 
\be
F_{\nu}[\Gamma^{(2)},q]_{p}(M,\partial M) \simeq {\rm Hom}(E^p_\nu[\Gamma^{(2)},q](M, \partial M); \mathbb{R}/\mathbb{Z}) \;,
\ee
Although there is no canonical isomorphism, an isomorphism can be fixed universally by choosing a homotopy equivalence between the Brown-Comenetz dual of $E_\nu[\Gamma^{(2)},q]$ and $F_\nu[\Gamma^{(2)},q]$ \cite{1976Brown-Comenetz}. (All these constructions are understood fiberwise over $BSO[\Gamma^{(2)}, q+1]$.) Pick an embedding $M \rightarrow \mathbb{R}^n$ for $n$ large enough, and let $\mathscr{N}M$ be the normal bundle. Atiyah duality further ensures that
\be
F_\nu[\Gamma^{(2)},q]_{p}(M,\partial M) \simeq F_\nu[\Gamma^{(2)},q]^{n-p}(T_{\mathscr{N}M}) \;,
\ee
where $T_{\mathscr{N}M}$ is the Thom space of $\mathscr{N}M$. Therefore, it is sufficient to construct a canonical element in $F_\nu[\Gamma^{(2)},q]^{n-m}(T_{\mathscr{N}M})$ to prove the proposition.

Let $MSO[\Gamma^{(2)},q+1]$ be the Thom spectrum of the universal bundle over $BSO[\Gamma^{(2)},q+1]$. We have the diagram over $MSO[\Gamma^{(2)},q+1]$
\be
\label{EqThomClassETheory}
\begin{tikzcd}
F_\nu[\Gamma^{(2)},q] \arrow[r] & H\mathbb{Z} \wedge MSO[\Gamma^{(2)},q+1] \arrow[r, "\tilde{\alpha}"] & \Sigma^{-q-1} H \bar{\Gamma}^{(2)} \wedge_{BO(\Gamma^{(2)})} MSO[\Gamma^{(2)},q+1] \\
& MSO[\Gamma^{(2)},q+1] \arrow[ul, "u_F", dotted] \arrow[u, "u"] \arrow[ur, "{\rm cst}"] & \;.
\end{tikzcd} 
\ee
where $u$ is an $H\mathbb{Z}$-valued map representing the Thom class. (The labels describe the maps between the fibers over $MSO[\Gamma^{(2)},q+1]$.) The properties of the Wu class ensure that $\tilde{\alpha} \circ u $ is homotopic to the constant map. Making a universal choice of homotopy, we obtain a homotopy class of maps $u_F$ as shown in \eqref{EqThomClassETheory}. This means that a $\mathfrak{W}$-manifold $M$ is associated to a canonical class in $F_\nu[\Gamma^{(2)},q]^{n-m}(T_{\mathscr{N}M})$, proving the proposition.
\end{proof}
\begin{proposition}
\label{PropFundHomEClassDisUn}
Let $M$ as in Proposition \ref{PropConstrFundHomEClass}, with the additional assumption that $M = M_1 \sqcup M_2$ decomposes into the disjoint union of two $\mathfrak{W}$-manifolds. Then
\be
\label{EqFundHomEClassDisUn}
[M,\partial M]_E = [M_1, \partial M_1]_E + [M_2, \partial M_2]_E \;.
\ee
\end{proposition}
\begin{proof}
The Mayer-Vietoris sequence associated to the generalized cohomology $F[\nu,\Gamma,q]$ defined in \eqref{DefSpecFHomFiber}, together with the fact that $M_1 \cap M_2$ is the empty manifold, yields \eqref{EqFundHomEClassDisUn}.
\end{proof}
We recall that the orientation flip is denoted by a minus sign.
\begin{proposition}
\label{PropFundHomEClassOrFl}
Let $M$ be as in Proposition \ref{PropConstrFundHomEClass}. Then
\be
\label{EqFundHomEClassOrFl}
[-M,\partial (-M)]_E = -[M,\partial M]_E \;.
\ee
\end{proposition}
\begin{proof}
The Thom class $u$ picks up a sign under orientation flip, which yields the sign in \eqref{EqFundHomEClassOrFl}.
\end{proof}

\paragraph{E-homology chains}

The construction above provides an integration map on the degree $m$ relative E-cohomology groups of a $m$-dimensional $\mathfrak{W}$-manifold with boundary. We would like to extend this integration map to arbitrary cochains: this is the analogue to picking a cycle  representative of the fundamental homology class in ordinary cohomology. Unfortunately, we do not have a model for E-homology in terms of ordinary chains. This comes from the fact that the cochain model for E-cohomology crucially uses the cup products, and that there is no analogue of the cup product for chains. One may still try to define E-homology chains (or "E-chains") as homomorphisms from the cochain group $\bar{C}^\bullet(M)$ into $\mathbb{R}/\mathbb{Z}$. It turns out that this definition, which appeared in a previous version of the present paper, is too naive. The non-abelian nature of the E-cochain group would require such E-chains to vanish on a large subset of E-cochains. 

The useful definition turns out to be the following. Let $M$ be a $m$-dimensional $\mathfrak{W}$-manifold, possible with boundary. A degree $m$ E-chain $\bar{m}$ is a function from $\bar{C}^m(M)$ into $\mathbb{R}/\mathbb{Z}$ subject to the following constraint. Let $\bar{s}_i = (s_i, y_i) \in \bar{C}^m(M)$, $i = 1,2$ and let us write the evaluation of $\bar{m}$ on $\bar{s} \in \bar{C}^m(M)$ as $\langle \bar{s}, \bar{m} \rangle$. We require that
\be
\label{EqLinRelEChain}
\langle \bar{s}_1 \boxplus \bar{s}_2, \bar{m} \rangle = \langle \bar{s}_1, \bar{m} \rangle + \langle \bar{s}_2, \bar{m} \rangle + c_{\bar{m}}(y_1, y_2)
\ee
for some $c_{\bar{m}}: C^m(M;\Gamma^{(2)}) \otimes C^m(M;\Gamma^{(2)}) \rightarrow \mathbb{R}/\mathbb{Z}$ satisfying
\be
c_{\bar{m}}(y_1, y_2) + c_{\bar{m}}(y_2, y_1) = \langle(dy_1 \cup_{m-2q+2} dy_2, 0), \bar{m}\rangle
\ee
and such that $c_{\bar{m}}(y_1, y_2) = 0$ whenever $y_1$ or $y_2$ is closed. The condition \eqref{EqLinRelEChain} makes $\bar{m}$ as close to a group homomorphism as possible. The necessity of the third term on the right-hand side can be derived by comparing 
$\langle \bar{s}_1 \boxplus \bar{s}_2, \bar{m} \rangle$ and $\langle \bar{s}_2 \boxplus \bar{s}_1, \bar{m} \rangle$. There is no obstruction to the existence of 
$c_{\bar{m}}$, although there is no canonical choice. For instance, one can enumerate
all unordered pairs of degree 4 $\Gamma^{(2)}$-cochains $(y_1, y_2)$, pick an arbitrary order on each and set $c_{\bar{m}}(y_1, y_2) = \langle(dy_1 \cup_{m-2q-2} dy_2, 0), \bar{m} \rangle$ and $c_{\bar{m}}(y_2, y_1) = 0$. 

%

E-chains are group homomorphisms when restricted to E-cocycles. A \emph{fundamental E-chain} on $M$ is an E-chain whose associated group homomorphism on relative E-cocyles coincides with the pullback of the integration map on E-cohomology constructed previously. We end this appendix by enumerating a few properties of fundamental E-chains. 

\begin{enumerate}
\item Let $\bar{x}$ be a degree $m$ E-cochain on $\partial M$, extending to an E-cochain $\bar{x}'$ on $M$, then $\bar{m}$ induces a closed fundamental E-chain $\partial \bar{m}$ on $\partial M$ through
\be
\langle \bar{x}, \partial \bar{m} \rangle = \langle d\bar{x}', \bar{m} \rangle \;.
\ee

\item A fundamental E-chain $\bar{m}$ is a group homomorphism when restricted to the subgroup of E-cochains of the form $(\hat{s},0)$. $\bar{m}$ therefore has an associated cycle representative $\hat{m}_{\mathbb{R}/\mathbb{Z}}$ of the fundamental $\mathbb{R}/\mathbb{Z}$-valued homology class of $M$. We have
\be
\langle (\hat{s},0), \bar{m} \rangle = \langle \hat{s}, \hat{m}_{\mathbb{R}/\mathbb{Z}} \rangle \;.
\ee

\item Let $f(\hat{y}) := \langle (0,\hat{y})$, $\hat{y} \in C^{m-q}(M;\bar{\Gamma}^{(2)})$. Given $\hat{y}_1, \hat{y}_2 \in C^{m-q}(M;\bar{\Gamma}^{(2)})$, we have
\be
f(\hat{y}_1 + \hat{y}_2) - f(\hat{y}_1) - f(\hat{y}_2) =  \langle d\hat{y}_1 \cup_{m-2q+1} \hat{y}_2 + \hat{y}_1 \cup_{m-2q} \hat{y}_2, \hat{m}_{\mathbb{R}/\mathbb{Z}} \rangle + c(\hat{y}_1, \hat{y}_2) \;,
\ee
so $f$ is a quadratic refinement of the pairing on $C^{m-q}(M;\bar{\Gamma}^{(2)})$ appearing on the right-hand side.\footnote{We thank Greg Moore for pointing this out to us.}
\end{enumerate}

\section{Bordism categories of $\mathfrak{W}$-manifolds}

\label{ApBordCatWMan}

In this appendix, we define the domain categories of the prequantum theory functor $\mathcal{S}^e$ and of the quantum theory functor $\mathcal{R}$.

\paragraph{Bordism category of $\mathfrak{W}$-manifolds} We define the bordism category $\mathcal{B}_\mathfrak{W}^{4\ell+3,1}$ of $4\ell+2$-dimensional $\mathfrak{W}$-manifolds as follows. 

The objects of $\mathcal{B}_\mathfrak{W}^{4\ell+3,1}$ are closed $4\ell+2$-dimensional $\mathfrak{W}$-manifolds $N$, endowed with germs of smooth structures on $N \times \{0\} \subset N \times [-1,1]$. In the rest of the paper, although we do not state it explicitly, all the $4\ell+2$-dimensional $\mathfrak{W}$-manifolds are endowed with such germs. 

The morphisms from an object $N^-$ to an object $N^+$ are triplets $(M, \theta^-, \theta^+)$ called $4\ell+3$-dimensional $\mathfrak{W}$\emph{-bordisms}. Here $M$ is a $4\ell+3$-dimensional $\mathfrak{W}$-manifold whose boundary has been decomposed into two disjoint components: $\partial M = \partial^- M \sqcup \partial^+ M$. The restriction from $M$ to $\partial^\pm M$ induce germs of smooth structures on $\partial^\pm M$. $\theta^\pm: \pm N^\pm \rightarrow \partial^\pm M$ are isomorphisms of $\mathfrak{W}$-manifolds preserving the germs of smooth structures on the domain and target. $-N_-$ denotes $N_-$ with the opposite orientation.

The composition in $\mathcal{B}_\mathfrak{W}^{4\ell+3,1}$ of a bordism $(M_1, \theta^-_1, \theta^+_1)$ from $N^-$ to $N$ with a bordism $(M_2, \theta^-_2, \theta^+_2)$ from $N$ to $N^+$ is the bordism $(M_1 \sqcup_{(\theta^-_2)^{-1} \circ \theta^+_1} M_2, \theta^-_1, \theta^+_2)$ obtained by quotienting the disjoint union $M_1 \sqcup M_2$ by the equivalence relation $\theta^+_1(p) \sim \theta^-_2(p)$ for $p \in   N$. The compatibility of $\theta^+_1$ and $\theta^-_2$ with the germs of smooth structures ensure that there is a unique smooth structure on $M_1 \sqcup_{(\theta^-_2)^{-1} \circ \theta^+_1} M_2$. Moreover, as $(\theta^-_2)^{-1} \circ \theta^+_1$ is a morphism of $\mathfrak{W}$-manifolds, the local systems of lattice and Wu structure on $M_1$ and $M_2$ induce a $\mathfrak{W}$-manifold structure on $M_1 \sqcup_{(\theta^-_2)^{-1} \circ \theta^+_1} M_2$.

The category $\mathcal{B}_\mathfrak{W}^{4\ell+3,1}$ is symmetric monoidal, with the monoidal structure given by the disjoint union of closed manifolds and bordisms. It also admits a involutive functor to the opposite category given by the orientation reversal of closed manifolds and bordisms. We will refer to this structure as a dagger structure, although strictly speaking it is slightly more general than the textbook version of a dagger structure, which requires the involutive functor to act as the identity on objects.

\paragraph{Bordism category of $(\mathfrak{W}, \mathsf{Z}^{\rm flat})$-manifolds} We describe here a related category $\mathcal{B}_{\mathfrak{W},{\rm flat}}^{4\ell+3,1}$. $\mathcal{B}_{\mathfrak{W},{\rm flat}}^{4\ell+3,1}$ is the domain of the prequantum field theory functor defined in Section \ref{SecSubPreqTh}. It is also the domain of the quantum field theory $\mathcal{R}$ when there is no anomaly.

The objects of $\mathcal{B}_{\mathfrak{W},{\rm flat}}^{4\ell+3,1}$ are closed $4\ell+2$-dimensional $\mathfrak{W}$-manifolds $N$ endowed with a germ of smooth structure as above, as well as a flat differential cocycle $\check{y} \in \mathsf{Z}^{\rm flat}_N$.

The morphisms from an object $(N^-, \check{y}^-)$ to an object $(N^+, \check{y}^+)$ are $4\ell+3$-dimensional $\mathfrak{W}$-bordisms $(M,\theta^-, \theta^+, \check{x})$ endowed with a differential cocycle $\check{x} \in \mathsf{Z}^{\rm flat}_{M, \check{y}}$, where $\check{y} = ((\theta^-)^{-1})^\ast(\check{y}^-) + ((\theta^+)^{-1})^\ast(\check{y}^+) \in \mathsf{Z}_{\partial M}$.

Suppose we have a bordism $(M_1, \theta^-_1, \theta^+_1, \check{x}_1)$, $\check{x}_1 \in \mathsf{Z}^{\rm flat}_{M_1, \check{y}_1}$ from $(N^-,\check{y}^-)$ to $(N,\check{y})$ and a bordism $(M_2, \theta^-_2, \theta^+_2, \check{x}_2)$, $\check{x}_2 \in \mathsf{Y}^{\rm flat}_{M_2, \check{y}_2}$ from $(N,\check{y})$ to $(N^+,\check{y}^+)$. The composition of the two bordisms is given by gluing like in the category $\mathcal{B}_\mathfrak{W}^{4\ell+3,1}$, but we have to describe the differential cohomology class carried by $M_{12} := M_1 \sqcup_{(\theta^-_2)^{-1} \circ \theta^+_1} M_2$. To this end, remark that by hypothesis, $(\theta^+_1)^\ast \check{y}_1 = (\theta^-_2)^\ast \check{y}_2 = \check{y}$. This means that $\check{x}_1$ and $\check{x}_2$ are the restriction to $M_1$ and $M_2$ of a differential cocycle $\check{x}_{12}$ on $M_{12}$. We define the composed bordism to be $(M_{12}, \theta_1^-, \theta_2^+, \check{x}_{12})$.

$\mathcal{B}_{\mathfrak{W},{\rm flat}}^{4\ell+3,1}$ is a symmetric monoidal category with respect to the disjoint union, and the orientation flip provides a dagger structure in the sense described above.

The category $\mathcal{B}_{\mathfrak{W},{\rm flat}}^{4\ell+3,1}$ is based on manifolds and bordisms endowed with a differential cocycle (a background gauge field), as opposed to a differential cohomology class (a gauge equivalence class of fields). It is not possible to define straightforwardly a bordism category of manifolds endowed with a differential cohomology class, because of ambiguities in the definition of the composition of bordisms. Suppose we try to repeat the definition above with differential cohomology classes. Then we will have $(\theta^+_1)^\ast y_1 = (\theta^-_2)^\ast y_2 = y$. After choosing cocycle representatives, we obtain $(\theta^+_1)^\ast \hat{y}_1 = (\theta^-_2)^\ast (\hat{y}_2 + d\hat{c})$ for some flat differential cochain $\check{c}$ on $\theta^-_2 N \subset \partial M_2$. We can extend $\check{c}$ to $M_2$ by zero and define $\check{x}'_2 = \check{x}_2 + d\check{c} \in \mathsf{Z}^{\rm flat}_{M_2}$. By construction, we have $\check{x}'_2|_N = \check{x}_1|_N$, where $N$ refers to the gluing locus as a submanifold of $M_{12}$. So there is a cocycle $\check{x}_{12}$ on $M_{12}$ that coincides with $\check{x}_1$ on $M_1$ and with $\check{x}_2$ on $M_2$. The trouble is that the differential cohomology class $x_{12}$ depends on the choice of $\check{c}$. More precisely, if we shift $\check{c}$ by a cocycle on $N$ that does not extend to $M_1$, the differential cohomology class of $x_{12}$ will change. This is the reason we endow manifolds and bordism with cocycles and not differential cohomology classes.

\paragraph{Bordism category of $(\mathfrak{W}, \mathsf{Z}^{\rm flat},\epsilon)$-manifolds} When the order of $\Gamma$ is even, the construction of the state space requires an extra structure $\epsilon$ on $4\ell+2$-dimensional manifolds, as described in Section \ref{SecExStruct}.

We define a bordism category $\mathcal{B}_{\mathfrak{W},{\rm flat},\epsilon}^{4\ell+3,1}$ as follows. Its objects are the objects of $\mathcal{B}_{\mathfrak{W},{\rm flat}}^{4\ell+3,1}$ endowed with an extra structure $\epsilon$. The morphisms between two objects of $\mathcal{B}_{\mathfrak{W},{\rm flat},\epsilon}^{4\ell+3,1}$ are the morphisms between these objects in $\mathcal{B}_{\mathfrak{W},{\rm flat}}^{4\ell+3,1}$ after we forget the extra structures.

\section{Heisenberg groups and modules}

\label{AppHeisGr2Gr}

We review here a type of Heisenberg groups that are central extensions of finite groups and their modules. (More general Heisenberg groups are treated in \cite{mumford2006tata} and in Appendix A of \cite{Freed:2006ya}.)

\paragraph{Alternating pairings on finite groups} Let $\mathsf{G}$ be a finite abelian group, whose group law will be written additively. We write $g, g_1, g_2,...$ for arbitrary elements of $\mathsf{G}$. We assume that $\mathsf{G}$ is endowed with a bilinear alternating perfect pairing $(\bullet,\bullet)_\mathsf{G}: \mathsf{G} \times \mathsf{G} \rightarrow \mathbb{R}/\mathbb{Z}$. We will also write
\be
B(\bullet, \bullet) := \exp 2 \pi i (\bullet,\bullet)_\mathsf{G}
\ee
for the corresponding $U(1)$-valued bimultiplicative pairing.

A \emph{Lagrangian subgroup} $L$ of $\mathsf{G}$ is a maximal subgroup of $\mathsf{G}$ such that $(l_1,l_2)_\mathsf{G} = 0$ for all $l_1,l_2 \in L$. Lagrangian subgroups $\mathsf{L},\mathsf{K}$ of $\mathsf{G}$ are called \emph{transverse} if $\mathsf{L} \cap \mathsf{K} = \{0\}$, the trivial subgroup. If $\mathsf{L}$ is transverse to $\mathsf{K}$, $\mathsf{G} \simeq \mathsf{L} \oplus \mathsf{K}$.

\paragraph{Heisenberg group} Recall that a \emph{2-cocycle} on $\mathsf{G}$ valued in $U(1)$ is a function $\phi : \mathsf{G} \times \mathsf{G} \rightarrow U(1)$ satisfying
\be
\phi(g_1, g_2) \phi(g_1 + g_2, g_3) = \phi(g_1, g_2 + g_3) \phi(g_2, g_3) \;.
\ee
A 2-cocycle $\phi$ is a \emph{refinement}  of $B$ if the commutator associated to $\phi$ coincides with $B$:
\be
\phi(g_1, g_2)(\phi(g_2, g_1))^{-1} = B(g_1, g_2) \;.
\ee
Given $B$, a refinement can for instance be constructed from any pair of transverse Lagrangian subgroups $\mathsf{L}$ and $\mathsf{K}$ by defining
\be
\phi(g_1, g_2) := B(k_1, l_2)  \;,
\ee
where $g_i = l_i + k_i$, $l_i \in \mathsf{L}$, $k_i \in \mathsf{K}$, $i = 1,2$.

A 2-cocycle $\phi$ defines a central extension of $\mathsf{G}$ by $U(1)$
\be
\label{EqExSeqCentExtGU1}
1 \rightarrow U(1) \rightarrow {\rm Heis}_\mathsf{G} \rightarrow \mathsf{G} \rightarrow 1 
\ee
as follows. As a set, ${\rm Heis}_\mathsf{G} = \mathsf{G} \times U(1)$, and the group law is given (multiplicatively) by 
\be
(g_1, t_1) \cdot (g_2, t_2) = \left(g_1 + g_2, t_1 t_2 \phi(g_1, g_2) \right) \;.
\ee
Proposition A.1 of \cite{Freed:2006ya} shows that up to isomorphism, ${\rm Heis}_\mathsf{G}$ depends only on the pairing $B$, and not on the choice of refinement $\phi$. We call ${\rm Heis}_\mathsf{G}$ the \emph{Heisenberg group} associated to $(\mathsf{G},B)$. ${\rm Heis}_\mathsf{G}$ satisfies the familiar Heisenberg commutation relations
\be
(g_1, t_1) \cdot (g_2, t_2) = (0, B(g_1, g_2)_G) \cdot (g_2, t_2) \cdot (g_1, t_1) \;.
\ee
where $0$ is the identity element of $\mathsf{G}$.

The exact sequence \eqref{EqExSeqCentExtGU1} shows that ${\rm Heis}_\mathsf{G}$ is composed of a $U(1)$-torsor over each $g \in \mathsf{G}$, with the trivial $U(1)$-torsor over $0 \in \mathsf{G}$. It can therefore be seen as a $U(1)$-bundle over $\mathsf{G}$, with a group law lifting the group law of $\mathsf{G}$.

\paragraph{Heisenberg modules} Consider the space $\mathbb{C}({\rm Heis}_\mathsf{G})$ consisting of complex valued functions on ${\rm Heis}_\mathsf{G}$ satisfying $f(g,t't) = t'f(g,t)$, $g \in \mathsf{G}$, $t,t' \in U(1)$. The left and right regular representations $\rho_L$ and $\rho_R$ of ${\rm Heis}_\mathsf{G}$ on $\mathbb{C}({\rm Heis}_\mathsf{G})$ are given by
\be
\rho_L(h)(f)(h') = f(hh') \;, \quad \rho_R(h)(f)(h') = f(h'h)
\ee
for $h,h' \in {\rm Heis}_\mathsf{G})$, $f \in \mathbb{C}({\rm Heis}_\mathsf{G})$. The left and right regular representations have dimension $|\mathsf{G}|$ and are reducible.

The analogue of the Stone-Von Neumann theory states that ${\rm Heis}_\mathsf{G}$ has a unique irreducible representation up to isomorphism (see for instance Lemma A.4 of \cite{Freed:2006ya}). It can be constructed as follows. Lagrangian subgroups of $\mathsf{G}$ lift (non-uniquely) to finite commutative subgroups of ${\rm Heis}_\mathsf{G}$. We call such subgroups of ${\rm Heis}_\mathsf{G}$ \emph{LL subgroups}. Pick a Lagrangian subgroup $\mathsf{L}$ of $\mathsf{G}$ and an associated LL subgroup $\tilde{\mathsf{L}}$. Let $\mathbb{C}({\rm Heis}_\mathsf{G})^{\tilde{\mathsf{L}}}$ be the subspace of functions invariant under the right action $\rho_R(\tilde{\mathsf{L}})$. $\rho_L$ defines a Heisenberg module structure on $\mathbb{C}({\rm Heis}_\mathsf{G})^{\tilde{\mathsf{L}}}$ and provides a realization of the unique irreducible representation of ${\rm Heis}_\mathsf{G}$, of dimension $|\mathsf{G}|^{1/2}$. We study in great detail the intertwiners between these representations for different choices of Lagrangian subgroups in Section \ref{SecStateSpace}.

\section{2-groups and their representations}

\label{App2Groups}

\paragraph{2-groups} Convenient references on 2-groups and their 2-modules are \cite{2003math......7200B, 2008arXiv0812.4969B}. A \emph{weak monoidal category} $\mathcal{M}$ is composed of
\begin{itemize}
\item A category $\mathcal{M}$.
\item A functor from $\mathcal{M} \times \mathcal{M}$ to $\mathcal{M}$, which we write with a dot $\cdot$.
\item An identity object $\bm{1} \in \mathcal{M}$.
\item Natural isomorphisms $a_{X,Y,Z}: X \cdot (Y \cdot Z) \rightarrow (X \cdot Y) \cdot Z$, $l_X: 1 \cdot X \rightarrow X$ and $r_X: X \cdot 1 \rightarrow X$, for $X,Y,Z \in \mathcal{M}$. The isomorphisms $a_{X,Y,Z}$ have to satisfy the pentagon identity and we should also have $1_X \cdot l_Y = (r_X \cdot 1_Y) \circ a_{X,1,Y}$, see \cite{2003math......7200B} for the corresponding commutative diagrams.
\end{itemize}
A \emph{weak inverse} of an object $X \in \mathcal{M}$ is an object $Y$ such that $X \cdot Y \simeq 1$ and $Y \cdot X \simeq 1$, where $\simeq$ denotes the existence of an isomorphism. $X$ is then called \emph{weakly invertible}

A \emph{weak 2-group} is a weak monoidal category in which all the morphisms are invertible and all the objects are weakly invertible. We have a \emph{strict 2-group} when the isomorphisms defining the weak inverses as well as the isomorphisms $a_{X,Y,Z}$, $l_X$, $r_X$ are all identities.

Another conceptual point of view on 2-groups is the following. We can see a group as a category with a single object such that all its morphisms are invertible. The elements of the group being in bijection with the morphisms and the group law is given by the composition law of morphisms. Similarly, a (weak) 2-group can be seen as a (weak) 2-category \cite{1998math.....10017L} with a single object, all of whose morphisms and 2-morphisms are invertible. The category $\mathcal{M}$ of the previous definition is the category of morphisms from the unique object to itself. The multiplication functor is given by the composition functor that is part of the definition of a 2-category.

\paragraph{Representations} When we picture a group $G$ as a category $\mathcal{G}$, a representation of $G$ is simply a functor from $\mathcal{G}$ to $\mathcal{V}_1$, the category of vector spaces. Indeed the unique object of $\mathcal{G}$ is sent to the vector space underlying the representation, and the morphisms of $\mathcal{G}$, which are in bijection with the elements of $G$, are sent to endomorphisms of $V$. The compatibility of this assignment with the composition of morphisms ensures that we obtain a representation.

This pattern can be copied to define representations of 2-groups. Given a 2-group $\mathcal{G}$, pictured as a weak 2-category with a single object, a \emph{representation} is a 2-functor from $\mathcal{G}$ to $\mathcal{V}_2$, the 2-category of 2-vector spaces (see for instance Appendix A.2 of \cite{Monnierd}).

Such a 2-functor assigns a 2-vector space $V$ to the unique object of $\mathcal{G}$. Assuming that we trivialized $V$ into a vector of vector spaces, the 2-functor assigns a matrix of vector spaces to each 1-morphism of $\mathcal{G}$ and a matrix of homomorphisms to each 2-morphisms of $\mathcal{G}$.

{
\small  \linespread{.6}

\providecommand{\href}[2]{#2}\begingroup\raggedright\endgroup

}

\end{document}